%% file: msc.tex
\documentclass[12pt,english,letterpaper]{report}
\usepackage[ansinew]{inputenc}
\usepackage[activeacute]{babel}

\usepackage{amsmath}
\usepackage{amssymb}
\usepackage{multirow}
\usepackage[left=1.5in, right=1in, top=1.6in, bottom=1in, includefoot,headheight=16.0pt]{geometry}
\usepackage{url}
\usepackage[color]{thesisuchile}
\usepackage{float}
\input{thesissetup}
\usepackage[plainpages=false,
pdfborderstyle={/S/U/W 0},
pdftitle={\title}, 
pdfauthor={\author}, 
pdfsubject={Succinct Data Structures}, pdfkeywords={compression,algorithms},
pdfcreator={},pdfstartview={FitH}]{hyperref}
\usepackage[all]{hypcap}

\usepackage[noend]{algorithmic}
\usepackage{color}

\usepackage[title, header, titletoc]{appendix}

\input{misc}

\begin{document}
\pagenumbering{Roman}
\maketitle
\vspace{-15cm}
\centertitle{Resumen}
\vspace{-0.37cm}
\input{resumen}
\newpage
\pagenumbering{roman}
\maketitleenglish
\centertitle{Abstract}
\input{abstract}
\thispagestyle{empty}
\tableofcontents
\newpage
\listoffigures
\listoftables
\clearpage
\pagenumbering{arabic}
\setlength{\parskip}{1.5ex}
\input{introduction}
\input{basic_concepts}
\input{repetitive_corpus}
\input{lz_parsing}

\input{lz77_index}

\input{experimental_evaluation}
\input{conclusions}

\newpage
\bibliographystyle{alpha}
\bibliography{msc}
\begin{appendices}
\renewcommand{\chaptermark}[1]{%
\markboth{Appendix \thechapter\ #1}{}}
\input{results_appendix}

\end{appendices}

\end{document}

%% file: thesissetup.tex
\universidad{Universidad de Chile}
\university{University of Chile}
\facultad{Facultad de Ciencias Físicas y Matemáticas}
\faculty{Faculty of Physics and Mathematics}
\departamento{Departamento de Ciencias de la Computación}
\department{Department of Computer Science}
\titulo{Auto-Índice de Texto basado en LZ77}
\title{Self-Index based on LZ77}
\trabajoygrado{Tesis para optar al Grado de Magíster en Ciencias, Mención Computación}
\degree{MSc.~in Computer Science}
\autor{Sebastián Andrés Kreft Carreño}
\author{Sebastian Kreft}
\profguia{Gonzalo Navarro Badino}
\profinta{Diego Arroyuelo Billiardi}
\profintb{Jérémy Barbay}
\profintc{Nieves Brisaboa}
\profguiaen{Gonzalo Navarro}
\profintaen{Diego Arroyuelo}
\profintben{Jérémy Barbay}
\profintcen{Nieves Brisaboa}
\ciudad{Santiago}
\pais{Chile}
\mespub{Agosto}
\monthpub{August}
\yearpub{2010}
\financiamiento{Este trabajo ha sido financiado en parte por la Beca Conicyt de Magíster Nacional y por el Instituto Milenio de Dinámica Celular y Biotecnología.}
\funding{This work was partially funded by Conicyt's Master Scholarship and by the Millennium Institute for Cell Dynamics and Biotechnology (ICDB).}

%% file: misc.tex
\newcommand{\source}{\textit{source}}
\newcommand{\chars}{\textit{char}}
\newcommand{\BWS}{\textrm{BWS}}
\newcommand{\etal}{\emph{et al.}}
\newcommand{\delete}[1]{}
\newcommand{\phrase}[1]{
    \ensuremath{
        \boxed{\texttt{#1\vphantom{\$}}}
    }
}
\newcommand{\rwidth}{0.35}
\newcommand{\results}[3]{
\begin{figure}[!ht]
\centering
\includegraphics[angle=-90,scale=\rwidth]{img/experiments/#1_extract}
\includegraphics[angle=-90,scale=\rwidth]{img/experiments/#1_locatep}
\includegraphics[angle=-90,scale=\rwidth]{img/experiments/#1_existf}
\includegraphics[angle=-90,scale=\rwidth]{img/experiments/#1_existn}
\caption[#2 (1)]{#2 (1). Note the logscales.}
\label{#3:1}
\end{figure}

\begin{figure}[!ht]
\centering
\includegraphics[angle=-90,scale=\rwidth]{img/experiments/#1_extract_tradeoff}
\includegraphics[angle=-90,scale=\rwidth]{img/experiments/#1_locate10}
\includegraphics[angle=-90,scale=\rwidth]{img/experiments/#1_locate15}
\includegraphics[angle=-90,scale=\rwidth]{img/experiments/#1_locate20}
\includegraphics[angle=-90,scale=\rwidth]{img/experiments/#1_locate02}
\includegraphics[angle=-90,scale=\rwidth]{img/experiments/#1_locate04}
\includegraphics[angle=-90,scale=\rwidth]{img/experiments/#1_existf_tradeoff}
\includegraphics[angle=-90,scale=\rwidth]{img/experiments/#1_existn_tradeoff}
\caption[#2 (2)]{#2 (2). Note the logscales.}
\label{#3:2}
\end{figure}

}

\newcommand{\figuresres}{\ref{res:t29:1}-\ref{res:kernel:2} and \ref{res:f41:1}-\ref{res:leaders:2}~}
\newcommand{\corpusurl}{\url{http://pizzachili.dcc.uchile.cl/repcorpus.html}}
\newcommand{\indexurl}{\url{http://pizzachili.dcc.uchile.cl/indexes/LZ77-index}}
\newcommand{\libcdsurl}{\url{http://code.google.com/p/libcds}}
\usepackage{rotating}
\newcommand{\trotate}[1]{\begin{sideways}#1\end{sideways}}

\usepackage{array}
\newcolumntype{x}[1]{%
>{\raggedleft\hspace{0pt}}p{#1}}%
\newcommand{\tn}{\tabularnewline}

\usepackage{ulem}
\newcommand{\jeremy}[1]{#1}
\newcommand{\djeremy}[1]{}
\newcommand{\nieves}[1]{#1}
\newcommand{\dnieves}[1]{}
\newcommand{\diego}[1]{#1}
\newcommand{\ddiego}[1]{}
\newcommand{\cdelete}[1]{}
\newcommand{\dmine}[1]{}

%% file: resumen.tex
Los dominios como bioinformática, sistemas de versionamiento de código, sistemas de edición colaborativos (wikis), y otros, producen grandes colecciones de texto que son sumamente repetitivas. Esto es, existen pocas diferencias entre los elementos de la colección. Esto permite que la compresibilidad de la colección sea extremadamente alta. Por ejemplo, una colección con versiones de un mismo artículo de Wikipedia puede ser comprimida a un $0.1\%$ de su espacio original, utilizando el esquema de compresión Lempel-Ziv de 1977 (LZ77).

Muchas de estas colecciones repetitivas contienen grandes volúmenes de texto. Es por eso que se requiere un método que permita almacenarlas eficientemente y a la vez operar sobre ellas. Las operaciones más comunes son extraer porciones aleatorias de la colección y encontrar todas las ocurrencias de un patrón dentro de la colección.

Un auto-índice es una estructura que almacena un texto en forma comprimida y permite encontrar eficientemente las ocurrencias de un patrón. Adicionalmente los auto-índices permiten extraer cualquier porción de la colección. 
Uno de los objetivos de estos índices es que puedan ser almacenados en memoria principal.
Esta característica es sumamente importante ya que el disco puede llegar a ser un millón de veces más lento que la memoria principal. 

La mayoría de los auto-índices existentes están basados en un esquema de compresión que predice los símbolos siguientes en base a una cantidad fija de símbolos anteriores. Este esquema, sin embargo, no funciona con textos repetitivos, ya que no es capaz de reconocer todos los elementos repetidos en la colección. Un esquema que sí captura las repeticiones es el LZ77, pero tiene el problema de no poder acceder aleatoriamente el texto.

En este trabajo se presenta un algoritmo para extraer substrings de un texto comprimido con un esquema Lempel-Ziv. Adicionalmente se presenta LZ-End, una variante de LZ77 que permite extraer el texto eficientemente usando espacio cercano al de LZ77. LZ77 extrae del orden de 1 millón de caracteres por segundo, mientras que LZ-End extrae más del doble.

Nuestro resultado más importante es el desarrollo del primer auto-índice orientado a textos repetitivos basado en LZ77/LZ-End. Su desempeño supera al auto-índice RLCSA, el estado del arte para textos repetitivos. 
La compresión de nuestros índices llega a ser dos veces mejor en ADN y colecciones de Wikipedia que la del RLCSA.
Cabe destacar que nuestro índice basado en LZ77 se construye en 35\% del tiempo requerido por el RLCSA, usando el 60\% de espacio de construcción.
La búsqueda de patrones cortos es más rápida que en el RLCSA, y para patrones largos la relación entre espacio y tiempo es favorable a nuestros índices.

Finalmente, se presenta también una colección de textos repetitivos provenientes de diversos dominios. Esta colección está disponible públicamente con el objetivo que se pueda convertir en un referente en experimentación.

%% file: abstract.tex
Domains like bioinformatics, version control systems, collaborative editing systems (wiki), and others, are producing huge data collections that are very repetitive. That is, \jeremy{there} are few differences between the elements of the collection. This fact makes the compressibility of the collection extremely high. For example, a collection with all different versions of a Wikipedia article can be compressed up to the $0.1\%$ of its original space, using the Lempel-Ziv 1977 (LZ77) compression scheme.

Many of these repetitive collections handle huge amounts of text data. For that reason, we require a method to store them efficiently, while providing the ability to operate \jeremy{on} them. The most common operations are the extraction of random portions of the collection and the search for all the occurrences of a given pattern inside the whole collection.

A self-index is a data structure that stores a text in compressed form and allows to find the occurrences of a pattern efficiently. On the other hand, self-indexes can extract any substring of the collection, hence they are able to replace the original text. 
One of the main goals when using these indexes is to store them within main memory.  
This characteristic is very important, as the disk may be 1 million times slower than main memory. 

Most current self-indexes are based on a compression scheme that predicts the following symbol based on the previous $k$ symbols. However, this scheme is not well suited for repetitive texts as \jeremy{it does not capture} long-range repetitions\djeremy{ are not captured}. The LZ77 compression scheme does capture such repetitions, but it is not able to access the text at random.

In this thesis we present a scheme for random text extraction from text compressed with a Lempel-Ziv parsing. Additionally, we present a variant of LZ77, called LZ-End, that efficiently extracts text using space close to that of LZ77. LZ77 extracts around 1 million characters per second, while LZ-End extracts over 2 million. 

The main contribution of this thesis is the first self-index based on LZ77/LZ-End and oriented to repetitive texts, which outperforms the state of the art (the RLCSA self-index) in many aspects. 
The compression of our indexes is better than that of RLCSA, being two times better for DNA and for Wikipedia articles. Our index is built using just 60\% of the space required by the RLCSA and within 35\% of the time.
Searching for short patterns is faster than on the RLCSA, and for longer patterns the space/time trade-off is in favor of our indexes.

Finally, we present a corpus of repetitive texts, coming from several application domains. We aim at providing a standard set of texts for research and experimentation, hence this corpus is publicly available.

%% file: introduction.tex
\chapter{Introduction}
In recent times we have seen a rise in the amount of digital information. This may be attributable to the drop of the data acquisition and storage costs. Most of this information is text, that is, symbol sequences representing natural language, music, source code, time series, biological sequences like DNA and proteins, and others.

Despite that the examples presented above seem very different, there is an operation that arises in most applications handling those types of sequences. This operation is called \emph{text search} and consists in finding all positions on the text where a given pattern appears. This operation serves as a basis for building more complex and meaningful operations, like finding the most common words, or finding approximate patterns.

Text search can be solved by two different approaches. The first scans the text sequentially looking for matches of the pattern. Classical examples of this type of search are Knuth-Morris-Pratt \cite{KMP77} and Boyer-Moore \cite{BM77} algorithms. The second way of searching is by querying an \emph{index} of the text, a data structure we have to build before performing the queries. This structure allows us to find the occurrences of a given pattern without scanning the whole text.

To index the text we need enough space in order to store the index, and most importantly we need to be able to access it efficiently. Nowadays, storage is not a difficult problem, however efficient access is. In the last years the speed of hard\diego{-}drives has not experienced significant\djeremy{ speed} improvements. Hard\diego{-}drive access times \diego{are} \ddiego{is} around $10ms=10^7ns$, while main memory access (RAM) is around $10ns$; in other words, accessing secondary storage is 1 million times slower than accessing main memory. This problem is still present despite the appearance of solid state drives (SSD), which have access times around $0.1ms=10^5ns$, being 10 thousand times slower than main memory. For this reason, indexes using space proportional to the compressed text have been proposed, aiming at storing them in main memory and handling the data directly in compressed form, rather than decompressing before using it \cite{ZMNBY00,NM07}. There are some indexes that, within that compressed space, are able to replace the original text; these are called self-indexes and are obviously preferable as one can discard the original text.

A particular kind of texts not yet fully benefited by current self-indexes are repetitive \ddiego{texts} \diego{ones}. These \ddiego{texts} arise from domains that handle huge collections of very similar entries or documents. For example, in a DNA collection of human genomes of different individuals, the similarity between any two DNA sequences would be close to 99.9\% \cite{dna_sim}. Source code collections are also very repetitive, as the changes between one version and the next are not substantial, except in the case of a major release. Versioning systems, like \emph{wikis}, also generate very repetitive collections because each revision is very similar to the previous one. The main problem is that existing self-indexes do not sufficiently exploit these repetitions, being the self-index orders of magnitude larger than the space achievable with a compression scheme that does exploit the repetitions, like LZ77 \cite{ZL77}. \jeremy{LZ77 parses the text into \emph{phrases} so that each phrase, except its last letter, appears previously in the text (these previous occurrences are called \emph{sources})). It compresses by essentially replacing each phrase by a backward pointer.} A recent work, aiming at adapting current self-indexes to handle large DNA databases of the same species \cite{MNSV08} found that LZ77 compression was still much superior to capture this repetitiveness, yet it was inadequate as a format for compressed storage because of its inability to retrieve individual sequences from the collection. Another work \cite{CN09,CFMPNbibe10} shows that grammar-based compression can allow extraction of substrings while capturing such repetitions, yet LZ77 compression is superior to grammar compression \cite{Rytter03,CLLP+05}. 

For these reasons in this thesis we focus on the \ddiego{creation} \diego{definition} of a self-index oriented to repetitive texts and based on LZ77-like compression schemes. Our main contributions are two: (1) a scheme for random text extraction in LZ77-like parsing, as well as a space-competitive variant, called LZ-End, achieving faster text extraction; (2) a self-index based on LZ77/LZ-End that achieves a better space/time trade-off than the best self-indexes oriented to repetitive texts.
\section{Contributions of the Thesis}
\begin{description}
\item[Chapter 3]: We create a public corpus of highly repetitive texts. The corpus is composed of texts coming from different real domains like biology, source code repositories, document repositories, and others, as well as artificial texts having interesting combinatorial properties. This corpus is available at \corpusurl.
\item[Chapter 4]: 
The worst-case extraction time of a substring of length $m$ in an LZ77 parsing is $O(mH)$, where $H$ is the maximum number of times a character is transitively copied in the parsing. We present an alternative parsing, called LZ-End, that performs very close to LZ77 in terms of compression but permits faster text extraction, $O(m+H)$ worst-case time. This work was published in the \emph{20th Data Compression Conference} \cite{KN10}.
\item[Chapter 5]: We introduce a new self-index oriented to repetitive texts and based on the \ddiego{Lempel-Ziv parsing} \diego{LZ77, LZ-End, and similar parsings}. Let $n'$ be the number of phrases of the parsing \diego{(for highly repetitive texts, $n'$ will be a small value)}. This  index uses in theory $2n' \log n + n' \log n' + n' \log D + O(n'\log \sigma) + o(n)$ bits of space, where $\sigma$ is the size of the alphabet and $D$ is upper-bounded by the maximum number of sources covering each other. It finds the $occ$ occurrences of a pattern of length $m$ in time $O(m^2H+m\log n' + occ \cdot D \log n')$. We present several practical variants that achieve better results, both in time and space, than the \emph{Run-length Compressed Suffix Array} (RLCSA) \cite{MNSV08} and the \emph{Grammar-based Self-index} \cite{CN09,CFMPNbibe10}, the state-of-the art self-indexes oriented to repetitive texts.
\end{description}
\section{Outline of the Thesis}
\begin{description}
\item[Chapter 2] describes basic concepts and related work relevant to this thesis.
\item[Chapter 3] presents a text corpus intended for repetitive text.
\item[Chapter 4] explains the Lempel-Ziv (LZ) parsing and some of its properties. It also introduces a new LZ variant called LZ-End, able to extract an arbitrary substring in constant time per extracted symbol in some cases.
\item[Chapter 5] presents a new self-index based on LZ77-like parsings. It covers the theoretical proposal and the considerations we made when implementing the index.
\item[Chapter 6] shows the experimental results of our proposed index, comparing it with the state-of-the-art self-indexes for repetitive texts. 
\item[Chapter 7] presents our conclusions and gives some lines of research that can be further investigated. 
\end{description}

%% file: basic_concepts.tex
\chapter{Basic Concepts}
In this chapter we introduce the basic concepts and notation used through this thesis. Then, we present the data structures used to build our index. Finally, we present two self-indexes oriented to repetitive texts. All logarithms in this thesis will be in base 2 and we will assume \jeremy{that} $0\log 0 = 0$.
\section{Strings}
\label{sec:strings}
\begin{definition} A \emph{string} $T$ is a sequence of characters drawn from an alphabet $\Sigma$. 
The alphabet is an ordered and finite set of size $\left | \Sigma \right |=\sigma$. 
The $i$-th character of a string is represented as $T[i]$. The symbol $\varepsilon$ represents the empty string of length $0$.
\end{definition}

\begin{definition}
Given a string $T$, and positions $i$ and $j$, the \emph{substring} of $T$ starting at $i$ and ending at $j$ is defined as $T[i,j]=T[i]T[i+1]\ldots T[j]$. If $i>j$, then $T[i,j]=\varepsilon$.
\end{definition}

\begin{definition}
Let $T$ \diego{be} a string of length $n$. The \emph{prefixes} of $T$ are the strings $T[1,j], \forall\, 0\leq j \leq n$ and \djeremy{the}\jeremy{its} \emph{suffixes} \jeremy{are} the strings $T[i,n], \forall\, 1\leq i \leq n+1$.
\end{definition}

\begin{definition}
Let $T_1$, $T_2$ be strings of length $n_1$ and $n_2$\diego{, respectively}. We define the concatenation of these strings as  $T_1T_2=T_1[1]\ldots T_1[n_1]T_2[1]\ldots T_2[n_2]$.
\end{definition}

\begin{definition}
Given a string $T$ of length $n$, the reverse of $T$ is $T^{rev}=T[n]T[n-1]\ldots T[2]T[1]$.
\end{definition}
\begin{definition}
The \emph{lexicographic order} ($<$) between strings is defined as follows:
Let $a,b$ \diego{be} characters \diego{in $\Sigma$} and $X,Y$ \diego{be} strings \diego{over $\Sigma$}.
\begin{eqnarray*}
&\varepsilon < X, \, \forall X \neq \varepsilon& \\
&aX < bY\,\, \textit{if}\,\, a<b\, \vee \, ( a=b \wedge X<Y)&
\end{eqnarray*}
\end{definition}
\section{Search Problems}
\label{sec:queries}
\begin{definition}
Given a string $T$ and a pattern $P$ (a string of length $m$) both over an alphabet $\Sigma$, the \emph{occurrence positions} of $P$ in $T$ are defined as
$O = \lbrace 1 + |X|,\, \exists X,Y,\, T = XPY \rbrace$. 
\end{definition}

\begin{definition}
Given a string $T$ and a pattern $P$, the following search problems are of interest:
\begin{itemize}
\item $exists(P,T)$ returns true iff $P$ is in $T$, i.e., returns true iff $O\neq \emptyset$.
\item $count(P,T)$ counts the number of occurrences of $P$ in $T$, i.e., returns $occ = |O|$.
\item $locate(P,T)$  finds the occurrences \diego{of $P$ in $T$}, i.e., returns the set $O$ in some order.
\item $extract(T,l,r)$  extracts the text substring $T[l,r]$.
\end{itemize}
\end{definition}

\begin{remark}
Note that $exists$ and $count$ can be answered after performing a $locate$ query.
\end{remark}
\section{Entropy}
\label{sec:entropy}
\begin{definition}
Let $T$ be a string of length $n$. The zero-th order empirical entropy is defined as
\begin{equation*}
H_0(T)= - \sum_{c\in \Sigma} \frac{n_c}{n} \log \frac{n_c}{n} 
\end{equation*}
where $n_c$ is the number of times the character $c$ appears in $T$, that is, $n_c/n$ is the empirical probability of appearance of character $c$. 
\end{definition}
It is worth noticing that the zero-th order entropy is invariant to permutations in the order of the text characters. The value $nH_0(T)$ is the least number of bits needed to represent $T$ using a compressor that gives each character a fixed encoding.

\begin{definition}
Let $T$ be a string of length $n$. The $k$-th order empirical entropy \cite{Man2001} is defined as
\begin{equation*}
H_k(T)= \sum_{S\in \Sigma^k} \frac{\left|T^S\right|}{n} H_0\left(T^S\right) 
\end{equation*}
where $T^S$ is the sequence composed of all characters preceded by string $S$ in $T$.
\end{definition}

The value $nH_k(T)$ is the least number of bits needed to represent $T$ using a compressor that encodes each character taking into account the $k$ preceding characters \diego{in $T$}. This value assumes the first $k$ characters are encoded for free, thus it gives a relevant lower bound only when $n\gg k$.

\djeremy{It can be proved that }$H_k$ is a decreasing function in $k$, that is, 
\begin{equation*}
0 \leq H_k(T) \leq H_{k-1}(T) \leq \ldots \leq H_1(T) \leq H_0(T) \leq \log \sigma.
\end{equation*}

The following lemma \jeremy{yields the ground to} show\djeremy{s} that the empirical entropy $H_k$ is not a good lower-bound measure for the compressibility of repetitive texts.
\begin{lemma}
\label{lemma:entropytt}
Let $T$ be a string of length $n$. For any $k\le n$ it holds $H_k(TT)\ge H_k(T)$. 
\end{lemma}
\begin{proof}
As new relevant contexts may have arisen in the concatenation $TT$, we denote by $C(T,k)$ the contexts of length $k$ present in $T$, and $C(TT,k)$ the contexts of $TT$. We have that $C(T,k) \subseteq C(TT,k)$. The number of new contexts in $TT$ is at most $k$.
For each $S \in C(T,k)$, we have $(TT)^{S} = T^{S}A^ST^{S}$, for some $A^S$ such that $|A^S|\leq k$. Then, 
\begin{eqnarray*}
H_k(TT)&=&\frac{1}{|TT|}\sum_{S\in C(TT,k)}|(TT)^S|H_0((TT)^S)\\
&\ge &\frac{1}{2|T|}\sum_{S\in C(T,k)}|T^SA^ST^S|H_0(T^SA^ST^S)\\
&\ge &\frac{1}{2|T|}\sum_{S\in C(T,k)}|T^ST^S|H_0(T^ST^S)\\
&=&\frac{1}{|T|}\sum_{S\in C(T,k)}|T^S|H_0(T^S)\\
&=&H_k(T).
\end{eqnarray*}
In the first step we used $C(T,k)\subseteq C(TT,k)$; in the second we used $|T|H_0(T)\le |TA|H_0(TA)$, for $T=T^ST^S$ and $A=A^S$ (since $|T^SA^ST^S| H_0(T^SA^ST^S) = $ \linebreak $|T^ST^SA^S| H_0(T^ST^SA^S)$); and in the third we used $H_0(TT)=H_0(T)$. The second property holds because
\begin{eqnarray*}
|TA|H_0(TA) & = & \sum_{c\in \Sigma} (n_c^T+n_c^A) \log \frac{n^T+n^A}{n_c^T+n_c^A} \\
& \ge & n^T \sum_{c\in \Sigma} \frac{n_c^T}{n^T} \log \frac{n^T+n^A}{n_c^T+n_c^A} \\
& \ge & n^T \sum_{c\in \Sigma} \frac{n_c^T}{n^T} \log \frac{n^T}{n_c^T}\\
& = & |T|H_0(T)
\end{eqnarray*}
where $n_c^X$ is the number of occurrences of character $c$ in string $X$, and $n^X=|X|$ for $X=T\, \text{or}\, A$. The last line is justified by the Gibbs inequality \cite{gibbs}.
\end{proof}
It follows that $|TT|H_k(TT)\ge2|T|H_k(T)$, that is, to encode $TT$ this model uses at least twice the space of the one used to encode $T$. An LZ77 encoding would need just one more phrase, as seen later.
\section{Encodings}
Most data structures need to represent symbols and numbers. Classic data structures use a fixed amount of space to store them, for example 1 byte for characters and 4 bytes for integers. Instead, compressed data structures aim to use the minimum possible space, thus they represent symbols using variable-length prefix-free codes or just using a fixed amount $b$ of bits, where $b$ is as small as possible. Table \ref{tab:coders} shows different encodings for the \ddiego{first 9} integers \diego{1,\ldots,9}, which we describe next.
\begin{description}
\item{\textbf{Unary Codes}} This representation is the simplest and serves as a basis for other coders. It represents a positive $n$ as $\texttt{1}^{n-1}\texttt{0}$, thus it uses exactly $n$ bits.
\item{\textbf{Gamma Codes}} It represents a positive $n$ by concatenating the length of its binary representation in unary and the binary representation of the symbol, omitting the most significant bit. The space is $2 \lfloor\log n\rfloor + 1$, $\lfloor\log n\rfloor + 1$ for the length and $\lfloor\log n\rfloor$ for the binary representation.   
\item{\textbf{Delta Codes}} This is an extension of $\gamma$-codes that works better on larger numbers. It represents the length of the binary representation of $n$ using $\gamma$-codes and then $n$ in binary without its most significant bit, thus using $\lbrace 2 \lfloor\log(\lfloor\log n\rfloor + 1)\rfloor + +1 \rbrace + \lfloor\log n\rfloor$ bits.
\item{\textbf{Vbyte Coding} \cite{vbyte}} It splits the $\lfloor \log(n + 1)\rfloor$ bits needed to represent $n$ into blocks of $b$ bits and stores each block in a chunk of $b + 1$ bits. The highest bit is 0 in the chunk holding the most significant bits of $n$, and 1 in the
rest of the chunks. For clarity we write the chunks from most to least significant, just like the binary representation of $n$. For example, if $n = 25 = 11001$ and $b = 3$, then we need two chunks and the representation is $0011 \cdot 1001$. Compared to an optimal encoding of $\lfloor \log(n + 1)\rfloor$ bits, this code loses one bit per $b$ bits of $n$, plus possibly an almost empty final chunk. Even when the best choice for $b$ is used, the total space achieved is still worse than $\delta$-encoding's performance. In exchange, Vbyte codes are very fast to decode.

\end{description}
\begin{table}
\centering
\begin{tabular}{|l|l|l|l|l|l|}
\hline
Symbol & Unary Code & $\gamma$-Code & $\delta$-Code & Binary($b=4$)&Vbyte($b=2$)\\ \hline
\hline
1 & 0 & 0 & 0 & 0001 & 001\\ \hline
2 & 10 & 100 & 1000 & 0010 & 010\\ \hline
3 & 110 & 101 & 1001 & 0011 & 011\\ \hline
4 & 1110 & 11000 & 10100 & 0100 & 001100 \\ \hline
5 & 11110 & 11001 & 10101 & 0101 & 001101 \\ \hline
6 & 111110 & 11010 & 10110 & 0110 & 001110 \\ \hline
7 & 1111110 & 11011 & 10111 & 0111 & 001111\\ \hline
8 & 11111110 & 1110000 & 11000000 & 1000 & 001100100 \\ \hline
9 & 111111110 & 1110001 & 110000001 & 1001 & 001100101 \\ \hline
\end{tabular}
\caption{Example of different coders}
\label{tab:coders}
\end{table}
\subsection{Directly Addressable Codes}
\label{sec:dac}
\diego{In many cases} \ddiego{Many times} we need to store a set of numbers using the least possible space, yet providing fast random access to each element. Variable\diego{-}length codes complicate this task, as \ddiego{we need to store additional} \diego{they require storing, in addition,} pointers to sampled positions of the encoded sequence.

A simple solution that shows good performance in practice is the so\diego{-}called \nieves{\emph{Directly Addressable Codes} (DAC)} \ddiego{\emph{Reordered Vbytes}} \cite{DAC}, a variant of \emph{Vbytes} \cite{vbyte}.
They start with a sequence $C=C_1, \ldots, C_n$ of $n$ integers. Then they compute the Vbyte encoding of each number. The least significant blocks are stored contiguously in an array $A_1$, and the highest bits of the least significant chunks are stored in a bitmap $B_1$. The remaining chunks are organized in the same way in arrays $A_i$ and bitmaps $B_i$, storing contiguously the $i$-th chunks of the numbers that have them. Note that arrays $A_i$ store contiguously the bits $(i-1)\cdot b+1,\ldots,i\cdot b$ and bitmaps $B_i$ store whether a number has further chunks or not, hence the name \emph{Reordered Vbytes}.  

Figure \ref{fig:dac} shows an example of the \diego{resulting} structure. The first element is represented with two blocks, thus, $A_1[0] = C_{1,1}$, $A_2[0] = C_{1,2}$, $B_1[0] = \texttt{1}$ and $B_2[0] = \texttt{0}$.

\begin{figure}[ht]
\centering
\includegraphics[width=8cm]{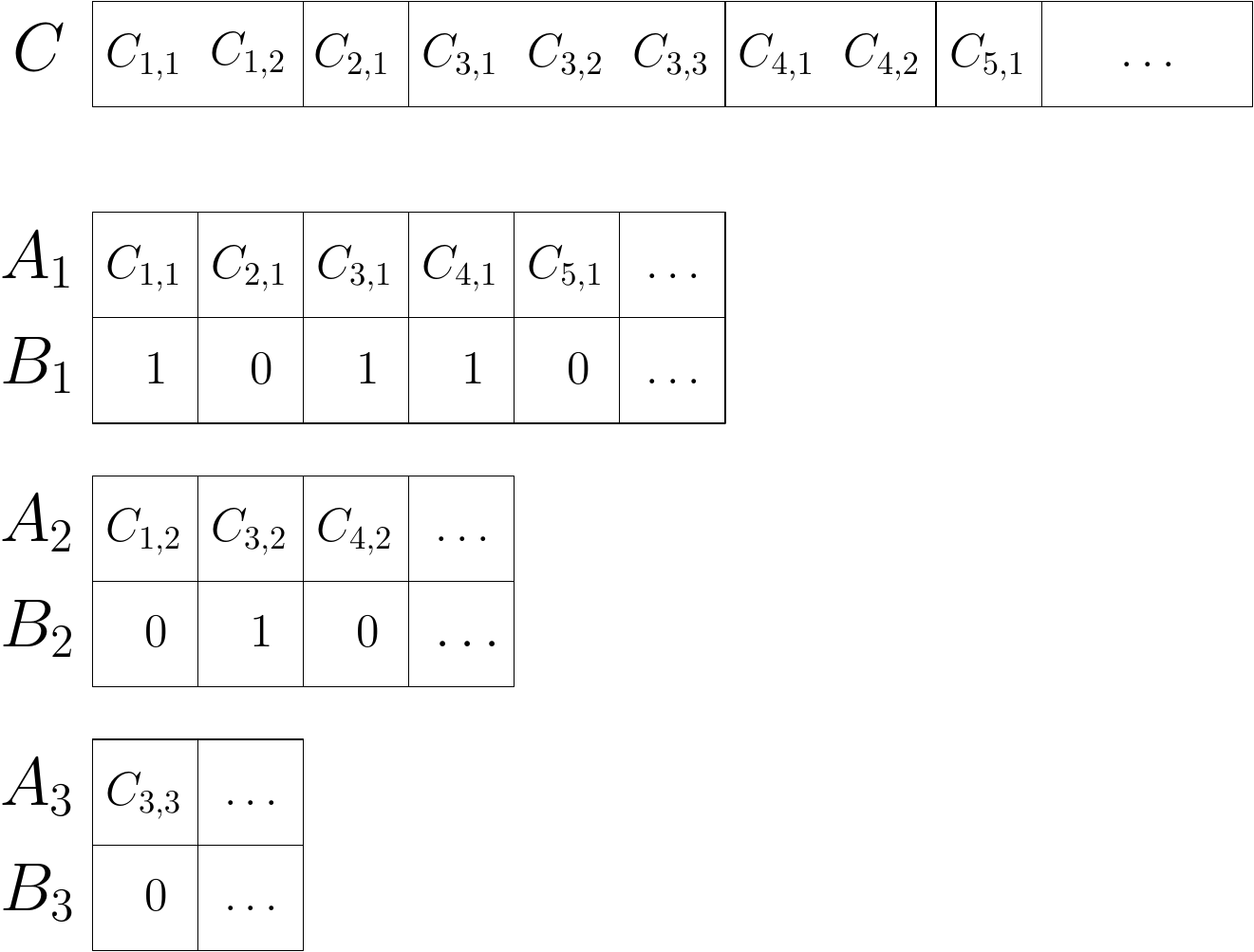}
\caption{Example of Directly Addressable Codes structure}
\label{fig:dac}
\end{figure}


To access the element at position $i=i_1$ we \ddiego{see if} \diego{check whether} $B_1[i_1]$ is set. If it is not set, this is the last chunk and we already have the value $C[i]=A_1[i_1]$, otherwise we have to fetch the following chunks. In that case, we recompute the position as $i_2 = rank_{1}(B_1, i_1)$, where $rank_{1}(B_1,i_1)$ is the number of ones up to position $i_1$ on bitmap $B_1$ (see Section \ref{sec:bitmaps} for further details). If $B_2[i_2]$ is not set we are done with $C[i] = A_1[i_1] + A_2[i_2]\cdot2^b$, otherwise we set $i_3 = rank_1(B_2,i_2)$ and continue in the following levels. Accessing a random element takes $O(\log(M)/b)$ worst case time, where $M=\max C_i$. However, the access time is lower for elements with shorter codewords, which are usually the most frequent ones.

We will use the implementation of Susana Ladra\jeremy{\footnote{Universidade da Coruña, Spain. \texttt{sladra@udc.es}}} (available by personal request) in this thesis.
\section{Bitmaps}
\label{sec:bitmaps}
Let $B$ a binary sequence over $\Sigma=\lbrace\texttt{0},\texttt{1}\rbrace$ (\diego{a} bitmap) of length $n$ and assume it has $m$ ones. We are interested in solving the following operations:
\begin{itemize}
    \item \textbf{$rank_b(B,i)$:} How many $b$'s are up to position $i$ \jeremy{(included)}.
    \item \textbf{$select_b(B,i)$:} The position of the $i$-th $b$ bit.
\end{itemize}

\begin{figure}
\centering
\includegraphics[width=12cm]{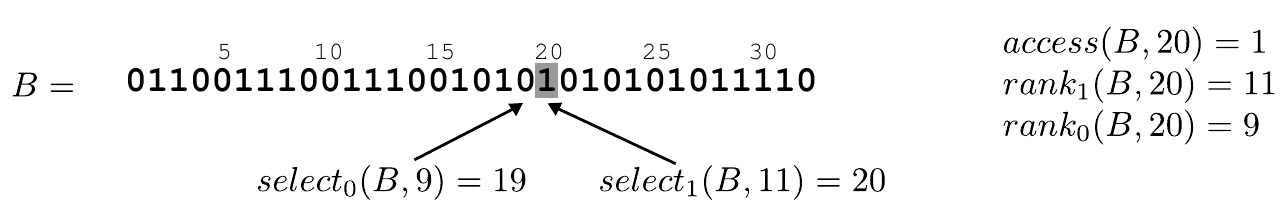}
\caption{Example of rank and select}
\label{fig:rank_select}
\end{figure}
\begin{example}
\nieves{Figure \ref{fig:rank_select} shows an example of the operations rank and select. We show the values of both $rank_1(B,20)=11$ and $rank_0(B,20)=9$. Note that these two values add up to 20, since the former returns the number of ones up to position 20, and the latter the number of zeroes. Also, $access$ simply returns the bit stored at that position, in our case at position 20 there is a 1. Finally, we show the value of $select_1(B,11)=20$, which was expected since $access(B,20)=1$. The value of $select_0(B,9)$ is 19.}
\end{example}

Several solutions have been proposed to address this problem. The first solution able to solve both \diego{kinds of} queries in constant time uses $n+O(\frac{n\log\log n}{\log n})$ bits of space \cite{Cla96}. 
Raman, Raman and Rao\jeremy{'s solution} (RRR) \cite{RRR02} achieve\jeremy{s} $nH_0(B) + O(\frac{n\log\log n}{\log n})$ bits and answer the queries in constant time. Okanohara and Sadakane \cite{OS07} proposed several alternatives tailored to the case of small $m$ (sparse bitmaps): \texttt{esp}, \texttt{recrank}, \texttt{vcode}, and \texttt{sdarray}. 
Table \ref{tab:bitmap} shows the time and space complexities of these solutions. Note that the reported space\jeremy{s} include\djeremy{s} the representation of the bitmap.

\begin{table}
\begin{center}
\begin{tabular}{l|c|c|c}
Variant & Size & Rank & Select \\ \hline
Clark & $n + o(n)$ & $O(1)$ & $O(1)$ \\
RRR & $nH_0(B) + o(n)$ & $O(1)$ & $O(1)$ \\
esp & $nH_0(B) + o(n)$ & $O(1)$ & $O(1)$ \\
recrank & $1.44 m \log \frac{n}{m} + m + o(n)$ & $O\left(\log \frac{n}{m}\right)$ & $O\left(\log \frac{n}{m}\right)$\\
vcode & $m\log(n/\log^2 n)+o(n)$ & $O(\log^2 n)$ & $O(\log n)$  \\
sdarray & $m \log \frac{n}{m} + 2m + o(m)$& $O\left(\log \frac{n}{m} + \frac{\log^4 m}{\log n}\right)$ &  $O\left(\frac{\log^4 m}{\log n}\right)$ 
\end{tabular}
\caption{Complexities for binary rank and select}
\label{tab:bitmap}
\end{center}
\end{table}

\subsection{Practical Dense Bitmaps}
\label{sec:prac_bitmap}
The extra $o(n)$ space of theoretical solutions \cite{Cla96} is large in practice.
\djeremy{A solution with good results in practice and small space overhead (up to 5\%) is the one proposed by González \etal ~\cite{GGMNwea05}.}
\jeremy{González \etal ~\cite{GGMNwea05} proposed a solution with good results in practice and small space overhead (up to 5\%).}
This implementation is very simple, yet its practical performance is better than classical solutions. They store the plain bitmap in an array $B$ and have a table $R_s$ where they store $rank_1(B,i \cdot s)$, where $s=32k$, where $k$ is a parameter \jeremy{for the frequency of the sampling of the bit vector}. They use a function called \emph{popcount} that counts the number of \diego{\texttt{1} bits} \ddiego{ones} in \djeremy{an integer} \jeremy{a word} \jeremy{(4 bytes)}. \nieves{This operation can be solved bit by bit, but it is easy to improve it, using either bit parallelism or precomputed tables, requiring thus just a few operations.}
\nieves{They} solve the operations as follows ($rank_0$ and $select_0$ are obvious variations):

\begin{itemize}
\item $rank_1(B,i)$: They start in the last entry of $R_s$ that precedes $i$ ($R_s[\lfloor i/s \rfloor]$), and then sequentially scan the array $B$, popcounting \djeremy{in chunks of $w = 32$ bits}  \jeremy{consecutive words}, until reaching the desired position. The popcounting of the last word is done by first setting all bits after position $i$ to zero, which is done in constant time using a mask. Thus the time is $O(k)$. 
\item $select_1(B,i)$: They first binary search the $R_s$ table for the last position $p$ where $R_s[p]\le i$. Then they scan $B$ sequentially using \emph{popcount} looking for the word where the desired select position is. Finally they find the desired position in the word by \djeremy{first popcounting byte by byte and finally they} sequentially scan\jeremy{ning} the \djeremy{byte} \jeremy{word} bit by bit. Thus the time is $O(k+\log \frac{n}{k})$.
\end{itemize} 

We will use the implementation of Rodrigo González (available at \libcdsurl) in this thesis.

\subsection{Practical Sparse Bitmaps}
\label{sec:deltacodes}
When the bitmap is very sparse (i.e., the number of ones in the bitmap is very low) one practical solution is to $\delta$-encode the \ddiego{consecutive} distances between \diego{consecutive} ones.
Additionally we need to store absolute sample\djeremy{s} \jeremy{values} $select_1(B, i\cdot s)$ for a sampling step $s$, plus pointers to the corresponding positions in the $\delta$-encoded sequence. We solve the operations as follows:
\begin{itemize}
\item $select_1(B,i)$ is solved within $O(s)$ time by going to the last sampled position preceding $i$ and decoding the $\delta$-encoded sequence from there. 
\item $rank_1(B,i)$ is solved in time $O(s + \log \frac{m}{s})$. First, we binary search the samples looking for the last sampled position such that $select_1(B, \ell\cdot s) \le i$. Starting from that position we sequentially decode the bitmap and stop as soon as $select_1(B, p) \ge i$.
\item $access(B,i)$ is solved in time $O(s + \log \frac{m}{s})$ in a way similar to rank.
\end{itemize}
The space needed by the structure is $W + n/s(\lfloor\log m\rfloor +1 + \lfloor\log W \rfloor + 1)$, where $W$ is the number of bits needed to represent all the $\delta$-codes. In the worst case $W=2m \lfloor\log(\lfloor\log \frac{n}{m}\rfloor + 1)\rfloor + m\lfloor\log \frac{n}{m}\rfloor + m=m\log \frac{n}{m} + O(m\log\log \frac{n}{m})$.

This structure allows a space-time trade-off related to $s$ and also has the property that several operations cost $O(1)$ after solving others. For example, $select_1(B, p)$ and $select_1(B, p+1)$ cost $O(1)$ after solving $p \leftarrow rank_1(B, i)$.
\section{Wavelet Trees}
\label{sec:wavelettree}
A wavelet tree \cite{GGV03} is an elegant data structure that stores a sequence $S$ of $n$ symbols from an alphabet $\Sigma$ of size $\sigma$. 
This structure supports some basic queries and is easily extensible to support others.

We split the alphabet into two halves $L$ and $R$, so that the elements of $L$ are lexicographically smaller than those of $R$. Then, we create a bitmap $B$ of size $n$ setting $B[i]=\texttt{0}$ if the symbol at position $i$ belongs to $L$ and $B[i]=\texttt{1}$ otherwise.
This bitmap is stored at the root of the tree. \jeremy{Afterward}\djeremy{Then}, we extract from $S$ all symbols belonging to $L$, generating sequence $S_L$, and all symbols belonging to $R$, generating sequence $S_R$ (these sequences are not stored). \jeremy{Finally, we recursively generate}\djeremy{Then we generate recursively} the left subtree on $S_L$ and the right subtree on $S_R$. We continue until we get a sequence over a one-letter alphabet. \nieves{Figure \ref{fig:wavelettree} shows the wavelet tree for the example text \texttt{alabar\_a\_la\_alabarda}. Only the bitmaps (black color) are stored in the tree. The labels of the tree show (gray color) the subsets $L$ and $R$ and the strings over the bitmaps (gray color) show the conceptual subsequences $S_L$ and $S_R$.}

\begin{figure}[ht]
\begin{center}
\includegraphics{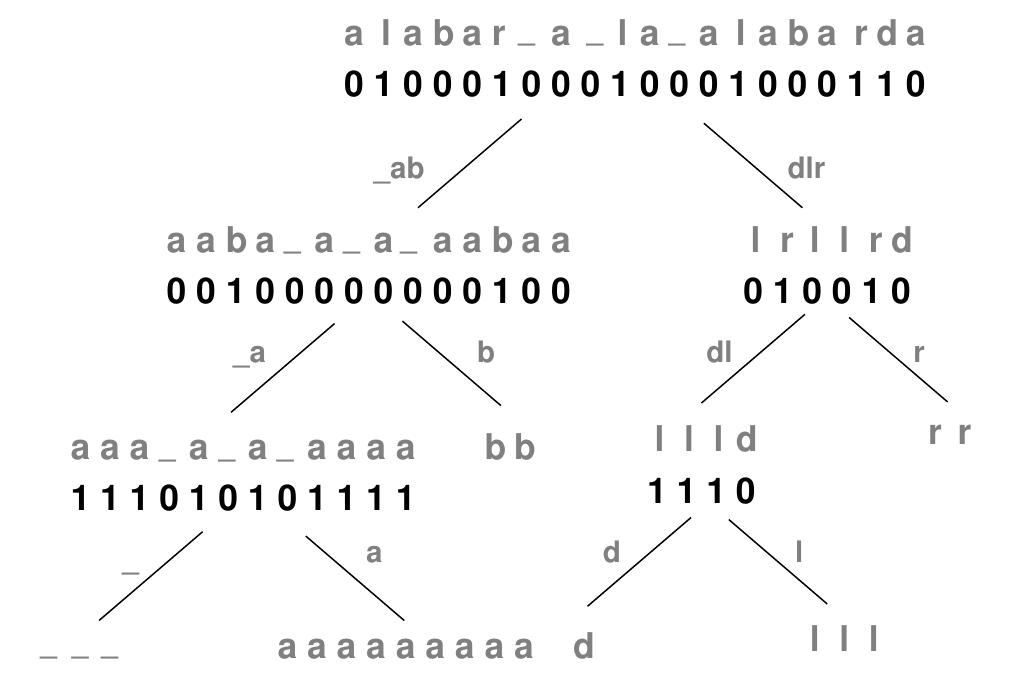}
\end{center}
\caption[Example of a wavelet tree]{Example of a wavelet tree for the text \texttt{alabar\_a\_la\_alabarda}}
\label{fig:wavelettree}
\end{figure}

The resulting tree has $\sigma$ leaves, height $\lceil \log \sigma \rceil$, and $n$ bits per level. Thus the space occupancy is $n \log \sigma$ bits, plus $o(n\log \sigma)$ (more precisely, \djeremy{$O(\frac{n\log \sigma \log\log n}{\log n})=O(n\log\log \sigma)$} \jeremy{$O(\frac{n\log \sigma \log\log n}{\log n})$}) additional bits to support \emph{rank} and \emph{select} queries on the bitmaps. 

In the following we \dmine{will} explain how this structure supports the operations \emph{access}, \emph{rank} and \emph{select} on $S$. The last two operations are just a generalization for larger alphabets of those defined in Section \ref{sec:bitmaps}.

\begin{itemize}
\item \textbf{Access:} To retrieve the symbol $S[i]$ we look at $B[i]$ at the root. If it is a \texttt{0} we go to left subtree, otherwise to the right subtree. The new position is $i \leftarrow rank_0(B,i)$ if we go to the left and $i \leftarrow rank_1(B,i)$ if we go to the right. This procedure continues recursively until we reach \ddiego{the last level} \diego{a leaf}. The bits read in the path from the root to the leaf represent the symbol sought.

\item \textbf{Rank:} To count how many $c$'s are up to position $i$ we  go to the left if $c$ is in $L$ and otherwise to the right. The new position is $i \leftarrow rank_0(B,i)$ if we go to the left and $i \leftarrow rank_1(B,i)$ if we go to the right, where $B$ is the bitmap of the root. When we reach a leaf the answer is $i$.

\item \textbf{Select:} To find the $i$-th symbol $c$ we first go to the leaf corresponding to $c$ and then \ddiego{we} go upwards to the root. Let $B$ the bitmap of the parent. If the current node is a left child then the position at the parent is $i \leftarrow select_0(B,i)$, otherwise it is $i \leftarrow select_1(B,i)$. When we reach the root the answer is \diego{the current} $i$ \diego{value}.
\end{itemize}

The running time of these operations is $O(\log \sigma)$, \diego{since} \ddiego{given that} we use a bitmap supporting constant-time rank, select and access.

\begin{figure}
\centering
\includegraphics{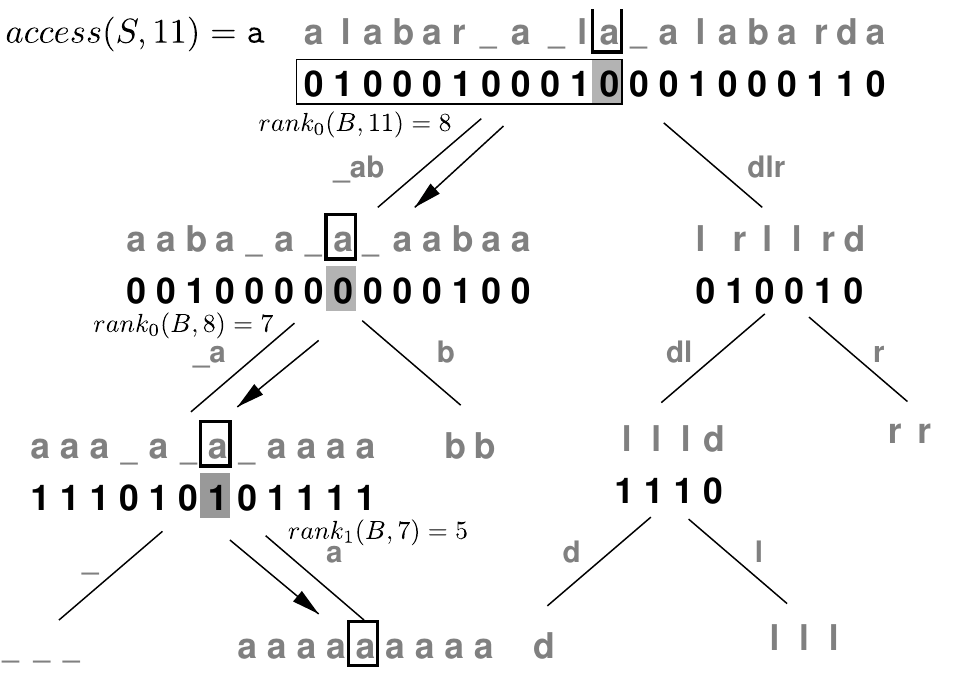}
\caption{Example of access in a wavelet tree}
\label{fig:wt_access}
\end{figure}
\begin{example}
\nieves{Figure \ref{fig:wt_access} shows an example of how we retrieve the 11th symbol of sequence $S$ ($access(S,11)=\texttt{a}$). First we access the bitmap of the root and see that at position 11 there is a 0. Hence we descend to the left. Then using $rank_0(B,11)=8$ we count how many zeroes are up to position 11. This value is our new position in the next level. Then we continue the process until we reach a leaf; in that case the symbol stored in that lead is the symbol sought, in our case an \texttt{`a'}.}
\end{example}

\begin{figure}
\centering
\includegraphics{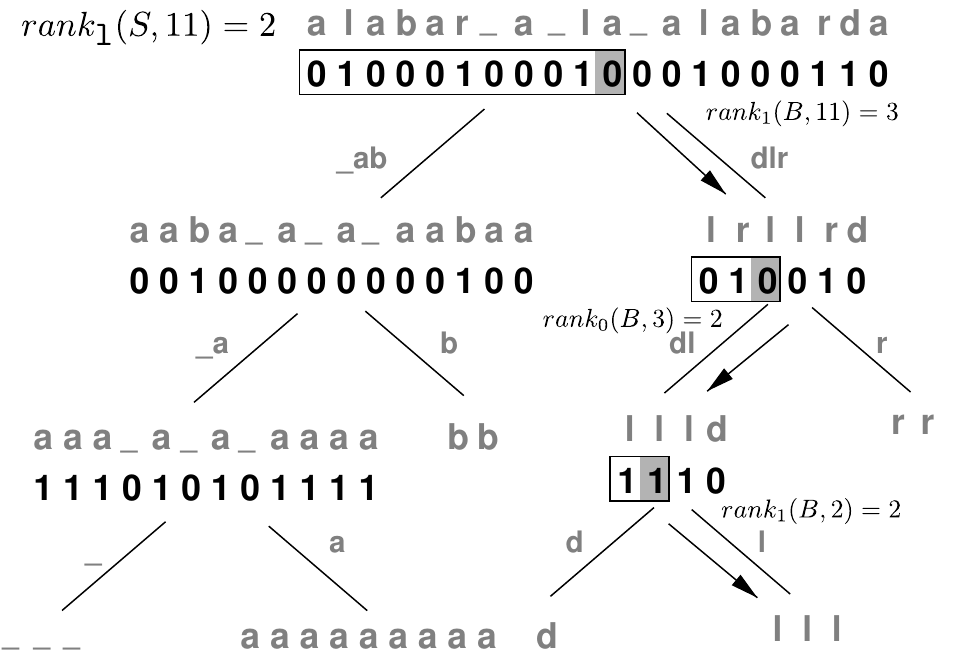}
\caption{Example of rank in a wavelet tree}
\label{fig:wt_rank}
\end{figure}
\begin{example}
\nieves{Figure \ref{fig:wt_rank} shows step by step how we compute $rank_\texttt{l}(S,11)=2$. Since symbol \texttt{`l'} is mapped to a 1 we descend from the root to the right child. Using $rank_1(B,11)=3$ we count the number of ones up to that position. This is our new position in the next level. Then we continue the process until we reach a leaf. The value sought is the last value of rank, in our case 2.}
\end{example}

\begin{figure}
\centering
\includegraphics{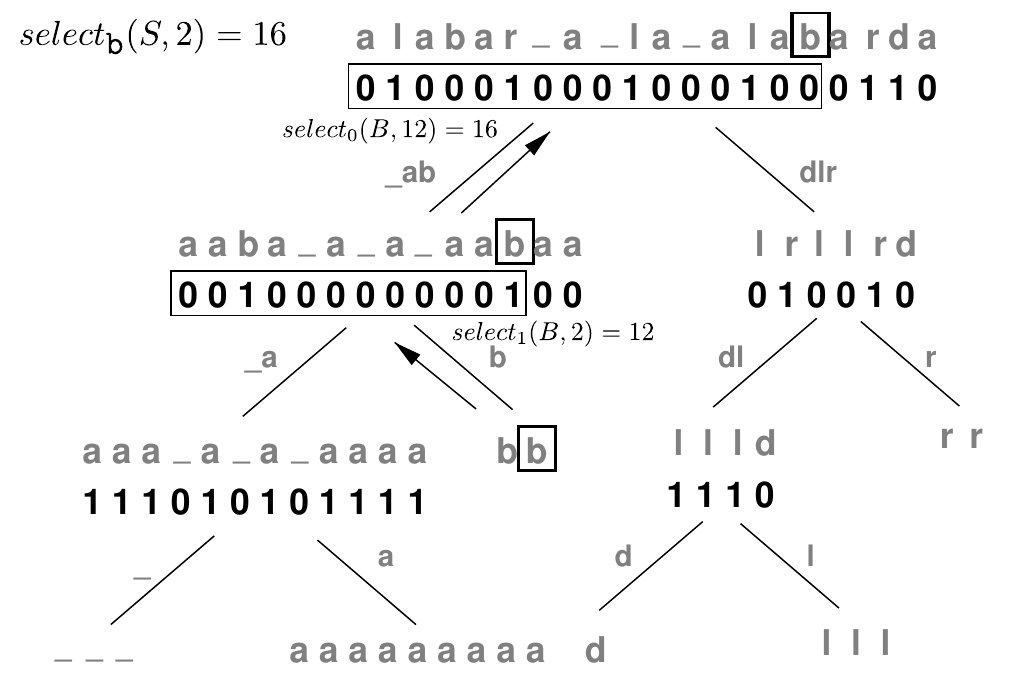}
\caption{Example of select in a wavelet tree}
\label{fig:wt_select}
\end{figure}
\begin{example}
\nieves{Figure \ref{fig:wt_select} shows an example of how to select the second \texttt{`b'} in the sequence $S$ ($select_\texttt{b}(S,2)$). First we descend to the leaf representing symbol \texttt{`b'}. Since that symbol was last mapped to a 1, we go to the parent and compute our new position as $select_1(B,2)=12$. In that level, \texttt{`b'} was mapped to a 0, so we go to the parent and the new position is $select_0(B,12)=16$, and that is the value sought.}
\end{example}

\subsection{Range Search}
\label{sec:range}
A direct application of wavelet trees is to answer range search queries \cite{MN07}. This method is very similar to the idea of Chazelle \cite{Cha88}.
\begin{definition}
Given a subset $R$ ($|R|=t$) of the discrete range \djeremy{$[1,n]\times[1,n]$} \jeremy{$[1,n]\times[1,\sigma]$}, a \emph{range query} returns the points $p \in R$ belonging to a range $[x_1,x_2]\times[y_1,y_2]$.
\end{definition}
\jeremy{An extension of the} \djeremy{The} wavelet tree supports range queries using $n+t\log n + o(n+t\log n)$ bits, counting the number of points \diego{within the range} in time $O(\log n)$ and reporting each occurrence in time $O(\log n)$. We will use a modified version of the implementation of Gonzalo Navarro\jeremy{\footnote{Check the LZ77-index source code (\indexurl) for the updated version}}.

We \djeremy{will} explain \jeremy{here} a simplified version \jeremy{for the case} in which there exists exactly one point for each value of $x$. We order the points of $R$ by their $x$ coordinate and create the sequence $S[1,n]$, such that for each $(x,y) \in R$, $S[x]=y$. Then we build the wavelet tree of $S$. 

\begin{figure}
\centering
\includegraphics[width=10cm]{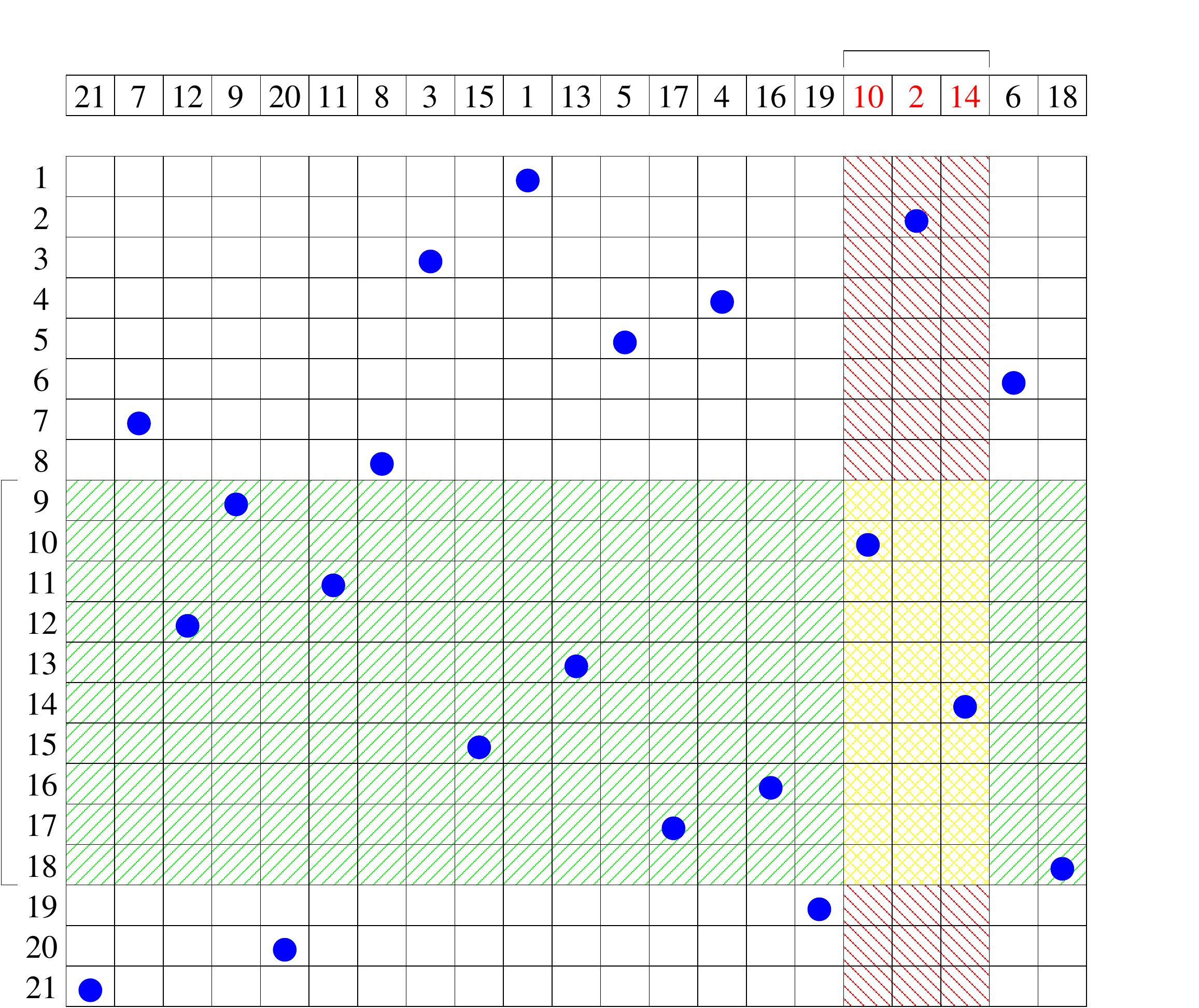}
\caption{Example of 2-dimensional range query}
\label{fig:range_example}
\end{figure}

\begin{example}
Figure \ref{fig:range_example} shows a grid, with exactly one $y$ value for each value of $x$. The figure shows in yellow the range $[17,19]\times[9,18]$, containing two occurrences; in red and yellow, the range $[17,19]\times[1,21]$, containing 3 occurrences; and in green and yellow, the range $[1,21]\times[9,18]$, containing 10  occurrences.
\end{example}

\subsubsection{Projecting}
A range in $S$ represents a range along the $x$ coordinate and the splits made by the wavelet tree define ranges along the $y$ coordinate. Every time we descend to a child of a node we need to know where \ddiego{is} the range represented in that child \diego{is}. The operation of determining the range defined by a child, given the range of the parent, is called \emph{projecting}. Using rank we project a range downwards. Given a node with bitmap $B$ the left projection of $[x,x']$ is $[1+rank_0(B,x-1), rank_0(B,x')]$ and the right projection is $[1+rank_1(B,x-1), rank_1(B,x')]$. A range $[y,y']$ along the $y$ coordinate is projected to the left as $[y,\lfloor(y+y')/2\rfloor]$ and to the right as $[\lfloor(y+y')/2\rfloor +1,y']$.

\subsubsection{Counting}
We start from the root with the \diego{one-dimensional} \ddiego{linear} ranges $[x,x']=[x_1,x_2]$ and $[y,y']=[1,n]$ and project them in both subtrees. We do this recursively until:
\begin{enumerate}
    \item $[x,x']=\emptyset$;
    \item $[y,y'] \cap [y_1,y_2] = \emptyset$; or
    \item $[y,y'] \subseteq [y_1,y_2]$, in which case we add $x'-x+1$ to the total.
\end{enumerate}
As the interval $[y_1,y_2]$ is covered by $O(\log n)$ maximal wavelet tree nodes, the total time to count the occurrences is $O(\log n)$.

\begin{figure}
\centering
\includegraphics[width=10cm]{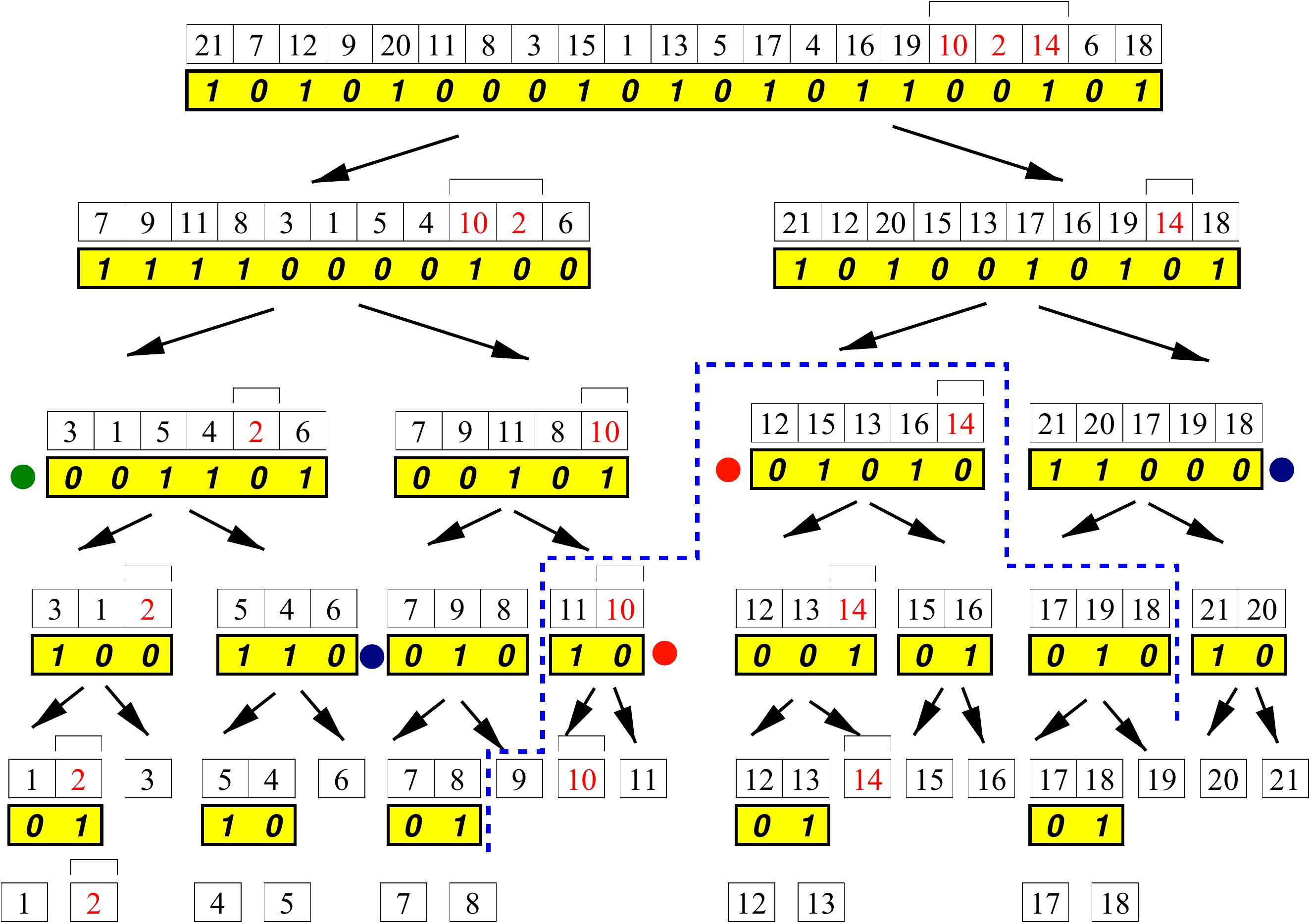}
\caption{Example of counting the occurrences in a 2-dimensional range query using a wavelet tree}
\label{fig:counting}
\end{figure}
\begin{example}
Figure \ref{fig:counting} shows the wavelet tree that represents the range of Figure \ref{fig:range_example}. The figure represents how to count the occurrences in the range $[17,19]\times[9,18]$. The figure shows in red how the range $[17,19]$ in the $x$ coordinate is projected downwards. The nodes below the blue line are those whose $y$ range is contained in the range $[9,18]$. Additionally, the nodes with a circle next to them are those in which the counting process ends. The blue circle represents rule number 1 (see above), the green one represents rule number 2 and finally rule number 3 is represented by the red circle. In our case, at each node marked with red, we report one occurrence, yielding a total of 2 occurrences. 
\end{example}
\subsubsection{Locating}
To locate the actual points we start from each node in which we were counting. If we want to know the $x$ coordinate we go up using select and if we want to know the $y$ coordinate we go down using rank. This operation takes $O(\log n)$ for each point \diego{located}.
\section{Permutations}
\label{sec:permutation}
A permutation is a bijection $\pi:[1,n]\rightarrow[1,n]$, and we are interested in computing efficiently both $\pi(i)$ and $\pi^{-1}(i)$ \jeremy{for any $1 \le i\le n$}. 
The permutation can be represented in a plain array using $n\log n$ bits, by storing $P=[\pi(1),\ldots,\pi(n)]$. This answers \djeremy{$\pi$}\jeremy{$\pi(i)$} in constant time. Solving \djeremy{$\pi^{-1}$}\jeremy{$\pi^{-1}(i)$} can be done by sequentially scanning $P$ for the position $j$ where $\pi(j)=i$. A more efficient solution \cite{MRRR03} is based on the cycles of a permutation. A cycle is a sequence $i,\pi(i),\pi^2(i), \ldots,\pi^k(i)$ such that $\pi^{k+1}(i)=i$. Every $i$ belongs to exactly one cycle. Then, to compute $\pi^{-1}$ we repeatedly apply $\pi$ over $i$, finding the element $e$ of the cycle such that $\pi(e)=i$. 
These solutions do not require any extra space to compute \djeremy{$\pi^{-1}$}\jeremy{$\pi^{-1}(i)$}, but they take $O(n)$ time in the worst case\dmine{ for computing $\pi^{-1}$}.

Representing the sequence $\pi[1,n]$ with a wavelet tree one \djeremy{could} \jeremy{can} answer both queries using $O(\log n)$ time and $n\log n + o(n \log n)$ bits of space. A faster solution \cite{MRRR03} is based on the cycles of the permutation. By introducing shortcuts in the cycles, it uses $(1 + \varepsilon)n \log n + O(n)$ bits and solves \djeremy{$\pi$}\jeremy{$\pi(i)$} in constant time and \djeremy{$\pi^{-1}$}\jeremy{$\pi^{-1}(i)$} in $O(1/\varepsilon)$ time, for any $\varepsilon>0$. 

We will use the implementation of \jeremy{Munro \etal's} \djeremy{permutations} \jeremy{shortcut technique} by Diego Arroyuelo\jeremy{\footnote{Yahoo! Research, Chile. \texttt{darroyue@dcc.uchile.cl}}}, available at \libcdsurl.
\section{Tree Representations}
\label{sec:dfuds}
A classical representation of a general tree of $n$ nodes requires $O(nw)$ bits of space, where $w \ge \log n$ is the bit length of a machine pointer. Typically only operations such as moving to the first child and to the next sibling, or to the $i$-th child, are supported in constant time. By further increasing the constant, some other simple operations are easily supported, such as moving to the parent, knowing the subtree size, or the depth of the node. However, the $\Omega(n \log n)$ bit space complexity is excessive in terms of information theory. The number of different general trees of $n$ nodes is $C_n \approx 4^n /n^{3/2}$, hence $\log C_n = 2n - \Theta(\log n)$ bits are sufficient to distinguish any one of them.

There are several succinct tree representations that use $2n+o(n)$ bits of space and answer most queries in constant time (see \jeremy{the review by Arroyuelo \etal} \cite{ACNSalenex10} for a detailed exposition); here we explain the DFUDS \cite{labeled_child} representation as this is the one that meets our requirements.

\begin{definition}
A sequence $S$ drawn from alphabet $\Sigma=\lbrace\texttt{0},\texttt{1}\rbrace$ is said to be \emph{balanced} if: (1) there are as many \texttt{0}s as \texttt{1}s and (2) at any position $i$ the number of zeroes to the left is greater or equal than the number of ones (i.e., $rank_\texttt{0}(S,i)\ge rank_\texttt{1}(S,i)$). Usually a balanced sequence is referred as \emph{balanced parentheses} by identifying \texttt{0} as `(' and \texttt{1} as `)', as the nesting of parentheses satisfies the above definition.
\end{definition}

The operations defined over a balanced sequence are: (1) \emph{findclose(S,i) (findopen(S,i))} finds the matching \texttt{1} (\texttt{0}) of the \texttt{0} (\texttt{1}) at position $i$, and (2) \emph{enclose(S,i)} is the position of tightest \texttt{0} enclosing node $i$.

\begin{definition}[\cite{labeled_child}]
The Depth-first unary degree sequence (DFUDS) is generated by a depth-first traversal of the tree, at each node appending the degree of the node in unary. Additionally a leading \texttt{1} is prepended to the sequence to make it balanced \jeremy{and allow the concatenation of several such encodings into a forest}. 
\end{definition}

The DFUDS sequence represents the topology of the tree using $2n$ bits. \jeremy{Tree nodes are identified in the DFUDS sequence according to their rank in the order given by the depth-first traversal (more precisely, the $i$-th node is identified by position $select_\texttt{1}(i)$ in the DFUDS encoding).}  \nieves{Figure \ref{fig:dfuds} shows the DFUDS bit-sequence for the example tree. The red 1 in the sequence is the preceding 1 added to make the sequence balanced. The green node is represented by the 10th 1 in the sequence, as it is the 10th node visited during a depth-first traversal. The blue sequence of five 1s and one 0 is the degree of the blue node.}

\begin{figure}[ht]
\begin{center}
\includegraphics{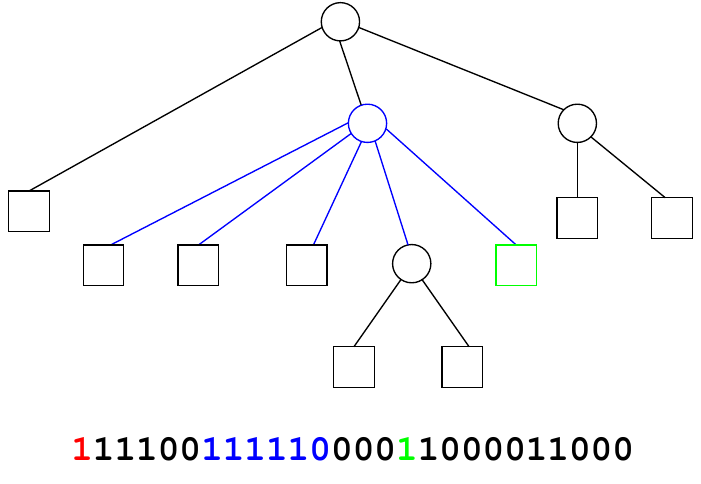}
\end{center}
\caption[Example of DFUDS representation]{Example of DFUDS representation}
\label{fig:dfuds}
\end{figure}

To solve the common operations over trees two data structures are built over the DFUDS sequence: a bitmap data structure supporting rank and select (Section \ref{sec:bitmaps}) and a data structure solving operations \emph{findclose}, \emph{findopen} and \emph{enclose} \cite{bp_jacobson,bp_munroraman,Nav09}. These structures allow one to compute the most common operations in constant time using $o(n)$ additional bits of space. Additionally, if we use labeled trees we need to store the labels of the edges in an array $chars$, using $n \log \sigma$ additional bits, where $\sigma$ is the labels' alphabet size. The label of the edge pointing to the $i$-th child of node $x$ is at \djeremy{$chars[rank_{\texttt{(}}(dfuds,x)+i]$} \jeremy{$chars[rank_{\texttt{1}}(dfuds,x)+i]$}. The operations we are interested in for this thesis are:
\begin{itemize}
\item \emph{degree($x$)}: number of children of node $x$.
\item \emph{isLeaf($x$)}: whether node $x$ is a leaf.
\item \emph{child($x$,$i$)}: $i$-th child of node $x$.
\item \emph{labeledChild($x$,$c$)}: child of node $x$ labeled by symbol $c$.
\item \emph{leftmostLeaf($x$)}: leftmost leaf of the subtree starting at node $x$.
\item \emph{rightmostLeaf($x$)}: rightmost leaf of the subtree starting at node $x$.
\item \emph{leafRank($x$)}: number of leaves to the left of node $x$. 
\item \emph{preorder($x$)}: preorder position of node $x$.
\end{itemize}
All these operations can be solved theoretically in constant time; however, in practice \emph{labeledChild} is solved by binary searching the \diego{labels of the} children\jeremy{, because it is much easier to implement and fast enough in practice}. To solve \emph{leftmostLeaf}, \emph{rightmostLeaf} and \emph{leafRank} we need to solve the queries $rank_{\texttt{00}}(i)$ and $select_{\texttt{00}}(i)$. $Rank_{\texttt{00}}(i)$ returns the number of occurrences of the substring \texttt{00} in the bitmap up to position $i$ and $select_{\texttt{00}}(i)$ returns the position $p$ of the $i$-th occurrence of the substring \texttt{00} in the bitmap. Solving these queries requires an additional data structure that uses $o(n)$ bits. It uses the same ideas as the one for solving rank and select for binary alphabets.

We will use a modified version of the implementation of Diego Arroyuelo available at \libcdsurl, adding support for leaf-related operations.
\section{Tries}
\label{sec:tries}
A \emph{trie} or digital tree is a data structure that stores a set of strings. It can find the elements of the set prefixed by a pattern in time proportional to the pattern length.
\begin{definition}
A \emph{trie} for a set $S$ of distinct strings is a tree where each node represents a distinct prefix in the set. The root node represents the empty
prefix $\varepsilon$. A node $v$ representing prefix $Y$ is a child of node $u$ representing prefix $X$ iff $Y = Xc$ for some character $c$, which labels the edge between $u$ and $v$.
\end{definition}
We suppose that all strings are ended by a special symbol \texttt{\$}, not present in the alphabet. We do this in order to ensure that no string $S_i$ is a prefix of some other string $S_j$. This property guarantees that the tree has exactly $|S|$ leaves. Figure \ref{fig:trie} shows an example of a trie.

\begin{figure}[ht]
\begin{center}
\includegraphics{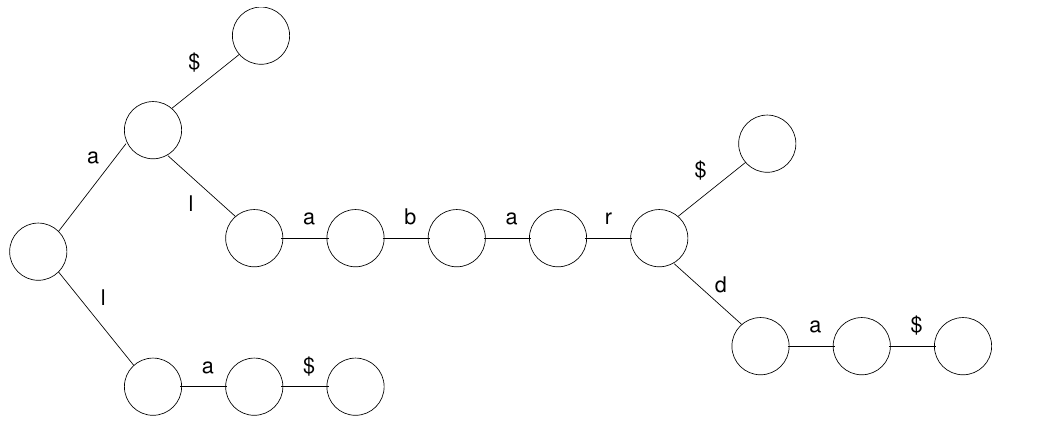}
\end{center}
\caption[Example of a trie]{Example of a trie for the set $S=\lbrace\texttt{`alabar'}, \texttt{`a'}, \texttt{`la'}, \texttt{`alabarda'}\rbrace$.}
\label{fig:trie}
\end{figure}

A trie for the set $S=\lbrace S_1,\ldots,S_n\rbrace$ is easily built on $O(|S_1|+\ldots+|S_n|)$ time by successive insertions \diego{(assuming we can descend to any child in constant time)}. A pattern $P$ is searched for in the trie starting from the root and following the edges labeled with the characters of $P$. This takes a total time of $O(|P|)$.

A \emph{compact trie} is an alternative representation that reduces the space of the trie by collapsing unary nodes \ddiego{to} \diego{into} a single node and labeling the edge with the concatenation of all labels. A \emph{PATRICIA tree} \cite{patricia}, an alternative that uses even less space, just stores the first character of the label string and its length. This variant is used when the strings $S$ are \ddiego{separately available} \diego{available separately}, as not all information is stored in the edges. In this variant, after the search we need to check if the prefix found actually matches the pattern. For doing so, we have to extract the text corresponding to any string with the prefix found and compare it with the pattern. It they are equal then all leaves will be occurrences (i.e., strings prefixed with the pattern), otherwise none will be an occurrence. Figure \ref{fig:trie2} shows an example of this kind of trie.  

\begin{figure}[ht]
\begin{center}
\includegraphics{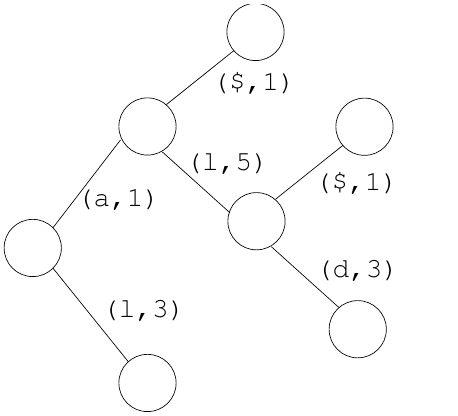}
\end{center}
\caption[Example of a PATRICIA trie]{Example of a PATRICIA trie for the set $S=\lbrace\texttt{`alabar'}, \texttt{`a'}, \texttt{`la'}, \texttt{`alabarda'}\rbrace$. The values in parentheses are respectively the first character of the label and the length of the label.}
\label{fig:trie2}
\end{figure}

\begin{definition}
A \emph{suffix trie} is a trie composed of all the suffixes $T[i,n]$ of a given text $T[1,n]$. The leaves of the trie store the positions where the suffixes start.
\end{definition}

\section{Suffix Trees}
\label{sec:stree}
\begin{definition}[\cite{Wei73,McC76}]
A \emph{suffix tree} is a PATRICIA tree built over all the suffixes $T[i,n]$ of a given text $T[1,n]$. The leaves in the tree indicate the text positions where the corresponding suffixes start.
\end{definition}

Figure \ref{fig:sa} shows the suffix tree for the text \texttt{`alabar\_a\_la\_alabarda\$'}.
\begin{figure}[ht]
\begin{center}
\includegraphics{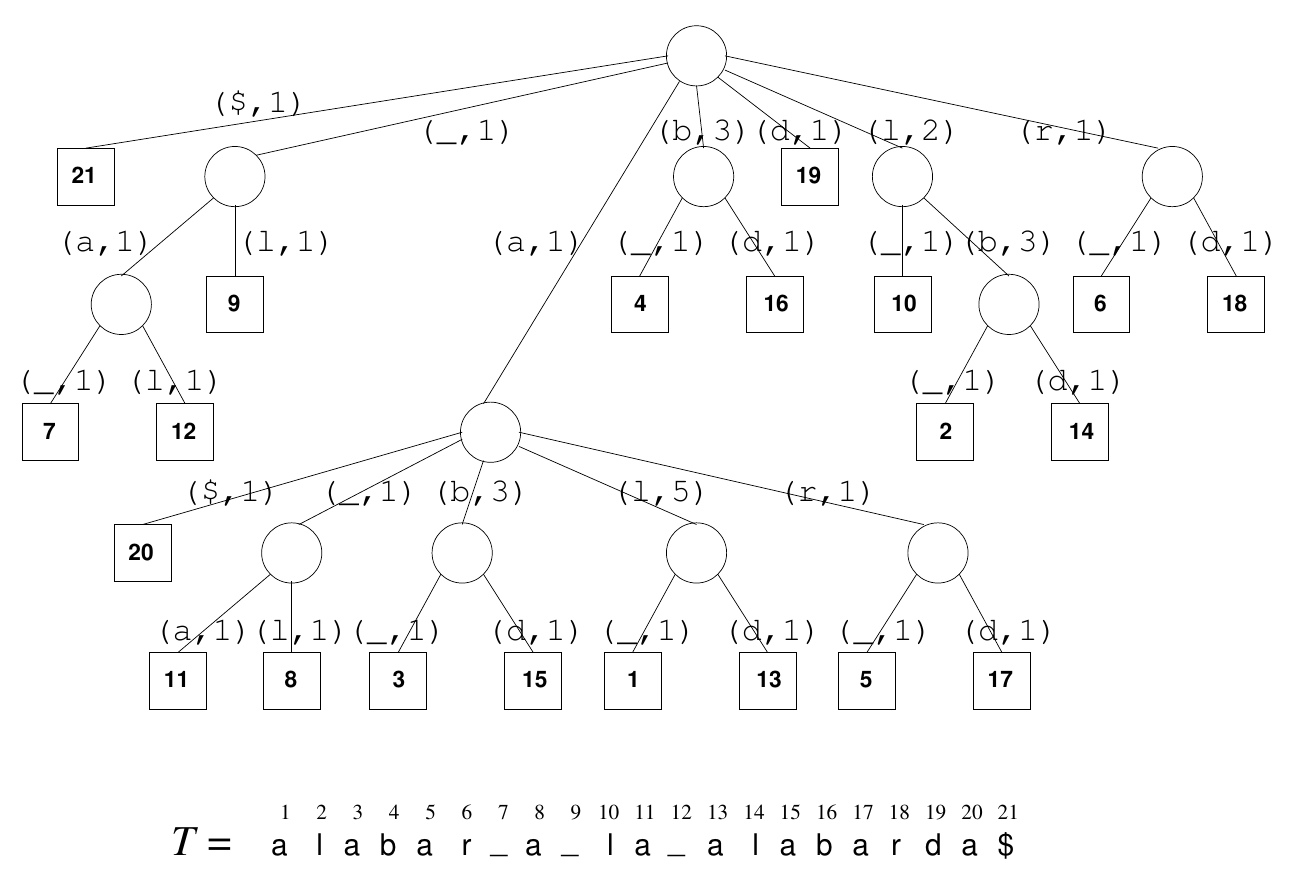}
\end{center}
\caption{\texorpdfstring{The suffix tree for the text \texttt{`alabar\_a\_la\_alabarda\$'}}{The suffix tree for the text 'alabar\_a\_la\_alabarda\$'}}
\label{fig:sa}
\end{figure}

The suffix tree can be built in $O(n)$ time using $O(n\log n)$ bits of space \cite{McC76,Ukk95}.

A suffix tree is able to find all the $occ$ occurrences of a pattern $P$ of length $m$ in time $O(m+occ)$, i.e., to solve the \emph{locate} query described in Section \ref{sec:queries}. After descending by the tree according to the characters of the pattern, we could be in three different cases: i) we reach a point in which there is no edge labeled with the current character of $P$, which means that the pattern does not occur in $T$; ii) we finish reading $P$ in an internal node (or in the middle of an edge), in which case the suffixes of the corresponding subtree are either all occurrences or none, therefore we only need to check if one of those suffixes matches the pattern $P$; iii) we end up in a leaf without consuming all the pattern, in which case at most one occurrence is found after checking the suffix with the pattern. As a subtree with $occ$ leaves has $O(occ)$ nodes\diego{,} the total time for reporting the occurrences is as stated above.

The suffix tree can solve the queries \emph{count} and \emph{exists} in $O(m)$ time. The process is similar to that of \emph{locate}. First we descend the tree according to the pattern. Then, we check if one of the suffixes of the subtree is a match. If it is a match the answer of \emph{count} is the number of leaves of the subtree \diego{(for which we need to store in each internal node the number of leaves that descend from it)}, otherwise it is zero.

\section{Suffix Arrays}
\label{sec:sarray}
\begin{definition}[\cite{MM93,sa_gonet}]
A \emph{suffix array} $A[1,n]$ is a permutation of the integer interval $[1,n]$, holding $T[A[i],n]<T[A[i+1],n]$ for all $1\le i < n$. In other words\diego{,} it is a permutation of the suffixes of the text such that the suffixes are lexicographically sorted.
\end{definition}

Figure \ref{fig:sarray} shows the suffix array for the text \texttt{`alabar\_a\_la\_alabarda\$'}. \jeremy{The character \texttt{\$} is the smallest one in lexicographical order.} The zone highlighted in gray represents those suffixes starting with \ddiego{an} \texttt{`a'}.
\begin{figure}[ht]
\begin{center}
\includegraphics{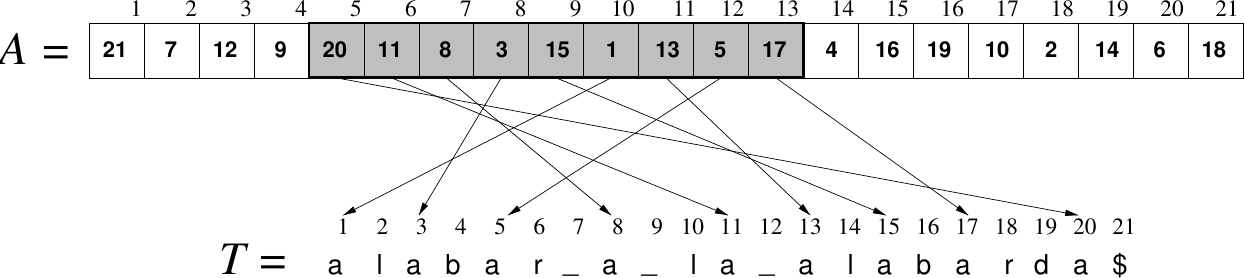}
\end{center}
\caption{\texorpdfstring{The suffix array for the text \texttt{`alabar\_a\_la\_alabarda\$'}}{The suffix array for the text 'alabar\_a\_la\_alabarda\$'}}
\label{fig:sarray}
\end{figure}

Note that the suffix array could be computed by collecting the values at the leaves of the suffix tree. 
\ddiego{Several} \diego{However, several} methods exist that compute the suffix array in $O(n)$ or $O(n\log n)$ time\diego{,} using significantly less space. For a complete survey see \cite{sa_construction}.

The suffix array can solve \emph{locate} queries in $O(m\log n + occ)$ time, and \emph{count} and \emph{exists} queries in $O(m\log n)$ time. First, we search for the interval $A[sp_1, ep_1]$ of the suffixes starting with $P[1]$. This can be done via two binary searches on $A$. The first binary search determines the starting position $sp$ for the suffixes lexicographically larger than or equal to $P[1]$, and the second determines the ending position $ep$ for suffixes that start with $P[1]$. Then, we consider $P[2]$, narrowing the interval to $A[sp_2, ep_2]$, holding all suffixes starting with $P[1,2]$. This process continues until $P$ is fully consumed or the \ddiego{interval is empty} \diego{current interval becomes empty}. Note that this algorithm searches for the pattern from left to right. For each character of the pattern, we do two binary searches taking at most time $O(\log n)$, hence the total time is $O(m\log n)$. Then \emph{locate} reports all occurrences in $O(occ)$ time and the answer to \emph{count} is $ep_m-sp_m+1$. We can also directly search for the interval $A[sp,ep]$ where the suffixes start with the pattern $P$ using just two binary searches on $A$, which find the first  and last position where the suffixes start with $P$. Each comparison between the pattern and a suffix will take at most $O(m)$ time, hence the total running time is also $O(m\log n)$. Yet, this is faster in practice than the previous method and is what we use in this thesis.

\section{Backward Search}
\label{sec:bwt}
Backward search is an alternative method for finding the interval $[sp,ep]$ corresponding to a pattern $P$ in the suffix array. It searches for the pattern from right to left, and is based on the Burrows-Wheeler transform.

\begin{definition}[\cite{BW94}]
Given a text $T$ terminated with the special character $T[n]=\texttt{\$}$ smaller than all others, and its suffix array $A[1, n]$, the \emph{Burrows-Wheeler transform (BWT)} of $T$ is defined as $T^{bwt}[i] = T[A[i]-1]$, except when $A[i] = 1$, where $T^{bwt}[i]=T[n]$. In other words, the transformation is \jeremy{conceptually} built first by generating all the cyclic shifts of the text, then sorting them lexicographically
, and finally taking the last character of each shift. \jeremy{In practice it can be built in linear time by building the suffix array first.}
\end{definition}
We can think of the sorted list of cyclic shift as a conceptual matrix $M[1,n][1,n]$. Figure \ref{fig:bwt} shows an example of how the BWT is computed for the text \linebreak
\texttt{`alabar\_a\_la\_alabarda\$'}. This transformation has the advantage of being easily compressed by local compressors \cite{Man2001}. It can be reversed as follows.
 
\begin{figure}[ht]
\begin{center}
\includegraphics[scale=1.2]{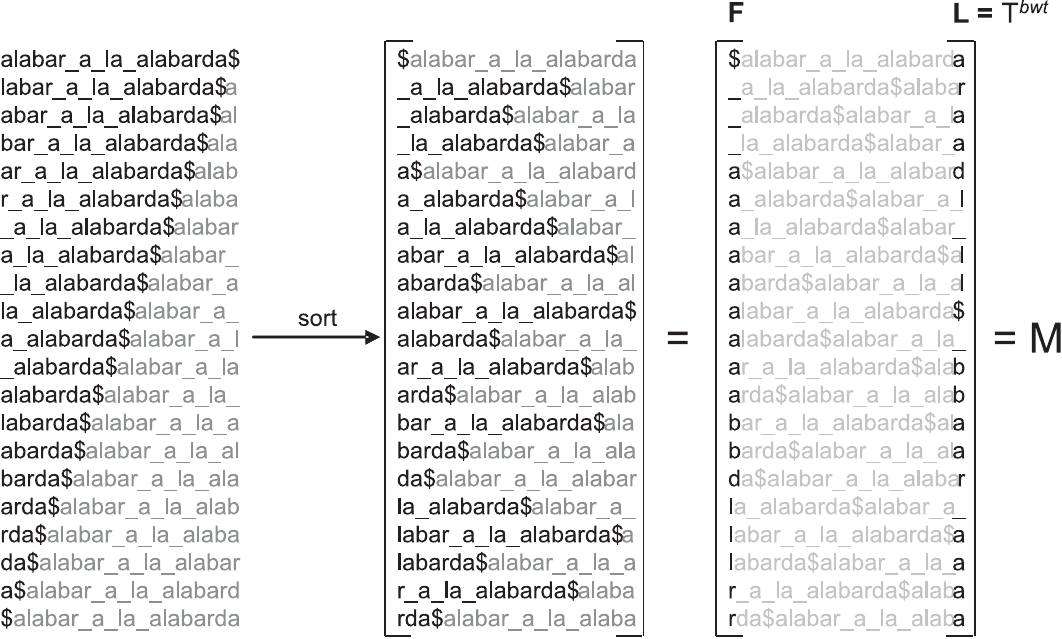}
\end{center}
\caption{\texorpdfstring{The BWT of the text \texttt{`alabar\_a\_la\_alabarda\$'}}{The BWT of the text 'alabar\_a\_la\_alabarda\$'}}
\label{fig:bwt}
\end{figure}
\begin{definition}
The LF-mapping $LF(i)$ maps a position $i$ in the last column of $M$ ($L=T^{bwt}$) to its occurrence in the first column of $M$ ($F$).
\end{definition}

\begin{lemma}[\cite{FM05}]
\label{lemma:lfmap}
It holds
\begin{equation*}
LF(i) = C[c] + rank_{c}(T^{bwt},i)
\end{equation*}
where $c=T^{bwt}[i]$ and $C[c]$ is the number of symbols smaller than $c$ in $T$.
\end{lemma}

\begin{lemma}[\cite{BW94}]
The LF-mapping allows one to reverse the Burrows-Wheeler transform.
\end{lemma}
\begin{proof}
We know that $T[n]=\texttt{\$}$ and since $\texttt{\$}$ is the smallest symbol, $T[n]=F[1]=\texttt{\$}$ and thus $T[n-1]=L[1]=T^{bwt}[1]$. Using the LF-mapping we compute $i=LF(1)$; knowing that $T[n-1]$ is at $F[i]$, we have $T[n-2]=L[i]$, as $L[i]$ always precedes $F[i]$ in $T$. In general, it holds $T[n-k] = T^{bwt}[LF^{k-1}(1)]$.
\end{proof}

Given the close relation between the suffix array and the BWT, it is natural to expect that a search algorithm can work on top of the BWT. Such algorithm is called \emph{backward search} (BWS), and at each stage it narrows the interval $[sp_i,ep_i]$ of the suffix array in which the suffixes start with $P[i,m]$, starting from $i=m$ and ending with $i=1$. Narrowing the interval $A[sp,ep]$ with a new character $c$ is called a $BWS(sp,ep,c)$ step and it is done very similarly to the LF-mapping (Lemma \ref{lemma:lfmap}). BWS searches a pattern from right to left, opposite to the search on suffix arrays, that searches for a pattern from left to right. 

Figure \ref{alg:bws} shows the backward search algorithm. Lines 5-7 correspond to the BWS step. 
\begin{figure}[ht]
\renewcommand{\algorithmiccomment}[1]{/*#1*/}
\algsetup{linenodelimiter=,indent=0.8em}
\textbf{BWS}$(P)$
\begin{algorithmic}[1]
\STATE $i \leftarrow len(P)$
\STATE $sp \leftarrow 1$
\STATE $ep \leftarrow n$
\WHILE{$sp \le ep$ \textbf{and} $i \ge 1$}
    \STATE $c \leftarrow P[i]$
    \STATE $sp \leftarrow C[c] + rank_{c}(T^{bwt},sp-1)+1$ 
    \STATE $ep \leftarrow C[c] + rank_{c}(T^{bwt},ep)$ 
    \STATE $i \leftarrow i-1$
\ENDWHILE
\IF{$sp>ep$}
    \STATE \textbf{return} $\emptyset$
\ENDIF
\STATE \textbf{return} $(sp,ep)$
\end{algorithmic}
\caption{Backward Search algorithm (BWS)}
\label{alg:bws}
\end{figure}

\section{Lempel-Ziv Parsings and Repetitions}
\label{sec:lzparsings}
Lempel and Ziv proposed in the seventies a new compression method \cite{LZ76,ZL77,ZL78}. The basic idea is to replace a repeated portion of the text with a pointer to some previous occurrence of that portion. To find the repetitions they keep a dictionary representing all the portions that can be copied. Many variants of these algorithms exist \cite{LZSS,LZW,LZRW} which differ in the way they parse the text or the encoding they use.

The LZ77 \cite{ZL77} parsing is a dictionary-based compression scheme in which the dictionary used is the set of substrings of the preceding text. This definition allows it to get one of the best compression ratios for repetitive texts.

\begin{definition}[\cite{ZL77}]
\label{def:lz77}
The \emph{LZ77 parsing} of text $T[1,n]$ is a sequence $Z[1,n']$ of
\emph{phrases} such that $T = Z[1] Z[2] \ldots Z[n']$, built as follows. 
Assume we
have already processed $T[1,i-1]$ producing the sequence $Z[1,p-1]$. Then, we
find the longest prefix $T[i,i'-1]$ of $T[i,n]$ which occurs in
$T[1,i-1]$,\footnote{The original definition allows the source of
$T[i,i'-1]$ to extend beyond position $i-1$, but we ignore this feature in
this thesis.} set $Z[p] = T[i,i']$ and continue with $i = i'+1$. The
occurrence in $T[1,i-1]$ of the prefix $T[i,i'-1]$ is called the \emph{source} of 
the phrase $Z[p]$.
\end{definition}

Note that each phrase is composed of the content of a source, which can be the 
empty string $\varepsilon$, plus a trailing character. Note also that all 
phrases of the parsing are different, except possibly the last one. To avoid 
that case, a special character \texttt{\$} 
is appended at the end, $T[n]=\texttt{\$}$. 

Typically a phrase is represented as a triple
$Z[p]=(start,len,c)$, where $start$ is the start position of the source, $len$ is the length of the source and $c$ is the trailing character.

\begin{example}
\label{ex:lz77}
Let $T=\texttt{`alabar\_a\_la\_alabarda\$'}$; the LZ77 parsing is as follows:
\begin{equation*}
 \phrase{a}\phrase{l}\phrase{ab}\phrase{ar}\phrase{\_}\phrase{a\_}\phrase{la\_}\phrase{alabard}\phrase{a\$}
\end{equation*}
In this example the seventh phrase copies two characters starting at position 2 and has a trailing character `\texttt{\_}'.
\end{example}

One of the greatest  advantages  of this algorithm 
is the simple and fast scheme of decompression, opposed to the construction algorithm which is more complicated. Decompression runs in linear time by copying the source content referenced by each phrase and then the trailing character. However, random text extraction is not as easy.

The LZ78 \cite{ZL78} compression scheme is also dictionary-based. Its dictionary is the set of all phrases previously produced. Because of this definition of the dictionary the construction process is much simpler than that of LZ77.

\begin{definition}[\cite{ZL78}]
The {\em LZ78 parsing} of text $T[1,n]$ is a sequence $Z[1,n']$ of \emph{phrases} such that $T = Z[1] Z[2] \ldots Z[n']$, built as follows. 
Assume we have already processed $T[1,i-1]$ producing the sequence $Z[1,p-1]$. Then, we find the longest phrase $Z[j]$\diego{, for $j\leq p-1$,} that is a prefix of $T[i,n]$, set $Z[p] = Z[j]T[i+|Z[j]|]$ and continue with $i = i+|Z[j]|+1$.
\end{definition}

Typically a phrase is represented as $Z[p]=(j,c)$, where $j$ is the phrase number of the source and $c$ is the trailing character.

\begin{example}
\label{ex:lz78}
Let $T=\texttt{`alabar\_a\_la\_alabarda\$'}$; the LZ78 parsing is as follows:
\begin{equation*}
 \phrase{a}\phrase{l}\phrase{ab}\phrase{ar}\phrase{\_}\phrase{a\_}\phrase{la}\phrase{\_a}\phrase{lab}\phrase{ard}\phrase{a\$}
\end{equation*}
In this example the ninth phrase copies two characters starting at position 2 and has a trailing character `\texttt{b}'.
\end{example}

With respect to compression, both LZ77 and LZ78 converge to the entropy of stationary ergodic sources \cite{LZ76,ZL78}. It also converges below the empirical entropy (Section \ref{sec:entropy}), as detailed next.

\begin{definition}[\cite{KM99}]
A parsing algorithm is said to be {\em coarsely optimal} if its compression
ratio $\rho(T)$ differs from the $k$-th order empirical entropy $H_k(T)$ by a
quantity depending only on the length of the text and that goes to zero as the
length increases. That is, $\forall k \, \exists f_k, \lim_{n \rightarrow
\infty} f_k(n) = 0$, such that for every text $T$,
$
\rho(T) \le H_k(T)+f_k(|T|).
$
\end{definition}

\begin{theorem}[\cite{KM99,PWZ92}]
The LZ77 and LZ78 parsings are coarsely optimal.
\end{theorem}

As explained in Section \ref{sec:entropy}, however, converging to $H_k(T)$ is not sufficiently good for repetitive texts. Repetitive texts are originated in applications where many similar versions of one base text are generated (i.e., DNA sequences); or where successive versions, each one similar to the preceding one (i.e., wiki), are generated. Statistical compressors are not able to capture this characteristic, because they predict a symbol based only on a short previous context, and such statistics do not change when the text is replicated many times (see Section \ref{sec:entropy} for the relation between $H_k(T)$ and $H_k(TT)$). Compressors based on repetitions, such as Lempel-Ziv parsings or grammar based ones, do exploit this repetitiveness.
\section{Self-Indexes}
\begin{definition}
A \emph{self-index} \cite{NM07} is an index that uses space proportional to that of the compressed text and solves the queries \emph{locate} and \emph{extract}. As this kind of indexes can reproduce any text substring, they replace the original text. Additionally, some indexes provide more efficient ways of computing \emph{exists} and \emph{count} queries.
\end{definition}

There are several general-purpose self-indexes, however most of them \jeremy{do not achieve high compression for repetitive texts, as they are only able to compress up to the $k$-th order empirical entropy (Section \ref{sec:entropy}).}
\djeremy{are based on the $k$-th order empirical entropy model (Section \ref{sec:entropy}), thus they do not achieve high compression for repetitive texts.} 
Most are based on the BWT or suffix array 
 (see \cite{NM07} for a complete survey).
In the last years some self-indexes oriented to repetitive texts have been proposed. We cover these now.

\subsection{Run-Length Compressed Suffix Arrays (RLCSAs)}
\label{sec:rlcsa}
The Run-Length Compressed Suffix Array (RLCSA) \cite{MNSV08} is based on the Compressed Suffix Array of Sadakane \cite{Sad03}. This is built around the so called $\Psi$ function. 
\begin{definition}[\cite{GV00}]
Let $A[1,\ldots,n]$ be the suffix array of a text $T$. Then $\Psi(i)$ is defined as 
\begin{equation*}
\Psi(i)=A^{-1}[(A[i]\!\!\!\!\mod n) + 1]
\end{equation*}
\end{definition}
The $\Psi$ function is the inverse of the LF mapping. $\Psi$ maps suffix $T[A[i],n]$ to suffix $T[A[i]+1,n]$, allowing one to scan the text from left to right. A \emph{run} in the $\Psi$ array is an interval $[a,b]$ for which it holds $\forall i \in [a,b-1],\,\Psi(i+1)=\Psi(i)+1$. 

In the RLCSA, \djeremy{they} \jeremy{one} run-length encodes the differences $\Psi[i]-\Psi[i-1]$ and store absolute samples of the array $\Psi$. This structure is very fast for \emph{count} and \emph{exists} queries. Its major drawback is the sampling it requires for \emph{locate} and \emph{extract} queries, as it takes $(n\log n)/s$ extra bits to achieve locating time $O(s)$, and time $O(s+r-l)$ for $extract(l,r)$, where $s$ is the sampling step.

The number of runs may be much smaller than $nH_k(T)$ (for example $runs(T)=runs(TT)$, whereas $|TT|H_k(TT)\ge 2|T|H_k(T)$ as shown in Section \ref{sec:entropy}). However, the difference between the number of runs and the number of phrases in an LZ77 parsing \cite{ZL77} may be a multiplicative factor as high as $\Theta(\sqrt{n})$.\footnote{Veli M\"{a}kinen, personal communication} For these reasons, the RLCSA seems to be an intermediate solution between LZ77 and empirical-entropy-based indexes. 

\subsection{Indexes based on sparse suffix arrays}
\label{sec:ku}
In this section we present two indexes \cite{KU96_sst,KU96} by K\"{a}rkk\"{a}inen and Ukkonen. Although these are not self-indexes, they set the ground for several self-indexes proposed later.

\begin{itemize}
\item First, they choose some indexing positions of the text. These can be evenly spaced points \cite{KU96_sst} or the points defined by a Lempel-Ziv parsing \cite{KU96}.
\item The suffixes starting at those points are indexed in a suffix trie, and the reversed prefixes in another trie.
\item The index in principle only allows one to find occurrences crossing an indexing point.
\item To find a pattern $P$ of length $m$, they partition it in all $m+1$ combinations of prefix and suffix.
\item For each partition, they search for the suffix in the suffix trie and for the prefix of the pattern in the reverse prefix trie.
\item The previous searches define a 2-dimensional range in a grid that relates each indexed text prefix (in lexicographic order) with the text suffix that follows (in lexicographic order). That is, related prefixes and suffixes are consecutive in the text.
\item A data structure supporting 2-dimensional range queries \cite{Cha88}, finds all pairs of related suffixes and prefixes, finding in this way the actual occurrences.
\item Additionally, using a Lempel-Ziv parsing they are able to find all the occurrences of the pattern. The occurrences are either found in the grid by the process described above (primary occurrences), or by considering the copies detected by the parsing (secondary occurrences), for which an additional method tracking the copies finds the remaining occurrences.
\end{itemize}

All following indexes can be thought as heirs of this general idea, which was improved by replacing or adding additional compact data structures to decrease the space usage. In most cases, the parsing was restricted only to LZ78 (Section \ref{sec:lz78_selfindexes}), since it simplifies the index, and in others to text grammars (SLPs, Section \ref{sec:slps}). In the following two subsections we list the results obtained in those cases. This thesis can also be thought as a heir of this fundamental scheme: For the first time compact data structures supporting the LZ77 parsing have been developed in this thesis, which show better performance on repetitive texts.

\subsection{LZ78-based Self-Indexes}
\label{sec:lz78_selfindexes}
In this section we present the space and running times of two indexes based on LZ78. Although they offer decent upper bounds and competitive performance on typical texts, experiments \cite{MNSV08} have demonstrated that LZ78 is too weak to profit from highly repetitive texts. There are other such self-indexes \cite{FM05}, not implemented as far as we know.
\subsubsection{Arroyuelo \etal's LZ-Index}
\label{sec:lz78index}
Navarro's LZ-Index \cite{Nav02} is the first self-index based on the LZ78 parsing \diego{using $O(nH_k(T))$ bits of space (it is also the first implemented in practice)}. It uses $4n'\log n'(1+o(1))$ bits and takes $O(m^3\log \sigma + (m+occ)\log n')$ time to \emph{locate} the $occ$ occurrences of a pattern of length $m$, where $\sigma$ is the size of the alphabet, and $n'$ is the number of phrases of the parsing.

Arroyuelo \diego{\etal} \ddiego{and Navarro} later improved the time and space of the index, achieving $(2+\epsilon)n'\log n'(1+o(1))$ bits and $O(m^2+(m+occ)\log n')$ locate time \ddiego{\cite{AN06}} \diego{\cite{ANSalgor10}}, or $(3+\epsilon)n'\log n'(1+o(1))$ bits and $O((m+occ)\log n')$ locate time \cite{AN07}.

\subsubsection{Russo and Oliveira's ILZI}
\label{sec:ilzi}
Russo and Oliveira present a self-index based on the so-called \emph{maximal parsing}, called ILZI \cite{ilzi}.
\begin{definition}[\cite{ilzi}]
Given a suffix trie $\mathcal{T}$ (of a set of strings), the \emph{$\mathcal{T}$-maximal parsing} of string $T$ is the sequence of nodes $v_1, \ldots , v_f$ such
that $T = v_1 \ldots v_f$ and, for every $j$, $v_j$ is the largest prefix of $v_j\ldots v_f$ that is a node of $T$.
\end{definition}

First, they compute the LZ78 parsing of $T^{rev}$, and then generate a suffix tree $\mathcal{T}_{78}$ over the set of the reverse phrases. Next they build the maximal parsing of $T$ using $\mathcal{T}_{78}$. \djeremy{They use this} \jeremy{This} parsing \djeremy{as it} improves the compression of LZ78, as shown by the following lemma. 

\begin{lemma}[\cite{ilzi}]
If the number of \ddiego{blocks} \diego{phrases} of the LZ78 parsing of $T$ is $n'$ then the $\mathcal{T}_{78}$-maximal parsing of $T$ has at most $n'$ \ddiego{blocks} \diego{phrases}.
\end{lemma}

Their index uses at most $5n'\log n'(1+o(1))$ bits and takes $O((m+occ)\log n')$ time to locate the $occ$ occurrences of a pattern of length $m$ ($n'$ is the number of blocks of the maximal parsing).
\subsection{Straight Line Programs (SLPs)}
\label{sec:slps}
Claude and Navarro \cite{CN09} propose\jeremy{d} a self-index based on \emph{straight-line programs} (SLPs). SLPs are grammars in which the rules are either $X_i \rightarrow \alpha \in \Sigma$ or $X_i \rightarrow X_lX_r, \, \text{for}\,\, l,r<i$. The LZ78 \cite{ZL78} parsing may produce an output exponentially larger than the smallest SLP. However, the LZ77 \cite{ZL77} parsing outperforms the smallest SLP \cite{CLLP+05}. On the other hand producing the smallest SLP is an NP-complete problem \cite{Rytter03,CLLP+05}. However, Rytter \cite{Rytter03} has shown how to generate in linear time a grammar using $O(\ell \log \ell)$ rules and height $O(\log \ell)$, where $\ell$ is the size of the LZ77 parsing. Again, SLPs are intermediate between LZ77 and other methods.

The index \cite{CN09} uses $n'\log n + O(n'\log n')$ bits of space, where $n'$ is the number of rules of the grammar. It solves $extract(l,r)$ in $O((r-l + h) \log n')$ time and \emph{locate} in $O((m(m + h) + h \cdot occ) \log n')$ time, where $h$ is the height of the derivation tree of the grammar and $m$ the length of the pattern.

Claude \etal ~\cite{CFMPNbibe10} evaluate\jeremy{d} a practical implementation using the grammar produced by Re-Pair \cite{LM00}. The results are competitive with the RLCSA only for extremely repetitive texts and short patterns.

%% file: repetitive_corpus.tex
\chapter{A Repetitive Corpus Testbed}
In this chapter we present a corpus of repetitive texts. These texts are categorized according to the source they come from into the following: Artificial Texts, Pseudo-Real Texts and Real Texts. The main goal of this collection is to serve as a standard testbed for benchmarking algorithms oriented to repetitive texts. The corpus can be downloaded from \corpusurl.
\section{Artificial Texts}
This subset is composed of highly repetitive texts that do not come from any real-life source, but are artificially generated through some mathematical definition and have interesting combinatorial properties.

\subsection{\texorpdfstring{Fibonacci Sequence ($F_n$)}{Fibonacci Sequence}}
This sequence is defined by the recurrence
\begin{align}
F_1 & = \texttt{0}\nonumber \\
F_2 & = \texttt{1}\nonumber \\
F_n & = F_{n-1}F_{n-2}
\end{align}
The length of the string $F_n$ is the Fibonacci number $f_{n}$ and the sequence is a \emph{sturmian word} \cite{ACW}, which means it has $i+1$ different substrings (factors) of length $i$.

\subsection{\texorpdfstring{Thue-Morse Sequence ($T_n$)}{Thue-Morse Sequence}}
This sequence \cite{thuemorse} is defined by the recurrence 
\begin{align}
T_1 & = \texttt{0}\nonumber \\
T_n & = T_{n-1}\overline{T_{n-1}}
\end{align}
where $\bar{F}$ is the bitwise negation operator (i.e., all \texttt{0} get converted to \texttt{1} and all \texttt{1} to \texttt{0}).
Because of the construction scheme of this sequence, there are many substrings of the form $XX$, for any string $X$. However, there are no overlapping squares, i.e., substrings of the form $\texttt{0}X\texttt{0}X\texttt{0}$ or $\texttt{1}X\texttt{1}X\texttt{1}$. Furthermore, this sequence is strongly cube-free, i.e., there are no substrings of the form $XXx$, where $x$ is the first character of the string $X$. Another interesting property of this string is that it is recurrent. 
That is, given any finite substring $w$ of length $n$, there is some length $n_w$ (often much longer than $n$) such that $w$ is contained in every substring of length $n_w$.
The length of these strings is $|T_n|=2^n$.


\subsection{\texorpdfstring{Run-Rich String Sequence ($R_n$)}{Run-Rich String Sequence}}
A measure of string complexity, related to the regularities of the text and strongly related to the LZ77 parsing \cite{runslz}, is the number of runs.
\jeremy{
\begin{definition}
A \emph{period} of string $T[1,n]$ is a positive integer $p$ holding that $\forall\, 1\le i \le n-p,\, T[i]=T[i+p]$. A string is said to be \emph{periodic} if its minimum period $p$ is such that $p \leq n/2$.  
\end{definition}
}
\begin{definition}[\cite{Main89}]
The substring $T[i,j]$ is a \emph{run} in a string $T$ iff $T[i,j]$ is periodic \dmine{the minimum period of $T[i,j]=p\leq |T[i,j]|/2$} and $T[i,j]$ is not extendable to the right ($j=n$ or $T[j+1]\neq T[j-p+1]$) or left ($i=1$ or $T[i-1]\neq T[i+p-1]$). 
\end{definition}
The higher the number of runs in a string, the more regularities it has.

It has been shown that the maximum number of runs in a string is greater than $0.944n$ \cite{MKIBS08} and lower than $1.029n$ \cite{CIT08}.
Franek \etal ~\cite{rich} show a constructive and simple way to obtain strings with many runs; the $n$-th of those strings is denoted $R_n$. The ratio of the runs of their strings to the length approaches $3/(1+\sqrt{5})=0.92705\ldots$.

\section{Pseudo-Real Texts}
Here we present a set of texts that were generated by artificially adding repetitiveness to real texts, thus we call them \emph{pseudo-real texts}.

To generate the texts, we took a prefix of 1MiB of all texts of Pizza\&Chili Corpus\footnote{\url{http://pizzachili.dcc.uchile.cl}}\diego{,} \ddiego{and} we mutated them\diego{, and we concatenated all of them in the order they were generated}. Our mutations take a random character position and change it to a random character different from the original one. 

We used two different schemes for the mutations. The first one, denoted by a `$^1$', generates different mutations of the first text. The second, denoted by a `$^2$', mutates the last text generated. The second scheme resembles the changes obtained through time in a software project or the versions of a document\jeremy{, while the first scheme produces changes analogous to the ones found in a collection of related DNA sequences}. 

The mutation rate, i.e., percentage of mutated characters, was set to $0.1\%$, $0.01\%$ and $0.001\%$.

The base texts (all from the Pizza\&Chili corpus) we mutated were the following:
\begin{itemize}
\item Sources: This file is formed by C/Java source code obtained by concatenating all the \texttt{.c}, \texttt{.h}, \texttt{.C} and \texttt{.java} files of the linux-2.6.11.6 and gcc-4.0.0  distributions.
\item Pitches: This file is a sequence of midi pitch values (bytes in 0-127, plus a few extra special values) obtained from a myriad of MIDI files freely available on Internet. 
\item Proteins: This file is a sequence of newline-separated protein sequences obtained from the Swissprot database. 
\item DNA: This file is a sequence of newline-separated gene DNA sequences obtained from files \texttt{01hgp10} to \texttt{21hgp10}, plus \texttt{0xhgp10} and \texttt{0yhgp10}, from Gutenberg Project. 
\item English: This file is the concatenation of English text files selected from \texttt{etext02} to \texttt{etext05} collections of Gutenberg Project. 
\item XML: This file is an XML that provides bibliographic information on major computer science journals and proceedings and it was obtained from \url{http://dblp.uni-trier.de}. 
\end{itemize}
\section{Real Texts}
This subset is composed of texts coming from real repetitive sources. These sources are DNA, Wikipedia Articles, Source Code, and Documents.

For the case of DNA we concatenated the texts in random order. For the others, we concatenated the texts according to the date they were created, from oldest to newest. 

\subsection{DNA}
Our DNA texts come from different sources. 
\begin{itemize}
\item The Saccharomyces Genome Resequencing Project\footnote{\url{http://www.sanger.ac.uk/Teams/Team71/durbin/sgrp}} provides two text collections: \emph{para}, which contains 36 sequences of \emph{Saccharomyces Paradoxus} and \emph{cere}, which contains 37 sequences of \emph{Saccharomyces Cerevisiae}.
\item From the National Center for Biotechnology Information (NCBI){\footnote{\url{http://www.ncbi.nlm.nih.gov}}} we collected some DNA sequences of the same bacteria. The species we collected are \emph{Escherichia Coli} (23), \emph{Salmonella Enterica} (15), \emph{Staphylococcus Aureus} (14), \emph{Streptococcus Pyogenes} (13), \emph{Streptococcus Pneumoniae} (11) and \emph{Clostridium Botulium} (10). \jeremy{We wrote in parentheses the total number of sequences of each collection.} We chose these species as they were the only ones with 10 or more different sequences.
\item A collection composed of 78,041 sequences of \emph{Haemophilus Influenzae}\footnote{\url{ftp://ftp.ncbi.nih.gov/genomes/INFLUENZA/influenza.fna.gz}}, also coming from the NCBI.
\end{itemize}
\begin{remark}Although there are four bases $\lbrace \texttt{A},\texttt{C},\texttt{G},\texttt{T}\rbrace$, DNA sequences may have alphabets of size up to $16=2^4$ because some characters denote an unknown choice among the four bases. The most common character used is $\texttt{N}$, which denotes a totally unknown symbol.
\end{remark}

\subsection{Wikipedia Articles}
We downloaded all versions of three Wikipedia articles, \emph{Albert Einstein}, \emph{Alan Turing} and \emph{Nobel Prize}. 
We downloaded them in English (denoted \emph{en}) and German (denoted \emph{de}). We chose these languages as they are among the most widely used on Internet and their alphabet may be represented using standard 1-byte encodings. The versions for all documents are up to January 12, 2010, except for the English article of \emph{Albert Einstein}, which was downloaded only up to November 10, 2006 because of the massive number of versions it has.
\subsection{Source Code}
We collected all versions 5.x of the \emph{Coreutils}\footnote{\url{ftp://mirrors.kernel.org/gnu/coreutils}} package and removed all binary files, making a total of 9 versions. We also collected all 1.0.x and 1.1.x versions of the \emph{Linux Kernel}\footnote{\url{ftp://ftp.kernel.org/pub/linux/kernel}}, making a total of 36 versions.
\subsection{Documents}
We took all \emph{pdf} files of CIA World Leaders\footnote{\url{https://www.cia.gov/library/publications/world-leaders-1}}
 ~from January 2003 to December 2009, and converted them to text (using software \texttt{pdftotext}).
\section{Statistics}
To understand the characteristics of the texts present in the \emph{Repetitive Corpus}, we provide below some statistics about them. The statistics presented are the following:

\begin{itemize}
\item \textbf{Alphabet Size:} 
We give the alphabet size and the inverse probability of matching (IPM), which is the inverse of the probability that two characters chosen at random match. IPM is a measure of the effective alphabet size. On a uniformly distributed text, it is precisely the alphabet size.

\item \textbf{Compression Ratio:}
Since we are dealing with compressed indexes it is useful to have an idea of the compressibility of the texts using general-purpose compressors. 
The following compressors are used: \texttt{gzip}\footnote{\url{http://www.gzip.org}} gives an idea of compressibility via dictionaries (an LZ77 parsing with limited window size); \texttt{bzip2}\footnote{\url{http://www.bzip.org}} gives an idea of block-sorting compressibility (using the BWT transform, Section \ref{sec:bwt}); \texttt{ppmdi}\footnote{\url{http://pizzachili.dcc.uchile.cl/utils/ppmdi.tar.gz}} gives an idea of partial-match-based compressors (related to the $k$-th order entropy, Section \ref{sec:entropy}); \texttt{p7zip}\footnote{\url{http://www.7-zip.org}} gives an idea of LZ77 compression with virtually unlimited window; and \texttt{Re-Pair}\footnote{\url{http://www.cbrc.jp/~rwan/en/restore.html}} \cite{LM00} gives an idea of grammar-based compression. All compressors were run with the highest compression options.

\item \textbf{Empirical Entropy:} Here we give the empirical entropy $H_k$ of the text with $k$ ranging from $0$ to $8$, measured as compression ratio. 
We also show, in parentheses, the \emph{complexity function} of $T$ \cite{ACW} (or the \emph{number of contexts}) which count how many different substrings of a given size does $T$ have. This is exactly our $C(T,k)$ of Lemma \ref{lemma:entropytt}.
This measure has the following properties:
\begin{eqnarray*}
C(T,1) & = & \sigma\\
C(T,n+m) & \leq & C(T,n)C(T,m)
\end{eqnarray*}
The lower this measure, the more repetitive the text is. For example, if $C(T,n)=1\, \forall n$, then $T=c^{m}$ for some character $c$. When $P(C,n) = n+1$ the sequence is said to be \emph{Sturmian} (the Fibonacci sequence is an example of a \emph{Sturmian} string).
\end{itemize}
\begin{remark}
The compression ratios are given as the percentage of the compressed file size over the uncompressed file
size, assuming the original file uses one byte per character. This means that 25\% compression can be
achieved over a DNA sequence having an alphabet \texttt{\{A,C,G,T\}} by simply using 2 bits per symbol. As seen from the real-life examples
given, these four symbols are usually predominant, so it is not hard to get very close to 25\% on
general DNA sequences as well.
\end{remark}

\subsection{Artificial Texts}
Tables \ref{tab:art_sigma}-\ref{tab:art_entropy} give the statistics of artificial texts.
\begin{table}[!ht]
\begin{center}
\begin{scriptsize}
\begin{tabular}{|l|r|r|r|}
\hline
\textbf{File} & \multicolumn{1}{c|}{\textbf{Size}} & \multicolumn{1}{c|}{$\Sigma$} & \multicolumn{1}{c|}{\textbf{IPM}} \\ \hline
$F_{41}$ & 256MiB & 2 & 1.894 \\ \hline
$T_{29}$ & 256MiB & 2 & 2.000 \\ \hline
$R_{13}$ & 207MiB & 2 & 2.000 \\ \hline
\end{tabular}
\end{scriptsize}
\end{center}
\caption{Alphabet statistics for Artificial Collection}
\label{tab:art_sigma}
\end{table}

\begin{table}[!ht]
\begin{center}
\begin{scriptsize}
\begin{tabular}{|l|c|c|c|c|c|}
\hline
\textbf{File} & \multicolumn{1}{c|}{\textbf{p7zip}} & \multicolumn{1}{c|}{\textbf{bzip2}} & \multicolumn{1}{c|}{\textbf{gzip}} & \multicolumn{1}{c|}{\textbf{ppmdi}} & \multicolumn{1}{c|}{\textbf{Re-pair}} \\ \hline
$F_{41}$ & 0.17624\% & 0.00572\% & 0.46875\% & 1.87500\% & 0.00002\% \\ \hline
$T_{29}$ & 0.35896\% & 0.01259\% & 0.54688\% & 2.18750\% & 0.00004\% \\ \hline
$R_{13}$ & 0.17172\% & 0.01227\% & 0.53140\% & 2.12560\% & 0.00009\% \\ \hline
\end{tabular}
\end{scriptsize}
\end{center}
\caption{Compression statistics for Artificial Collection}
\label{tab:art_compression}
\end{table}

\begin{table}[!ht]
\begin{center}
\begin{scriptsize}
\begin{tabular}{|l|c|c|c|c|c|c|c|c|c|}
\hline
\textbf{File} & \textbf{$H_0$} & \textbf{$H_1$} & \textbf{$H_2$} & \textbf{$H_3$} & \textbf{$H_4$} & \textbf{$H_5$} & \textbf{$H_6$} & \textbf{$H_7$} & \textbf{$H_8$} \\ \hline
\multirow{2}{*}{$F_{41}$} & 11.99\% & 7.41\% & 4.58\% & 4.58\% & 2.83\% & 2.83\% & 2.83\% & 1.75\% & 1.75\% \\
  & (1) & (2) & (3) & (4) & (5) & (6) & (7) & (8) & (9) \\ \hline
\multirow{2}{*}{$T_{29}$} & 12.50\% & 11.48\% & 8.34\% & 8.34\% & 4.16\% & 4.16\% & 4.16\% & 2.09\% & 2.09\% \\
  & (1) & (2) & (4) & (6) & (10) & (12) & (16) & (20) & (22) \\ \hline
\multirow{2}{*}{$R_{13}$} & 12.50\% & 9.85\% & 8.51\% & 6.55\% & 2.56\% & 2.33\% & 2.33\% & 2.33\% & 2.33\% \\
  & (1) & (2) & (4) & (6) & (8) & (10) & (12) & (14) & (16) \\ \hline
\end{tabular}
\end{scriptsize}
\end{center}
\caption{Empirical entropy statistics for Artificial Collection}
\label{tab:art_entropy}
\end{table}

\clearpage
\subsection{Pseudo-Real Texts}
Tables \ref{tab:pseudoreal_sigma}-\ref{tab:pseudoreal2_entropy} give the statistics of pseudo-real texts.
\begin{table}[!ht]
\begin{center}
\begin{scriptsize}
\begin{tabular}{|l|r|r|r|}
\hline
\textbf{File} & \multicolumn{1}{c|}{\textbf{Size}} & \multicolumn{1}{c|}{$\Sigma$} & \multicolumn{1}{c|}{\textbf{IPM}} \\ \hline
Xml 0.001\%$^1$ & 100MiB & 89 & 27.84 \\ \hline
Xml 0.01\%$^1$ & 100MiB & 89 & 27.84 \\ \hline
Xml 0.1\%$^1$ & 100MiB & 89 & 27.84 \\ \hline
DNA 0.001\%$^1$ & 100MiB & 5 & 3.98 \\ \hline
DNA 0.01\%$^1$ & 100MiB & 5 & 3.98 \\ \hline
DNA 0.1\%$^1$ & 100MiB & 5 & 3.98 \\ \hline
English 0.001\%$^1$ & 100MiB & 106 & 15.65 \\ \hline
English 0.01\%$^1$ & 100MiB & 106 & 15.65 \\ \hline
English 0.1\%$^1$ & 100MiB & 106 & 15.65 \\ \hline
Pitches 0.001\%$^1$ & 100MiB & 73 & 33.07 \\ \hline
Pitches 0.01\%$^1$ & 100MiB & 73 & 33.07 \\ \hline
Pitches 0.1\%$^1$ & 100MiB & 73 & 33.07 \\ \hline
Proteins 0.001\%$^1$ & 100MiB & 21 & 16.90 \\ \hline
Proteins 0.01\%$^1$ & 100MiB & 21 & 16.90 \\ \hline
Proteins 0.1\%$^1$ & 100MiB & 21 & 16.90 \\ \hline
Sources 0.001\%$^1$ & 100MiB & 98 & 28.86 \\ \hline
Sources 0.01\%$^1$ & 100MiB & 98 & 28.86 \\ \hline
Sources 0.1\%$^1$ & 100MiB & 98 & 28.86 \\ \hline
\end{tabular}
\end{scriptsize}
\end{center}
\caption{Alphabet statistics for Pseudo-Real Collection (Scheme 1)}
\label{tab:pseudoreal_sigma}
\end{table}

\begin{table}[!ht]
\begin{center}
\begin{scriptsize}
\begin{tabular}{|l|r|r|r|}
\hline
\textbf{File} & \multicolumn{1}{c|}{\textbf{Size}} & \multicolumn{1}{c|}{$\Sigma$} & \multicolumn{1}{c|}{\textbf{IPM}} \\ \hline
Xml 0.001\%$^2$& 100MiB & 89 & 27.84 \\ \hline
Xml 0.01\%$^2$& 100MiB & 89 & 27.84 \\ \hline
Xml 0.1\%$^2$& 100MiB & 89 & 27.86 \\ \hline
DNA 0.001\%$^2$& 100MiB & 5 & 3.98 \\ \hline
DNA 0.01\%$^2$& 100MiB & 5 & 3.98 \\ \hline
DNA 0.1\%$^2$& 100MiB & 5 & 3.98 \\ \hline
English 0.001\%$^2$& 100MiB & 106 & 15.65 \\ \hline
English 0.01\%$^2$& 100MiB & 106 & 15.66 \\ \hline
English 0.1\%$^2$& 100MiB & 106 & 15.74 \\ \hline
Pitches 0.001\%$^2$& 100MiB & 73 & 33.07 \\ \hline
Pitches 0.01\%$^2$& 100MiB & 73 & 33.07 \\ \hline
Pitches 0.1\%$^2$& 100MiB & 73 & 33.10 \\ \hline
Proteins 0.001\%$^2$& 100MiB & 21 & 16.90 \\ \hline
Proteins 0.01\%$^2$& 100MiB & 21 & 16.90 \\ \hline
Proteins 0.1\%$^2$& 100MiB & 21 & 16.92 \\ \hline
Sources 0.001\%$^2$& 100MiB & 98 & 28.86 \\ \hline
Sources 0.01\%$^2$& 100MiB & 98 & 28.86 \\ \hline
Sources 0.1\%$^2$& 100MiB & 98 & 28.92 \\ \hline
\end{tabular}
\end{scriptsize}
\end{center}
\caption{Alphabet statistics for Pseudo-Real Collection (Scheme 2)}
\label{tab:pseudoreal2_sigma}
\end{table}

\begin{table}[!ht]
\begin{center}
\begin{scriptsize}
\begin{tabular}{|l|r|r|r|r|r|}
\hline
\textbf{File} & \multicolumn{1}{c|}{\textbf{p7zip}} & \multicolumn{1}{c|}{\textbf{bzip2}} & \multicolumn{1}{c|}{\textbf{gzip}} & \multicolumn{1}{c|}{\textbf{ppmdi}} & \multicolumn{1}{c|}{\textbf{Re-Pair}} \\ \hline
Xml 0.001\%$^1$ & 0.15\% & 11.00\% & 18.00\% & 3.50\% & 0.19\% \\ \hline
Xml 0.01\%$^1$ & 0.18\% & 12.00\% & 18.00\% & 3.60\% & 0.46\% \\ \hline
Xml 0.1\%$^1$ & 0.46\% & 12.00\% & 18.00\% & 4.10\% & 2.00\% \\ \hline
DNA 0.001\%$^1$ & 0.27\% & 27.00\% & 28.00\% & 11.00\% & 0.34\% \\ \hline
DNA 0.01\%$^1$ & 0.29\% & 27.00\% & 28.00\% & 11.00\% & 0.58\% \\ \hline
DNA 0.1\%$^1$ & 0.51\% & 27.00\% & 28.00\% & 12.00\% & 2.50\% \\ \hline
English 0.001\%$^1$ & 0.31\% & 28.00\% & 37.00\% & 22.00\% & 0.39\% \\ \hline
English 0.01\%$^1$ & 0.35\% & 28.00\% & 37.00\% & 22.00\% & 0.65\% \\ \hline
English 0.1\%$^1$ & 0.59\% & 28.00\% & 37.00\% & 22.00\% & 2.70\% \\ \hline
Pitches 0.001\%$^1$ & 0.47\% & 54.00\% & 52.00\% & 47.00\% & 0.69\% \\ \hline
Pitches 0.01\%$^1$ & 0.50\% & 54.00\% & 52.00\% & 47.00\% & 0.95\% \\ \hline
Pitches 0.1\%$^1$ & 0.75\% & 54.00\% & 52.00\% & 48.00\% & 3.20\% \\ \hline
Proteins 0.001\%$^1$ & 0.32\% & 41.00\% & 39.00\% & 31.00\% & 0.42\% \\ \hline
Proteins 0.01\%$^1$ & 0.35\% & 41.00\% & 39.00\% & 31.00\% & 0.68\% \\ \hline
Proteins 0.1\%$^1$ & 0.59\% & 41.00\% & 39.00\% & 32.00\% & 2.70\% \\ \hline
Sources 0.001\%$^1$ & 0.20\% & 19.00\% & 25.00\% & 12.00\% & 0.28\% \\ \hline
Sources 0.01\%$^1$ & 0.23\% & 19.00\% & 25.00\% & 12.00\% & 0.56\% \\ \hline
Sources 0.1\%$^1$ & 0.50\% & 20.00\% & 25.00\% & 13.00\% & 2.60\% \\ \hline
\end{tabular}
\end{scriptsize}
\end{center}
\caption{Compression statistics for Pseudo-Real Collection (Scheme 1)}
\label{tab:pseudoreal_compression}
\end{table}

\begin{table}[!ht]
\begin{center}
\begin{scriptsize}
\begin{tabular}{|l|r|r|r|r|r|}
\hline
\textbf{File} & \multicolumn{1}{c|}{\textbf{p7zip}} & \multicolumn{1}{c|}{\textbf{bzip2}} & \multicolumn{1}{c|}{\textbf{gzip}} & \multicolumn{1}{c|}{\textbf{ppmdi}} & \multicolumn{1}{c|}{\textbf{Re-Pair}} \\ \hline
Xml 0.001\%$^2$ & 0.15\% & 12.00\% & 18.00\% & 3.50\% & 0.18\% \\ \hline
Xml 0.01\%$^2$ & 0.18\% & 14.00\% & 19.00\% & 4.40\% & 0.39\% \\ \hline
Xml 0.1\%$^2$ & 0.39\% & 25.00\% & 29.00\% & 17.00\% & 2.20\% \\ \hline
DNA 0.001\%$^2$ & 0.26\% & 27.00\% & 28.00\% & 11.00\% & 0.33\% \\ \hline
DNA 0.01\%$^2$ & 0.29\% & 27.00\% & 28.00\% & 11.00\% & 0.52\% \\ \hline
DNA 0.1\%$^2$ & 0.46\% & 27.00\% & 28.00\% & 13.00\% & 2.20\% \\ \hline
English 0.001\%$^2$ & 0.31\% & 28.00\% & 37.00\% & 22.00\% & 0.38\% \\ \hline
English 0.01\%$^2$ & 0.34\% & 29.00\% & 37.00\% & 23.00\% & 0.59\% \\ \hline
English 0.1\%$^2$ & 0.55\% & 38.00\% & 43.00\% & 31.00\% & 2.50\% \\ \hline
Pitches 0.001\%$^2$ & 0.46\% & 54.00\% & 52.00\% & 47.00\% & 0.68\% \\ \hline
Pitches 0.01\%$^2$ & 0.49\% & 54.00\% & 53.00\% & 48.00\% & 0.89\% \\ \hline
Pitches 0.1\%$^2$ & 0.71\% & 59.00\% & 57.00\% & 52.00\% & 2.80\% \\ \hline
Proteins 0.001\%$^2$ & 0.31\% & 41.00\% & 39.00\% & 32.00\% & 0.41\% \\ \hline
Proteins 0.01\%$^2$ & 0.34\% & 42.00\% & 40.00\% & 33.00\% & 0.62\% \\ \hline
Proteins 0.1\%$^2$ & 0.54\% & 47.00\% & 46.00\% & 40.00\% & 2.50\% \\ \hline
Sources 0.001\%$^2$ & 0.20\% & 20.00\% & 25.00\% & 13.00\% & 0.27\% \\ \hline
Sources 0.01\%$^2$ & 0.23\% & 21.00\% & 26.00\% & 14.00\% & 0.49\% \\ \hline
Sources 0.1\%$^2$ & 0.44\% & 34.00\% & 35.00\% & 26.00\% & 2.50\% \\ \hline
\end{tabular}
\end{scriptsize}
\end{center}
\caption{Compression statistics for Pseudo-Real Collection (Scheme 2)}
\label{tab:pseudoreal2_compression}
\end{table}

\begin{table}[!ht]
\begin{center}
\begin{scriptsize}
\begin{tabular}{|l|c|c|c|c|c|c|c|c|c|}
\hline
\textbf{File} & \textbf{$H_0$} & \textbf{$H_1$} & \textbf{$H_2$} & \textbf{$H_3$} & \textbf{$H_4$} & \textbf{$H_5$} & \textbf{$H_6$} & \textbf{$H_7$} & \textbf{$H_8$} \\ \hline
Xml & 65.25\% & 38.63\% & 21.00\% & 12.50\% & 8.13\% & 6.00\% & 5.25\% & 4.75\% & 4.13\% \\
0.001\%$^1$ & (1) & (89) & (3325) & (20560) & (56120) & (98084) & (134897) & (168846) & (200451) \\ \hline
Xml & 65.25\% & 38.63\% & 21.00\% & 12.50\% & 8.13\% & 6.00\% & 5.25\% & 4.75\% & 4.13\% \\
0.01\%$^1$ & (1) & (89) & (4135) & (30975) & (79379) & (131811) & (177924) & (220923) & (261651) \\ \hline
Xml & 65.25\% & 38.75\% & 21.25\% & 12.75\% & 8.25\% & 6.13\% & 5.38\% & 4.88\% & 4.25\% \\
0.1\%$^1$ & (1) & (89) & (5251) & (67479) & (196554) & (326296) & (440199) & (550570) & (661284) \\ \hline
DNA & 25.00\% & 24.25\% & 24.13\% & 24.00\% & 24.00\% & 23.75\% & 23.50\% & 22.88\% & 21.25\% \\
0.001\%$^1$ & (1) & (5) & (18) & (67) & (260) & (1029) & (4102) & (16349) & (62437) \\ \hline
DNA & 25.00\% & 24.25\% & 24.13\% & 24.00\% & 24.00\% & 23.75\% & 23.50\% & 22.88\% & 21.25\% \\
0.01\%$^1$ & (1) & (5) & (18) & (67) & (260) & (1029) & (4102) & (16368) & (63204) \\ \hline
DNA & 25.00\% & 24.25\% & 24.13\% & 24.00\% & 24.00\% & 23.75\% & 23.50\% & 22.88\% & 21.38\% \\
0.1\%$^1$ & (1) & (5) & (19) & (70) & (264) & (1034) & (4109) & (16399) & (65168) \\ \hline
English & 57.25\% & 45.13\% & 34.75\% & 25.88\% & 19.88\% & 15.88\% & 12.50\% & 9.63\% & 7.25\% \\
0.001\%$^1$ & (1) & (106) & (2659) & (18352) & (63299) & (145194) & (256838) & (379514) & (501400) \\ \hline
English & 57.25\% & 45.13\% & 34.75\% & 25.88\% & 19.88\% & 15.88\% & 12.50\% & 9.63\% & 7.25\% \\
0.01\%$^1$ & (1) & (106) & (3243) & (24063) & (82896) & (180401) & (305292) & (439387) & (572056) \\ \hline
English & 57.25\% & 45.25\% & 34.88\% & 26.13\% & 20.13\% & 16.00\% & 12.50\% & 9.75\% & 7.25\% \\
0.1\%$^1$ & (1) & (106) & (4491) & (46116) & (190765) & (439130) & (715127) & (983435) & (1237512) \\ \hline
Pitches & 66.13\% & 61.00\% & 53.50\% & 37.13\% & 16.38\% & 6.25\% & 2.88\% & 1.38\% & 0.75\% \\
0.001\%$^1$ & (1) & (73) & (3549) & (73664) & (376958) & (642406) & (767028) & (833456) & (871970) \\ \hline
Pitches & 66.13\% & 61.00\% & 53.50\% & 37.25\% & 16.38\% & 6.25\% & 2.88\% & 1.38\% & 0.75\% \\
0.01\%$^1$ & (1) & (73) & (3581) & (76900) & (399435) & (684445) & (821533) & (898126) & (946219) \\ \hline
Pitches & 66.13\% & 61.13\% & 53.63\% & 37.38\% & 16.63\% & 6.38\% & 2.88\% & 1.50\% & 0.88\% \\
0.1\%$^1$ & (1) & (73) & (3733) & (95838) & (598394) & (1096014) & (1363610) & (1543086) & (1687166) \\ \hline
Proteins & 52.25\% & 52.13\% & 51.63\% & 47.50\% & 25.13\% & 4.63\% & 0.75\% & 0.25\% & 0.25\% \\
0.001\%$^1$ & (1) & (21) & (422) & (8045) & (128975) & (463357) & (572530) & (589356) & (595906) \\ \hline
Proteins & 52.25\% & 52.13\% & 51.63\% & 47.50\% & 25.13\% & 4.63\% & 0.75\% & 0.25\% & 0.25\% \\
0.01\%$^1$ & (1) & (21) & (422) & (8045) & (131064) & (494845) & (626269) & (654067) & (670075) \\ \hline
Proteins & 52.25\% & 52.13\% & 51.63\% & 47.50\% & 25.50\% & 4.88\% & 0.88\% & 0.38\% & 0.38\% \\
0.1\%$^1$ & (1) & (21) & (425) & (8076) & (143879) & (768510) & (1150595) & (1293347) & (1403589) \\ \hline
Sources & 68.75\% & 46.88\% & 30.00\% & 19.63\% & 14.38\% & 11.00\% & 8.38\% & 6.88\% & 5.75\% \\
0.001\%$^1$ & (1) & (98) & (4557) & (29667) & (75316) & (130527) & (194105) & (259413) & (320468) \\ \hline
Sources & 68.75\% & 46.88\% & 30.00\% & 19.63\% & 14.38\% & 11.00\% & 8.50\% & 6.88\% & 5.75\% \\
0.01\%$^1$ & (1) & (98) & (5621) & (42303) & (102977) & (170525) & (244755) & (320237) & (391260) \\ \hline
Sources & 68.75\% & 47.00\% & 30.25\% & 19.88\% & 14.63\% & 11.13\% & 8.50\% & 7.00\% & 5.88\% \\
0.1\%$^1$ & (1) & (98) & (7359) & (104679) & (299799) & (498046) & (687941) & (872189) & (1049051) \\ \hline

\end{tabular}
\end{scriptsize}
\end{center}
\caption{Empirical entropy statistics for Pseudo-Real Collection (Scheme 1)}
\label{tab:pseudoreal_entropy}
\end{table}

\begin{table}[htbp]
\begin{center}
\begin{scriptsize}
\begin{tabular}{|l|c|c|c|c|c|c|c|c|c|}
\hline
\textbf{File} & \textbf{$H_0$} & \textbf{$H_1$} & \textbf{$H_2$} & \textbf{$H_3$} & \textbf{$H_4$} & \textbf{$H_5$} & \textbf{$H_6$} & \textbf{$H_7$} & \textbf{$H_8$} \\ \hline
Xml & 65.25\% & 38.63\% & 21.13\% & 12.63\% & 8.13\% & 6.00\% & 5.25\% & 4.75\% & 4.13\% \\
0.001\%$^2$ & (1) & (89) & (3325) & (20560) & (56120) & (98084) & (134897) & (168846) & (200451) \\ \hline
Xml & 65.25\% & 39.38\% & 22.00\% & 13.25\% & 8.63\% & 6.50\% & 5.63\% & 5.13\% & 4.50\% \\
0.01\%$^2$ & (1) & (89) & (4135) & (31042) & (79630) & (132163) & (178388) & (221499) & (262329) \\ \hline
Xml & 65.25\% & 44.00\% & 28.75\% & 18.50\% & 12.25\% & 9.25\% & 8.00\% & 7.13\% & 6.25\% \\
0.1\%$^2$ & (1) & (89) & (5255) & (72227) & (226418) & (378994) & (513539) & (645141) & (777226) \\ \hline
DNA & 25.00\% & 24.25\% & 24.13\% & 24.00\% & 24.00\% & 23.75\% & 23.50\% & 22.88\% & 21.25\% \\
0.001\%$^2$ & (1) & (5) & (18) & (67) & (260) & (1029) & (4102) & (16349) & (62436) \\ \hline
DNA & 25.00\% & 24.25\% & 24.13\% & 24.13\% & 24.00\% & 23.88\% & 23.50\% & 23.00\% & 21.38\% \\
0.01\%$^2$ & (1) & (5) & (18) & (67) & (260) & (1029) & (4102) & (16369) & (63242) \\ \hline
DNA & 25.00\% & 24.50\% & 24.38\% & 24.25\% & 24.25\% & 24.13\% & 23.88\% & 23.50\% & 22.38\% \\
0.1\%$^2$ & (1) & (5) & (19) & (70) & (264) & (1034) & (4109) & (16400) & (65387) \\ \hline
English & 57.25\% & 45.13\% & 34.75\% & 26.00\% & 20.00\% & 15.88\% & 12.50\% & 9.63\% & 7.13\% \\
0.001\%$^2$ & (1) & (106) & (2659) & (18353) & (63300) & (145195) & (256838) & (379514) & (501400) \\ \hline
English & 57.25\% & 45.50\% & 35.38\% & 26.50\% & 20.25\% & 15.88\% & 12.38\% & 9.50\% & 7.13\% \\
0.01\%$^2$ & (1) & (106) & (3243) & (24079) & (83037) & (180592) & (305458) & (439539) & (572186) \\ \hline
English & 57.38\% & 47.75\% & 39.50\% & 31.13\% & 23.00\% & 16.63\% & 12.13\% & 8.88\% & 6.38\% \\
0.1\%$^2$ & (1) & (106) & (4482) & (47357) & (202366) & (466838) & (749065) & (1015587) & (1265447) \\ \hline
Pitches & 66.13\% & 61.13\% & 53.63\% & 37.25\% & 16.38\% & 6.25\% & 2.88\% & 1.38\% & 0.75\% \\
0.001\%$^2$ & (1) & (73) & (3549) & (73664) & (376958) & (642406) & (767028) & (833456) & (871970) \\ \hline
Pitches & 66.13\% & 61.13\% & 53.88\% & 37.50\% & 16.50\% & 6.38\% & 2.88\% & 1.38\% & 0.88\% \\
0.01\%$^2$ & (1) & (73) & (3581) & (76917) & (399546) & (684518) & (821589) & (898152) & (946228) \\ \hline
Pitches & 66.13\% & 62.00\% & 55.88\% & 40.25\% & 17.38\% & 6.50\% & 3.13\% & 1.88\% & 1.38\% \\
0.1\%$^2$ & (1) & (73) & (3742) & (96359) & (606175) & (1103560) & (1367417) & (1545154) & (1688526) \\ \hline
Proteins & 52.25\% & 52.13\% & 51.63\% & 47.50\% & 25.25\% & 4.63\% & 0.75\% & 0.25\% & 0.25\% \\
0.001\%$^2$ & (1) & (21) & (422) & (8045) & (128975) & (463357) & (572529) & (589356) & (595906) \\ \hline
Proteins & 52.25\% & 52.13\% & 51.63\% & 47.63\% & 25.75\% & 5.00\% & 0.88\% & 0.50\% & 0.38\% \\
0.01\%$^2$ & (1) & (21) & (422) & (8045) & (131079) & (494846) & (626306) & (654107) & (670114) \\ \hline
Proteins & 52.25\% & 52.13\% & 51.75\% & 48.75\% & 30.13\% & 7.63\% & 2.13\% & 1.50\% & 1.38\% \\
0.1\%$^2$ & (1) & (21) & (426) & (8072) & (143924) & (771311) & (1154106) & (1297080) & (1407901) \\ \hline
Sources & 68.75\% & 47.00\% & 30.00\% & 19.75\% & 14.38\% & 11.00\% & 8.50\% & 6.88\% & 5.75\% \\
0.001\%$^2$ & (1) & (98) & (4557) & (29667) & (75316) & (130527) & (194105) & (259413) & (320468) \\ \hline
Sources & 68.75\% & 47.50\% & 30.75\% & 20.13\% & 14.63\% & 11.13\% & 8.63\% & 7.00\% & 5.88\% \\
0.01\%$^2$ & (1) & (98) & (5615) & (42337) & (103082) & (170646) & (244874) & (320346) & (391369) \\ \hline
Sources & 68.75\% & 51.25\% & 36.63\% & 24.38\% & 16.75\% & 12.13\% & 9.13\% & 7.25\% & 6.00\% \\
0.1\%$^2$ & (1) & (98) & (7372) & (108997) & (319310) & (525914) & (718657) & (904022) & (1080824) \\ \hline
\end{tabular}
\end{scriptsize}
\end{center}
\caption{Empirical entropy statistics for Pseudo-Real Collection (Scheme 2)}
\label{tab:pseudoreal2_entropy}
\end{table}

\clearpage
\subsection{Real Texts}
Tables \ref{tab:real_sigma}-\ref{tab:real_entropy} give the statistics of real texts.
\begin{table}[!ht]
\begin{center}
\begin{scriptsize}
\begin{tabular}{|l|r|r|r|}
\hline
\textbf{File} & \multicolumn{1}{c|}{\textbf{Size}} & \multicolumn{1}{c|}{$\Sigma$} & \multicolumn{1}{c|}{\textbf{IPM}} \\ \hline
Cere & 440MiB & 5 & 4.301 \\ \hline
Para & 410MiB & 5 & 4.096 \\ \hline
Clostridium Botulium & 34MiB & 4 & 3.356 \\ \hline
Escherichia Coli & 108MiB & 15 & 4.000 \\ \hline
Salmonella Enterica & 66MiB & 9 & 3.993 \\ \hline
Staphylococcus Aureus & 38MiB & 5 & 3.579 \\ \hline
Streptococcus Pneumoniae & 23MiB & 8 & 3.836 \\ \hline
Streptococcus Pyogenes & 24MIB & 10 & 3.800 \\ \hline
Influenza & 148MiB & 15 & 3.845 \\ \hline
Coreutils & 196MiB & 236 & 19.553 \\ \hline
Kernel & 247MiB & 160 & 23.078 \\ \hline
Einstein (en) & 446MiB & 139 & 19.501 \\ \hline
Einstein (de) & 89MiB & 117 & 19.264 \\ \hline
Nobel (en) & 85MiB & 126 & 20.070 \\ \hline
Nobel (de) & 31MiB & 118 & 17.786 \\ \hline
Turing (en) & 7.7MiB & 103 & 21.096 \\ \hline
Turing (de) & 85MiB & 100 & 19.719 \\ \hline
World Leaders & 45MiB & 89 & 3.855 \\ \hline
\end{tabular}
\end{scriptsize}
\end{center}
\caption{Alphabet statistics for Real Collection}
\label{tab:real_sigma}
\end{table}

\begin{table}[!ht]
\begin{center}
\begin{scriptsize}
\begin{tabular}{|l|r|r|r|r|r|}
\hline
\textbf{File} & \multicolumn{1}{c|}{\textbf{p7zip}} & \multicolumn{1}{c|}{\textbf{bzip2}} & \multicolumn{1}{c|}{\textbf{gzip}} & \multicolumn{1}{c|}{\textbf{ppmdi}} & \multicolumn{1}{c|}{\textbf{Re-Pair}} \\ \hline
Cere & 1.14\% & 2.50\% & 26.36\% & 24.09\% & 1.86\% \\ \hline
Para & 1.46\% & 26.34\% & 27.07\% & 24.88\% & 2.80\% \\ \hline
Clostridium Botulium & 8.53\% & 25.88\% & 26.47\% & 24.12\% & 20.00\% \\ \hline
Escherichia Coli & 4.72\% & 26.85\% & 28.70\% & 25.93\% & 9.63\% \\ \hline
Salmonella Enterica & 5.61\% & 27.27\% & 28.79\% & 25.76\% & 12.42\% \\ \hline
Staphylococcus Aureus & 2.89\% & 26.32\% & 28.95\% & 25.00\% & 5.26\% \\ \hline
Streptococcus Pneumoniae & 4.78\% & 26.52\% & 27.39\% & 24.78\% & 9.57\% \\ \hline
Streptococcus Pyogenes & 5.00\% & 26.25\% & 27.08\% & 25.00\% & 9.58\% \\ \hline
Influenza & 1.35\% & 6.62\% & 7.43\% & 3.78\% & 3.31\% \\ \hline
coreutils & 1.94\% & 16.33\% & 24.49\% & 12.76\% & 2.55\% \\ \hline
kernel & 0.81\% & 21.86\% & 27.13\% & 18.62\% & 1.13\% \\ \hline
einstein.en & 0.07\% & 5.38\% & 35.20\% & 1.61\% & 0.10\% \\ \hline
einstein.de & 0.11\% & 4.38\% & 31.46\% & 1.35\% & 0.16\% \\ \hline
nobel.en & 0.13\% & 2.94\% & 18.82\% & 1.76\% & 0.20\% \\ \hline
nobel.de & 0.18\% & 3.55\% & 27.74\% & 1.68\% & 0.30\% \\ \hline
turing.en & 1.09\% & 36.36\% & 285.71\% & 15.58\% & 1.71\% \\ \hline
turing.de & 0.03\% & 0.18\% & 0.10\% & 0.11\% & 0.05\% \\ \hline
world leaders & 1.29\% & 7.11\% & 17.78\% & 3.56\% & 1.78\% \\ \hline
\end{tabular}
\end{scriptsize}
\end{center}
\caption{Compression statistics for Real Collection}
\label{tab:real_compression}
\end{table}

\begin{table}[!ht]
\begin{center}
\begin{scriptsize}
\begin{tabular}{|l|c|c|c|c|c|c|c|c|c|}
\hline
\textbf{File} & \textbf{$H_0$} & \textbf{$H_1$} & \textbf{$H_2$} & \textbf{$H_3$} & \textbf{$H_4$} & \textbf{$H_5$} & \textbf{$H_6$} & \textbf{$H_7$} & \textbf{$H_8$} \\ \hline
\multirow{2}{*}{Cere} & 27.38\% & 22.63\% & 22.63\% & 22.50\% & 22.50\% & 22.50\% & 22.50\% & 22.38\% & 22.25\% \\
   & (1) & (5) & (25) & (125) & (610) & (2515) & (8697) & (28080) & (88624) \\ \hline
\multirow{2}{*}{Para} & 26.50\% & 23.50\% & 23.38\% & 23.38\% & 23.38\% & 23.38\% & 23.25\% & 23.25\% & 23.13\% \\
   & (1) & (5) & (25) & (125) & (625) & (3125) & (14725) & (51542) & (139149) \\ \hline
Clostridium & 23.25\% & 23.00\% & 22.88\% & 22.75\% & 22.75\% & 22.75\% & 22.63\% & 22.50\% & 22.25\% \\
Botulium & (1) & (4) & (16) & (64) & (256) & (1024) & (4096) & (16383) & (65118) \\ \hline
Escherichia & 25.00\% & 24.75\% & 24.50\% & 24.38\% & 24.25\% & 24.25\% & 24.13\% & 24.13\% & 23.88\% \\
Coli & (1) & (15) & (145) & (779) & (2715) & (7436) & (15641) & (32561) & (85363) \\ \hline
Salmonella & 25.00\% & 24.75\% & 24.50\% & 24.38\% & 24.25\% & 24.13\% & 24.13\% & 24.00\% & 23.75\% \\
Enterica & (1) & (9) & (35) & (97) & (299) & (1077) & (4159) & (16457) & (65618) \\ \hline
Staphylococcus & 23.88\% & 23.75\% & 23.75\% & 23.63\% & 23.63\% & 23.63\% & 23.50\% & 23.25\% & 22.75\% \\
Aureus & (1) & (5) & (18) & (67) & (260) & (1029) & (4102) & (16391) & (65282) \\ \hline
Streptococcus & 24.63\% & 24.38\% & 24.38\% & 24.25\% & 24.13\% & 24.13\% & 24.00\% & 23.75\% & 23.13\% \\
Pneumoniae & (1) & (8) & (31) & (133) & (574) & (2183) & (6928) & (21093) & (71592) \\ \hline
Streptococcus & 24.50\% & 24.38\% & 24.25\% & 24.13\% & 24.13\% & 24.13\% & 24.00\% & 23.88\% & 23.25\% \\
Pyogenes & (1) & (10) & (50) & (174) & (456) & (1291) & (4418) & (16758) & (65919) \\ \hline
\multirow{2}{*}{Influenza} & 24.63\% & 24.13\% & 24.13\% & 24.00\% & 23.88\% & 23.50\% & 22.00\% & 18.63\% & 13.25\% \\
   & (1) & (15) & (125) & (583) & (2329) & (7978) & (21316) & (44748) & (101559) \\ \hline
\multirow{2}{*}{coreutils} & 68.38\% & 51.25\% & 35.88\% & 23.88\% & 17.00\% & 12.88\% & 10.13\% & 8.00\% & 6.50\% \\
   & (1) & (236) & (18500) & (169716) & (606527) & (1335553) & (2258650) & (3258896) & (4247313) \\ \hline
\multirow{2}{*}{kernel} & 67.25\% & 50.50\% & 36.63\% & 25.75\% & 19.25\% & 15.13\% & 12.13\% & 9.63\% & 7.75\% \\
   & (1) & (160) & (7122) & (90396) & (351918) & (773818) & (1305616) & (1912604) & (2553008) \\ \hline
\multirow{2}{*}{einstein.en} & 62.00\% & 46.38\% & 33.38\% & 21.13\% & 13.25\% & 9.00\% & 6.50\% & 4.75\% & 3.50\% \\
   & (1) & (139) & (4546) & (28685) & (77333) & (142559) & (211506) & (276343) & (335151) \\ \hline
\multirow{2}{*}{einstein.de} & 63.00\% & 44.88\% & 32.63\% & 20.88\% & 13.25\% & 9.00\% & 6.13\% & 4.38\% & 3.13\% \\
   & (1) & (117) & (3278) & (16765) & (39010) & (64884) & (89914) & (112043) & (130473) \\ \hline
\multirow{2}{*}{nobel.en} & 62.63\% & 44.63\% & 30.50\% & 18.25\% & 11.50\% & 8.13\% & 6.00\% & 4.50\% & 3.38\% \\
   & (1) & (126) & (3566) & (18079) & (42334) & (69855) & (95644) & (119260) & (140401) \\ \hline
\multirow{2}{*}{nobel.de} & 61.13\% & 43.25\% & 31.13\% & 19.63\% & 12.50\% & 8.63\% & 6.00\% & 4.13\% & 3.00\% \\
   & (1) & (118) & (2726) & (12959) & (30756) & (49695) & (66108) & (80467) & (92184) \\ \hline
\multirow{2}{*}{turing.en} & 63.25\% & 45.75\% & 32.00\% & 19.13\% & 11.50\% & 7.63\% & 5.38\% & 3.88\% & 2.88\% \\
   & (1) & (103) & (2794) & (14091) & (33498) & (55489) & (75611) & (93402) & (108636) \\ \hline
\multirow{2}{*}{turing.de} & 62.38\% & 43.25\% & 29.25\% & 16.75\% & 9.50\% & 6.00\% & 3.88\% & 2.63\% & 2.00\% \\
   & (1) & (100) & (1806) & (7268) & (15407) & (23070) & (29038) & (33714) & (37335) \\ \hline
world & 43.38\% & 24.38\% & 17.25\% & 11.63\% & 7.63\% & 5.13\% & 4.00\% & 3.50\% & 3.13\% \\
leaders & (1) & (89) & (2526) & (23924) & (106573) & (246566) & (374668) & (468701) & (547040) \\ \hline
\end{tabular}
\end{scriptsize}
\end{center}
\caption{Empirical entropy statistics for Real Collection}
\label{tab:real_entropy}
\end{table}
\clearpage
\section{Discussion}
It can be seen in the tables presented above that only p7zip and Re-Pair capture the repetitiveness of the texts, achieving a compression ratio at least one order of magnitude better than bzip2, gzip or ppmdi. 
It can also be noted in Tables \ref{tab:pseudoreal_compression} and \ref{tab:pseudoreal2_compression} that p7zip is more robust to capture the repetitiveness than Re-Pair, as with mutation ratios of 0.1\% p7zip compresses 5 times better than Re-Pair. Table \ref{tab:real_compression} also shows that Re-Pair fails to capture some repetitions, as for all DNA texts except \emph{para} and \emph{cere} the compression of p7zip is two times better than that of Re-Pair. Tables \ref{tab:pseudoreal_compression} and \ref{tab:pseudoreal_entropy} also show that the compression ratio of bzip2, gzip and ppmdi does not change significantly when increasing the repetitiveness of the text (decreasing mutation ratio). However, Tables \ref{tab:pseudoreal2_compression} and \ref{tab:pseudoreal2_entropy} show that when decreasing the mutation ratio from 0.1\% to 0.01\% the gain in compression is greater than 10\%, but when decreasing the mutation to 0.001\% the compression ratio does not improve as much. It can also be seen that the compression ratios of bzip2 and gzip are close to the $H_2$-$H_3$, whereas, curiously, ppmdi compression ratios are not well predicted by any $H_k$.
\diego{Notice that, since artificial texts are extremely compressible, small constant overheads (usually irrelevant) may produce significant differences in the size of the compressed file.}

%% file: lz_parsing.tex
\chapter{LZ-End: A New Lempel-Ziv Parsing}
\label{chap:parsing}
In this chapter we explain \nieves{some properties of the LZ77 parsing (see Section \ref{sec:lzparsings})} \dnieves{the concept of a Lempel-Ziv parsing} and present a variant that has the advantage of faster text extraction. The results presented in this chapter were published in the \emph{20th Data Compression Conference (DCC)} \cite{KN10}.
\section{LZ77 on Repetitive Texts}
\label{sec:lz77}
An interesting property of the LZ77 parsing is that it captures the repetitions of the text. \diego{Text repetitions, as well as single-character edits on a text, alter the number of phrases of the parsing very little. This explains why LZ77 is so strong on highly repetitive collections.}
\begin{lemma}
\label{lemma:lz77rep}
Given \diego{the texts} \ddiego{a text} $T$, \diego{$T'$ and the characters $a$ , $b$;}  the following statements hold
\begin{eqnarray}
H^{LZ77}(TT) & = & H^{LZ77}(T) + 1 \label{eq:e1} \\
H^{LZ77}(TT\texttt{\$}) & \le & H^{LZ77}(T\texttt{\$}) + 1 \label{eq:e2}\\
H^{LZ77}(TT') & \le & H^{LZ77}(TaT') + 1 \label{eq:e3}\\
H^{LZ77}(TaT') & \le & H^{LZ77}(TT') + 1 \label{eq:e4}\\
H^{LZ77}(TaT') & \le & H^{LZ77}(TbT') + 1 \label{eq:e5}
\end{eqnarray} 
where $H^{LZ77}(T)$ is the number of phrases of the LZ77 parsing.
\end{lemma}

\begin{proof}
Assume the last phrase of the LZ77 parsing of $T\texttt{\$}$ is \phrase{\$} and that $H^{LZ77}(T\texttt{\$})=n'$ . That means the first $n'-1$ phrases cover the text $T$. Now, if we have the text $TT\texttt{\$}$, we have that the first $n'-1$ phrases are the same as for the parsing of $T$ and the last phrase would be \phrase{$T$\$}, hence inequality for Equation (\ref{eq:e2}) holds. Now, assume the last phrase of the parsing is \phrase{$A$\$} for some $A \neq \varepsilon$. Therefore the $n'$-th phrase of the parsing of $TT$ would be \phrase{$A$$B$} for some $B$ such that $1\le|B|<|T|$, thus this phrase does not completely cover $TT$. An additional phrase covers the remaining portion of the text, thus equality holds for Equation (\ref{eq:e2}). The proof of Equation (\ref{eq:e1}) is similar to the second part of Equation (\ref{eq:e2}).

Now consider Equation (\ref{eq:e3}). Let $Z[p]=XY$ the last phrase covering $T$, where $X$ is a suffix of $T$ and $Y$ is a prefix of $T'$. When adding the new character in the middle, in the worst case the phrase gets converted to $Xa$ (this is the phrase that may increase the total number of phrases). Then the following phrase will cover at least the prefix $Y$, and each successive phrase will cover at least the next phrase of the original parsing. Hence, the number of phrases is at most one more than the original number of phrases. The proofs for Equations (\ref{eq:e4}) and (\ref{eq:e5}) are similar to the one above.
\end{proof}

\ddiego{The lemma states that a repetition of a text, as well as single-character edits on a text, alter the number of phrases very little. This explains why LZ77 is so strong on highly repetitive collections. On the contrary, the} \diego{The} LZ78 parsing \cite{ZL78} described in Section \ref{sec:lz78_selfindexes} is not that powerful. On $T=a^n$ it produces $n'=\frac{\sqrt{n}}{2}+O(1)$ phrases, and this increases to $n'=\frac{\sqrt{2n}}{2}+O(1)$ on $TT$. LZ77, instead, produces $n'=\log_2(n)+O(1)$ phrases on $T$ and just one more phrase on $TT$.
\section{LZ-End}
\label{sec:lzend}
In this section we introduce a new LZ-like parsing. Its main characteristic is a faster random text extraction, while its compression is close to that of LZ77. 
\begin{definition}
\label{def:lzend}
The \emph{LZ-End parsing} of text $T[1,n]$ is a sequence $Z[1,n']$ of
{\em phrases} such that $T = Z[1] Z[2] \ldots Z[n']$, built as follows. 
Assume we
have already processed $T[1,i-1]$ producing the sequence $Z[1,p-1]$. Then, we
find the longest prefix $T[i,i'-1]$ of $T[i,n]$ that is a suffix of 
$Z[1] \ldots Z[q]$ for some $q<p$,
set $Z[p] = T[i,i']$ and continue with $i = i'+1$. 
\end{definition}

\begin{example}
Let $T=\texttt{`alabar\_a\_la\_alabarda\$'}$; the LZ-End parsing is as follows:
\begin{equation*}
\phrase{a}\phrase{l}\phrase{ab}\phrase{ar}\phrase{\_}\phrase{a\_}\phrase{la}\phrase{\_a}\phrase{labard}\phrase{a\$}
\end{equation*}
In this example, when generating the seventh phrase we cannot copy two characters as in Example \ref{ex:lz77}, because \texttt{`la'} does not end in a previous end of phrase. However, \texttt{`l'} does end in an end of phrase, hence we generate the phrase \texttt{`la'}. Notice that the number of phrases increased from 9 to 10 \diego{with respect to the original LZ77 scheme}.
\end{example}

The LZ-End parsing is similar to \djeremy{that}\jeremy{the one proposed} by Fiala and Green \cite{FG89}, in that
theirs restricts where the sources start, while ours restricts where the
sources end. This is the key feature that will allow us extract arbitrary
phrases in constant time per extracted symbol and, as shown \jeremy{in Section \ref{sec:lz_extraction}}\djeremy{later}\dmine{, the key to stay close to LZ77 in compression ratio for very repetitive sequences}.

\subsection{Encoding}
\label{sec:encoding}

The output of an LZ77 compressor is, essentially, the sequence of triplets
$z(p) = (j,\ell,c)$, such that the source of $Z[p] = T[i,i']$ is
$T[j,j+\ell-1]$, $\ell=i'-i$, and $c = T[i']$. This format allows fast decompression of $T$, but not decompressing an individual phrase $Z[p]$ \jeremy{in reasonable time (one must basically decompress the whole text)}.

The LZ-End parsing, although potentially \diego{generates} \ddiego{generating} more phrases than LZ77, permits a
shorter encoding of each, of the form $z(p) = (q,\ell,c)$, such that the
source of $Z[p] = T[i,i']$ is a suffix of $Z[1]\ldots Z[q]$, and the rest is
as above. \diego{This representation is shorter because it stores the phrase identifier rather than a text position.} We introduce a more sophisticated encoding that will, in addition,
allow us \jeremy{to} extract individual phrases in constant time per extracted symbol.

\begin{itemize}
    \item $\chars[1,n']$ (using $n'\lceil \log \sigma \rceil$ bits) 
encodes the trailing characters ($c$ above).
\vspace*{-2mm}
    \item $\source[1,n']$ (using $n'\lceil \log n' \rceil$ bits) 
encodes the phrase identifier where the source ends ($q$ above).
\vspace*{-2mm}
    \item $B[1,n]$ (using $n'\log\frac{n}{n'}+O(n'+\frac{n\log\log n}{\log n})$ bits in compressed form \cite{RRR02}, see Section \ref{sec:bitmaps}) marks the ending positions of the phrases in $T$.
\end{itemize}

Thus we have $z(p) = (q,l,c) = (\source[p],select_1(B,p+1)-select_1(B,p)-1,\chars[p])$. 
We can also know in constant time that phrase $p$ ends at 
$select_1(B,\source[p])$ and that it starts $\ell$ positions before. Finally, 
we can know that \jeremy{the} text position $i$ belongs to phrase $Z[rank_1(B,i-1)+1]$.

\subsection{Extraction Algorithm}
\label{sec:lz_extraction}
The algorithm to extract an arbitrary substring in LZ-End is given in Figure~\ref{alg:extraction}.
The extraction works from right to left. First we compute the last phrase $p$ \ddiego{covering} \diego{intersecting} the substring. If the last character is stored explicitly, i.e., it is an end of phrase (see Line 4),  we output $char[p]$ and recursively extract the remaining substring (line 5). Otherwise we split the substring into two parts. The right one is the intersection of the rightmost phrase covering the substring and the substring \diego{itself}, and is extracted recursively by going to the source of that phrase (line 10). The left part is also extracted recursively (line 13). 

While the algorithm works for extracting any substring, we can prove it \ddiego{is optimal} \diego{takes constant time per extracted symbol} when the substring ends at a phrase.
\begin{figure}[tb]
\renewcommand{\algorithmiccomment}[1]{/*#1*/}
\algsetup{linenodelimiter=,indent=0.8em}
\begin{center}
\begin{tabular}{l}
\begin{minipage}{14cm}
\textbf{Extract}$(start, len)$
\begin{algorithmic}[1]
\IF{$len > 0$} 
   \STATE $end \leftarrow start+len-1$
   \STATE $p \leftarrow rank_1(B,end)$
   \IF{$B[end]=1$}
       \STATE Extract$(start,len-1)$
       \STATE {\bf output} $\chars[p]$
   \ELSE
       \STATE $pos \leftarrow select_1(B,p)+1$
       \IF{$start<pos$}
           \STATE Extract$(start,pos-start)$
	       \STATE $len \leftarrow end-pos+1$
	       \STATE $start \leftarrow pos$
       \ENDIF
       \STATE Extract$(select_1(B,\source[p+1])-select_1(B,p+1)+start+1,len)$
   \ENDIF
\ENDIF
\end{algorithmic}
\end{minipage}
\end{tabular}
\end{center}
\caption[LZ-End extraction algorithm]{%
LZ-End extraction algorithm for $T[start,start+len-1]$.
}
\label{alg:extraction}
\end{figure}

\begin{theorem}
\label{lemma:extraction_lzend}
Function {\em Extract} outputs a text substring $T[start,end]$ ending
at a phrase in time $O(end-start+1)$.
\end{theorem}
\begin{proof}
If $T[start,end]$ ends at a phrase, then $B[end]=1$. We proceed by induction 
on $len=end-start+1$. The case $len\le 1$ is trivial by inspection.
Otherwise, we output $T[end]$ at line 6 after a recursive call on the same 
phrase and length $len-1$. This time we go to line 8. The current phrase (now
$p+1$) starts at $pos$. If $start<pos$, we carry out a recursive call at line 
10 to handle the segment $T[start,pos-1]$. As this segment ends \diego{at the end of} phrase $p$, 
induction shows that this takes time $O(pos-start+1)$. Now the segment
$T[\max(start,pos),end]$ is contained in $Z[p+1]$ and it finishes one symbol 
before the phrase ends. Thus a copy of it finishes where $Z[\source[p+1]]$ 
ends, so induction applies also to the recursive call at line 13, which
extracts the remaining string from the source instead of from $Z[p+1]$, also
in constant time per extracted symbol.
\end{proof}

We have shown that the algorithm extracts any substring by starting from the end of a phrase. Thus, extracting an arbitrary substring may be more expensive than an end-of-phrase aligned one.
\begin{definition}
\label{def:height}
Let $T=Z[1]Z[2]\ldots Z[n']$ be a LZ-parsing of $T[1,n]$. Then the \emph{height} of the parsing is defined as $H=\max_{1\le i \le n} C[i]$, where $C$ is defined as follows.  Let $Z[i]=T[a,b]$ be a phrase which source is $T[c,d]$, then 
\begin{eqnarray*}
C[k] & = & C[(k-a)+c]+1,\, \forall a \le k < b\\
C[b] & = & 1
\end{eqnarray*}
\end{definition}
Array $C$ counts how many times a character was transitively copied from its original source. This is also the extraction cost of that character. Hence, the value $H$ is the worst-case bound for extracting a single character in the LZ parse.  

\begin{lemma}
\label{lemma:hlzend}
In an LZ-End parsing it holds that $H$ is smaller than the longest phrase, i.e., $H \le \max_{1\le p \le n'} |Z[p]|$.
\end{lemma}
\begin{proof}
We will prove by induction that $\forall 1\le i < n,\, C[i]\le C[i+1]+1$. From this inequality the lemma follows. For all positions $i_p$ where a phrase $p$ ends, it holds by definition that $C[i_p]=1$. Thus, for all positions $i$ in the phrase $p$, we have $C[i]\le C[i_p]+i_p-i \le |Z[p]|$. \\

The first phrase of any LZ-End parsing is $T[0]$, and the second is either $T[1]$ or $T[1]T[2]$. In the first case, we have $C[1]C[2]=1,1$, in the latter $C[1]C[2]C[3]=1,2,1$. In both cases the property holds. Now, suppose the inequality is valid up to position $i_p$ where the phrase $Z[p]$ ends. Let $i_{p+1}$ be the position where the phrase $Z[p+1]=T[a,b]$ ends (so $a=i_p+1$ and $b=i_{p+1}$) and $T[c,d]$ its source. For all $i_p+1\le i < i_{p+1}$, $C[i] = C[(i-a)+c]+1$, and since $d\le i_p$, the inequality holds by inductive hypothesis for $i_p+1\le i \le i_{p+1}-2$. By definition of the LZ-End parsing the source of a phrase ends in a previous end of phrase, hence $C[i_{p+1}-1]=C[d]+1=2 \le 1+1=C[i_{p+1}]+1$. For position $i_{p+1}$ (end of phrase) the inequality trivially holds as it has by definition the least possible value.
\end{proof}

The above lemma does not hold for LZ77. Moreover, the LZ-End parsing yields a better extraction complexity.

\begin{lemma}
\label{lemma:cost_lzend}
Extracting a substring of length $\ell$ from an LZ-End parsing takes time $O(\ell+H)$.
\end{lemma}
\begin{proof}
Theorem \ref{lemma:extraction_lzend} already shows that the cost to extract a substring ending at a phrase boundary is \diego{constant per extracted symbol}. The only piece of the code in Figure \ref{alg:extraction} that does not amortize in this sense is line 13, where we recursively unroll the last phrase, removing the last character each time, until hitting the end of the substring to extract. By definition of $H$, this line cannot be executed more than $H$ times. So the total time is $O(\ell+H)$.
\end{proof}

\begin{remark}
On a text coming from an ergodic Markov source of entropy $h$, the expected value of the longest phrase is $O(\frac{\log n}{h})$. However, as we are dealing with highly repetitive texts this expected length does not hold.
\end{remark}

\begin{remark}
Algorithm Extract (Figure \ref{alg:extraction}) also works on parsing LZ77, but in this case the best theoretical bound we can prove for extracting a substring of length $\ell$ is $O(\ell H)$ . However, the results in Section \ref{sec:lzend_ext_res} suggest that on average it may be much better.
\end{remark}



\section{Compression Performance}
We study now the compression performance of LZ-End, first with respect to the empirical $k$-th order entropy and then on repetitive texts.
\subsection{Coarse Optimality}
We prove that LZ-End is coarsely optimal. The main tool is the following lemma.
\label{sec:coarse_op}
\begin{lemma} \label{lem:unique}
All the phrases generated by an LZ-End parse are different.
\end{lemma}
\begin{proof}
Assume by contradiction $Z[p] = Z[p']$ for some $p < p'$. When $Z[p']$ was generated, we could have taken $Z[p]$ as the source, yielding phrase $Z[p']c$, longer than $Z[p']$. This is clearly a valid source as $Z[p]$ is a suffix of $Z[1] \ldots Z[p]$. So this is not an LZ-End parse.
\end{proof}

\begin{lemma}[\cite{LZ76}]  \label{lem:folk}
Any parsing of $T[1,n]$ into $n'$ distinct phrases satisfies
$n' = O\left(\frac{n}{\log_{\sigma}n}\right)$,
where $\sigma$ is the alphabet size of $T$.
\end{lemma}

\begin{lemma}[\cite{KM99}] \label{lem:entropy}
For any text $T[1,n]$ parsed into $n'$ different phrases, it holds
$ n' \log n' \le n H_k(T) + n' \log \frac{n}{n'} + \Theta(n'(1+k\log \sigma))$,
for any $k$.
\end{lemma}

\begin{lemma} \label{lemma:phrase_bounds}
For any text $T[1,n]$ parsed into $n'$ different phrases, using LZ77 or LZ-End, it holds
$n'\log n \le nH_k(T)+o(n \log \sigma)$ for any $k = o(\log_\sigma n)$.
\end{lemma}

\begin{proof}
Arroyuelo and Navarro \cite{AN_LZ} prove that the property holds for any LZ parsing for which Lemmas \ref{lem:folk} and \ref{lem:entropy} hold. In particular it holds for LZ77 and for our proposal, the LZ-End.
\end{proof}

\begin{theorem}
The LZ-End compression is coarsely optimal.
\end{theorem}
\begin{proof}
The proof is based on the one by Kosaraju and Manzini \cite{KM99} for LZ77.
Here we consider in addition our particular encoding (the result holds for 
triplets $(q,\ell,c)$ as well).
The size of the parsing in bits is
\begin{eqnarray*}
\textit{LZ-End}(T) & = & n'\lceil \log \sigma \rceil + n' \lceil \log n'
\rceil + n'\log\frac{n}{n'} + O\left(n' + \frac{n\log\log n}{\log n}\right) \\
& = & n'\log n + O\left(n'\log\sigma + \frac{n\log\log n}{\log n}\right).
\end{eqnarray*}
Thus from Lemmas~\ref{lem:unique} and \ref{lem:entropy} we have
\[
\textit{LZ-End}(T) ~~\le~~ nH_k(T) + 2n' \log \frac{n}{n'} + 
		O\left(n'(k+1)\log\sigma + \frac{n\log\log n}{\log n}\right).
\]
Now, by means of Lemma~\ref{lem:folk} and since $n'\log\frac{n}{n'}$ is
increasing in $n'$, we get
\begin{eqnarray*}
\textit{LZ-End}(T) &\le& nH_k(T) + 
             O\left (\frac{n\log\sigma \log \log n}{\log n}\right ) + 
             O\left (\frac{n(k+1)\log^2 \sigma}{\log n} + \frac{n\log\log n}{\log n}\right)\\
&=&nH_k(T) + O\left (\frac{n\log\sigma (\log \log n+(k+1)\log \sigma)}{\log n}\right ).
\end{eqnarray*}
Thus, diving by $n$ and taking $k$ and $\sigma$ as constants, we get that the 
compression ratio is
\[
\rho(T) ~~\le~~ H_k(T) + O\left (\frac{n\log \log n}{\log n} \right ).
\]
\end{proof}

\subsection{Performance on Repetitive Texts}
We have not found a worst-case bound for the competitiveness of LZ-End compared to LZ77. 
However, we show, on the negative side, a sequence that produces almost twice the number of phrases when parsed with LZ-End, so LZ-End is at best 2-competitive with LZ77. On the positive side, we show that LZ-End satisfies some of the properties of Lemma \ref{lemma:lz77rep}.

\begin{example}
Let $T=\texttt{112}\cdot \texttt{113}\cdot \texttt{214}\cdot \texttt{325}\cdot \texttt{436}\cdot \texttt{547}\cdot\ldots\cdot(\sigma-2)(\sigma-3)\sigma$. The length of the text is $n=3(\sigma-1)$. The LZ parsings are: 
\begin{eqnarray*}
\text{LZ77} & \phrase{1}\phrase{12}\phrase{113}\phrase{214}\phrase{325}\phrase{436}\phrase{547}\ldots\boxed{(\sigma-2)(\sigma-3)\sigma\vphantom{\texttt{\$}}}\\
\text{LZ-End} & \phrase{1}\phrase{12}\phrase{11}\phrase{3}\phrase{21}\phrase{4}\phrase{32}\phrase{5}\phrase{43}\phrase{6}\phrase{54}\phrase{7}\ldots\boxed{(\sigma-2)(\sigma-3)\vphantom{\texttt{\$}}}\boxed{\sigma\vphantom{(\sigma-3)}}
\end{eqnarray*}
The size of LZ77 is $n'=\sigma$ and the size of LZ-End is $n'=2(\sigma-1)$.
\end{example}

For LZ-End we can only prove the following lemma regarding the concatenation of texts.
\begin{lemma}
\label{lemma:lzendrep}
Given a text $T$, the following statements hold
\begin{eqnarray}
H^{LZ-End}(TT) & \le & H^{LZ-End}(T) + 2 \label{eq:e1_end} \\ 
H^{LZ-End}(TT\texttt{\$}) & \le & H^{LZ-End}(T\texttt{\$}) + 1 \label{eq:e2_end}
\end{eqnarray} 
where $H^{LZ-End}(T)$ is the number of phrases of the LZ-End parsing.
\end{lemma}

\begin{proof}
Assume $H^{LZ-End}(T\texttt{\$})=n'$ and that the last phrase of the LZ-End parsing of $T\texttt{\$}$ is \phrase{\$}. That means the first $n'-1$ phrases cover the text $T$. Now, if we have the text $TT\texttt{\$}$, we have that the first $n'-1$ phrases are the same as for the parsing of $T$ and the last phrase would be \phrase{$T$\$} (since $T$ ends at the end of the $(n'-1)$-th phrase), hence inequality for Equation (\ref{eq:e2_end}) holds (actually it holds $H^{LZ-End}(TT\texttt{\$}) = H^{LZ-End}(T\texttt{\$})$). Now, assume the last phrase of the parsing is \phrase{$A$\$} for some $A \neq \varepsilon$, and that the prefix of $T$ in the $n'-1$ first phrases is $xX$, where $x$ is a character. Therefore the $n'$-th phrase of the parsing of $TT$ would be at least \phrase{$A$$x$}. Then the $(n'+1)$-th phrase will be \phrase{$X$$A$\texttt{\$}}, thus equality holds for Equation (\ref{eq:e2_end}). The situation is analogous if the $n'$-th phrase extends beyond $Ax$. For Equation (\ref{eq:e1_end}) consider that $T$ is parsed into $n'-1$ phrases covering the prefix $xX$ and the last phrase is $aAb$ (where $x$, $a$ and $b$ are characters and $A\neq \varepsilon$ is a string). Then, the $(n'+1)$-th phrase of $TT$ is at least \phrase{$x$$X$$a$}, and thus, the $(n'+2)$-th phrase is \phrase{$Ab$}. Because there must exist a phrase ending in $A$ for phrase \phrase{$aAb$} to exist. If, instead, the $n'$-th phrase is just \phrase{$a$}, then the $(n'+1)$-th phrase is \phrase{$xXa$}$=$\phrase{$T$}.
\end{proof}

\section{Construction Algorithm}
\label{sec:constr}

\begin{figure}[tb]
\renewcommand{\algorithmiccomment}[1]{/*#1*/}
\algsetup{linenodelimiter=,indent=0.8em}
\begin{center}
\begin{tabular}{l}
\begin{minipage}{10cm}
\begin{algorithmic}[1]
\STATE $\mathcal{F} \leftarrow \{\langle -1,n+1\rangle\}$
\STATE $i \leftarrow 1$, $p \leftarrow 1$
\WHILE{$i\le n$}
    \STATE $[sp,ep] \leftarrow [1,n]$ \\
    \STATE $j \leftarrow 0$, $\ell \leftarrow j$
    \WHILE{$i+j\le n$}
        \STATE $[sp,ep] \leftarrow \BWS(sp,ep,T[i+j])$ 
        \STATE $mpos \leftarrow \arg \max_{sp \le k \le ep}A[k]$
	\IF{$A[mpos] \le n+1-i$}
        \STATE \textbf{break}
    \ENDIF
	\STATE $j \leftarrow j+1$
        \STATE $\langle q,fpos\rangle \leftarrow \textrm{Successor}(\mathcal{F},sp)$
	\IF{$fpos \leq ep$}
        \STATE $\ell \leftarrow j$
    \ENDIF
    \ENDWHILE
    \STATE Insert$(\mathcal{F},\langle p,A^{-1}[n+1-(i+\ell)]\rangle)$
    \STATE \textbf{output} $(q, \ell, T[i+\ell])$
    \STATE $i \leftarrow i+\ell+1$,~ $p \leftarrow p+1$
\ENDWHILE
\end{algorithmic}
\end{minipage}
\end{tabular}
\end{center}
\caption[LZ-End construction algorithm]{%
LZ-End construction algorithm. $\mathcal{F}$ stores pairs 
$\langle$phrase identifier, suffix array position$\rangle$ and answers successor queries
on the text position. $BWS(sp,ep,c)$ was defined in Section \ref{sec:bwt}.
}
\label{alg:construction}
\end{figure}

We present an algorithm to compute the parsing LZ-End, inspired \jeremy{by the}\djeremy{on} algorithm CSP2 
by Chen~\etal~\cite{CPS08}. We compute the range of all text prefixes ending with a 
pattern $P$, rather than suffixes starting with $P$ \cite{FG89}. 

We first build the suffix array (Section \ref{sec:sarray}) $A[1,n]$ of the
{\em reverse} text, $T^{rev} = T[n-1]\ldots T[2] T[1]\texttt{\$}$, so that
$T^{rev}[A[i],n]$ is the lexicographically $i$-th smallest suffix of $T^{rev}$. 
We also build its inverse permutation: $A^{-1}[j]$ is the 
lexicographic rank of $T^{rev}[j,n]$. Finally, we build the Burrows-Wheeler 
Transform (BWT) (Section \ref{sec:bwt}) of $T^{rev}$, $T^{bwt}[i] = T^{rev}[A[i]-1]$ (or
$T^{rev}[n]$ if $A[i]=1$). 

On top of the BWT we will apply \emph{backward search} (Section \ref{sec:bwt}) to find out whether there are occurrences of a $T[i,i'-1]$ (Definitions \ref{def:lz77} and \ref{def:lzend}). 

Since, for LZ-End, the phrases must in addition finish at a previous 
phrase end, we maintain a dynamic set $\mathcal{F}$ where we add the ending 
positions of the successive phrases we create, mapped to $A$. That is, once 
we create phrase $Z[p] = T[i,i']$, we insert $A^{-1}[n+1-i']$ into 
$\mathcal{F}$.

Backward search over $T^{rev}$ adapts very well to our purpose. By considering the
patterns $P = (T[i,i'-1])^{rev}$ for consecutive values of $i'$, we are searching
backwards for $P$ in $T^{rev}$, and thus finding the {\em ending} positions of 
$T[i,i'-1]$ in $T$, by carrying out one further BWS step for each new $i'$
value. Thus we can use $\mathcal{F}$ naturally.

As we advance $i'$ in $T[i,i'-1]$, we test
whether $A[sp,ep]$ contains some occurrence finishing before $i$ in $T$, that
is, starting after $n+1-i$ in $T^{rev}$. If it does not, then we stop looking for
larger $i'$ values as there are no matches preceding $T[i]$. For this, we
precompute a {\em Range Maximum Query (RMQ)} data structure \cite{FH07} on $A$,
which answers queries $mpos = \arg\max_{sp\le k\le ep} A[k]$. Then 
if $A[mpos]$ is not large enough, we stop.

In addition, we must know if $i'$ finishes at some phrase end, i.e., if
$\mathcal{F}$ contains some value in $[sp,ep]$. A {\em successor} 
query on $\mathcal{F}$ finds the smallest value $fpos \ge sp$ in $\mathcal{F}$. 
If $fpos \le ep$, then it
represents a suitable LZ-End source for $T[i,i']$. Otherwise, as the condition
could hold again for a later $[sp,ep]$ range, we do not stop but recall the last $j=i'$ where it was 
valid. Once we stop because no matches ending before $T[i]$ exist, we insert 
phrase $Z[p] = T[i,j]$ and continue from $i = j+1$. This may retraverse 
some text since we had processed up to $i' \ge j$. We call $N \ge n$ the
total number of text symbols processed.

The algorithm is depicted in Figure \ref{alg:construction}.

\begin{figure}
\centering
\includegraphics[width=14cm]{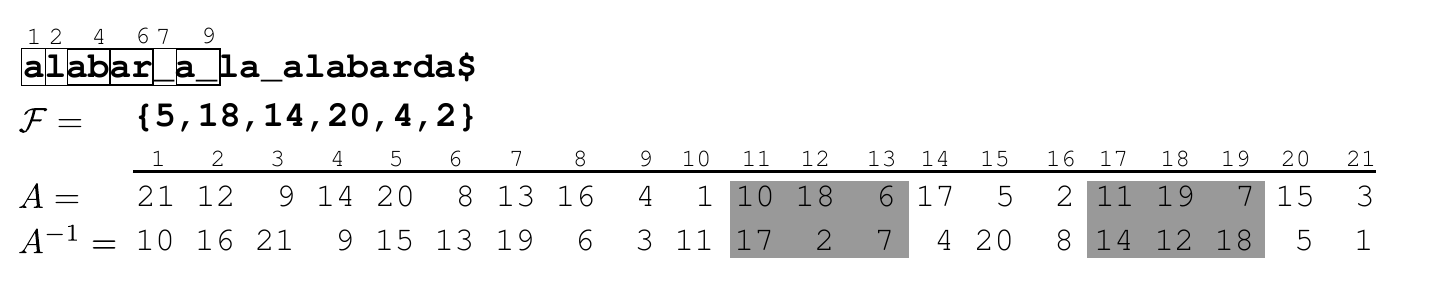}
\caption{Example of LZ-End construction algorithm}
\label{fig:construction_lzend_ex}
\end{figure}

\begin{example}
\nieves{Figure \ref{fig:construction_lzend_ex} shows the structures used during the parsing of the string \texttt{`alabar\_a\_la\_alabarda\$'}. The array $A$ corresponds to the suffix array of the reversed text and $A^{-1}$ to its inverse permutation. The figure shows the parsing up to the 6th phrase and the values inserted in the dictionary $\mathcal{F}$. The values inserted in $\mathcal{F}$ are $A^{-1}[len-i]$, where $i$ is the ending position of a phrase and $len$ the length of the text. For example, the second phrase ends in position 2, thus the value inserted corresponds to $A^{-1}[21-2]=18$.}\\

\nieves{Now we continue the process to generate the next phrase. First, using BWS we find the interval of $A$ that represents the suffixes (of the reverse text) starting with \texttt{`l'}, obtaining the range $[17,18]$ (right gray zone). Then we look in $\mathcal{F}$ for the successor of $17$, obtaining the value $18$, which is still in the range. Hence, we have found a valid source. Afterward, we continue with the next character. Again, with BWS we find the interval of $A$ representing the suffixes (of the reverse text) starting with \texttt{`al'} (left gray zone), which are the prefixes of the text ending with \texttt{`la'}. This gives us the range $[11,13]$. We look for the successor of $11$ in $\mathcal{F}$, which is $14$. Since this value is outside the interval, there are no valid sources. We continue this process until there are no more possible sources. Finally, we get that the only valid source is \texttt{`l'}, generating the new phrase \texttt{`la'}.}
\end{example}

In theory the construction algorithm can work within bit space 
(1) $nH_k(T^{rev})+o(n\log\sigma) = nH_k(T)+o(n\log\sigma)$ (since $nH_k(T) = nH_k(T^{rev}) + O(\log n)$ \cite[Theorem A.3]{FM05}) for building the BWT incrementally \cite{GN08}; plus
(2) $2n+o(n)$ bits for the RMQ structure \cite{FH07}; plus
(3) $O(n'\log n)$ bits for a successor data structure.
After building the BWT incrementally in time $O(n\log n \lceil
\frac{\log\sigma}{\log\log n} \rceil)$ \cite{GN08}, we can make it static, so
that it supports access to the successive characters of $T$ in time
$O(\lceil\frac{\log\sigma}{\log\log n}\rceil)$, as well to $A$ and $A^{-1}$ 
in time $O(\log^{1+\epsilon} n)$ for any constant $\epsilon>0$ \cite{FMMN07}.
The RMQ structure is built in $O(n)$ time and within the same final space,
and answers queries in constant time.
The successor data structure could be a simple balanced search tree, with
leaves holding $\Theta(\log n)$ elements, so that the access time is
$O(\log n)$ and the space is $n'\log n (1+o(1))$ \cite{Munro86}.
Thus, using Lemmas \ref{lem:folk} and \ref{lem:entropy}, the overall construction space is
$2n(H_k(T)+1)+o(n\log\sigma)$ bits, for any $k=o(\log_\sigma n)$. 
The time is dominated by the BWT construction, 
$O(n\log n \lceil \frac{\log\sigma}{\log\log n} \rceil)$, 
plus the $N$ accesses to $A$, $O(N\log^{1+\epsilon} n)$.
If, instead, we use $O(n\log n)$ bits of space, we can build and store 
explicitly $A$ and $A^{-1}$ in $O(n)$ time \cite{KS03}. The overall time
becomes $O(N\lceil\frac{\log\sigma}{\log\log n}\rceil)$.

Note that a simplification of our construction algorithm, disregarding
$\mathcal{F}$ (and thus $N=n$) builds the LZ77 parsing using just
$n(H_k(T)+2)+o(n\log\sigma)$ bits and 
$O(n\log n(\log^\epsilon n + o(\log\sigma)))$ time, which is less than 
the best existing solutions \cite{OS08,CPS08}. 

In practice our implementation of the algorithm works within \emph{byte} space 
(1) $n$ as we maintain $T$ explicitly; plus
(2) $2.02n$ for our implementation of BWT (following Navarro's ``large''
FM-index implementation \cite{Nav09}, where $L$ is maintained explicitly); plus
(3) $4n$ for $A$, which is explicitly maintained; plus
(4) $0.7n$ for Fischer's implementation of RMQ \cite{FH07}; plus
(5) $n$ for $A^{-1}$, using a sampling-based implementation of inverse
permutations~\cite{MRRR03} (Section \ref{sec:permutation}); plus
(6) $12n'$ for a balanced binary tree implementing the successor structure.
This adds up to less than $10n$ bytes in practice.
$A$ is built in time $O(n\log n)$ in practice; other construction times are $O(n)$. 
After this, the time of the algorithm is $O(N\log n') = O(N\log n)$.

As we see soon, $N$ is usually (but not always) only slightly larger than $n$;
we now prove it is limited by the phrase lengths. 

\begin{lemma}
The amount of text retraversed at any step is $< |Z[p]|$ for some $p$.
\end{lemma}
\begin{proof}
Say the last valid match $T[i,j-1]$ was with suffix $Z[1]\ldots Z[p-1]$
for some $p$, thereafter we worked until $T[i,i'-1]$ without finding
any other valid match, and then formed the phrase (with source $p-1$). Then
we will retraverse $T[j+1,i'-1]$, which must be shorter than $Z[p]$ since
otherwise $Z[1]\ldots Z[p]$ would have been a valid match.
\end{proof}

\begin{remark}
On ergodic sources with entropy $h$, $N=O(\frac{n}{h}\log n)$, but as explained this is not a realistic model on repetitive texts.
\end{remark}




\section{Experimental Results}

We implemented two different LZ-End encoding schemes. The first is as 
explained in Section~\ref{sec:encoding}. In the second ({\em LZ-End2}) we 
store the starting position of the source, $select_1(B,\source[p])$, rather than 
the identifier of the source, $\source[p]$. This in theory raises the $nH_k(T)$
term in the space to $2nH_k(T)$ (and noticeably in practice, as seen soon), 
yet we save one $select$ operation at extraction time (line 13 in
Figure~\ref{alg:extraction}), which has a significant 
impact in performance. In both implementations,
bitmap $B$ is represented by $\delta$-encoding the consecutive phrase lengths (Section \ref{sec:deltacodes}). 
Recall that, in a $\delta$-encoded bitmap, $select_1(B,p)$ and $select_1(B,p+1)$ cost $O(1)$ after solving $p \leftarrow rank_1(B,end)$, thus LZ-End2 does no $select$ operations for extracting.

We compare our compressors with {\em LZ77} and {\em LZ78} implemented by 
ourselves. \jeremy{LZ77 triples are encoded in the same way as LZ-End2.} We include the best performing compressors of Chapter 3, p7zip and Re-Pair. Compared
to p7zip, LZ77 differs in the final
encoding of the triples, which p7zip does better. This is orthogonal to the
parsing issue we focus on in this thesis. We also implemented 
{\em LZB} \cite{lzb}, which limits the distance \emph{dist} at
which the phrases can be from their original (not transitive) sources, so one 
can decompress any window by starting from that distance behind; and 
{\em LZ-Cost}, a novel proposal where we limit the number of times any text 
character can be copied (i.e., its $C[\cdot]$ value in Definition \ref{def:height}), thus directly limiting the maximum cost per character 
extraction.
We have found no efficient parsing algorithm for {\em LZB} and 
{\em LZ-Cost}, thus we test them on small texts only. 
We also implemented {\em LZ-Begin}, the ``symmetric'' variant of LZ-End, which
also allows random phrase extraction in constant time per extracted symbol. LZ-Begin forces the
source of a phrase to start where some previous phrase starts, just like
Fiala and Green \cite{FG89}, yet phrases have a leading rather than a 
trailing character. Although the parsing is much simpler, the compression ratio 
is noticeably worse than that of LZ-End, as we \jeremy{will} see \jeremy{in Section \ref{sec:compression_lzend}}\djeremy{soon}.
 
We used the
texts of the Canterbury corpus (\url{http://corpus.canterbury.ac.nz}), the 
50MB texts from the Pizza\&Chili corpus 
(\url{http://pizzachili.dcc.uchile.cl}), and highly repetitive texts from the previous chapter.
We use a 3.0 GHz Core 2 Duo processor with 4GB of main memory, running Linux 
2.6.24 and g++ (gcc version 4.2.4) compiler with -O3 optimization. 

\subsection{Compression Ratio}
\label{sec:compression_lzend}

Table~\ref{tab:compression} gives compression ratios for the different
collections and parsers. \jeremy{Figure \ref{fig:lzparsing_compression} shows the same results graphically for one representative text of each collection.} For LZ-End we omit the sampling for bitmap $B$, as it 
can be reconstructed on the fly at loading time. LZ-End is usually 5\% worse than 
LZ77, and at most 10\%
over it on general texts and 20\% on the highly repetitive collections, where
the compression ratios are nevertheless excellent. LZ78 is from 20\% better
to 25\% worse than LZ-End on typical texts, but it is orders of magnitude worse
on highly
repetitive collections. With parameter $\log(n)/2$, LZ-Cost is usually close 
to LZ77, yet sometimes it is much worse, and it is never better than LZ-End
except by negligible margins. LZB is not competitive at all. Finally,
LZ-Begin is about 30\% worse \jeremy{than LZ77} on typical texts, and up to 40 times worse for 
repetitive texts. This is because not all phrases of the parsing are unique (Lemma \ref{lemma:lzbegin}). This property was the key to prove the coarse optimality of the LZ parsings.
\begin{lemma}
\label{lemma:lzbegin}
Not all the phrases generated by an LZ-Begin parsing are different.
\end{lemma}
\begin{proof}
We prove this lemma by showing a counter-example. Let $T=A\texttt{xy}A\texttt{y}A\texttt{z}$, where $\texttt{x},\texttt{y},\texttt{z}$ are distinct characters and $A$ is a string. Suppose we have parsed up to $A\texttt{x}$, then the next phrase will be $\texttt{y}A$, and the following phrase will also be $\texttt{y}A$. 
\end{proof}
\begin{figure}[!ht]
	\centering
	\includegraphics[angle=-90,width=8cm]{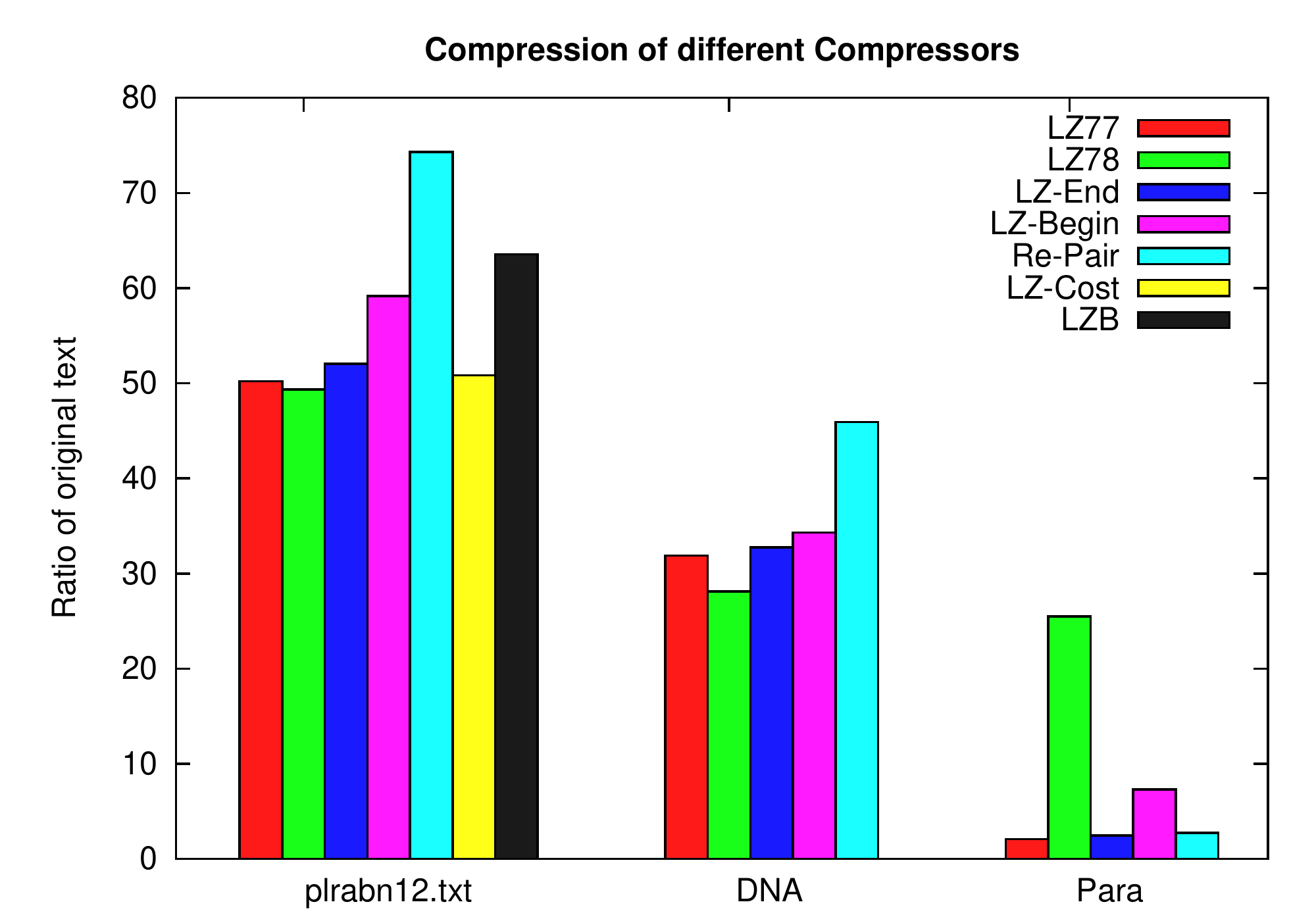}
	\caption{Compression ratio for different compressors}
\label{fig:lzparsing_compression}
\end{figure}

\begin{table}[tbp]
\begin{center}
{\scriptsize
\begin{tabular}{|l|r|r|r|r|r|r|r|r|}
\hline
 & \multicolumn{1}{l|}{} & \multicolumn{1}{c|}{\textbf{LZ77}} &
\multicolumn{1}{c|}{\textbf{LZ78}} & \multicolumn{1}{c|}{\textbf{LZ-End}} &
\multicolumn{1}{c|}{\textbf{LZ-Cost}} & \multicolumn{1}{c|}{\textbf{LZB}} &
\multicolumn{1}{c|}{\textbf{LZ-Begin}} & \multicolumn{1}{c|}{\textbf{Re-Pair}} \\ \hline
\hline
\textbf{Canterbury} & Size(KiB) & \multicolumn{7}{|l|}{} \\ \hline
alice29.txt  & 148.52  & 47.17\% & 49.91\% & 49.32\% & 48.51\% & 61.75\% & 59.02\% & 72.29\% \\ \hline
asyoulik.txt & 122.25  & 51.71\% & 52.95\% & 53.51\% & 52.41\% & 66.42\% & 62.34\% & 81.52\% \\ \hline
cp.html      & 24.03   & 43.61\% & 53.60\% & 45.53\% & 46.27\% & 66.26\% & 56.93\% & 78.65\% \\ \hline
fields.c     & 10.89   & 39.21\% & 54.73\% & 41.69\% & 44.44\% & 61.32\% & 60.61\% & 65.19\% \\ \hline
grammar.lsp  & 3.63    & 48.48\% & 57.85\% & 50.41\% & 56.30 \% & 67.02\% & 67.14\% & 85.60\% \\ \hline
lcet10.txt   & 416.75  & 42.62\% & 46.83\% & 44.65\% & 43.44\% & 56.72\% & 54.21\% & 57.47\% \\ \hline
plrabn12.txt & 470.57  & 50.21\% & 49.34\% & 52.06\% & 50.83\% & 63.55\% & 59.15\% & 74.32\% \\ \hline
xargs.1      & 4.13    & 57.87\% & 65.38\% & 59.56\% & 59.45\% & 86.37\% & 73.14\% & 107.33\% \\ \hline
aaa.txt      & 97.66   & 0.055\% & 0.51\% & 0.045\% & 1.56\% & 0.95\% & 0.040\% & 0.045\% \\ \hline
alphabet.txt & 97.66   & 0.110\% & 4.31\% & 0.105\% & 0.23\% & 1.15\% & 0.100\% & 0.081\% \\ \hline
random.txt   & 97.66   & 107.39\% & 90.10\% & 105.43\% & 107.40\% & 121.11\% & 106.9\% & 219.24\% \\ \hline
E.coli       & 4529.97 & 34.13\% & 27.70\% & 34.72\% & - & - & 35.99\% & 57.63\% \\ \hline
bible.txt    & 3952.53 & 34.18\% & 36.27\% & 36.44\% & - & - & 43.98\% & 41.81\% \\ \hline
world192.txt & 2415.43 & 29.04\% & 38.52\% & 30.99\% & - & - & 41.52\% & 38.29\% \\ \hline
pi.txt       & 976.56  & 55.73\% & 47.13\% & 55.99\% & - & - & 57.36\% & 108.08\% \\ \hline
\end{tabular}
\vspace{1cm}\\
%
\begin{tabular}{|p{2.5cm}|x{1.35cm}|x{1.67cm}|x{1.67cm}|x{1.67cm}|x{1.67cm}|x{1.67cm}|}
\hline
 & \multicolumn{1}{l|}{} & \multicolumn{1}{c|}{\textbf{LZ77}} &
\multicolumn{1}{c|}{\textbf{LZ78}} & \multicolumn{1}{c|}{\textbf{LZ-End}} &
\multicolumn{1}{c|}{\textbf{LZ-Begin}} & \multicolumn{1}{c|}{\textbf{Re-Pair}} \tn \hline
\textbf{Pizza Chili} & Size(MiB) & \multicolumn{5}{l|}{}\tn \hline
Sources  & 50 & 28.50\% & 41.14\% & 31.00\% & 41.95\% & 31.07\% \tn \hline
Pitches  & 50 & 44.50\% & 59.30\% & 45.78\% & 57.22\% & 59.90\% \tn \hline
Proteins & 50 & 47.80\% & 53.20\% & 47.84\% & 54.95\% & 71.29\% \tn \hline
DNA      & 50 & 31.88\% & 28.12\% & 32.76\% & 34.28\% & 45.90\% \tn \hline
English  & 50 & 31.12\% & 41.80\% & 31.12\% & 38.54\% & 30.50\% \tn \hline
XML      & 50 & 17.00\% & 21.24\% & 17.64\% & 25.49\% & 18.50\% \tn \hline
\end{tabular}
\vspace{1cm}\\


\begin{tabular}{|p{2.5cm}|x{1.35cm}|x{1.67cm}|x{1.67cm}|x{1.67cm}|x{1.67cm}|x{1.67cm}|}
\hline
 & \multicolumn{1}{l|}{} & \multicolumn{1}{c|}{\textbf{LZ77}} &
\multicolumn{1}{c|}{\textbf{LZ78}} & \multicolumn{1}{c|}{\textbf{LZ-End}} &
\multicolumn{1}{c|}{\textbf{LZ-Begin}} & \multicolumn{1}{c|}{\textbf{Re-Pair}} \tn \hline
\textbf{Repetitive} & Size(MiB) & \multicolumn{5}{l|}{} \tn \hline
Wikipedia Einstein         & 357.40 & 9.97$\times 10^{-2}$\% & 9.29\% & 1.01$\times 10^{-1}$\% & 4.27\% & 1.04$\times 10^{-1}$\% \tn \hline
World Leaders          &  40.65 & 1.73\% & 15.89\% & 1.93\% & 7.97\% & 1.89\% \tn \hline
Rich String 11             &  48.80 & 3.20$\times 10^{-4}$\% & 0.82\% & 4.18$\times 10^{-4}$\% & 0.01\% & 3.75$\times 10^{-4}$\% \tn \hline
Fibonacci 42               & 255.50 & 7.32$\times 10^{-5}$\% & 0.40\% & 5.32$\times 10^{-5}$\% & 6.07$ \times 10^{-5}$\% & 2.13$\times 10^{-5}$\% \tn \hline
Para                       & 409.38 & 2.09\% & 25.49\% & 2.48\% & 7.29\% & 2.74\% \tn \hline
Cere                       & 439.92 & 1.48\% & 25.33\% & 1.74\% & 6.15\% & 1.86\% \tn \hline
Coreutils                  & 195.77 & 3.18\% & 27.57\% & 3.35\% & 7.33\% & 2.54\% \tn \hline
Kernel                     & 246.01 & 1.35\% & 30.02\% & 1.43\% & 3.43\% & 1.10\% \tn \hline
\end{tabular}
}
\end{center}
\caption[Compression ratio of different LZ-like parsings]{Compression ratio of different parsings, in percentage of compressed
over original size. We use parameter $cost=(\log n)/2$ for {\em LZ-Cost} and $dist=n/5$ for {\em LZB}.}
\label{tab:compression}
\end{table}

The above results show that LZ-End achieves very competitive compression ratios, even in
the challenging case of highly repetitive sequences, where LZ77 excells.

Consistently with Chapter 3, Re-Pair results show that grammar-based compression is a relevant
alternative. Yet, we note that it is only competitive on highly repetitive
sequences, where most of the compressed data is in the dictionary. This
implementation applies sophisticated compression to the dictionary,
which we do not apply on our compressors. Those sophisticated dictionary compression techniques prevent direct access to
the grammar rules, essential for extracting substrings.

\subsection{Parsing Time}

Figure~\ref{fig:construction_time} shows parsing times on two files for 
LZ77 (implemented following CSP2 \cite{CPS08}), LZ-End with the algorithm of
Section~\ref{sec:constr}, and p7zip. We show separately the time of the suffix
array construction algorithm we use, \texttt{libdivsufsort}
(\texttt{http://code.google.com/p/libdivsufsort}),
common to LZ77 and LZ-End. 

\begin{figure}[!ht]
	\centering
	\includegraphics[angle=-90,scale=0.372]{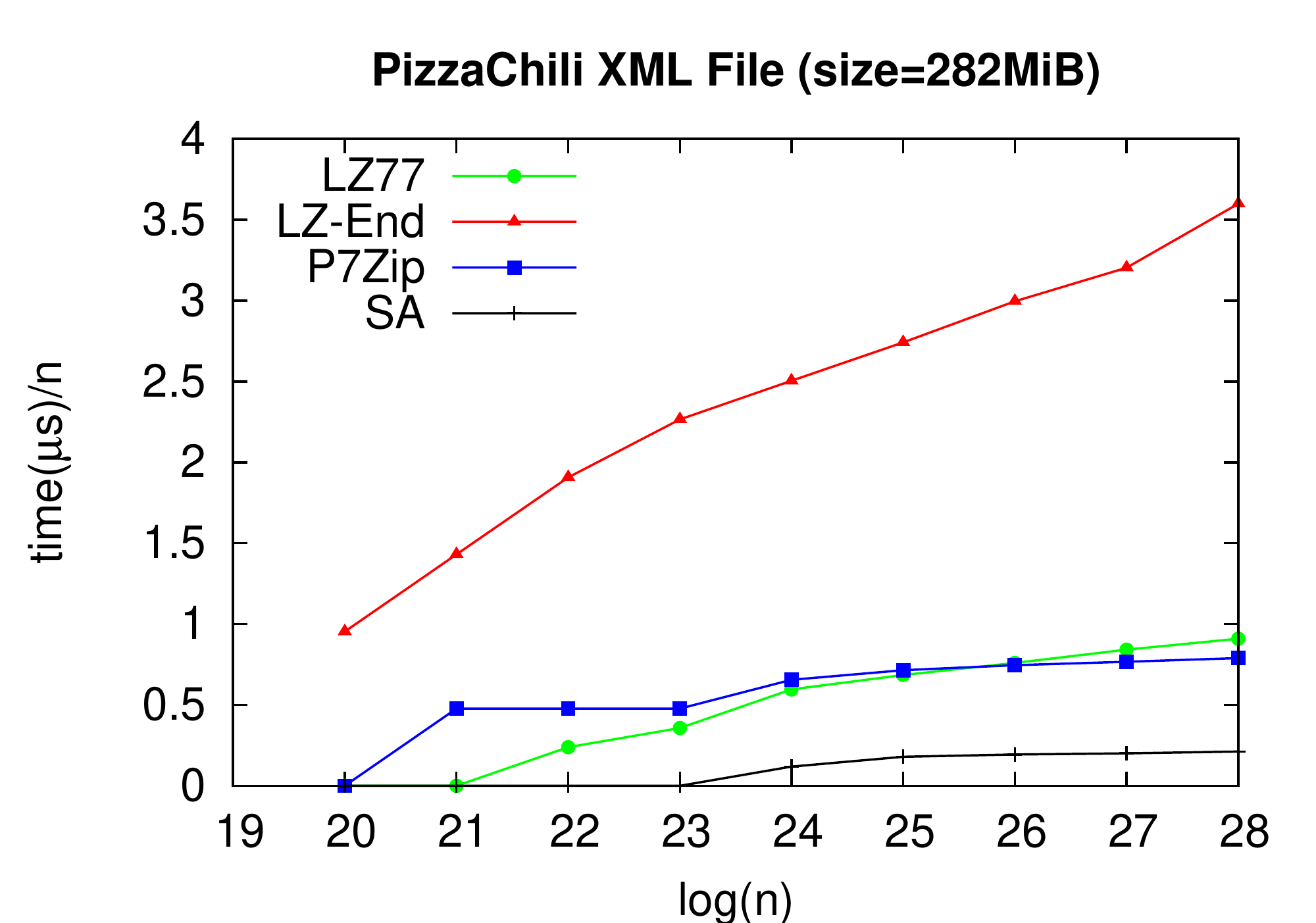}
    \includegraphics[angle=-90,scale=0.372]{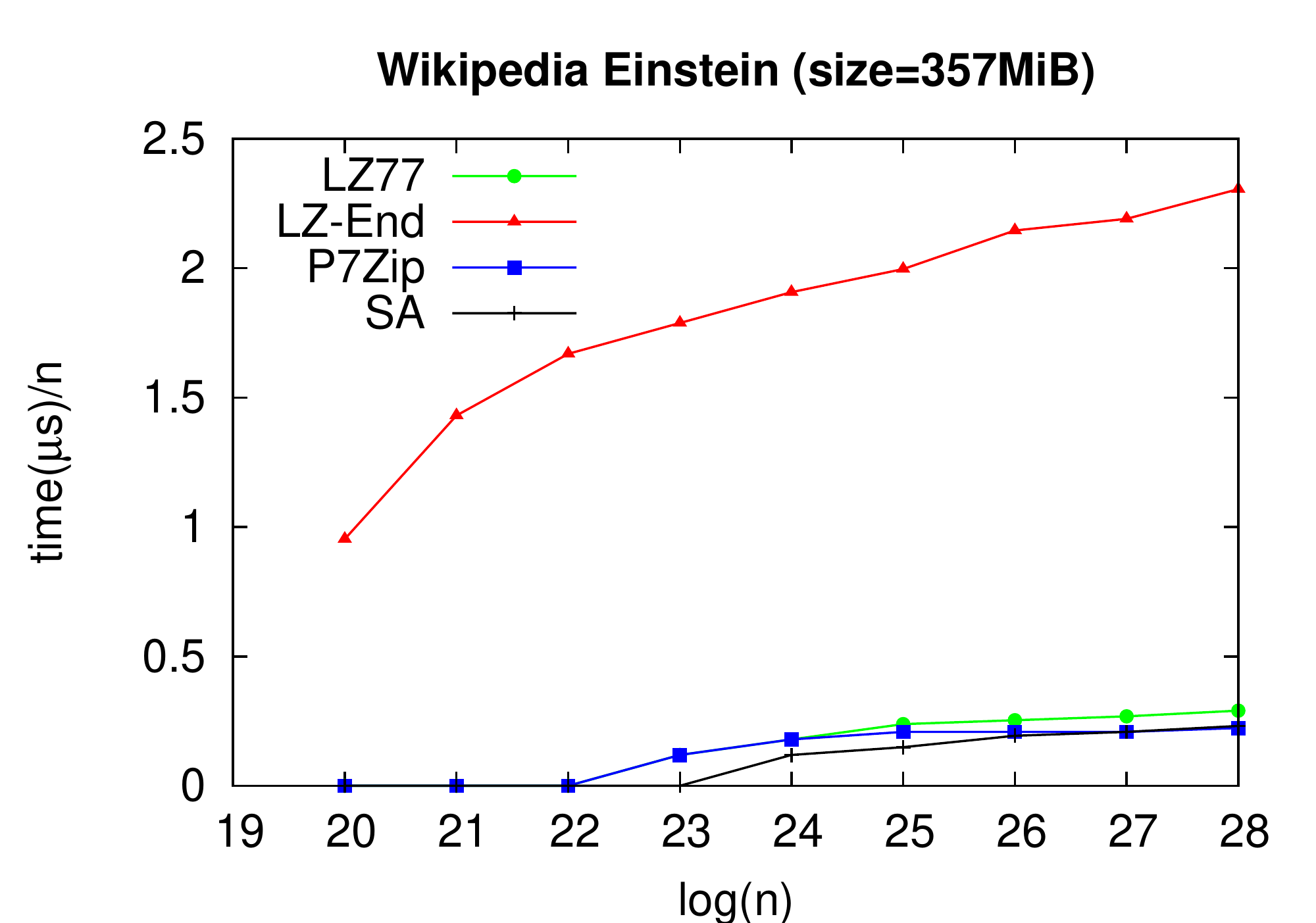}
	\caption[LZ77 and LZ-End parsing times]{
Parsing times for XML and Wikipedia Einstein, in microseconds per character.
}
	\label{fig:construction_time}
\end{figure}

Our LZ77 construction time is competitive with the state
of the art (p7zip), thus the excess of LZ-End is due to the more
complex parsing. Least squares fitting for the nanoseconds/char yields 
$10.4 \log n+O(1/n)$ (LZ77) and $82.3 \log n +O(1/n)$ (LZ-End) for 
Einstein text, and $32.6 \log n +O(1/n)$ (LZ77) and $127.9 \log n +O(1/n)$ 
(LZ-End) for XML. The correlation coefficient is always over 0.999, which
suggests that $N=O(n)$ and our parsing \djeremy{is}takes $O(n\log n)$ time in practice. 
Indeed, across all of our collections, the ratio $N/n$ stays between 1.05 and 
1.37, except on {\tt aaa.txt} and {\tt alphabet.txt}, where it is 10--14 (which
suggests that $N=\omega(n)$ in the worst case). Figure \ref{fig:lzend_total_work} shows the total text traversed by LZ-End parsing algorithm for two different texts.

\begin{figure}[!t]
	\centering
	\includegraphics[angle=-90,scale=0.372]{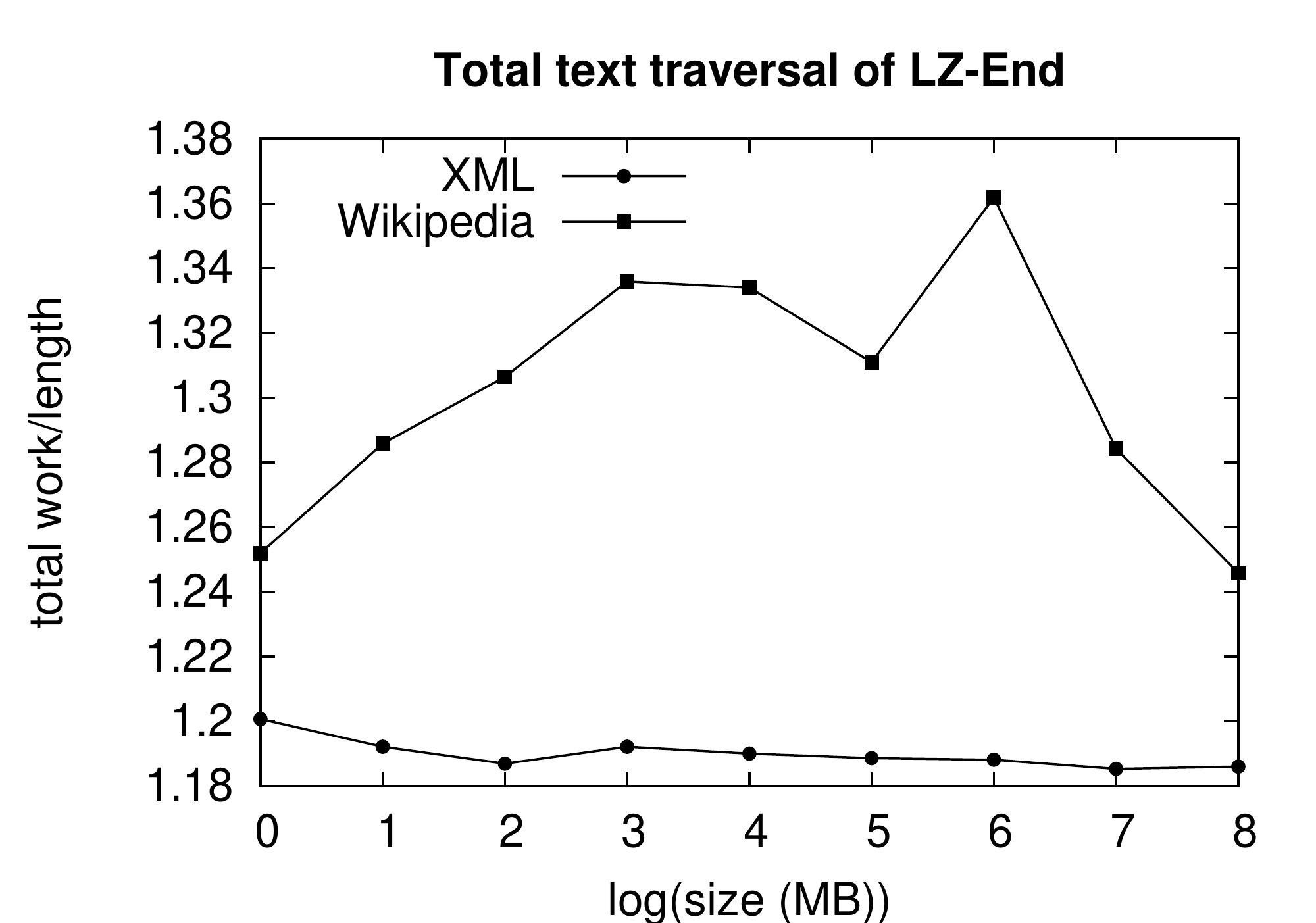}
	\caption{Total text traversed during LZ-End construction algorithm.}
	\label{fig:lzend_total_work}
\end{figure}

The LZ-End parsing time breaks down as follows. For XML: BWS 36\%, RMQ 19\%, 
tree operations 33\%, SA construction 6\% and inverse SA lookups 6\%. For
Einstein: BWS 56\%, RMQ 19\%, tree operations 17\%, and SA construction 8\% (the inverse SA lookups take negligible time).


\subsection{Text Extraction Speed}
\label{sec:lzend_ext_res}
Figure~\ref{fig:extraction_length} shows the extraction speed of 
{\em arbitrary} substrings of 
increasing length. The three parsings (LZ77, LZ-End and LZ-End2) are parameterized to use 
approximately the same space, 550KiB for Wikipedia Einstein and 64MiB for XML. This is achieved by adjusting the sample step $s$ of the $\delta$-encoded bitmaps (Section \ref{sec:deltacodes}).
It can be seen that (1) the time per character stabilizes after some extracted
length, as expected from Lemma \ref{lemma:cost_lzend}, (2) LZ-End variants extract faster than LZ77,
especially on very repetitive collections, and (3) LZ-End2 is faster than
LZ-End, even if the latter invests its better compression in a denser
sampling. Least squares fitting for the extraction time of a substring of length $m$ are given in Table \ref{tab:model}.

\begin{table}[!ht]
\begin{center}
\begin{tabular}{|l|r|}
\multicolumn{2}{c}{Pizza\&Chili XML} \\ \hline
Scheme & Model \\ \hline
LZ77 &  $4.44+0.33m$\\
LZ-End & $7.40+0.36m$\\
LZ-End2 & $6.41+0.27m$\\ \hline
\end{tabular}
\begin{tabular}{|l|r|}
\multicolumn{2}{c}{Wikipedia Einstein} \\ \hline
Scheme & Model \\ \hline
LZ77 &  $19.09+ 0.38m$\\
LZ-End & $5.75+0.19m$\\
LZ-End2 & $5.64+0.19m$\\ \hline
\end{tabular}
\end{center}
\caption[Least squares fitting for extraction time]{Least squares fitting for extraction time. All correlation coefficients are always over 0.999.}
\label{tab:model}
\end{table}

\begin{figure}[!ht]
	\centering
	\includegraphics[angle=-90,scale=0.372]{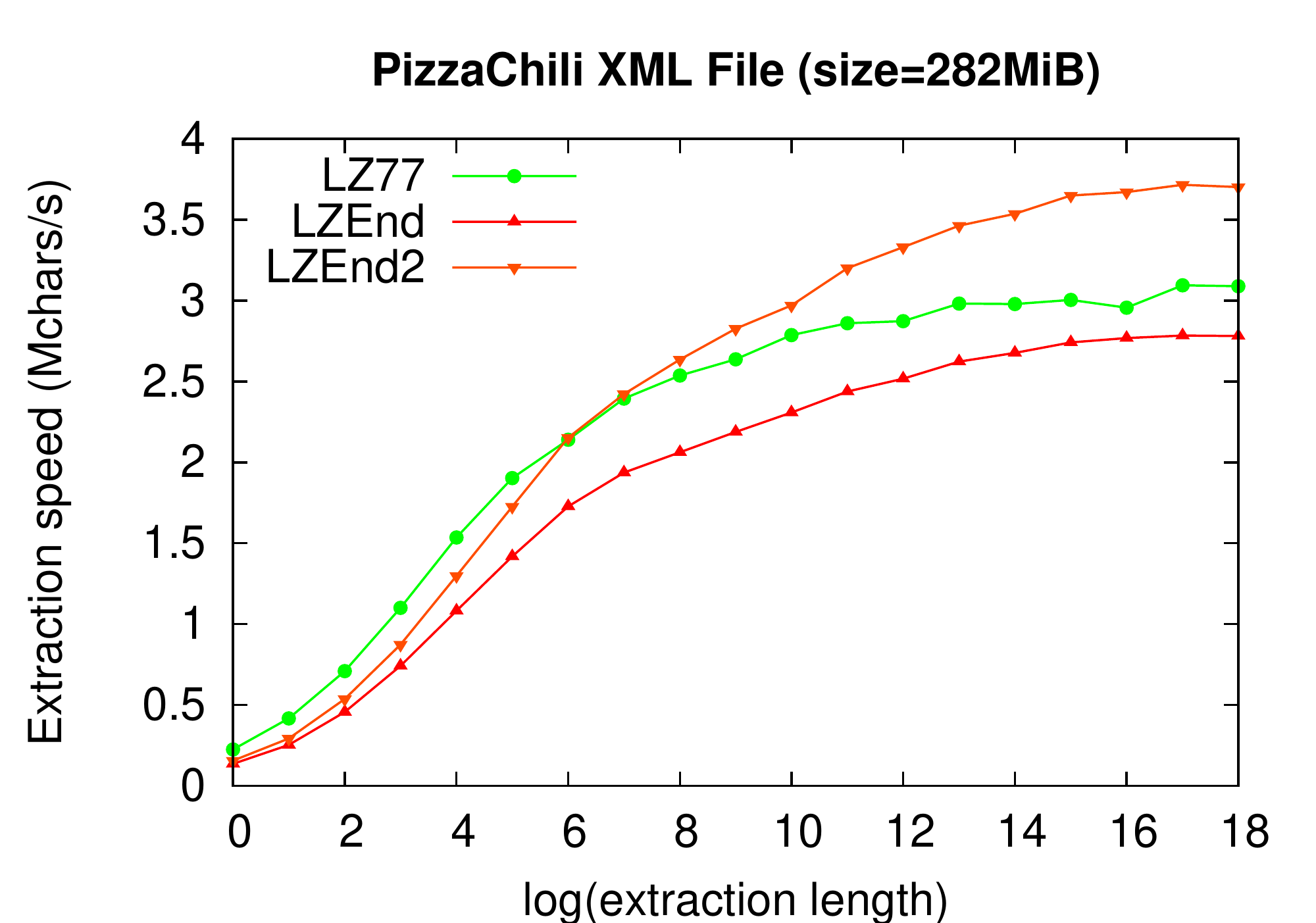}
    \includegraphics[angle=-90,scale=0.372]{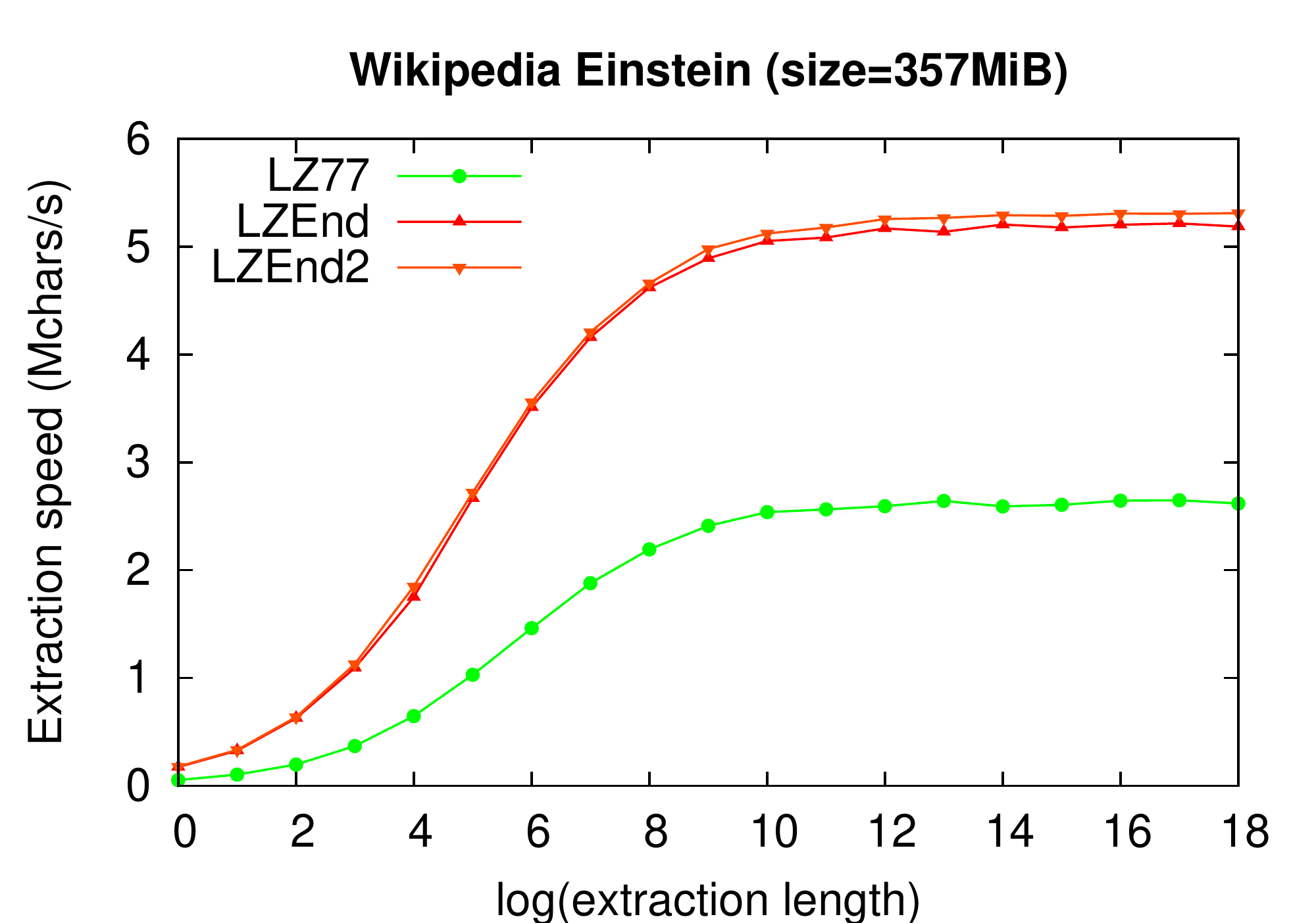}
	\caption[LZ extraction speed vs extracted length]{
Extraction speed vs extracted length, for XML and Wikipedia Einstein.
}
	\label{fig:extraction_length}
\end{figure}

We now set the extraction length to 1,000 and measure the extraction speed per character,
as a function of the space used by the data and the sampling. Here we use
bitmap $B$ and its sampling for the other formats as well. LZB and LZ-Cost have
also their own space/time trade-off parameter; we tried several
combinations and chose the points dominating the others.
Figure \ref{fig:extraction} shows the results for small and large files.

\begin{figure}[!ht]
	\centering
	\includegraphics[angle=-90,scale=0.372]{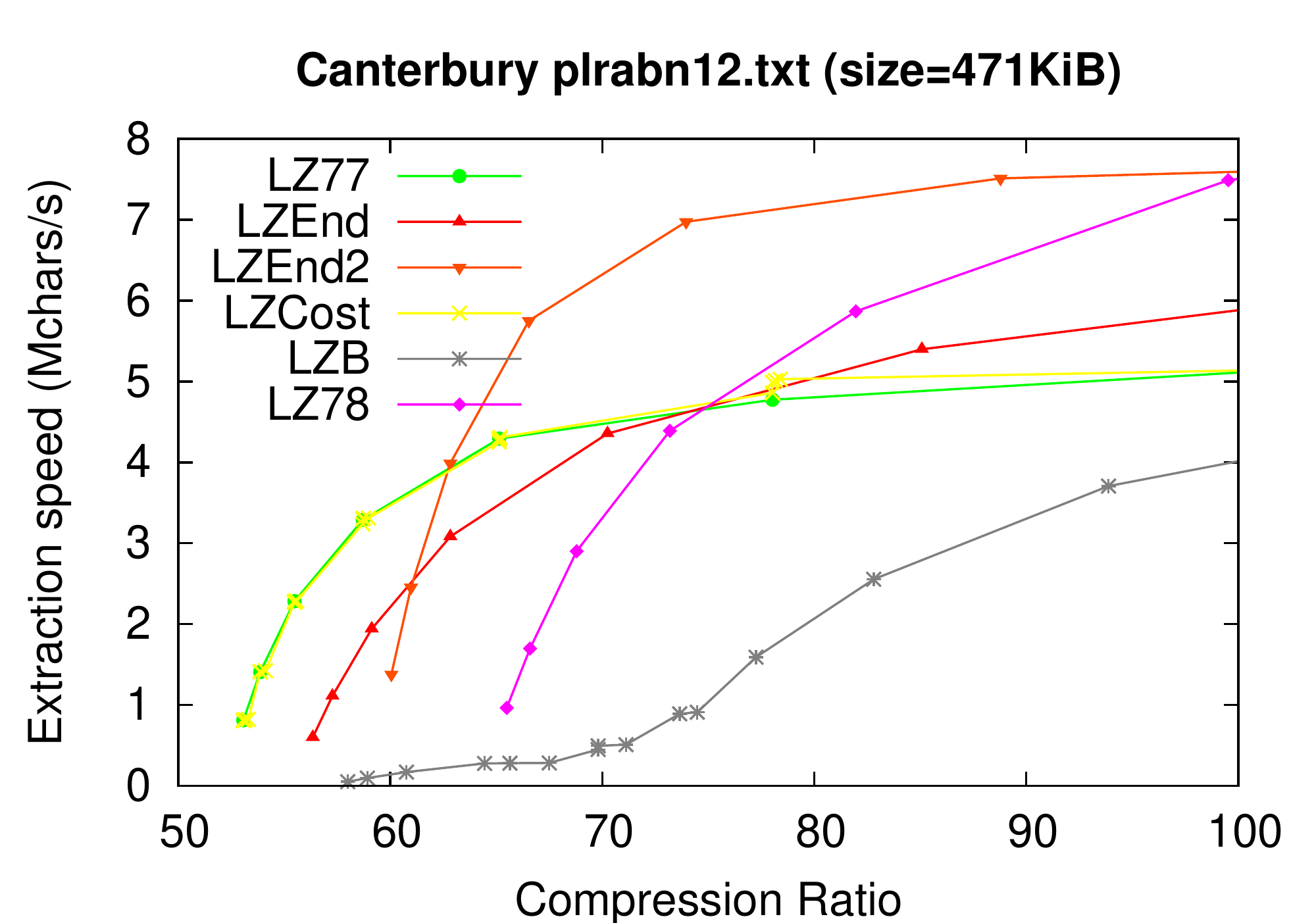}
    \includegraphics[angle=-90,scale=0.372]{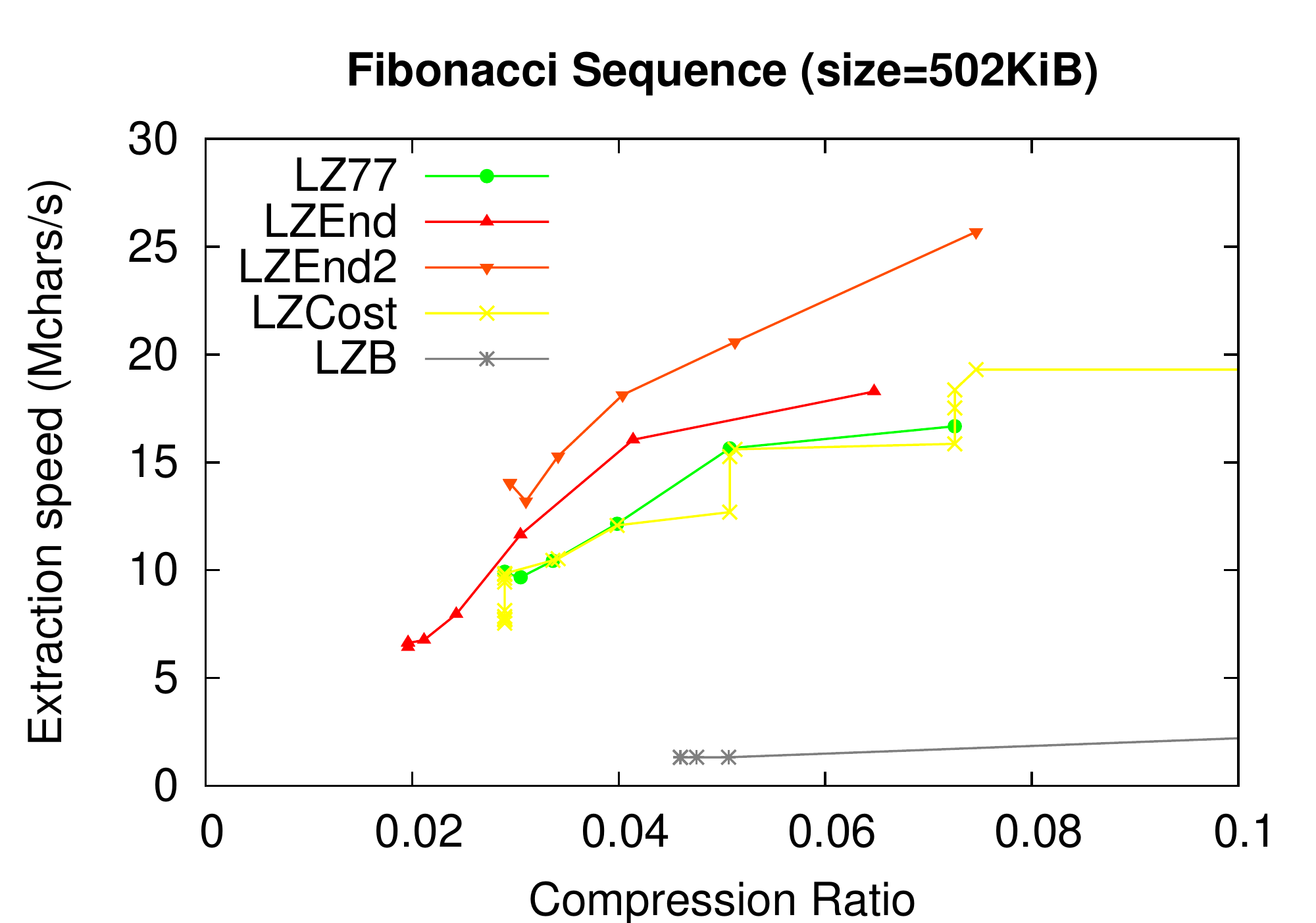}
	\includegraphics[angle=-90,scale=0.372]{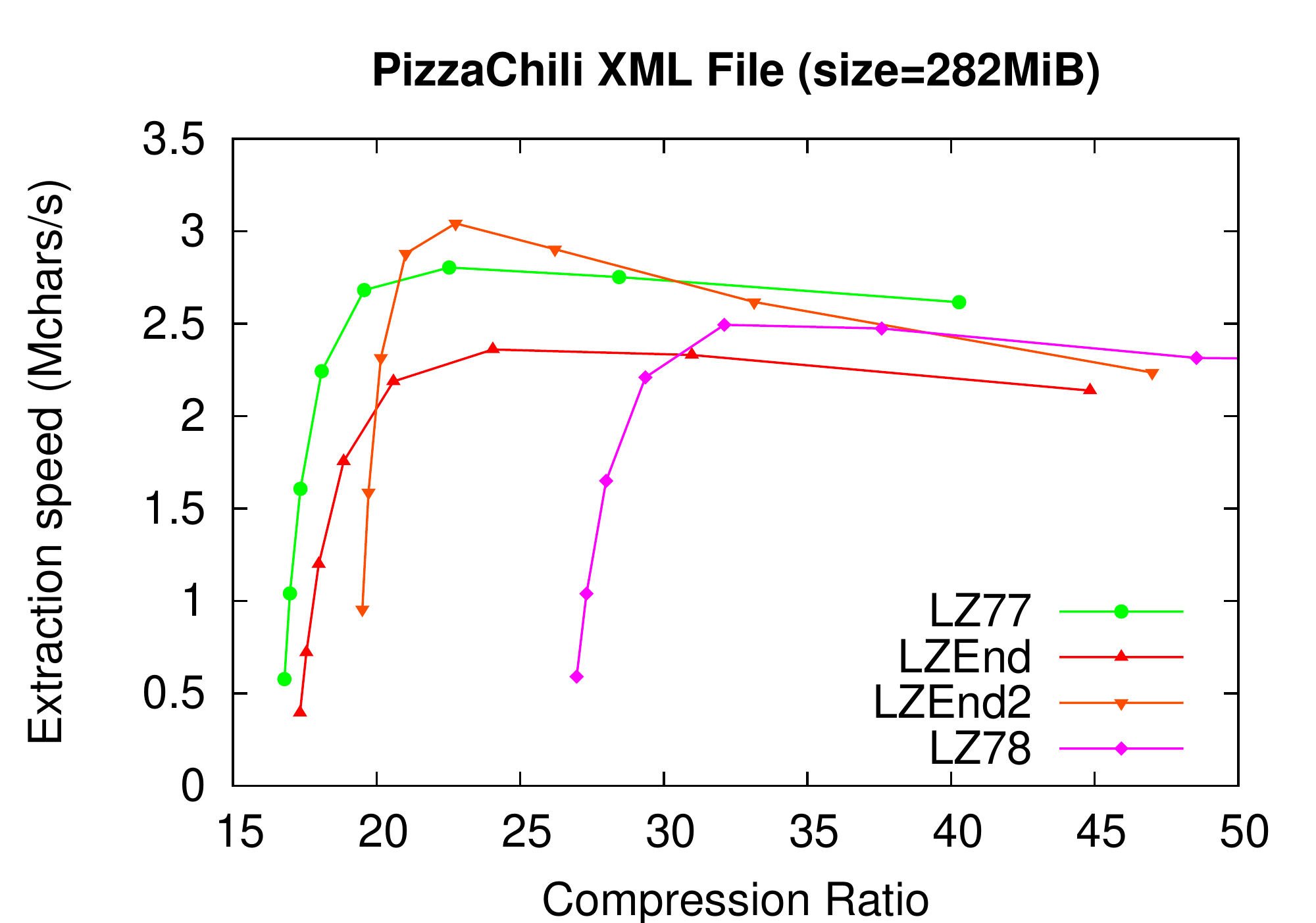}
    \includegraphics[angle=-90,scale=0.372]{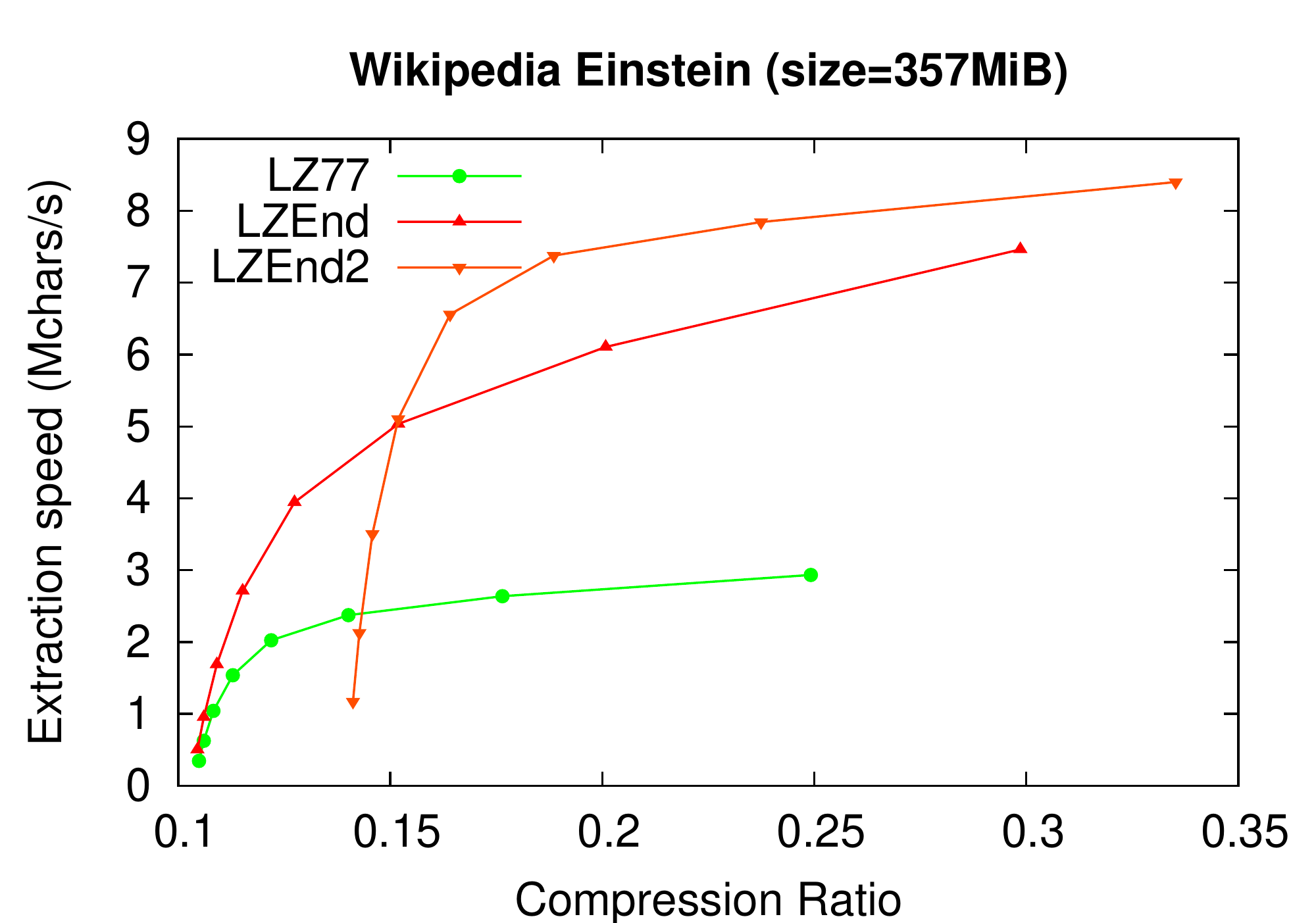}
    \caption[LZ extraction speed vs parsing size]{
Extraction speed vs parsing and sampling size, on different texts.
}
	\label{fig:extraction}
\end{figure}


It can be seen that LZB is not competitive, whereas LZ-Cost follows LZ77
closely (while offering a worst-case guarantee). The LZ-End variants dominate
all the trade-off except when LZ77/LZ-Cost are able of using less space. On
repetitive collections, LZ-End2 is more than 2.5 times faster than LZ77
at extraction.

%% file: lz77_index.tex
\chapter{An LZ77-Based Self-Index}
In this chapter we describe a self-index based on the LZ77 parsing. It builds on the ideas of the original LZ-based index proposed by K{\"a}rkk{\"a}inen and Ukkonen \cite{KU96,Kar99} and the ideas presented by Navarro for reducing its space usage \cite{Nav08}. Our index will be mostly independent \jeremy{of}\djeremy{on} the type of Lempel-Ziv parsing used, and we will combine it with LZ77 and LZ-End.
We use compact data structures to achieve the minimum possible space. These structures also allow one to convert the original index into a self-index, \jeremy{so that}\djeremy{hence} we do not need the text anymore.

As we will show, the index includes all the structures needed to randomly extract any substring from the text, introduced in the previous chapter. 
The worst-case time to extract a substring of length $\ell$ is $O(\ell H)$ for LZ77 and $O(\ell+H)$ for LZ-End (see Section \ref{sec:lz_extraction}). 
Additionally, the proposed index only supports \emph{count} queries by performing a full \emph{locate}, and \emph{exists} queries by essentially locating one occurrence. For these reasons, in the following we focus \diego{only} on \ddiego{the} \emph{locate} \ddiego{query} \diego{queries}.
\section{Basic Definitions}
Assume we have a text $T$ of length $n$, which is partitioned into $n'$ phrases using a LZ77-like compressor (see Chapter \ref{chap:parsing}).
Let $P[1,m]$ be a search pattern. We will call \emph{primary occurrences} of $P$ those covering more than one phrase; \emph{special primary occurrences} those ending at the end of a phrase and being completely covered by the phrase; and \emph{secondary occurrences} those occurrences completely covered by a phrase and not ending at an end of phrase.

\begin{example}
\begin{displaymath}
\setlength\arraycolsep{0em}
\begin{array}{ccccccccc}
& & & & & & \text{\scriptsize{1 1 1}} & \text{\scriptsize{1 1 1 1 1 1 1}} & \text{\scriptsize{2 2}}\\[-0.5em]
\text{\scriptsize{1}} & \text{\scriptsize{2}} & \text{\scriptsize{3 4}} & \text{\scriptsize{5 6}} & \text{\scriptsize{7}} & \text{\scriptsize{8 9}} & \text{\scriptsize{0 1 2}} &  \text{\scriptsize{3 4 5 6 7 8 9}} & \text{\scriptsize{0 1}}\\
\phrase{a} & \phrase{\textcolor{red}{l}} & \phrase{\textcolor{red}{ab}} & \phrase{ar} & \phrase{\_} & \phrase{a\_} & \phrase{la\_} & \phrase{a\textcolor{blue}{lab}a\textcolor{green}{rd}} & \phrase{a\$} 
\end{array}
\end{displaymath}
In this example the occurrence of \texttt{`lab'} starting at position 2 (red color) is primary as it spans two phrases. The second occurrence, starting at position 14 (blue color) is secondary. The occurrence of \texttt{`rd'} starting at position 18 (green color) is special primary.
\end{example}

\ddiego{We distinguish between these three types of occurrences, as we are first finding the primary occurrences (including the special ones), and recursively from those the secondary ones.}
\diego{We need to distinguish between these three types of occurrences, as we will find first the primary occurrences (including the special ones), which will be then used to recursively find the secondary ones (which, in turn, will be used to find further secondary occurrences).}

\section{Primary Occurrences}
\djeremy{From the definition given above it holds that}\jeremy{By definition,} a primary occurrence covers at least two phrases. Thus, each primary occurrence can be seen as $P=LR$, where the left side $L$ is a suffix of a phrase and the right side $R$ is the concatenation of zero or more consecutive phrases plus a prefix of the next phrase. For this reason, to find this type of occurrences we partition the pattern in two (in every possible way). Then, we search for the occurrences of the left part of the pattern in the suffixes of the phrases and for the right part in the prefixes of the suffixes of the text starting at  beginning of phrases. Then, we need to find which pairs of left and right occurrences actually represent an occurrence of pattern $P$: 
\begin{enumerate}
    \item Partition the pattern $P[1,m]$ into $P[1,i]$ and $P[i+1,m]$ for each $1 \le i < m$.
    \item Search for the right part $P[i+1,m]$ in the prefixes of the suffixes of the text starting at phrases.
    \item Search for the left part $P[1,i]$ in suffixes of phrases.
    \item Connect both results, generating all primary occurrences.
\end{enumerate}
\subsection{Right Part of the Pattern}
\label{sec:right_part}
To find the right side \diego{$P[i+1,m]$} of the pattern we \ddiego{have} \diego{use} a \ddiego{sparse} suffix \diego{trie} (recall Sections \ref{sec:tries} and \ref{sec:ku}) that indexes all suffixes of $T$ starting at the beginning of a phrase. In the leaves of the tree we store the identifier (id) of the phrases. Conceptually, these form an array $id$ that stores the phrase ids in lexicographic order (i.e., the leaves of the \ddiego{sparse} suffix \diego{trie}). As we see later, we do not need to store $id$ explicitly.
\begin{figure}[!ht]
\centering
\begin{displaymath}
\setlength\arraycolsep{0em}
\begin{array}{ccccccccc}
\text{\scriptsize{1}} & \text{\scriptsize{2}} & \text{\scriptsize{3}} & \text{\scriptsize{4}} & \text{\scriptsize{5}} & \text{\scriptsize{6}} & \text{\scriptsize{7}} &  \text{\scriptsize{8}} & \text{\scriptsize{9}}\\
\phrase{a} & \phrase{l} & \phrase{ab} & \phrase{ar} & \phrase{\_} & \phrase{a\_} & \phrase{la\_} & \phrase{alabard} & \phrase{a\$} 
\end{array}
\end{displaymath}
\includegraphics[width=8cm]{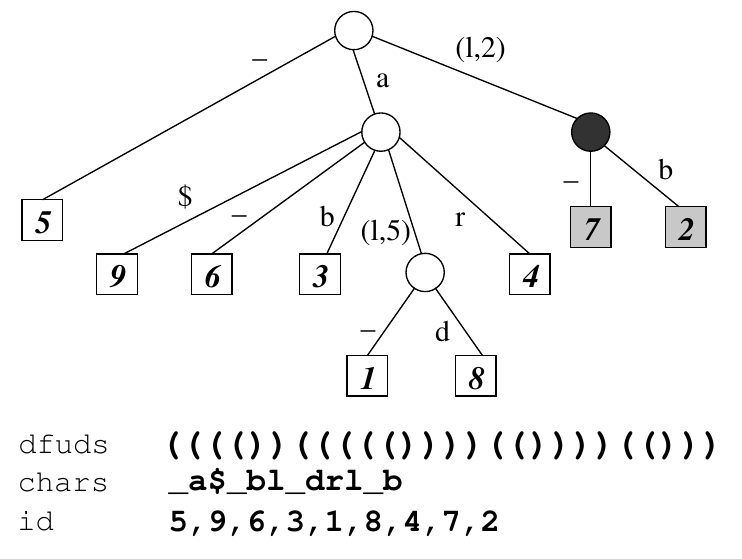}
\caption[\texorpdfstring{The \ddiego{sparse} suffix \diego{trie} for the string \texttt{`alabar\_a\_la\_alabarda\$'}}{The suffix trie for the string 'alabar\_a\_la\_alabarda\$'}]{The \ddiego{sparse} suffix \diego{trie} for the string \texttt{`alabar\_a\_la\_alabarda\$'}. The dark node is the note at which we stop searching for the pattern \texttt{`la'}, and the gray leaves represent the phrases that start with that pattern.}
\label{fig:sst}
\end{figure}

We will represent the \ddiego{sparse} suffix \diego{trie} as a labeled tree using DFUDS (Section \ref{sec:dfuds}). To search for a pattern we descend through the tree using $labeled\_child$ (recall Section \ref{sec:dfuds}), and then discard as many characters of the pattern as the skip of the branch indicates.
We continue this process \diego{either} until we reach a leaf, the pattern is completely consumed, or we cannot descend anymore. Our answer will be an interval of the array of ids, representing all \ddiego{matches} \diego{phrases starting with the pattern $P[i+1,m]$}.
In case we consume the pattern in an internal node, we need to go to the leftmost and rightmost leaves in order to obtain the interval, which is computed using $leaf\_rank$ and represents the start and end positions in the array of the ids.

\begin{example}
\nieves{Suppose we are looking for the right pattern \texttt{`la'}. Figure \ref{fig:sst} shows in dark the node at which we stop searching for the pattern, and in gray the phrases that start with that pattern. The answer is the range $[8,9]$ (i.e., the lexicographical order of the phrases).}
\end{example}

\begin{remark}
\label{rem:check}
Recall from Section \ref{sec:tries} that in a PATRICIA tree, after searching for the positions we need to check if they are actually a match, as some characters are not checked because of the skips. \nieves{In the example presented above, the answer would have been the same if we were searching for any right pattern of the form $\texttt{l}x$, where $x$ is a character distinct from $\texttt{a}$.} We use a different method here, which is explained in Section \ref{sec:connect}. 
\end{remark} 

We do not explicitly store the skips in our theoretical proposal, as they can be computed from the tree and the text. Given a node in the trie, if we go to the leftmost and rightmost leaves, we can extract the corresponding suffixes until computing how many characters they share. This value will be the sum of all the skips from the root to the given node. However, we already know they share $S$ characters, where $S$ is the sum of all skips from the root to the previous node (i.e., the parent node). Therefore, to compute the skip, we \ddiego{skip the first $S$ characters and extract from there} \diego{extract the suffixes of both leaves skipping the first $S$ characters}. The amount of symbols shared by both \diego{extracted} strings will be the skip. Extracting a skip of length $s$ will take at most $O(sH)$ time both for LZ77 and for LZ-End, since the extraction is from left to right and we have to extract one character at a time until they differ. Thus, the total time for extracting the skips as we descend is $O(mH)$.


\subsection{Left Part of the Pattern}
To find the left \diego{part $P[1, i]$} \ddiego{side} of the pattern we have a trie (actually a PATRICIA trie, Section \ref{sec:tries}) that indexes all the reversed phrases, stored as a compact labeled tree (Section \ref{sec:dfuds}). Thus to find the left part of the pattern in the text we need to search for $(P[1,i])^{rev}$ in this trie. The array that stores the leaves of the trie is called $rev\_id$ and is stored explicitly.
\begin{figure}[!ht]
\centering
\includegraphics[width=8cm]{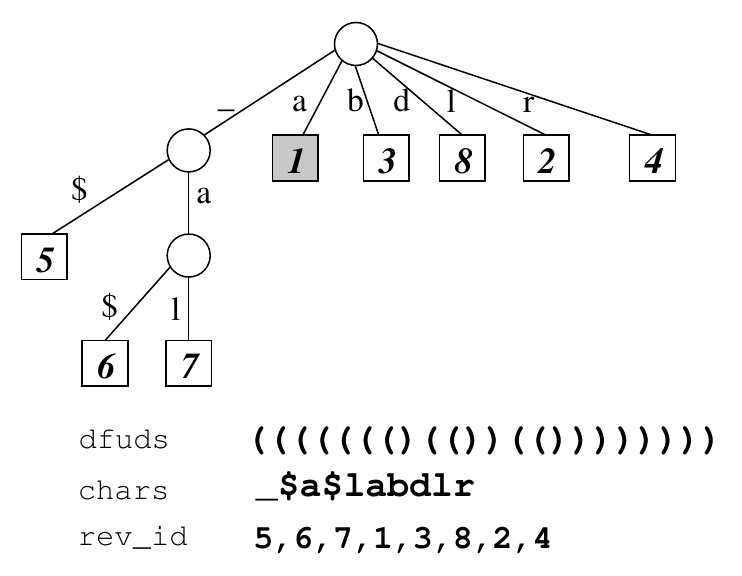}
\caption[\texorpdfstring{The reverse trie for the string \texttt{`alabar\_a\_la\_alabarda\$'}}{The reverse trie for the string 'alabar\_a\_la\_alabarda\$'}]{The reverse trie for the string \texttt{`alabar\_a\_la\_alabarda\$'}. The gray leaf is the node at which we stop searching for the pattern \texttt{`a'}.}
\label{fig:rev}
\end{figure}

The search process and the considerations for this tree are exactly the same as the ones for Section \ref{sec:right_part}.
The only difference with \ddiego{respect \jeremy{to}} the \ddiego{sparse} suffix \diego{trie} is that \ddiego{in this case the computation of the skips is faster} \diego{the computation of the skips is faster now}. Our text extraction algorithm works from right to left and since the text is reversed our algorithm outputs the characters in the correct order. Thus, extracting a skip of length $s$ takes $O(sH)$ time for LZ77 and $O(s+H)$ time for LZ-End. However, in the worst case the total time would \djeremy{be} still \jeremy{be} $O(mH)$ as all skips may be of length 1.

\begin{example}
\nieves{Suppose we are looking for the left pattern \texttt{`a'}. Figure \ref{fig:rev} shows in gray the node at which we stop searching for the pattern. In this case we end up in a leaf, so that is the only phrase that ends with the given pattern. The answer is the range $[4,4]$.}
\end{example}

\subsection{Connecting Both Parts}
\label{sec:connect}
In the previous steps we found two intervals, one in the $id$ array and the other in the $rev\_id$ array. These intervals represent the sets of phrases where the matches of the right side of the pattern start ($id$ array interval) and the phrases ending with the left side of the pattern ($rev\_id$ array interval). Actual occurrences of the pattern are composed of consecutive phrases. Hence, to find the occurrences of the pattern, we need to find which ids in the right interval are consecutive to those rev\_ids in the left interval.   
For doing so we use a range structure (see Section \ref{sec:range}) that connects the consecutive phrases in both trees. Figure \ref{fig:range} shows the range data structure connecting both trees for our example string \nieves{and below the sequence that is represented with the wavelet tree}. 
\begin{figure}[!ht]
\centering
\includegraphics[width=10cm]{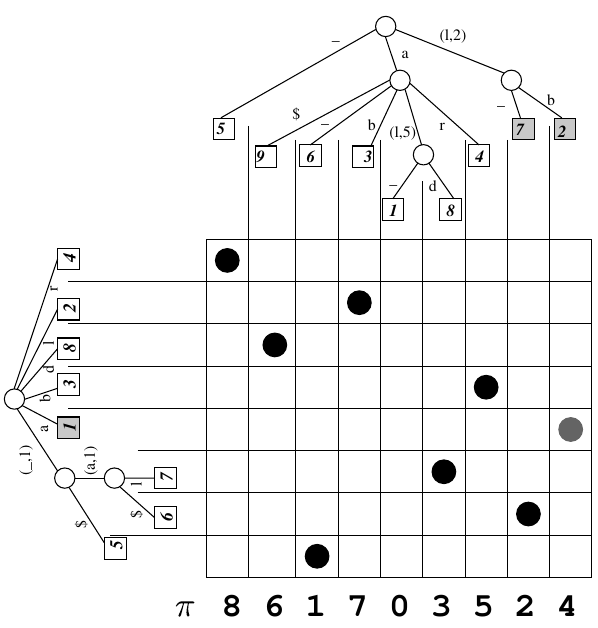}
\caption[\texorpdfstring{The range structure for the string \texttt{`alabar\_a\_la\_alabarda\$'}}{The range structure for the string 'alabar\_a\_la\_alabarda\$'}]{The range structure for the string \texttt{`alabar\_a\_la\_alabarda\$'}. The gray circle marks the only primary occurrence of the pattern \texttt{`ala'}, and the gray nodes show the ranges defined by the left and right part of the pattern.}
\label{fig:range}
\end{figure}

This structure is built from a permutation $\pi$ on $[1,\ldots,n']$. 
This permutation is just an array containing for each id (column) the corresponding rev\_id (row). In other words, the permutation holds that $id[i] = 1 + rev\_id[\pi(i)]$. For our example the permutation array would be \texttt{\{8,6,1,7,0,3,5,2,4\}} \nieves{(note that we count from left to right and from bottom to top, and that we assume that $rev\_id[0]=0$)}. 

\begin{example}
\nieves{Suppose we are looking for the pattern \texttt{`ala'}. The possible partitions are $(\texttt{a},\texttt{la})$ and $(\texttt{al},\texttt{a})$. Figure \ref{fig:range} shows in gray the ranges obtained when searching for the left and right part of partition $(\texttt{a},\texttt{la})$. Then we look for all points inside those ranges, obtaining the only primary occurrence that starts at phrase 1. The same procedure is carried out for the other partition.}
\end{example}

\begin{remark}
\label{rem:ids}
The range structure allows us to compute $id[i]$, just storing the $rev\_id$ array. Say we want to compute $id[i]$. We extract the value $S[i]$ from the wavelet tree, giving us the row $p$ where the corresponding reverse id is. Then we compute $id[i] = 1 + rev\_id[p]$. 
\end{remark}

\begin{example}
\nieves{Say we want to compute $id[6]$ (i.e., the phrase id of the 6th lexicographical smallest phrase). We extract from the wavelet tree the 6th symbol, getting the value 3. This value is the lexicographical order of the reversed 5th phrase. Computing $rev\_id[3]=7$, we know that the 5th phrase is phrase number 7. Hence, the 6th phrase is phrase number 8 (i.e., $id[6]=8$).}
\end{example}

At this stage we also have to validate that the answers returned by the search query are actual occurrences, as the PATRICIA tries by themselves do not guarantee the pattern found is actually a match (see Remark \ref{rem:check}).
For the first occurrence reported by the range data structure we extract the substring of length $m$ starting at the reported position and check if it matches the pattern. If so we can ensure that all the other reported occurrences match the pattern as well, otherwise no occurrence is a match. This process works because all occurrences reported by both tries share all characters, thus all occurrences reported by the range query share all characters. We check the validity of the occurrences here as the range check is cheaper than extracting text and we want to extract text only when a candidate to complete occurrence is found.

This structure adds $O(\log n')$ time to the search phase, and $O(\log n')$ time per primary occurrence found.

Note that we are able to answer \emph{exists} queries with the structures explained so far. If the number of occurrences reported by the range search is greater than one, then we check if one of those queries is an actual match. If there is a match, then the pattern is present in the text.

\subsection{Special Primary Occurrences}
The special primary occurrences could be found using the same steps explained above for primary occurrences, taking the left part of the pattern as the pattern itself and the right side of the pattern as the empty string $\varepsilon$. However, we do know that looking for $\varepsilon$ in the \ddiego{sparse} suffix \diego{trie} will return the complete tree, thus making the search in the range structure unnecessary. For this reason we call this type of occurrence special primary, as we search for them slightly differently from the primary ones. For these occurrences we just need to search for $P^{rev}$ in the reverse trie.

Since the search $P^{rev}$ in the reverse trie gives us a range in the rev\_id array, we decided to store it explicitly instead of the id array. Furthermore, the result of the range search gives us positions \jeremy{in}\djeremy{of} the rev\_id array.

\subsection{Converting Phrase Ids to Text Positions}
\label{sec:convertid}
From the range structure we obtain the phrase id where an occurrence lies. Then we need to convert it to a real text position. For doing so, we use a bitmap that marks the ends of phrases. This bitmap is the same $B$ used in Chapter \ref{chap:parsing} for extracting text. Figure \ref{fig:bitmap} shows the bitmap for the example string. The bitmap is below the parsing. 

\begin{figure}[!ht]
\centering
\includegraphics[width=8cm]{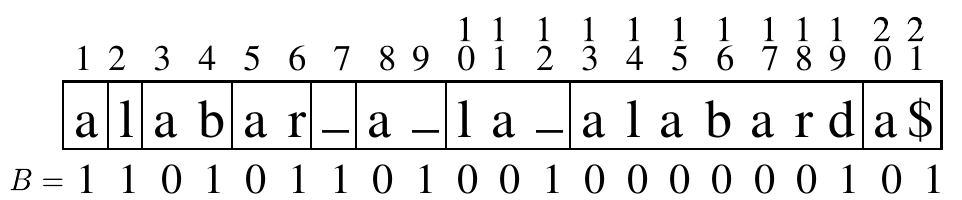}
\caption{\texorpdfstring{The bitmap $B$ of phrases for the string \texttt{`alabar\_a\_la\_alabarda'}}{The bitmap B of phrases for the string 'alabar\_a\_la\_alabarda'}}
\label{fig:bitmap}
\end{figure}

The conversion between phrase ids and text positions takes constant time as follows:
\begin{itemize}
\item $phrase(pos) = 1+rank_1(B,pos-1)$: it gives the phrase id containing any text position $pos$.
\item $first\_pos(id) = select_1(B,id-1)+1$: position of the first character of phrase $id$. 
\item $last\_pos(id) = select_1(B,id)$: position of the trailing character of phrase $id$. 
\end{itemize}

Recall from Section \ref{sec:encoding} that this bitmap also allows us to compute the length of the phrase as $length(id) = select_1(B,id+1)-select_1(B,id)$.

\subsection{Implementation Considerations}
\label{sec:implementing}
Here we explain some considerations we made when implementing our index.

\begin{itemize}
\item \textbf{Skips:} as the average value for the skips is usually very low and computing them from the text phrases is slow in practice, we considered storing the skips, for one or for both tries, using the \emph{Directly Addressable Codes} (Section \ref{sec:dac}). Note that in this case we never access array $id$ nor $rev\_id$ during the trie traversal, they are only accessed when checking and reporting the occurrences. 

\item \textbf{Binary Search:} instead of storing the trie we can do a binary search over the ids (rev\_ids) of the \ddiego{sparse} suffix \diego{trie} (reverse trie). For the \ddiego{sparse} suffix \diego{trie}, we do not have explicitly the array of ids, but as shown in Remark \ref{rem:ids} we can retrieve them using the range structure and the rev\_ids array. This alternative modifies the complexity of searching for a prefix/suffix of $P$ to $O(mH\log n')$ for LZ77 or $O((m+H)\log n')$ for LZ-End (actually, since we extract the phrases right-to-left, binary search on the reverse trie costs $O(m\log n')$ for LZ-End). Additionally, we could store explicitly the array of ids, instead of accessing them through the rev ids. Although this alternative increases the space usage of the index and does not improve the complexity, it gives an interesting trade-off in practice.
\end{itemize}
\section{Secondary Occurrences}
Secondary occurrences are found from the primary occurrences and, recursively, from \jeremy{other previously discovered} secondary occurrences. 
\subsection{Basic Idea}
\label{sec:secondary}
The idea to find the secondary occurrences is to locate all sources (of the LZ parsing) covering the occurrence and then mapping their corresponding phrases to real text positions. To do this we use  another bitmap, called \emph{bitmap of sources} $B_S$. The bitmap is built by first writing in unary the \ddiego{number} \diego{amount} of empty sources ($\varepsilon$) and then for each position of the text writing in unary how many sources start at that position. In this way each \texttt{1} corresponds to a source and a \texttt{0} represents the position where the sources (\texttt{1}s) immediately preceding it start.
Figure \ref{fig:sources} shows the sources and the corresponding phrases they generate (except the empty sources), and below the resulting bitmap. Since there are 3 empty sources the bitmap starts with \texttt{1110}, then are 5 sources starting at position 1, hence \texttt{111110} follows, then just one source starting at position 2, adding \texttt{10}, and finally one \texttt{0} for each remaining position.

Additionally, we need a permutation $P_S$ connecting the \texttt{1}s in the bitmap $B$ of phrases (recall Section \ref{sec:convertid}) to the \texttt{1}s in the bitmap $B_S$ of sources. The sources starting at a given position are sorted by increasing length, thus the last \texttt{1} before a \texttt{0} marks the longest source starting at that position. An example is given in Figure \ref{fig:perm}. 
This permutation replaces the array $source$ of Section \ref{sec:encoding}

\begin{figure}[!ht]
\centering
\includegraphics[width=10cm]{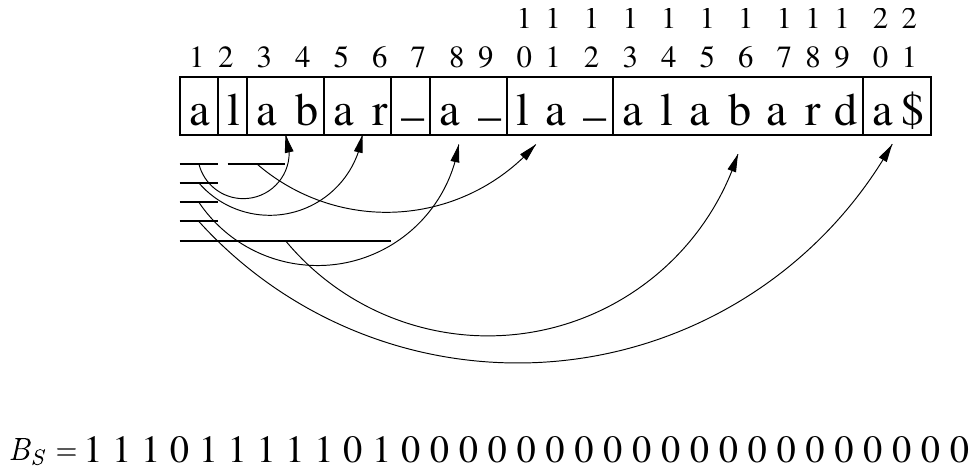}
\caption{\texorpdfstring{Marking sources on bitmap $B_S$}{Marking sources on bitmap B_S}}
\label{fig:sources}
\end{figure}

\begin{figure}[!ht]
\centering
\includegraphics[width=10cm]{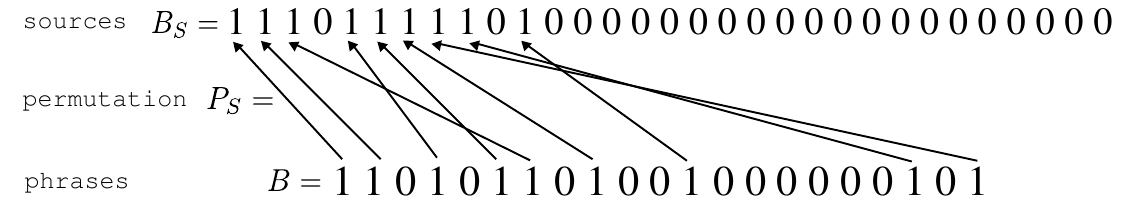}
\caption[Permutation connecting bitmap of phrases and bitmap of sources]{\texorpdfstring{Permutation connecting bitmap of phrases $B$ (bottom) and bitmap of sources $B_S$ (top)}{Permutation connecting bitmap of phrases and bitmap of sources}}
\label{fig:perm}
\end{figure}

For each occurrence found, we find the position \emph{pos} of the \texttt{0} corresponding to its starting position in the bitmap of sources. Then we consider all the \texttt{1}s to the left of \emph{pos}. We convert each source to its target phrase, compute its length and see if the source covers the occurrence. If so, we report it as a secondary occurrence and recursively generate all secondary occurrences from this new occurrence. In case the source does not cover the occurrence, we stop the process and continue processing the remaining occurrences. The algorithm is depicted in Figure \ref{alg:secondary}.

\begin{figure}[!ht]
\renewcommand{\algorithmiccomment}[1]{//#1}
\algsetup{linenodelimiter=,indent=0.8em}
\begin{center}
\begin{tabular}{l}
\begin{minipage}{12cm}
\textbf{secondaryOcc}$(start,len)$
\begin{algorithmic}[1]
    \STATE $pos \leftarrow select_0(B_S,start+1)$
    \STATE $source\_id \leftarrow pos-start-1$
    \WHILE{$source\_id>0$}
        \STATE $phrase\_id \leftarrow P_S^{-1}(source\_id)$
        \STATE $source\_start \leftarrow select_1(B_S,source\_id)-source\_id$ 
        \IF{$source\_start + len(phrase\_id) \ge start+len$}
            \STATE $occ\_pos \leftarrow first\_pos(source\_id)+start-source\_start$
            \STATE \textbf{report} $occ\_pos$
            \STATE secondaryOcc$(occ\_pos,len)$
        \ELSE
            \STATE \textbf{return}        
        \ENDIF
        \STATE $source\_id \leftarrow source\_id - 1$
    \ENDWHILE
\end{algorithmic}
\end{minipage}
\end{tabular}
\end{center}
\caption{\texorpdfstring{Searching for secondary occurrences from $T[start,start+len]$ (preliminary version)}{Searching for secondary occurrences}}
\label{alg:secondary}
\end{figure}

\begin{example}
Consider the only primary occurrence of the pattern \texttt{`la'} starting at position 2. We find the third \texttt{0} in the bitmap of sources at position 12. Then we consider all ones starting from position 11 to the left. The first \texttt{1} at position 11 maps to a phrase of length 2 that covers the occurrence, hence we report an occurrence at position 10. The second \texttt{1} maps to a phrase of length 6 that also covers the occurrence, thus we report another occurrence at position 15. The third \texttt{1} maps to a phrase of length $1$, hence it does not cover the occurrence and we stop. We proceed recursively for the secondary occurrences found at position 10 and 15.
\end{example} 

\begin{remark}
The method explained above is just introductory, as it does not work for general LZ77-like parsings. It only works for parsings in which no source strictly contains another source. Is it easy to see that if a source $S_2$ is strictly covered by another $S_1$ some secondary occurrences are lost. Let $M$ be a match of the pattern sought and let $M$ be between the rightmost positions of $S_2$ and $S_1$.  Then, as $S_2$ is the first source to the left of $M$, we test if it covers $M$, stopping the process. However, $S_1$ does cover $M$ and produces a secondary occurrence, which was not detected by the algorithm presented above.
\end{remark}

\begin{example}
Let us start with the primary occurrence of the pattern \texttt{`ba'} starting at position 4. The first source to the left is \texttt{`la'}, at position 2 and of length 2, which does not cover the pattern. Hence, the algorithm explained above would stop, reporting no secondary occurrences. However, to the left of this source is the source \texttt{`alabar'} that does cover the pattern and generates the secondary occurrence starting at position 16.
\end{example}

\subsection{Complete Solution}
\label{sec:secondary2}
K{\"a}rkk{\"a}inen in his thesis \cite{Kar99} proposes a method for converting the LZ77 parsing into one in which no source contains another. However, we decided not to use it as it increases excessively the number of phrases. Recall that our index will use space proportional to the number of phrases of the parsing, thus any increase in the number of phrases affects directly the final size of the index.

Another proposal of K{\"a}rkk{\"a}inen is to separate the sources by levels, so that within a level no source strictly contains another, and then apply the method explained in Section \ref{sec:secondary} within each level.

\begin{definition}
\label{def:depth_source}
The \emph{depth} of a source is defined as
\begin{displaymath}
    depth(s) = \left \lbrace      
                    \begin{array}{ll}
                           0 &  cover(s)=\emptyset\\
                           1+\max_{s' \in cover(s)}depth(s') &  otherwise
                         \end{array}
               \right. ,
\end{displaymath} 
where $cover(s)$ is the set of all sources containing the source $s$. Let $S_1, S_2$ be two sources starting at $p_1, p_2$ and of lengths $l_1, l_2$. $S_1$ is said to cover $S_2$ if $p_1 < p_2$ and $p_1+l_1\ge p_2+l_2$. Note that, by definition, sources starting at the same position are not covered by each other. However, sources ending at the same position may cover each other.  
For $s=\varepsilon$ we define $depth(\varepsilon)=0$.
\end{definition}


Figure \ref{fig:depth} shows the additional array storing the depths of each source. The four sources \texttt{`a'} and the source \texttt{`alabar'} have depth equal to 0 as all of them start at the same position. Source \texttt{`la'} has depth 1, as it is contained by the source \texttt{`alabar'}. 

\begin{figure}[!ht]
\centering
\includegraphics[width=10cm]{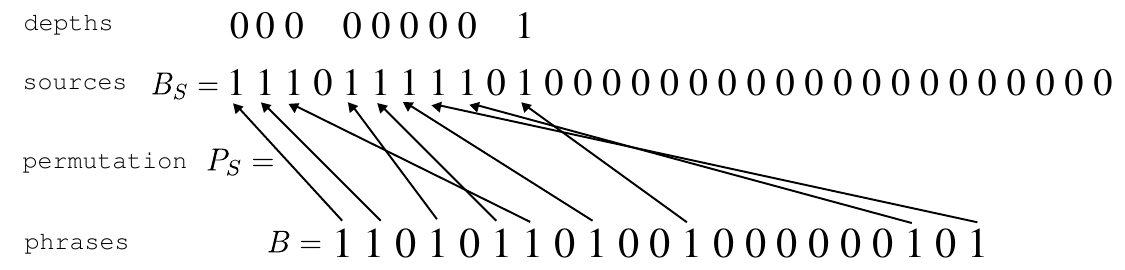}
\caption{\texorpdfstring{The depth of the sources for the string \texttt{`alabar\_a\_la\_alabarda\$'}}{The depth of the sources for the string 'alabar\_a\_la\_alabarda\$'}}
\label{fig:depth}
\end{figure}

The process now is similar to the idea presented earlier; however, now when we find a source not covering the occurrence we look for its depth $d$ and then consider to the left only sources with depth $d'<d$, as those at depth $\ge d$ are guaranteed not to contain the occurrence. This process works because in each level it holds that sources to the left will end earlier than the current source, 
because of the definition of depth. Moreover, sources at higher depths to the left will also end earlier as they are contained in a source of the current depth to the left.

Now the total running time to find all $occ$ secondary occurrences given a seed occurrence is $\Omega(\frac{1}{\varepsilon}occ\cdot L)$ and $O(\frac{1}{\varepsilon}occ\cdot L + D)$, where $\varepsilon$ is the parameter for computing the inverse permutation $P_S^{-1}$ (Section \ref{sec:permutation}), \djeremy{$O(L)$}\jeremy{$L$} is the time to find the next element to the left with depth less than a given value (an operation we consider next), and $D$ is the maximum depth. 
The additional $O(D)$ cost is because in the worst case after finding the last occurrence we will be in a source of depth $D$, then move to a source of depth $D-1$, that does not yield an occurrence, and so on up to a source of depth~1.

\subsection{Prev-Less Data Structure}
\label{sec:prevless}
As explained above, we need to be able to, given a position $pos$ in an array $U$ and a value $v$, find the rightmost position $p$ preceding $pos$ for which it holds $U[p]\le v$. We will call this query $prevLess(U,p,v)$. 

To solve this query we will encode $U$ (i.e., the array of depths of Section \ref{sec:secondary2}) using a wavelet tree (Section \ref{sec:wavelettree}) supporting this additional operation. The algorithm 
descends according to the bits of value $v$.
If the value $v$ gets mapped to a \texttt{0} we recursively search in the left subtree. If the value $v$ gets mapped to a \texttt{1} we recursively search in the right subtree. In this case, as the answer could be at the left side, we look for the rightmost \texttt{0} preceding $pos$ in the bitmap of the wavelet tree node. Finding this \texttt{0} takes constant time using \emph{rank} and \emph{select}. Finally, we return the maximum of the value returned by the right subtree and the rightmost zero.

The pseudocode of the algorithm is presented in Figure \ref{alg:prevless}. The algorithm receives as parameters a wavelet tree $tree$, a position $pos$, and a value $v$, and returns $prevLess(array(tree),pos,v)$, where $array(tree)$ are the values represented by the wavelet tree. The bitmap of the wavelet tree is denoted $tree.B$. Function \emph{toBit} returns to which side the value goes, and its output depends on the level.

\begin{figure}[!ht]
\renewcommand{\algorithmiccomment}[1]{//#1}
\algsetup{linenodelimiter=,indent=0.8em}
\begin{center}
\begin{tabular}{l}
\begin{minipage}{10.5cm}
\textbf{prevLess}$(tree, pos, v)$
\begin{algorithmic}[1]
    \STATE \COMMENT{toBit depends on the level}
    \IF{$toBit(v,tree)=0$}
        \STATE $lpos \leftarrow prevLess(tree.left, rank_0(tree.B, pos), v)$
        \STATE \textbf{return} $select_0(tree.B, lpos)$         
    \ELSE
        \STATE \COMMENT{rightmost zero}
        \STATE $rm0 \leftarrow $ $select_0(tree.B, rank_0(tree.B, pos))$ 
        \STATE $lpos \leftarrow prevLess(tree.right, rank_1(tree.B, pos), v)$
        \STATE \textbf{return} $\max\lbrace rm0, select_1(tree.B, lpos)\rbrace$
    \ENDIF
\end{algorithmic}
\end{minipage}
\end{tabular}
\end{center}
\caption{PrevLess algorithm}
\label{alg:prevless}
\end{figure}

As the algorithm just performs constant-time operations at each level, its total running time is $L=O(\log D)$.

If, \ddiego{instead of the bitmap of sources,} we label each source with its depth, and \jeremy{label the} changes to the next text position with a $D+1$, \ddiego{where $D$ is the maximum depth}, we can get rid of the original bitmap of sources. \dmine{As}Since the wavelet tree also supports \emph{rank} and \emph{select} queries we have the same functionality as the bitmap of sources, yet with the ability to answer \emph{prevLess} queries. 
However, as in practice the bitmap of sources is very sparse, we preferred to use a $\delta$-encoded bitmap to represent it and the wavelet tree for the depths.

Using this operation we can now modify algorithm \emph{secondaryOcc} of Figure \ref{alg:secondary}. We keep track of the maximum depth $d$ for which there may be sources covering the occurrence. When a source does not cover the occurrence, we update the value of $d$. Using the operation \emph{prevLess}, we move to the next candidate source. The final algorithm is presented in Figure \ref{alg:secondary2}.
\begin{figure}[!ht]
\renewcommand{\algorithmiccomment}[1]{//#1}
\algsetup{linenodelimiter=,indent=0.8em}
\begin{center}
\begin{tabular}{l}
\begin{minipage}{12cm}
\textbf{secondaryOcc}$(start,len)$
\begin{algorithmic}[1]
    \STATE $pos \leftarrow select_0(B_S,start+1)$
    \STATE $source\_id \leftarrow pos-start-1$
    \STATE \COMMENT{$D$ is the maximum depth}
    \STATE $d \leftarrow D$
    \WHILE{$source\_id>0$}
        \STATE $phrase\_id \leftarrow P_S^{-1}(source\_id)$
        \STATE $source\_start \leftarrow select_1(B_S,source\_id)-source\_id$ 
        \IF{$source\_start + len(phrase\_id) \ge start+len$}
            \STATE $occ\_pos \leftarrow first\_pos(source\_id)+start-source\_start$
            \STATE \textbf{report} $occ\_pos$
            \STATE secondaryOcc$(occ\_pos,len)$
        \ELSE
            \STATE $d \leftarrow depth[source\_id]-1$
            \IF{$d<0$}
                \STATE \textbf{return}        
            \ENDIF
        \ENDIF
        \STATE $source\_id \leftarrow prevLess(depth\_tree,source\_id,d)$
    \ENDWHILE
\end{algorithmic}
\end{minipage}
\end{tabular}
\end{center}
\caption{\texorpdfstring{Searching for secondary occurrences from $T[start,start+len]$}{Searching for secondary occurrences}}
\label{alg:secondary2}
\end{figure}

\section{Query Time}
Combining all the steps gives us the total running time to find the occurrences.
\begin{itemize}
\item \textbf{Primary Occurrences:} the total time is $O(m(Find_m^{sst} + Find_m^{rev} + \log n' + Extract_m)+ occ_1 \log n')$, where \djeremy{$O(Find_m^{sst})$} \jeremy{$Find_m^{sst}$} is the time to search for a subpattern of length $m$ in the \ddiego{sparse} suffix \diego{trie}, \djeremy{$O(Find_m^{rev})$} \jeremy{$Find_m^{rev}$} is the time to search for a subpattern of length $m$ in the reverse trie, and $occ_1$ is the number of primary occurrences. Thus the time to count the occurrences in the range structure is $O(\log n')$, the total time to locate the primary occurrences in the range structure is $O(occ_1 \log n')$, and \djeremy{$O(Extract_m)$} \jeremy{$Extract_m$} is the time to extract $m$ characters to verify the PATRICIA searches. \djeremy{$O(Extract_m)$} \jeremy{$Extract_m$}, as said in Section \ref{sec:lz_extraction}, depends on the parsing and is $O(mH)$ for LZ77 and $O(m+H)$ for LZ-End in the worst case. As our experiments show later (Section \ref{sec:results}), in practice the difference is not as drastic: LZ77 is about 3 times (for most texts) slower for long substrings and not much slower for short substrings.

The $Find$ times depend on the structures used:
    \begin{itemize}
        \item Tries: $O(Find_m^{sst})=O(Find_m^{rev}) = O(m + Skips) = O(m + mH) = O(mH)$ (as the time to compute all skips is $O(mExtract_1)$ in the worst case).
        \item Tries+Skips: $O(Find_m^{sst})=O(Find_m^{rev}) = O(m + Skips) = O(m)$ (as the skips are stored).
        \item Binary Search: $O(Find_m^{sst}) = O(\log n' \cdot Extract_m)$ if we store the array of ids explicitly, otherwise the time is $O(Find_m^{sst}) = O(\log n'(\log n'+ Extract_m))$. $O(Find_m^{rev}) = O(\log n' \cdot Extract_m)$ on LZ77, and $O(m\log n')$ on LZ-End (as the extraction takes constant time per extracted symbol  in this case).
Here we save the verification of PATRICIA trees but this has no effect on the total complexity.
    \end{itemize}
    With this the total time using tries is $O(m^2H + m\log n' + occ_1 \log n')$, independent of the parsing. 
When adding skips the time drops to $O(m^2+mH+m\log n'+occ_1 \log n')$ on LZ-End.
When using, instead, binary searching, the time is $O(m(m+H)\log n'+occ_1 \log n')$ for LZ-End and $O(m^2H\log n'+occ_1 \log n')$ for LZ77 if we store the $id$ array explicitly, otherwise the time increases to $O(m(\log n'+m + H)\log n'+occ_1 \log n')$ for LZ-End and $O(m(\log n' + mH)\log n'+occ_1 \log n')$ for LZ77.

\item \textbf{Secondary Occurrences:} the total time is $O(\frac{1}{\varepsilon}occ(\log D + D))$, where $D$ is the maximum depth and $\varepsilon$ is the parameter \jeremy{for the representation} of the permutation (Section \ref{sec:permutation}). Recall from Sections \ref{sec:secondary} and \ref{sec:secondary2} that the time to find the secondary occurrences from a seed is $O(\frac{1}{\varepsilon}occ\log D + D)$. However in this case we are recursively locating the secondary occurrences from all the occurrences found, and in a worst case we could pay $O(D)$ for each occurrence, not finding new ones. 

\end{itemize}
Taking $\varepsilon=\frac{1}{\log n'}$ gives us a total time similar to the \jeremy{one} given for the primary occurrences, except that $occ_1 \log n'$ changes to $occ \cdot D \log n'$.

To solve \emph{exists} queries, we basically search for the first primary occurrence. Hence the total time is as given for the primary occurrences replacing $occ_1=0$ (the details can be seen in Table \ref{tab:summary_index}).
\section{Construction}
In this section we explain the construction algorithm of the proposed index. We propose a practical construction algorithm, with bounded space usage and decent times. \jeremy{(See table \ref{tab:summary_parameters} for a reminder of the definitions of the variables.)}
\begin{enumerate}
    \item \textbf{Alphabet mapping:} as we work with standard texts that represent each symbol using 1 byte, we map the byte values to effective alphabet positions and vice versa.

    \item \textbf{LZ parsing:} For LZ77 we use the algorithm CPS2 of Chen \etal ~\cite{CPS08} and for LZ-End we use the algorithm \djeremy{explained} \jeremy{described} in Section \ref{sec:constr}. At this stage we generate three different files containing the trailing characters, the lengths of the sources, and the starting positions of the sources. Using suffix trees (Section \ref{sec:stree}), the LZ77 parsing can be computed in $O(n)$ time using $O(n)$ words of space, and the LZ-End parsing in $O(N)$ time using the same space. Additionally, the LZ77 parsing can be computed theoretically in $O(n\log n(\log^\epsilon n + o(\log\sigma)))$ time using $n(H_k(T)+2)+o(n\log\sigma)$ bits and the LZ-End parsing in time $O(\log n(N\log^\epsilon n + n\, o(\log\sigma)))$ using $n(H_k(T)+2)+n'\log n + o(n\log \sigma)$ bits (Section \ref{sec:constr}). The practical algorithms take total time $O(n\log n)$ for LZ77 and $O(N\log n)$ for LZ-End (recall Section \ref{sec:constr}). The space usage is around $5.7n$ bytes for LZ77 and $9n$ bytes for LZ-End, and this is the peak space usage for the self-index construction. In the index we only store explicitly the trailing characters using $n'\log \sigma$ bits.

    \item \textbf{Bitmap of phrases:} this bitmap is easily computed from the array containing the lengths of the sources in time $O(n')$. It uses $n'\log\frac{n}{n'}+O(n'+\frac{n\log\log n}{\log n})$ bits (Section \ref{sec:bitmaps}). In practice we use $\delta$-encoded bitmaps, using $n'\log \frac{n}{n'}+O(n'\log\log n+\frac{n'\log n'}{s})$ bits, where $s$ is the sampling step. All query times are then multiplied by $s+\log \frac{n'}{s}$. For the analysis we will assume $s=\log n'$.

    \item \textbf{\ddiego{Sparse} Suffix trie and reverse trie:} for constructing these trees, we decided to insert all indexed substrings in a PATRICIA trie. This is $O(n)$ time for the reverse trie, but it could be quadratic for the \ddiego{sparse} suffix \diego{trie} (there are complex $O(n)$-time algorithms for building \ddiego{sparse} suffix \diego{tries} \cite{sst_constr}). In practice this does not happen and the running time is good, as the number of phrases $n'$ generated by the parsing is relatively small. Of course we insert and store \ddiego{in the trie pointers to the text rather than the whole strings} \diego{pointers to the text in the trie, rather than the whole strings}. From the PATRICIA trie, we can extract the sorted ids, the skips and the DFUDS representation of the tree in time $O(n')$. Each tree will have at most $2n'$ nodes, hence they require $4n'+o(n')$ bits (Section \ref{sec:dfuds}) for the topology of the tree, plus $2n'$ label characters encoded using $2n'\log \sigma$ bits. Additionally, the rev\_ids are stored using $n' \log n'$ bits. 

    \item \textbf{Range structure:} to build the range structure we need a permutation from the ids of the \ddiego{sparse} suffix \diego{tries} to the ids of the reverse trie. This is done in $O(n')$ time, inverting the permutation of the ids and then traversing the rev\_ids and assigning each to the corresponding id. 
Then, the range structure is built starting from the permutation in time $O(n'\log n')$. It uses $n'\log n' + O(n'\log \log n')$ bits (Section \ref{sec:range}).

    \item \textbf{Sources depths:} for computing the secondary occurrences related structures we need to first compute the depth of each source. First we sort all sources by increasing starting position, breaking ties by decreasing length. Doing this we know that all parents of a source are to its left. We keep track of the rightmost source of depth $d$ for each possible depth. Then for each source we binary search the rightmost sources and find the deepest source $d$ that covers the current phrase. Afterward, we set the current source as the rightmost source of depth $d+1$. The running time of the algorithm is $O(n'\log n')$.  

    \item \textbf{Prev-Less Depth Structure:} 
this structure is constructed in $O(n'\log D)$ time as it is just a wavelet tree
  It uses $n'\log D + O(n'\log \log D)$ bits (Section \ref{sec:prevless}).    

    \item \textbf{Source-Phrase Permutation:} it takes $O(n')$ time starting from the ids of the sorted sources. It is stored using $(1 + \varepsilon)n' \log n' + O(n')$ bits (Section \ref{sec:permutation}), and as we set $\varepsilon = \frac{1}{\log n'}$, the total space is $n'\log n'+O(n')$ bits. 

    \item \textbf{Bitmap of Sources:} it takes $O(n')$ time to build from the starting positions of the sorted sources. It uses $n'\log\frac{n}{n'}+O(n'+\frac{n\log\log n}{\log n})$ bits (Section \ref{sec:bitmaps}). In practice we use $\delta$-encoded bitmaps, so the same considerations as for the bitmap of phrases apply.
\end{enumerate}

Adding up the space of all structures we get that the index requires $2n'\log n + n'\log n' + n'\log D + O(n'\log \sigma + \frac{n\log\log n}{\log n})$ bits of space, which in our practical implementation is $2n'\log n + n'\log n' + n'\log D + O(n'\log \sigma + n'\log\log n)$ bits plus the skips we store. Note that in the case of binary searching we do not use tries, yet the asymptotic space complexity is not reduced.

Note \jeremy{that} our practical index space is fully proportional to $n'$, depending on $n$ only logarithmically.

For the construction time and space of the index we have given practical figures. We give now two \djeremy{theoretical} \jeremy{trade-offs for the theoretical upper} bounds. The first, Theory$_1$, uses the least possible construction time, and the second, Theory$_2$, the least possible construction space.
\begin{itemize}
\item Theory$_1$: The space gets dominated by the $O(n\log n)$ bits needed to build the parsing. All construction times are $O(n'\log n')$, except the parsing and creating the PATRICIA trees. Hence, the index is built in time $O(n+n'\log n')$ ($O(N+n'\log n')$ for LZ-End) using $O(n\log n)$ bits. 

\item Theory$_2$: In Section \ref{sec:constr} we showed that the LZ77 parsing can be computed in $O(n\log n(\log^\epsilon n + o(\log\sigma)))$ time using $O(nH_k(T))+o(n\log\sigma)$ bits and the LZ-End parsing in time $O(\log n(N\log^\epsilon n + n\, o(\log\sigma)))$ using the same space. All data structures (except PATRICIA trees) are constructed as explained above, each structure requiring at most $O(n'\log n) = O(nH_k(T))+ o(n\log \sigma)$ (recall Lemma \ref{lemma:phrase_bounds}) and $O(n' \log n')$ time. 
In the following we show that we can build the PATRICIA trees in time $O(n' \log^{1+\epsilon} n) $ using $O(nH_k(T)+ o(n\log \sigma))$ bits. The idea is similar to that presented by Claude and Navarro \cite{CN10} enhanced with some ideas from Russo \etal ~\cite{RNO08}\\

1) First we build the FM-index \cite{FM05} of $T$ in $O(n\log n \frac{\log \sigma}{\log \log n})$ time within $nH_k(T)+o(n\log \sigma)$ bits of space. 2) Then, we build the Fully-Compressed Suffix Tree (FCST) \cite{RNO08}, that supports all tree operations in $O(\log^{1+\varepsilon}n)$ time. For building the FCST, we simulate a traversal of the suffix tree starting from the root using Weiner links \cite{Wei73}, which are simulated using the LF mapping (see Section \ref{sec:bwt}) over the FM-Index. During the traversal we mark all nodes that are at depth multiple of $\delta$ in the implicit tree defined by the Weiner links, where $\delta$ is the space/time trade-off parameter of the FCST (which takes $o(n \log \sigma)$ bits of space for $\delta = \log^{1+\epsilon} n$). These nodes are stored in a simple array using $o(n\log \sigma)$ bits with all the information required by the FCST construction algorithm. The running time of this process is $O(n\frac{\log \sigma}{\log \log n})$ and the space is $o(n\log \sigma)$ bits. 3) We use a dynamic balanced tree to mark some of the nodes of the FCST; these will be the nodes of our PATRICIA tree. For each phrase starting position we convert it using $A^{-1}$ to an FM-Index position, and then \emph{selectLeaf} (which gives the $i$-th leaf) converts it to a position in the FCST. We mark in the balanced tree the \emph{preorder} of the FCST node, as well as the phrase id. Then, we traverse the balanced tree from left to right computing $LCA(x_i,x_{i+1})$ (where $x_i$ is the current node and $x_{i+1}$ is the node to the right, and LCA is their lowest common ancestor in the FCST) and inserting the value in the tree. To build the PATRICIA tree, we traverse the balanced tree again from left to right, creating the PATRICIA nodes, generating the parentheses representation and the labels of the edges. Using the operation \emph{letter} of the FCST (which gives any letter of the path leading to a node) we retrieve the label. Using the FCST we determine the topology of the tree (we keep the current PATRICIA path in a stack; add closing parentheses and pop the stack until the top of the stack is an ancestor of the new node; and then we add the opening parenthesis for the current node and push it to the stack).  This step runs in $O(n' \log^{1+\epsilon} n) $ time and within $O(n'\log n)=O(nH_k(T))+o(n\log \sigma)$ (recall Lemma \ref{lemma:phrase_bounds}) bits of space.\\ 

Hence, the total time is $O(n\log n \frac{\log \sigma}{\log \log n}+n'\log^{1+\varepsilon}n)$ and the total space is $O(nH_k(T))+o(n\log \sigma)$. The process for constructing the reverse trie is almost the same, but now we do not consider the whole suffixes, because they are limited by the phrase length. Given $A^{-1}(pos)$, we use the operation $LAQ_S(d)$ of the FCST (which retrieves the ancestor with string depth $d$) to find which node we need to mark.\\

As the time for constructing the PATRICIA trees gets dominated by the parsing algorithm the total space required is $O(nH_k(T))+ o(n\log \sigma)$ and the total running time is $O(n\log n(\log^\epsilon n + o(\log\sigma)))$ for LZ77 and $O(\log n(N\log^\epsilon n + n\, o(\log\sigma)))$ for LZ-End.

\end{itemize}

\section{Summary}
We have presented a self-index that given a text of length $n$, parsed into $n'$ different phrases by a Lempel-Ziv like parsing, uses space proportional to that of the compressed text, i.e., $O(n'\log n)+o(n)$ bits. Table \ref{tab:summary_index} summarizes the space and time of the operations over the index and Table \ref{tab:summary_parameters} summarizes all the parameters of the index. In practice, due to our sparse bitmap representation, the $o(n)$ bits  disappear from the space but the times of \emph{extract}, \emph{exists} and \emph{locate} are multiplied by $O(\log n')$.
\begin{table}[ht]
\begin{center}
\begin{small}
\begin{tabular}{|l|l|l|}
\hline
 & Tries & Binary Search \\ \hline
Construction Time & \multicolumn{2}{|p{11cm}|}{Theory$_1$: $O(n+n'\log n')$ for LZ77 and $O(N+n'\log n')$ for LZ-End\newline
Theory$_2$: $O(n\log n(\log^\epsilon n + o(\log\sigma)))$ for LZ77 and $O(\log n(N\log^\epsilon n + n\, o(\log\sigma)))$ for LZ-End\newline
Practice: $O(n\log n)$ for LZ77, $O(N\log n)$ for LZ-End.} \\ \hline

Construction Space & \multicolumn{2}{|p{11cm}|}{Theory$_1$: $O(n \log n)$ bits\newline
Theory$_2$: $O(nH_k(T))+ o(n\log \sigma)$\newline
Practice: LZ77 $\approx 6n$ bytes, LZ-End $\approx 9n$ bytes} \\ \hline

Index Space & \multicolumn{2}{|p{11cm}|}{Theory:  $2n'\log n + n'\log n' + n'\log D + O(n'\log \sigma + \frac{n\log\log n}{\log n})$ bits\newline
Practice: $2n'\log n + n'\log n' + n'\log D + O(n'\log \sigma + n'\log\log n)$ bits} \\ \hline

Extract Time & \multicolumn{2}{|c|}{LZ77: $O(mH)$ , LZ-End: $O(m+H)$} \\ \hline

Exists Time & \multicolumn{1}{|p{3.5cm}|}{$O(m^2H + m\log n')$\newline
With skips and LZ-End: $O(m^2+mH+m\log n)$
} & \multicolumn{1}{|p{7.5cm}|}{
Using $n'\log n'$ additional bits: LZ77: $O(m^2H\log n')$, LZ-End: $O(m(m+H)\log n')$.\newline
Otherwise: LZ77: $O(m(\log n'+mH)\log n')$, LZ-End: $O(m(\log n'+m+H)\log n')$
} \\ \hline

Locate Time & \multicolumn{1}{|p{3.5cm}|}{$O(m^2H + m\log n' +  occ \cdot D \log n')$\newline
With skips and LZ-End: $O(m^2+mH+m\log n'+occ \cdot D \log n' )$
} & \multicolumn{1}{|p{7.5cm}|}{
Using $n'\log n'$ additional bits: LZ77: $O(m^2H\log n' + occ \cdot D \log n')$, LZ-End: $O(m(m+H)\log n' + occ \cdot D \log n')$.\newline
Otherwise: LZ77: $O(m(\log n'+mH)\log n' + occ \cdot D \log n')$, LZ-End: $O(m(\log n' + m+H)\log n' + occ \cdot D \log n')$
}
 \\ \hline
\end{tabular}
\end{small}
\end{center}
\caption[Summary table of LZ77-Index]{Summary table of LZ77-Index. Adding skips requires at most $4n'\log n$ more bits, but far less in practice. In practice times are multiplied by $O(\log n')$.}
\label{tab:summary_index}
\end{table}

\begin{table}[ht]
\begin{center}
\begin{small}
\begin{tabular}{|l|l|l|}
\hline
Parameter & Description & Defined in\\ \hline
$\sigma$ & size of the alphabet & Section \ref{sec:strings} \\ \hline
$n$ & length of the text & Section \ref{sec:strings} \\ \hline
$n'$ & length of the LZ parsing & Definitions \ref{def:lz77} and \ref{def:lzend}\\ \hline
$m$ & length of the pattern & Section \ref{sec:queries}\\ \hline
$s$ & sampling step of $\delta$-encoded bitmap & Section \ref{sec:deltacodes}\\ \hline
$\varepsilon$ & parameter of the permutation & Section \ref{sec:permutation}\\ \hline
$D$ & maximum depth of the sources & Section \ref{sec:secondary2}\\ \hline
$H$ & height of the LZ parsing & Definition \ref{def:height}\\ \hline
$N$ & total text retraversed in the LZ-End parsing & Section \ref{sec:constr}\\ \hline
\end{tabular}
\end{small}
\end{center}
\caption[Summary table of parameters of LZ77-Index]{Summary table of parameters of LZ77-Index}
\label{tab:summary_parameters}
\end{table}

%% file: experimental_evaluation.tex
\chapter{Experimental Evaluation}
\label{chap:results}
\section{Experimental Setup}
\label{sec:results}
In our tests we compared the proposed index against RLCSA \cite{NM07}. We did not test the SLP Index \cite{CN09} because we could not make it run consistently in our collections, yet some comparison results can be inferred from their experimental evaluation \cite{CFMPNbibe10}, as we do \djeremy{later} \jeremy{in Section \ref{sec:analysis_index}}.

We used in our experiments the LZ77 and the LZ-End parsings. For the LZ indexes we used the following variants (\jeremy{defined in} \djeremy{see} Section \ref{sec:implementing}), ordered by decreasing space requirement. In all variants we stored the skips of the trees using DAC (Section \ref{sec:dac}), because not using them lead to results worse than using binary search \jeremy{(our slowest variant)}.
\begin{enumerate}
\item \ddiego{Sparse} Suffix \diego{trie} and reverse trie (original proposal).
\item Binary search on ids with the explicit ids and reverse trie.
\item Binary search on reverse ids and \ddiego{sparse} suffix \diego{trie}.
\item Binary search on ids with explicit ids and binary search on reverse ids.
\item Binary search on ids with implicit ids and binary search on reverse ids.
\end{enumerate}

Recall from Section \ref{sec:implementing} that the array of the ids is not stored in the index, only the array of reverse ids. Thus, if we want to binary search over ids we have two alternatives: (1) spend $n'\log n'$ additional bits to store explicitly the array of ids, or (2) using Remark \ref{rem:ids} to access the array implicitly by paying $O(\log n')$ access time. The index variants with \emph{explicit ids} refer to the alternative (1) and the ones with \emph{implicit ids} refer to alternative (2). The alternative using the \ddiego{sparse} suffix \diego{trie} do not need to access the id array, but rather the reverse ids array, which is always maintained in explicit form

\begin{remark}
The reader may note that the \jeremy{results concerning the} alternative
\begin{itemize}
\item Binary search on ids with implicit ids and reverse trie
\end{itemize}
is not present. We omit \jeremy{the empirical results of} this alternative as the compression ratio is about the same obtained using alternative number 3 and the performance of \emph{locate} is noticeably worse. Remember that accessing the implicit array of ids takes time $O(\log n')$.
\end{remark}

 
The parameters used for the data structure are as follows: $s=16$ for the $\delta$-codes bitmap (Section \ref{sec:deltacodes}), $\varepsilon=1/32$ for the permutation (Section \ref{sec:permutation}) and sampling step $b=20$ for the bitmaps of González \etal ~(Section \ref{sec:prac_bitmap}). We used these parameter values as they are the typical ones used in experimentation, and additionally with these values our indexes achieve a good space/time trade-off.

For RLCSA we used sampling with steps 512, 256, 128 and 64. The index was built using a buffer of 100MiB.   

All our experiments were conducted on a machine with two Intel Xeon CPU running at 2.00GHz with 16GiB main memory. The operating system is Ubuntu 8.04.4 LTS with Kernel 2.6.24-27-server. The compiler used was \texttt{g++} (\texttt{gcc} version 4.2.4) executed with the \texttt{-O3} optimization flag.

We present the results obtained for the following texts \nieves{(the results of only one text from each collection are presented in this section, the remaining results are presented in Appendix \ref{appendix:results})}:
\begin{itemize}
\item Artificial: $F_{41}$, $R_{13}$, $T_{29}$.
\item Pseudo-Real: DNA 0.1\% (Scheme 1), Proteins 0.1\% (Scheme 1), English 0.1\% (Scheme 2), Sources 0.1\% (Scheme 2).
\item Real: Para, Cere, Influenza, Escherichia Coli, Coreutils, Kernel, Einstein (en), Einstein (de), World Leaders. 
\end{itemize}

We restricted the experiments to the texts listed above since many of the texts produced similar results \jeremy{during preliminary experiments}. For DNA and Wiki domains, we chose the largest texts. For the case of pseudo-real texts we kept the DNA, Proteins, English and Sources texts, as these kind of texts naturally form repetitive collections. 

Two types of experiments were carried out. One, labeled ``results (1)'' in Figures \figuresres, considers the time of the operations as a function of the pattern length $|P|$. The second, labeled ``results (2)'' in Figures \figuresres, shows the space/time trade-off of the operations. 
Although we do not show a space/time tuning for LZ77/LZ-End index, the plots of figures labeled ``results (2)'' show a line formed by 5 points. These refer to the variants 1-5 described above.
The results are presented and discussed in Section \ref{sec:results}. The experiments conducted are the following:
\begin{itemize}
\item \textbf{Construction time and space:} we present the build time for each index as well as the peak memory usage. Results are given in Figure \ref{fig:construction}.

\item \textbf{Compression Ratio:} we present the compression ratio for different self-indexes. We show alternatives 1 and 5 of our indexes, which are respectively the largest and smallest variants. For RLCSA we show the space achieved with a sampling step of 512, and without the samples, which is the lowest space reachable by that index. For the LZ78-index \cite{AN06} (Section \ref{sec:lz78index}) we used $\varepsilon=\frac{1}{128}$ as the sampling step of the permutation. Additionally, we show the compression ratio of ILZI \cite{ilzi} (Section \ref{sec:ilzi}). The results are shown in Table \ref{tab:index_compression}, where we also include p7zip as a baseline.

\item \textbf{Structures Space:} we present the space usage of the different data structures used in our indexes. The results are given in Tables \ref{tab:lz77_det} and \ref{tab:lzend_det} as percentage of the size of the index.

\item \textbf{Parsing Statistics:} we present the value of $D$ and $H$, which affect the performance of our indexes. The results are displayed in Tables \ref{tab:dvalues} and \ref{tab:hvalues}.
 
\item \textbf{Extraction speed:} we extracted 10,000 substrings of length $2^i, i\in\lbrace 0, \ldots, 12 \rbrace$. We show only one line for the LZ77 and the LZ-End index, as all the variants have the same extraction speed. See the top-left plots of Figures \ref{res:t29:1}-\ref{res:kernel:2} labeled ``results (1)'', which are representative of all the results (the rest are in Appendix \ref{appendix:results}, in Figures \ref{res:f41:1}-\ref{res:leaders:2}). We also show the space/time trade-off of the indexes for extracting a pattern of length $2^{12}$. Extraction times per character stabilize at this length. See the top-left plot of Figures \figuresres labeled ``results (2)''.

\item \textbf{Search time:} we located 1,000 patterns of length 10, 15, and 20. We limited the number of occurrences reported to 30,000. See plots 2-4 in reading order of Figures \figuresres labeled ``results (2)''. We also located patterns of increasing length from 5 to 40. In this case we only show the results for alternative 1 (original proposal) and 5 (minimum space) of both LZ77 and LZ-End.
See top-right plot of Figures \figuresres labeled ``results (1)''.

\item \textbf{Locate time:} we located 1,000 patterns of length 2 and 4. This test highlights the time needed to find the occurrences in our indexes, as it dominates the time for traversing the tries. We limited the number of occurrences reported to 100,000. See plots 5-6 of Figures \figuresres labeled ``results (2)''.

\item \textbf{Exists Time:} we generated 20,000 patterns of lengths 5, 10, 20, 40 and 80; half of them were present in the text and the other half were a random concatenation of symbols of the text. For RLCSA we check the existence using a \emph{count} query. For this reason, we only show one line for RLCSA, as \emph{count} time is independent of the sampling size. The \emph{exists} query of the LZ77 index is basically a search of a primary occurrence, and thus it illustrates the time for traversing the tries. We only show the results for alternative 1 (original proposal) of both LZ77 and LZ-End, since the other alternatives are orders of magnitude slower for this. See left-bottom plot of  Figures \figuresres labeled ``results (1)'' for existing patterns and right-bottom plot for non-existing patterns. We also show the space/time trade-off of these two queries for patterns of length 20, see plots 7-8 of Figures \figuresres labeled ``results (2)''.
\end{itemize}

\begin{figure}[!ht]
\centering
\includegraphics[angle=-90,scale=0.35]{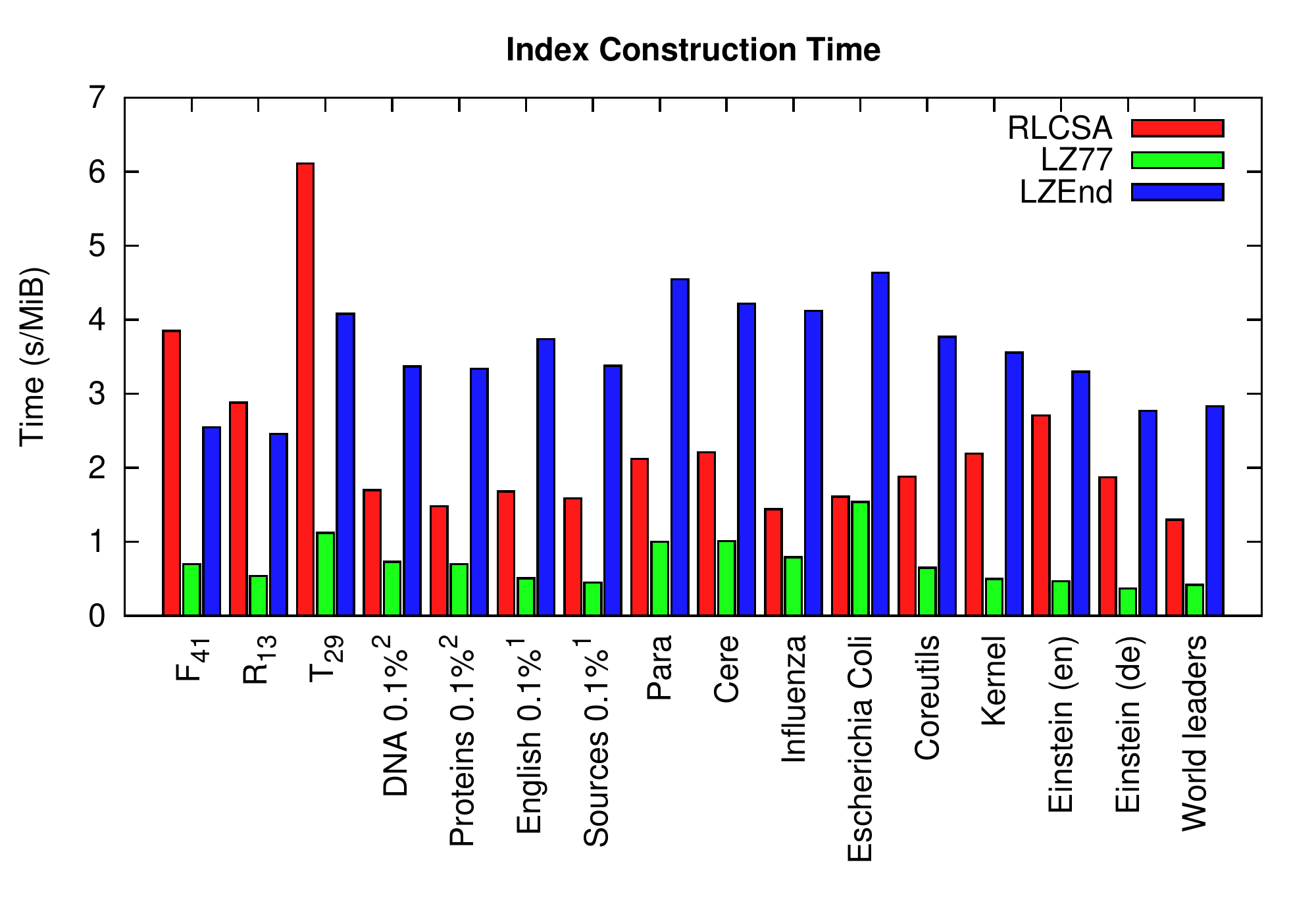}
\includegraphics[angle=-90,scale=0.35]{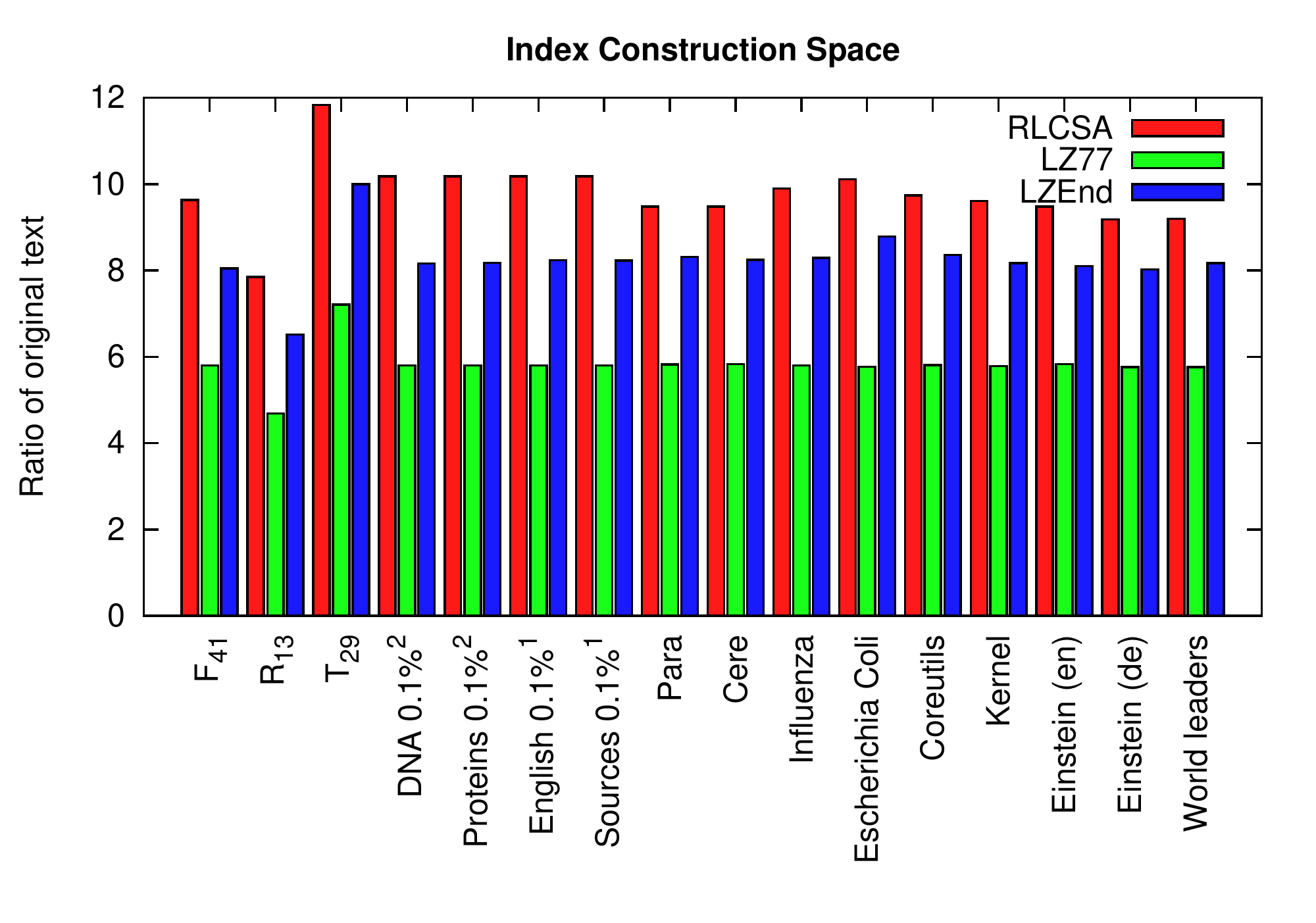}
\begin{scriptsize}
\begin{tabular}{|l|r|r|r|r|r|r|}
\hline
 & \multicolumn{2}{|c|}{LZ77 Index} & \multicolumn{2}{|c|}{LZ-End Index} & \multicolumn{2}{|c|}{RLCSA} \\ \hline
\textbf{File} & Time & Space & Time & Space & Time & Space \\ \hline
$F_{41}$ & 0.70 & 5.80 & 2.55 & 8.05 & 3.85 & 9.63 \\ \hline
$R_{13}$ & 0.54 & 4.69 & 2.46 & 6.52 & 2.88 & 7.85 \\ \hline
$T_{29}$ & 1.12 & 7.21 & 4.08 & 10.00 & 6.11 & 11.83 \\ \hline \hline
DNA 0.1\% $^2$ & 0.73 & 5.80 & 3.37 & 8.16 & 1.70 & 10.18 \\ \hline
Proteins 0.1\% $^2$ & 0.70 & 5.80 & 3.34 & 8.18 & 1.48 & 10.18 \\ \hline
English 0.1\% $^1$ & 0.51 & 5.80 & 3.74 & 8.24 & 1.68 & 10.18 \\ \hline
Sources 0.1\% $^1$ & 0.45 & 5.80 & 3.38 & 8.23 & 1.59 & 10.18 \\ \hline \hline
Para & 1.00 & 5.82 & 4.55 & 8.31 & 2.12 & 9.48 \\ \hline
Cere & 1.01 & 5.83 & 4.22 & 8.25 & 2.21 & 9.48 \\ \hline
Influenza & 0.79 & 5.80 & 4.12 & 8.29 & 1.44 & 9.90 \\ \hline
Escherichia Coli & 1.54 & 5.77 & 4.64 & 8.79 & 1.61 & 10.11 \\ \hline
Coreutils & 0.65 & 5.81 & 3.77 & 8.36 & 1.88 & 9.74 \\ \hline
Kernel & 0.50 & 5.79 & 3.56 & 8.17 & 2.19 & 9.61 \\ \hline
Einstein (en) & 0.47 & 5.83 & 3.30 & 8.10 & 2.71 & 9.48 \\ \hline
Einstein (de) & 0.37 & 5.76 & 2.77 & 8.02 & 1.87 & 9.18 \\ \hline
World leaders & 0.42 & 5.76 & 2.83 & 8.17 & 1.30 & 9.20 \\ \hline
\end{tabular}
\end{scriptsize}
\caption[Construction time and space for the indexes]{Construction time and space for the indexes. Times are in seconds per MiB and spaces are the ratio between construction space and text space.}
\label{fig:construction}
\end{figure}

\begin{table}[htbp]
\begin{center}
\begin{scriptsize}
\newcommand{\e}[1]{E#1}
\begin{tabular}{|l||r|r||r|r||r|r|r|r||r|}
\hline 
Text                & LZ78  & ILZI   & RLCSA                 & RLCSA$_{512}$& LZ77$_5$              & LZ77$_1$              & LZ-End$_5$            & LZ-End$_1$            & p7zip \\ \hline \hline
$F_{41}$            & 2.43  & 0.59   & 4.73\e{-4}   & 0.96 & \textbf{3.28\e{-4}}   & 6.04\e{-4}   & \textbf{3.75\e{-4}}   & 5.75\e{-3}   & 0.18 \\ \hline
$R_{13}$            & 3.68  & 0.59   & 7.47\e{-4}   & 0.94 & \textbf{4.15\e{-4}}   & 7.36\e{-4}   & \textbf{6.72\e{-4}}   & 1.10\e{-3}   & 0.36 \\ \hline
$T_{29}$            & 4.18  & 0.76   & 6.03\e{-4}   & 0.52 & \textbf{3.80\e{-4}}   & 6.93\e{-4}   & \textbf{5.95\e{-4}}   & 1.02\e{-3}   & 0.17 \\ \hline\hline
DNA 0.1\% $^2$      & 83.28 & 60.55  & 4.25                  & 4.73 & \textbf{2.12}         & \textbf{3.47}         & \textbf{4.21}         & 6.19                  & 0.51 \\ \hline
Proteins 0.1\% $^2$ & 82.94 & 71.67  & 3.80                  & 4.30 & \textbf{2.37}         & \textbf{3.69}         & 4.69                  & 6.71                  & 0.59 \\ \hline
English 0.1\% $^1$  & 90.71 & 69.42  & 4.52                  & 5.00 & \textbf{2.49}         & \textbf{3.91}         & 7.32                  & 10.62                 & 0.55 \\ \hline
Sources 0.1\% $^1$  & 82.82 & 58.40  & 3.77                  & 4.25 & \textbf{2.10}         & \textbf{3.31}         & 6.83                  & 10.13                 & 0.44 \\ \hline\hline
Para                & 82.04 & 72.99  & 9.86                  & 10.83 & \textbf{5.36}         & \textbf{8.45}         & \textbf{9.20}         & 13.38                 & 1.46 \\ \hline
Cere                & 79.37 & 66.93  & 7.60                  & 8.57 & \textbf{3.74}         & \textbf{5.94}         & \textbf{6.16}         & 8.96                  & 1.14 \\ \hline
Influenza           & 46.88 & 36.19  & 3.84                  & 4.77 & 4.55                  & 7.49                  & 9.20                  & 13.89                 & 1.35 \\ \hline
Escherichia Coli    & 88.36 & 83.69  & 24.71                 & 25.62 & \textbf{18.82}        & 29.71                 & 30.36                 & 43.58                 & 4.72 \\ \hline
Coreutils           & 81.36 & 58.38  & 6.53                  & 7.46 & 7.89                  & 12.47                 & 11.89                 & 17.47                 & 1.94 \\ \hline
Kernel              & 85.94 & 60.19  & 3.78                  & 4.71 & \textbf{3.31}         & 5.26                  & 5.12                  & 7.50                  & 0.81 \\ \hline
Einstein (en)       & 33.26 & 13.28  & 0.23                  & 1.20 & \textbf{0.18}         & 0.30                  & 0.32                  & 0.48                  & 0.07 \\ \hline
Einstein (de)       & 48.99 & 21.61  & 0.48                  & 1.39& \textbf{0.32}         & 0.54                  & 0.59                  & 0.89                  & 0.11 \\ \hline
World leaders       & 49.74 & 35.08  & 3.32                  & 4.20 & 3.85                  & 6.27                  & 6.44                  & 9.63                  & 1.29 \\ \hline
\end{tabular}
\end{scriptsize}
\end{center}
\caption[Compression ratio of different self-indexes]{Compression ratio (given in percentage of original file size) of different self-indexes. In bold are highlighted those LZ-based indexes outperforming the best compression achievable by RLCSA.}
\label{tab:index_compression}
\end{table}

\begin{table}[!ht]
\centering
\begin{scriptsize}
\begin{tabular}{|l|r|r|r|r|}
\hline
 & \multicolumn{2}{|c|}{LZ77 Index} & \multicolumn{2}{|c|}{LZ-End Index} \\ \hline
Text & Mean & $D$ & Mean & $D$ \\ \hline \hline
$F_{41}$  & 0.00 & 0 & 0.10 & 1 \\ \hline
$R_{13}$  & 0.35 & 2 & 1.47 & 4 \\ \hline
$T_{29}$  & 0.51 & 2 & 1.58 & 3 \\ \hline \hline
DNA 0.1\% $^2$  & 13.45 & 27 & 11.43 & 26 \\ \hline
Proteins 0.1\% $^2$  & 14.77 & 38 & 12.25 & 43 \\ \hline
English 0.1\% $^1$  & 1.65 & 15 & 2.81 & 23 \\ \hline
Sources 0.1\% $^1$  & 1.55 & 16 & 2.65 & 34 \\ \hline \hline
Para  & 2.76 & 24 & 3.31 & 21 \\ \hline
Cere  & 3.53 & 37 & 4.32 & 21 \\ \hline
Influenza  & 2.00 & 33 & 1.58 & 24 \\ \hline
Escherichia Coli  & 2.15 & 13 & 2.40 & 15 \\ \hline
Coreutils  & 2.51 & 32 & 2.29 & 40 \\ \hline
Kernel  & 2.57 & 28 & 3.05 & 41 \\ \hline
Einstein (en)  & 2.53 & 14 & 2.62 & 22 \\ \hline
Einstein (de)  & 2.86 & 13 & 3.27 & 16 \\ \hline
World leaders  & 2.81 & 46 & 2.63 & 29 \\ \hline
\end{tabular}
\end{scriptsize}
\caption[\texorpdfstring{$D$ value and mean depth for the LZ indexes}{D value and mean depth for the LZ indexes}]{$D$ value (i.e., maximum depth) and mean depth for the LZ indexes}
\label{tab:dvalues}
\end{table}

\begin{table}[!ht]
\centering
\begin{scriptsize}
\begin{tabular}{|l|r|r|r|r|}
\hline
 & \multicolumn{2}{|c|}{LZ77 Index} & \multicolumn{2}{|c|}{LZ-End Index} \\ \hline
Text & Mean & $H$ & Mean & $H$ \\ \hline \hline
$F_{41}$  & 26.56 & 39 & 24.56 & 36 \\ \hline
$R_{13}$  & 19.54 & 27 & 21.7 & 45 \\ \hline
$T_{29}$  & 14.41 & 22 & 13.53 & 27 \\ \hline
DNA 0.1\% $^2$  & 6.63 & 26 & 5.81 & 25 \\ \hline
Proteins 0.1\% $^2$  & 4.77 & 19 & 4.86 & 22 \\ \hline
English 0.1\% $^1$  & 30.36 & 96 & 9.67 & 43 \\ \hline
Sources 0.1\% $^1$  & 31.06 & 98 & 10.08 & 64 \\ \hline
Para  & 9.38 & 31 & 8.39 & 259 \\ \hline
Cere  & 9.25 & 34 & 8.66 & 257 \\ \hline
Influenza  & 11.93 & 46 & 11.65 & 80 \\ \hline
Escherichia Coli  & 7.95 & 28 & 6.86 & 28 \\ \hline
Coreutils  & 9.03 & 39 & 11.21 & 175 \\ \hline
Kernel  & 12.36 & 51 & 9.45 & 46 \\ \hline
Einstein (en)  & 176.28 & 1003 & 23.66 & 118 \\ \hline
Einstein (de)  & 73.01 & 507 & 15.72 & 61 \\ \hline
World leaders  & 11.31 & 41 & 14.77 & 103 \\ \hline
\end{tabular}
\end{scriptsize}
\caption[\texorpdfstring{$H$ value and mean extraction cost for the LZ indexes}{H value and mean extraction cost for the LZ indexes}]{$H$ value (i.e., maximum extraction cost) and mean extraction cost for the LZ indexes}
\label{tab:hvalues}
\end{table}

\begin{table}[htbp]
\centering
\begin{scriptsize}
\begin{tabular}{|l|r|r|r|r|r|r|r|r|r|r|r|}
\hline
 & \trotate{Trailing characters} & \trotate{$B$} & \trotate{Suffix trie} & \trotate{Suffix trie skips} & \trotate{Reverse trie} & \trotate{Reverse trie skips} & \trotate{$rev\_ids$} & \trotate{Range} & \trotate{Depth} & \trotate{$P_S$} & \trotate{$B_S$} \\ \hline \hline
$F_{41}$  & 0.03 & 0.35 & 47.00 & 0.76 & 47.00 & 0.74 & 0.08 & 1.93 & 0.51 & 0.54 & 0.32 \\ \hline
$R_{13}$  & 0.03 & 0.34 & 46.87 & 0.72 & 46.86 & 0.67 & 0.08 & 1.92 & 0.94 & 0.53 & 0.31 \\ \hline
$T_{29}$  & 0.04 & 0.46 & 46.63 & 0.79 & 46.62 & 0.80 & 0.11 & 1.91 & 0.94 & 0.57 & 0.41 \\ \hline \hline
DNA 0.1\% $^2$  & 1.98 & 10.45 & 13.69 & 7.14 & 13.36 & 6.52 & 11.86 & 12.51 & 3.51 & 12.93 & 6.04 \\ \hline
Proteins 0.1\% $^2$  & 3.48 & 9.69 & 12.94 & 6.94 & 12.09 & 5.64 & 12.51 & 13.19 & 4.10 & 13.64 & 5.78 \\ \hline
English 0.1\% $^1$  & 4.50 & 9.59 & 13.15 & 5.99 & 13.40 & 5.59 & 11.56 & 12.19 & 2.76 & 12.60 & 8.67 \\ \hline
Sources 0.1\% $^1$  & 4.36 & 9.92 & 13.42 & 5.94 & 13.43 & 5.36 & 11.22 & 11.84 & 3.09 & 12.23 & 9.18 \\ \hline \hline
Para  & 1.90 & 8.52 & 12.53 & 6.81 & 12.23 & 6.43 & 13.28 & 13.95 & 3.32 & 14.36 & 6.65 \\ \hline
Cere  & 1.87 & 8.60 & 12.72 & 6.85 & 12.42 & 6.45 & 13.06 & 13.72 & 3.55 & 14.12 & 6.64 \\ \hline
Influenza  & 2.34 & 9.35 & 13.08 & 7.82 & 12.71 & 7.26 & 11.70 & 12.30 & 3.46 & 12.68 & 7.31 \\ \hline
Escherichia Coli  & 2.56 & 8.20 & 12.66 & 6.65 & 12.40 & 6.38 & 13.43 & 14.11 & 2.68 & 14.52 & 6.42 \\ \hline
Coreutils  & 4.92 & 7.50 & 13.23 & 5.99 & 13.21 & 5.62 & 12.91 & 13.57 & 3.47 & 13.96 & 5.62 \\ \hline
Kernel  & 5.09 & 7.47 & 13.41 & 6.01 & 13.25 & 5.72 & 12.71 & 13.37 & 3.37 & 13.78 & 5.82 \\ \hline
Einstein (en)  & 5.19 & 8.23 & 15.06 & 6.59 & 14.45 & 5.95 & 11.02 & 11.70 & 2.88 & 12.06 & 6.86 \\ \hline
Einstein (de)  & 4.55 & 7.93 & 16.80 & 6.42 & 17.52 & 5.78 & 9.75 & 10.55 & 3.08 & 10.77 & 6.78 \\ \hline
World leaders  & 4.47 & 8.28 & 13.57 & 6.83 & 13.98 & 6.16 & 11.49 & 12.13 & 4.08 & 12.52 & 6.49 \\ \hline
\end{tabular}
\end{scriptsize}
\caption[Detailed space of LZ77 index structures]{Detailed space of LZ77 index structures. Values are in percentage of the total size.}
\label{tab:lz77_det}
\end{table}

\begin{table}[htbp]
\centering
\begin{scriptsize}
\begin{tabular}{|l|r|r|r|r|r|r|r|r|r|r|r|}
\hline
  & \trotate{Trailing characters} & \trotate{$B$} & \trotate{Suffix trie} & \trotate{Suffix trie skips} & \trotate{Reverse trie} & \trotate{Reverse trie skips} & \trotate{$rev\_ids$} & \trotate{Range} & \trotate{Depth} & \trotate{$P_S$} & \trotate{$B_S$} \\ \hline \hline
$F_{41}$  & 0.03 & 0.51 & 47.13 & 0.47 & 47.12 & 0.46 & 0.08 & 1.94 & 0.51 & 0.54 & 0.48 \\ \hline
$R_{13}$  & 0.05 & 0.82 & 45.70 & 0.78 & 45.70 & 0.78 & 0.15 & 2.24 & 1.75 & 0.61 & 0.74 \\ \hline
$T_{29}$  & 0.06 & 1.02 & 45.55 & 1.05 & 45.54 & 0.84 & 0.20 & 2.22 & 1.33 & 0.65 & 0.83 \\ \hline \hline
DNA 0.1\% $^2$  & 1.70 & 15.01 & 11.25 & 6.12 & 10.66 & 5.67 & 10.78 & 11.35 & 2.99 & 11.71 & 12.75 \\ \hline
Proteins 0.1\% $^2$  & 2.91 & 14.38 & 10.70 & 5.62 & 10.17 & 5.20 & 11.07 & 11.65 & 3.55 & 12.03 & 12.71 \\ \hline
English 0.1\% $^1$  & 3.90 & 14.23 & 11.58 & 5.43 & 10.73 & 4.67 & 10.58 & 11.13 & 2.78 & 11.49 & 13.48 \\ \hline
Sources 0.1\% $^1$  & 3.76 & 13.93 & 12.55 & 5.33 & 11.53 & 4.48 & 10.22 & 10.75 & 3.10 & 11.10 & 13.24 \\ \hline \hline
Para  & 1.63 & 13.66 & 11.10 & 5.69 & 10.48 & 5.28 & 11.93 & 12.53 & 2.67 & 12.88 & 12.16 \\ \hline
Cere  & 1.66 & 14.13 & 10.99 & 5.89 & 10.31 & 5.48 & 11.60 & 12.18 & 2.71 & 12.54 & 12.52 \\ \hline
Influenza  & 2.09 & 13.99 & 11.41 & 6.54 & 11.03 & 6.28 & 10.44 & 10.97 & 2.75 & 11.32 & 13.18 \\ \hline
Escherichia Coli  & 2.21 & 13.61 & 10.69 & 5.65 & 10.05 & 5.24 & 12.16 & 12.77 & 2.33 & 13.12 & 12.18 \\ \hline
Coreutils  & 4.26 & 12.63 & 11.70 & 5.10 & 11.53 & 4.81 & 11.18 & 11.75 & 3.17 & 12.09 & 11.78 \\ \hline
Kernel  & 4.40 & 12.84 & 11.63 & 5.11 & 11.42 & 4.85 & 11.01 & 11.57 & 3.30 & 11.93 & 11.92 \\ \hline
Einstein (en)  & 4.47 & 13.82 & 12.73 & 5.69 & 12.10 & 5.05 & 9.49 & 10.05 & 2.92 & 10.38 & 13.28 \\ \hline
Einstein (de)  & 3.94 & 13.69 & 14.00 & 5.60 & 13.38 & 4.92 & 8.99 & 9.65 & 2.80 & 9.89 & 13.11 \\ \hline
World leaders  & 3.90 & 13.77 & 11.77 & 5.65 & 12.12 & 5.33 & 10.02 & 10.57 & 3.02 & 10.93 & 12.91 \\ \hline
\end{tabular}
\end{scriptsize}
\caption[Detailed space of LZ-End index structures]{Detailed space of LZ-End index structures. Values are in percentage of the total size.}
\label{tab:lzend_det}
\end{table}

\clearpage

\results{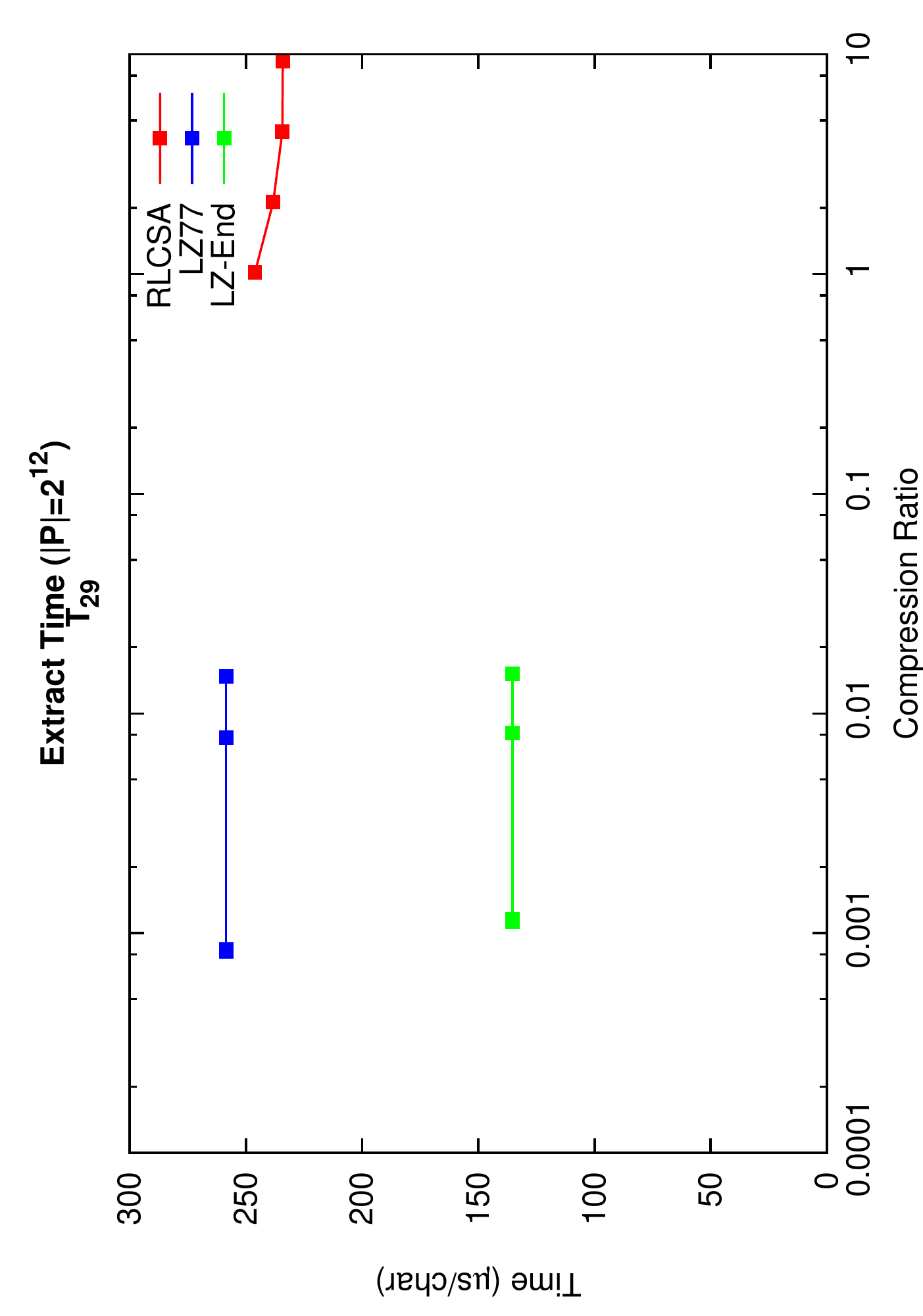}{$T_{29}$ results}{res:t29}
\clearpage
\results{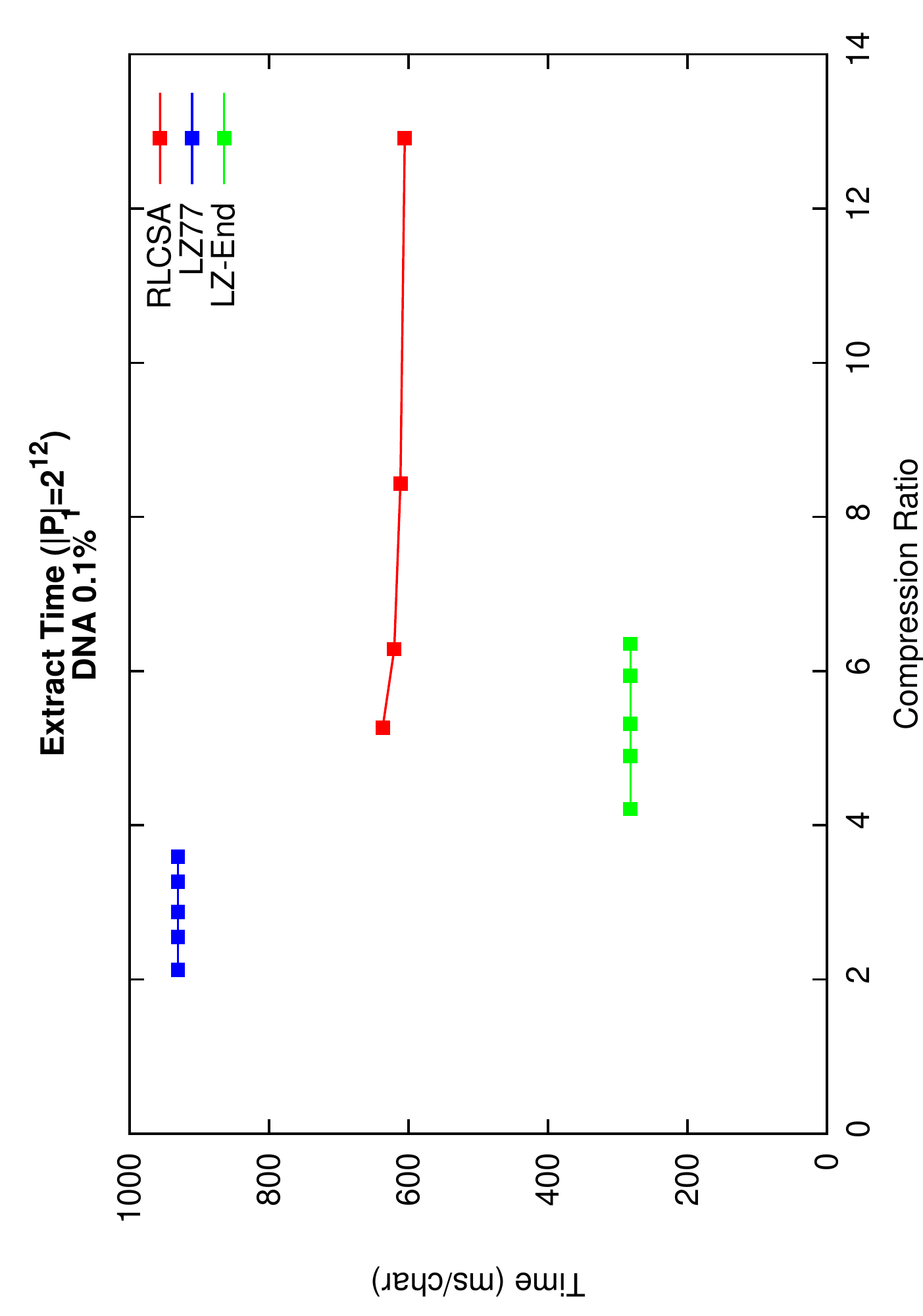}{DNA 0.1\% $^1$ results}{res:dna}
\clearpage
\results{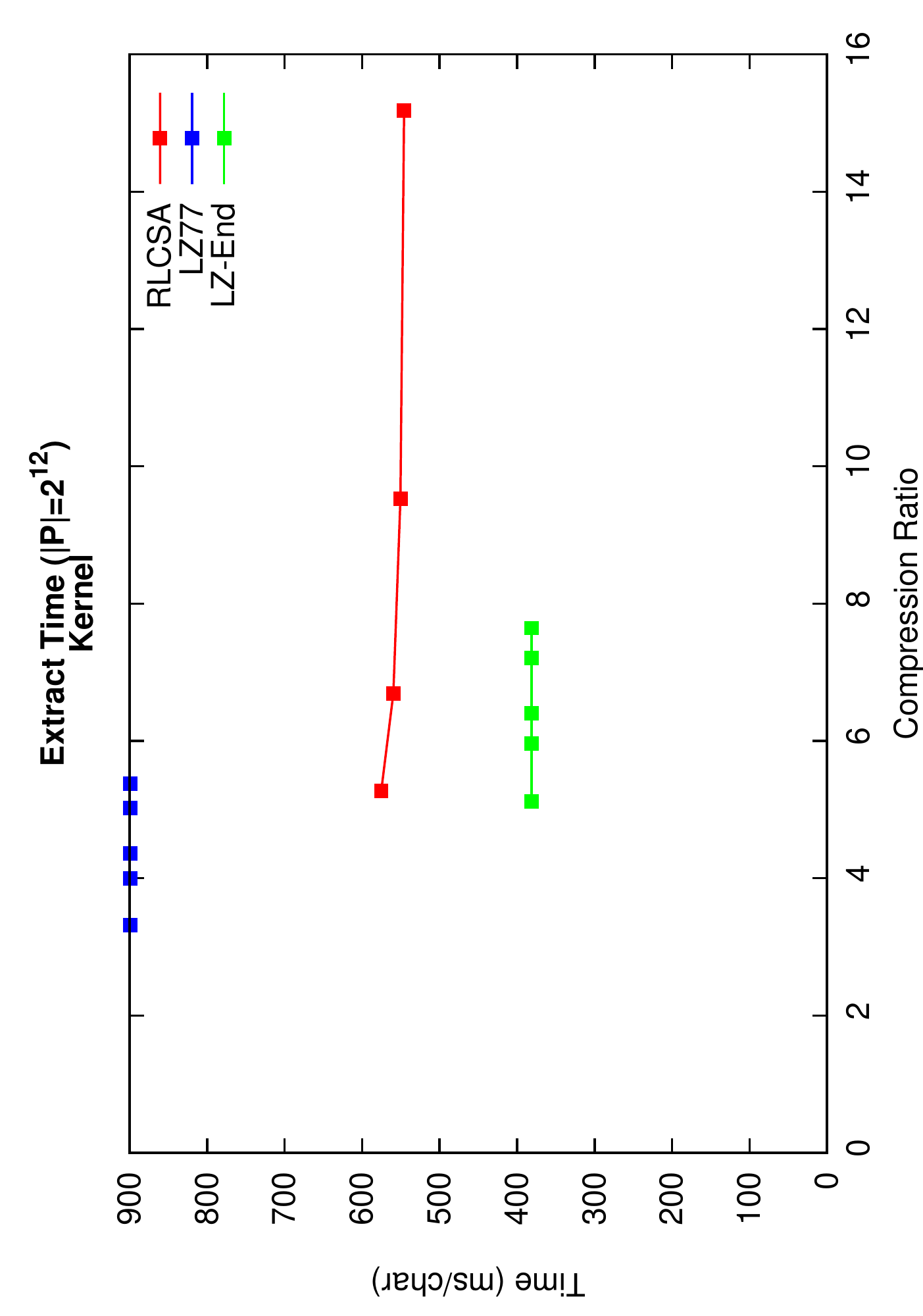}{Kernel results}{res:kernel}
\clearpage
\section{Analysis of the Results}
\label{sec:analysis_index}
It can be seen from the results presented above that our indexes are competitive with RLCSA and in most cases they show better space/time trade-offs.

Figure \ref{fig:construction} shows that our indexes are built more efficiently than RLCSA. The space needed to build the LZ77 index is about 60\% of that of RLCSA, and for the case of \jeremy{the} LZ-End index the space is about 85\% of that of RLCSA. The construction time for LZ77 is about 35\% of that of RLCSA, yet for LZ-End it is about 185\% of that of RLCSA (i.e., slower). Our space occupancy during construction is a great advantage against RLCSA as it allows us to build the index for larger texts using the same resources. We could reduce construction space further by sacrificing construction time, recall Section \ref{sec:constr}.

The compression ratio of our indexes is usually superior to that of RLCSA (see Table \ref{tab:index_compression}). When considering the lower bound of RLCSA, which only supports \emph{count} and \emph{exists} \jeremy{queries}, our smallest index compresses better than RLCSA for all except \djeremy{3} texts \jeremy{influenza, coreutils and World Leaders}. When considering RLCSA with a sampling step equal to 512, our compression is better for all except \djeremy{1} text \jeremy{coreutils}. From now on we compare our indexes with the RLCSA with sampling step 512. The compression difference is more noticeable for artificial texts, where our compression is 100-1000 times better than RLCSA. For DNA collections Para and Cere our best compression space (always achieved with LZ77) is about 45\% of RLCSA's, yet it raises to 70-80\% for Influenza and Escherichia Coli. The space is also 80\% on World Leaders. On kernel our space is 60\% of RLCSA, but they are almost the same on Coreutils. On Wikipedia articles our space is 20-30\% of that of RLCSA. LZ-End needs more space that LZ77, losing to RLCSA in Influenza, Escherichia Coli, Coreutils and World Leaders. For texts of pseudo-real collections the compression ratio of LZ77 is about 60\% of that of RLCSA, and even the alternative using more space has better compression ratios than RLCSA. It can also be seen in Figures \figuresres labeled ``results (2)'' that in most cases the space/time trade-off is in favor of LZ77/LZ-End. Our LZ77 indexes use 2.63-7.52 times the space of p7zip and our LZ-End indexes use 4.57-23.03 times the space, depending on the type of text and the number of structures we use (e.g., tries).
Finally, the compression ratios of LZ78-based indexes are not competitive at all (note that ILZI's maximal parsing performs better than LZ78).

\jeremy{We expected to improve the compression of RLCSA for highly repetitive texts, since LZ77 is more powerful to detect repetitions (see Lemma \ref{lemma:lz77rep} and Section \ref{sec:rlcsa}). 
For artificial texts, the most repetitive ones, this situation is even more clear. For such texts, most of the space of the RLCSA corresponds to the space to represent the samplings, which is difficult to compress \cite{RLCSA_journal}.}

Tables \ref{tab:dvalues} and \ref{tab:hvalues} show that $D$ is generally moderate (below 42), and that the greatest extraction cost is also moderate (at most 257 \jeremy{steps} in LZ-End and at most 98 \jeremy{steps} in LZ77 except for the texts of Wikipedia). When taking into account the mean values of the depth and extraction cost, the values decrease noticeably, being the average extraction cost below 25 \jeremy{steps} for LZ-End and below 32 \jeremy{steps} for LZ77 (except for the texts of Wikipedia). 

Tables \ref{tab:lz77_det} and \ref{tab:lzend_det} show that the \ddiego{sparse} suffix \diego{trie} and the reverse trie use more space than the rest of the data structures. Then we see that the range structure, the reverse ids and the permutation use roughly the same space (they require $n'\log n'+o(n'\log n')$ bits). For LZ-End we see that the space needed to represent both bitmaps increases noticeably, being even higher to that of the range structure, or the permutation. Finally, we see that the skips only use 11-13\% and the wavelet tree of depths uses 2-4\% of the total space. Additionally, for the artificial texts, we see that more than 90\% of the space is used to represent the \ddiego{sparse} suffix \diego{trie} and the reverse trie. This is because our implementation of DFUDS stores a boosting table of constant size 1kiB, and for these texts the number of phrases $n'$ is less than 100, being the space of the table considerably larger than the space of the remaining data structures. We also note that our implementation stores the labels of the tree using 1 byte for each symbol.

The extract time of LZ-End index is better than RLCSA in all texts, being at least twice as fast, and up to 10 times faster for short passages. The LZ77 index extracts substrings of length up to 50 faster than RLCSA. When taking into account the space/time trade-offs (top-left plot of Figures \figuresres labeled ``results (2)'') we see that our indexes improve RLCSA both in extraction time and space, dominating the curve defined by RLCSA, excepting the texts where LZ-End loses in space, in which case no one dominates. \jeremy{The extraction space/time trade-off is better than that of RLCSA, because since RLCSA cannot compress the sampling, it has to use a sparse sampling to be competitive in space.}

The performance of locate queries is related to the pattern length. This is because in our indexes the locating cost is quadratic on the pattern length (yet, this is in the worst case; in practice many searches are abandoned earlier). It can be seen in plots 2-6 of Figures \figuresres labeled ``results (2)'' that for patterns of length 2 or 4 all of our indexes are significantly faster than RLCSA. This is because our indexes are much faster to locate each occurrence, as this is the cost that dominates for short patterns. However, when we increase the length of the patterns, the increase in cost is noticeable for the alternatives using binary search, which are those using least space. However, the alternative number 1, using tries, shows a time basically independent of $|P|$ (except somewhat in DNA 0.1\%$^1$ and Escherichia Coli). This is also seen in the top-right plot of Figures \figuresres labeled ``results (1)'' . RLCSA time is almost insensitive to $|P|$, thus in several cases it becomes faster for longer patterns (which also have fewer occurrences, for reporting which RLCSA is slower).

By analyzing the performance results of SLPs \cite{CFMPNbibe10} it is clear that the compression ratio of SLPs (at least when using Re-Pair to create the grammar) is worse than that of RLCSA. For the case of DNA (Para, Cere and Influenza) the compression ratio is more than twice that of the LZ77 indexes. Furthermore, the locate time of SLPs is only comparable to RLCSA for small patterns ($|P|\le6$), in which case our indexes show a space/time trade-off much better than that of RLCSA.

Finally, we have that \emph{exists} queries are solved consistently faster by RLCSA than by our indexes. Looking at plots 3 and 4 of Figures \figuresres labeled ``results (1)'' we notice that our larger variants are comparable to RLCSA, although always slower, in the case of patterns present in the text. The difference widens in the case of non-existent patterns, as RLCSA improves more sharply. Moreover, in our indexes the time increases with the length, opposite to RLCSA where the time is practically constant when the pattern does not exist. In plots 7 and 8 of Figures \figuresres labeled ``results (2)'' one can see the trade-off of \emph{exists} queries. They show that binary search is not an alternative if we are interested in this type of queries. For the case of patterns present in the text, binary searching the queries takes about 10 times more, and for patterns not present in the text about 10-1000 times more, than the time needed using tries.


%% file: conclusions.tex
\chapter{Conclusions}
We have presented a new full-text self-index based on the Lempel-Ziv parsing. This index is especially well suited for applications in which the text is highly repetitive and the user is interested in finding patterns in the text (\emph{locate}) and accessing arbitrary substrings of the text (\emph{extract}). Our indexes provide a much better space/time trade-off than the previous ones for these operations.

The compression ratio of our indexes is more than 10 times better than previous indexes based on LZ78, which are shown to be inappropriate for very repetitive texts. Additionally, the compression ratio of our smallest index is, for almost all texts (13 out of 16)\jeremy{\footnote{We could not devise any characteristic that explains why the compression ratio of RLCSA for those texts is superior to our indexes.}}, better than the lower bound achievable using RLCSA \cite{MNSV08}, the best previous self-index for these texts. When compared to the smallest practical RLCSA, the compression of our index is better for all except one text, usually by a factor of 2 at least. Compared to pure LZ77 compression, our index takes usually 3-6 times the space achieved by p7zip.

We also introduced a new LZ-parsing called LZ-End, which is close to LZ77 in compression ratio and gives faster access to text substrings.
The extraction speed when using LZ-End is always better than that of RLCSA, and the extraction speed of our LZ77-based index is also superior for small substrings. Our indexes are always better for locating the occurrences of short patterns (of length up to 10), and the results are mixed for longer ones. \jeremy{This is because our locating time is quadratic and depends also in the extraction speed, which shows different behaviors according to the text.}

The only operation for which RLCSA is consistently better than our indexes is for answering if a pattern is present in the text (\emph{exists}), \djeremy{being} the difference \jeremy{being} even more notorious for the case of non-existent patterns. Similarly, our indexes cannot count the number of occurrences without locating them all, whereas the RLCSA can do this very fast. Nevertheless, \djeremy{we believe} \jeremy{it has been argued \cite{ANjea10} that} these two queries are \jeremy{used in much more specific applications serving as a basis for complex tasks such as approximate pattern matching or text categorization, while} \djeremy{less relevant than} extracting and locating \jeremy{are the most important for general applications}.


An interesting goal for future research would be to reduce the $m^2$ factor of the \emph{locate} query time to just $m$. This improvement would make our index even more attractive. This has been achieved for other LZ-based indexes \cite{AN07,ilzi}, yet these have been built on LZ parsings that are too weak for very repetitive texts.

Another line of research would be to design new LZ-like compression schemes allowing fast decompression of random substrings of the LZ parsing. Note that our only trade-off related to extraction speed is the use of LZ-End instead of LZ77, and still LZ-End takes constant time per extracted symbol only in certain cases. In a recent work, Kuruppu \etal ~\cite{KPZ10} use a single document as the dictionary for the LZ77 algorithm, storing that source document in plain form. This method is a heuristic and works fairly well enough only when the documents of the collection are not successive versions, as in collections of DNA. However, even the compression achieved for DNA collections (para and cere) is almost the same than the compression we achieved using our best LZ77 variant, yet we have a self-index and they are only able to extract text, although their extraction times are more than 100 times faster than ours. \ddiego{Anyways} \diego{Nevertheless}, this method is orthogonal to our index proposal and we could build our self-index \diego{on} \ddiego{in} top of their compression scheme.
 
Another important line of research is to devise an LZ parsing algorithm that uses space proportional to that of the final compressed text. Currently, to build the LZ77 parsing one needs about six times the space of the original text. Although this space is lower than that of RLCSA, is still too much to handle very large text collections. Alternatively, a parsing algorithm working in secondary storage would also be useful to handle very large collections. We are aware that this is a more than challenging task, as the parsing process is strongly related with dynamic self-indexes for repetitive texts. That is, if we have a dynamic self-index (or at least an index able to insert strings at the end) we can easily produce the LZ parsing of the text. Hence, studying how we can build a dynamic LZ-based index is a natural research direction.

It would be also interesting to study if counting could be answered more efficiently, and if more meaningful operations like approximate pattern matching could be implemented, or if some operations of the suffix tree could be simulated on the index.

Another interesting research goal is to decrease the space factor, both in theory and in practice. Compared to a pure LZ77 compressor, the factor is 4 in theory and 3-6 in practice, as mentioned. Such a reduction has been achieved for Arroyuelo and Navarro's LZ-index, reducing the factor from 4 \cite{Nav02} to $(2+\epsilon)$ \cite{AN06} (see Section \ref{sec:lz78index}). This was possible because there was some redundancy between the components of the index. We could also reduce the factor by coding the bitmaps of the wavelet tree of depths in compressed form \cite{RRR02} (see Section \ref{sec:bitmaps}), since most depths are very low in practice and only some are high. However, the space improvement would not be too impressive, since the space of wavelet tree of depths is just 2-4\% of the index size.
 
We have also presented a text corpus oriented to repetitive texts. The main goal of this corpus is to become a reference set in experimentation with this kind of texts. The corpus is publicly available at \corpusurl.

Finally, our implementation has been left public in the site \indexurl, to promote its use in real-world and research applications and to serve as a baseline for future developments in repetitive text indexing.

%% file: results_appendix.tex
\chapter{Experimental Results}
\label{appendix:results}
In this appendix we present the results of the experiments described in Section \ref{sec:results} for the remaining texts.
\results{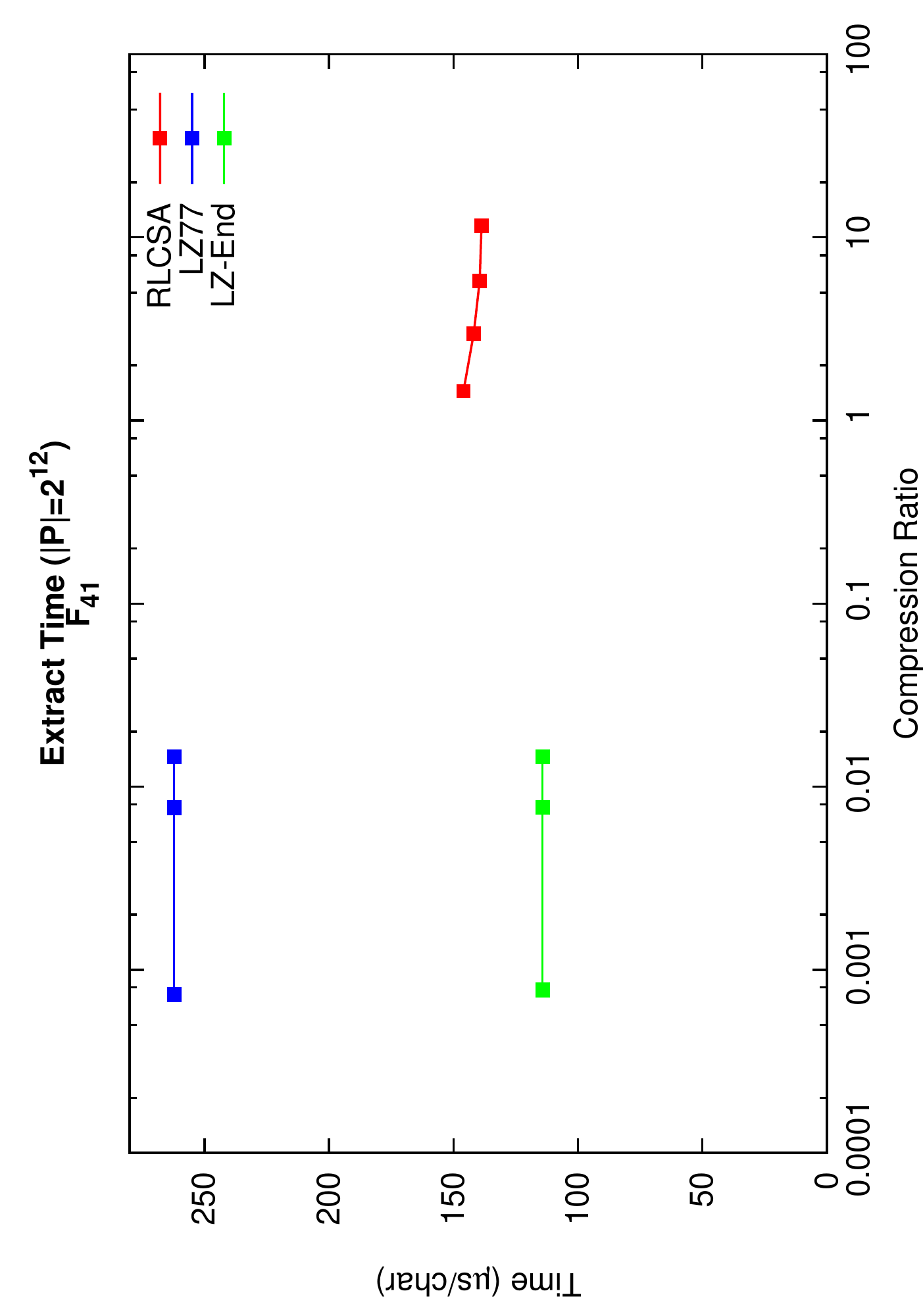}{$F_{41}$ results}{res:f41}
\results{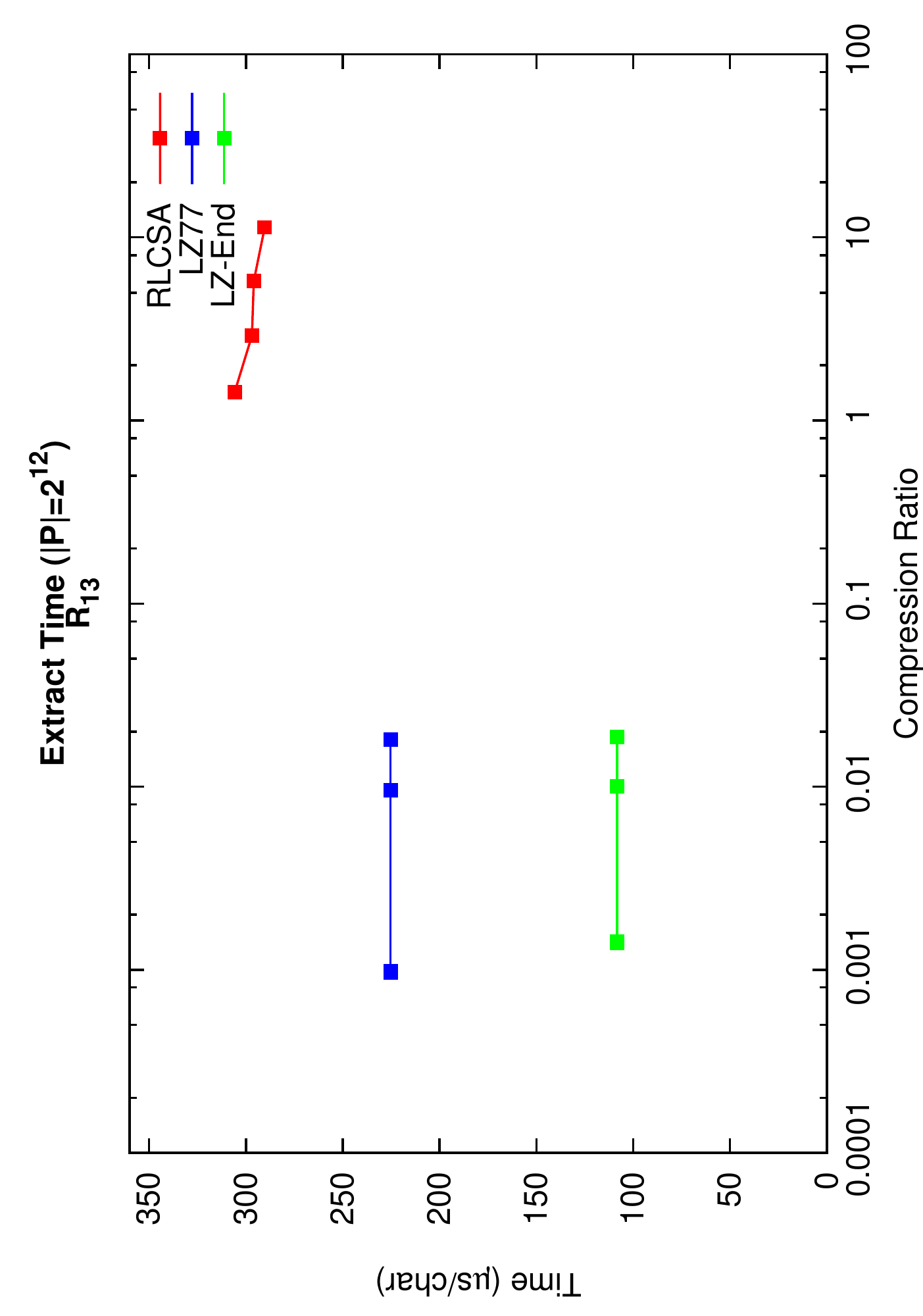}{$R_{13}$ results}{res:r13}
\clearpage
\results{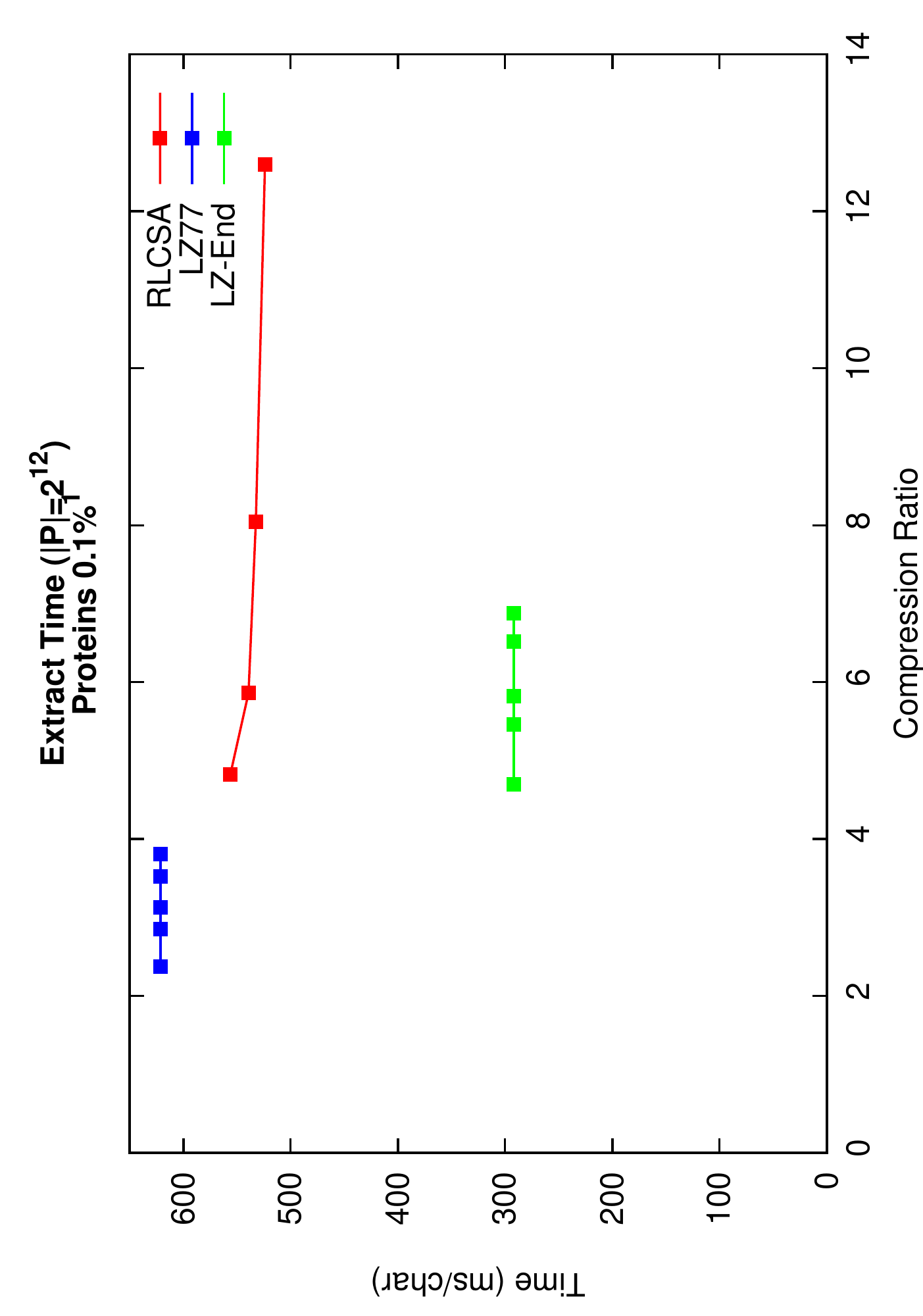}{Proteins 0.1\% $^1$ results}{res:proteins}
\results{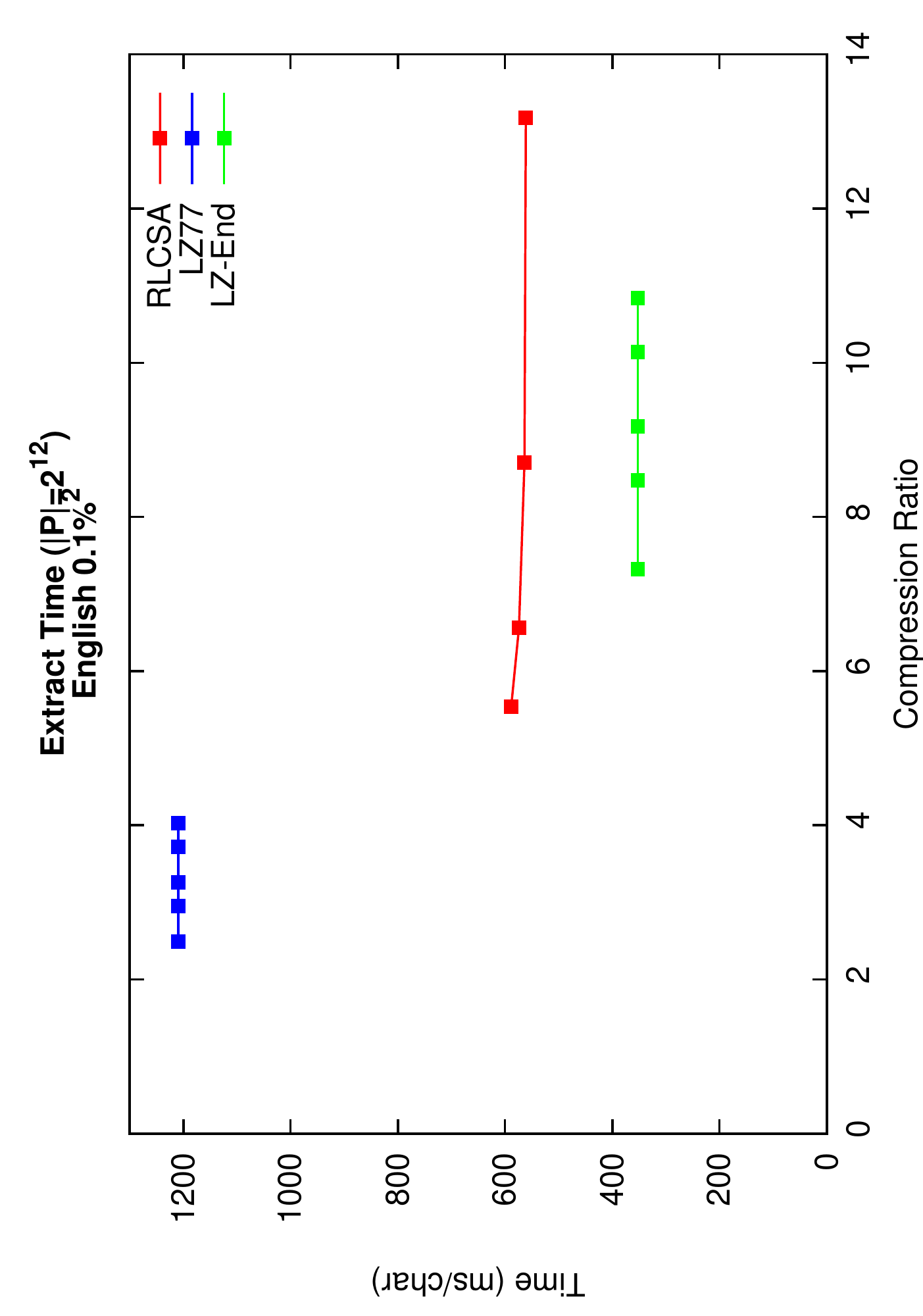}{English 0.1\% $^2$ results}{res:english}
\results{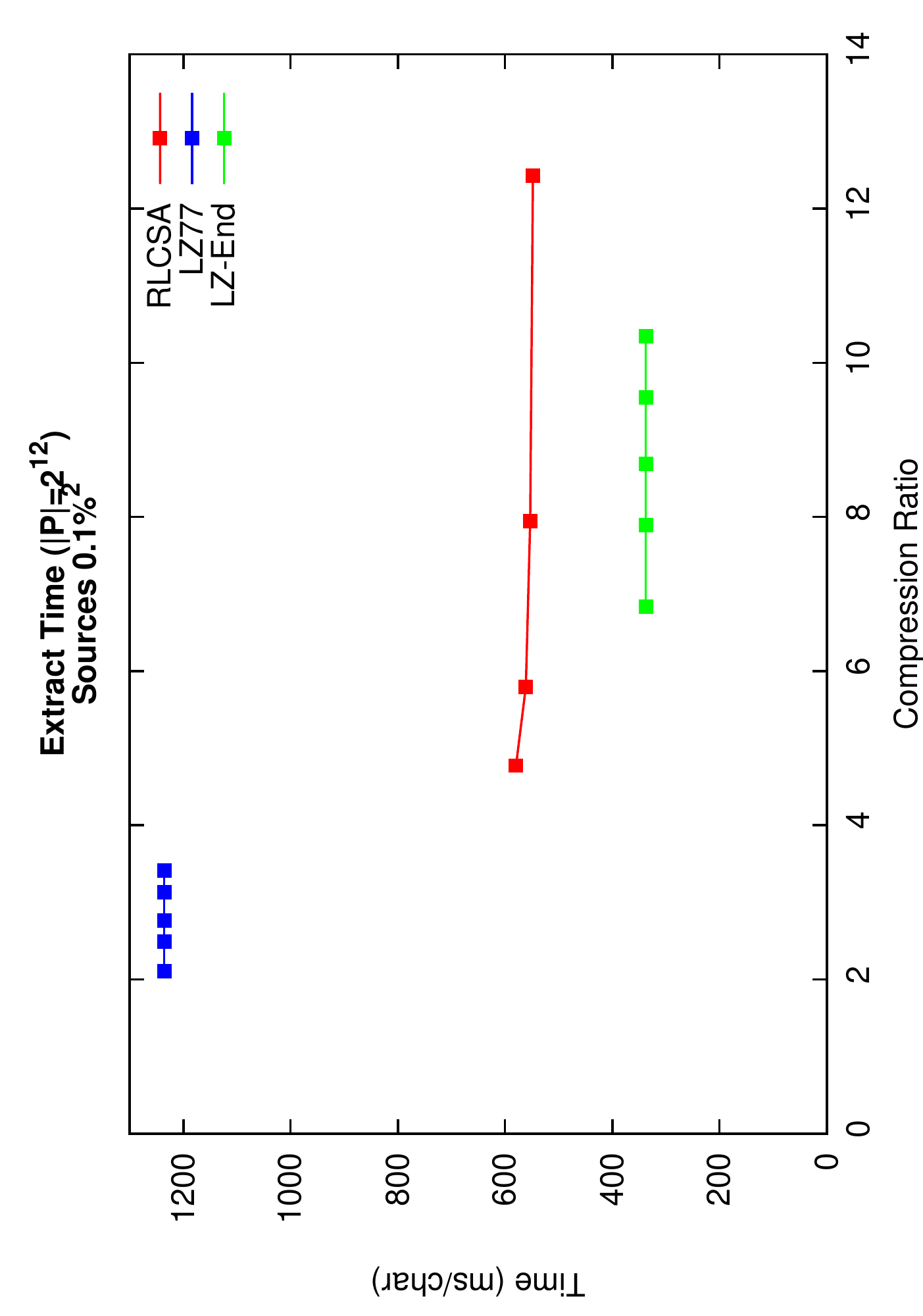}{Sources 0.1\% $^2$ results}{res:sources}
\clearpage
\results{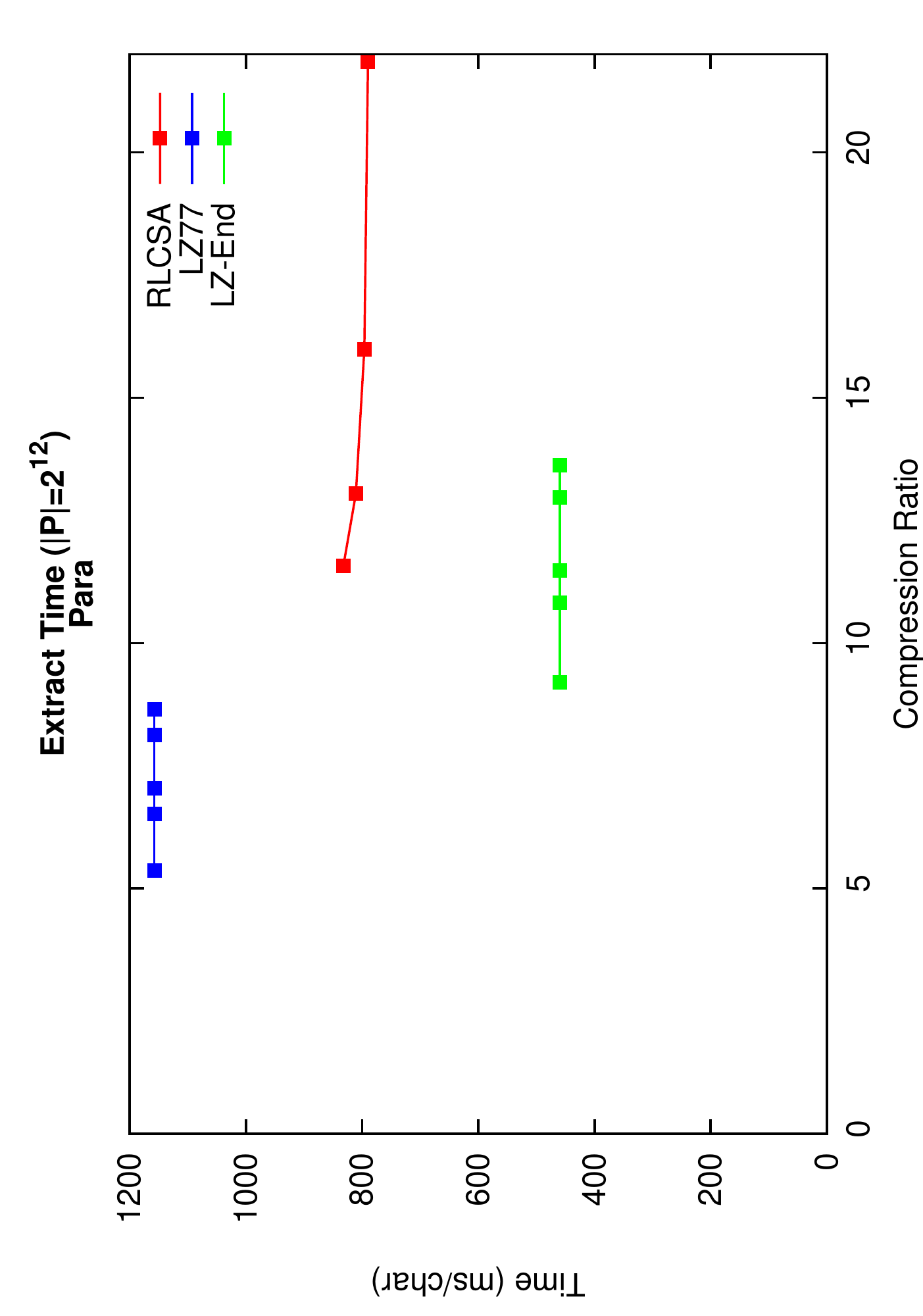}{Para results}{res:para}
\results{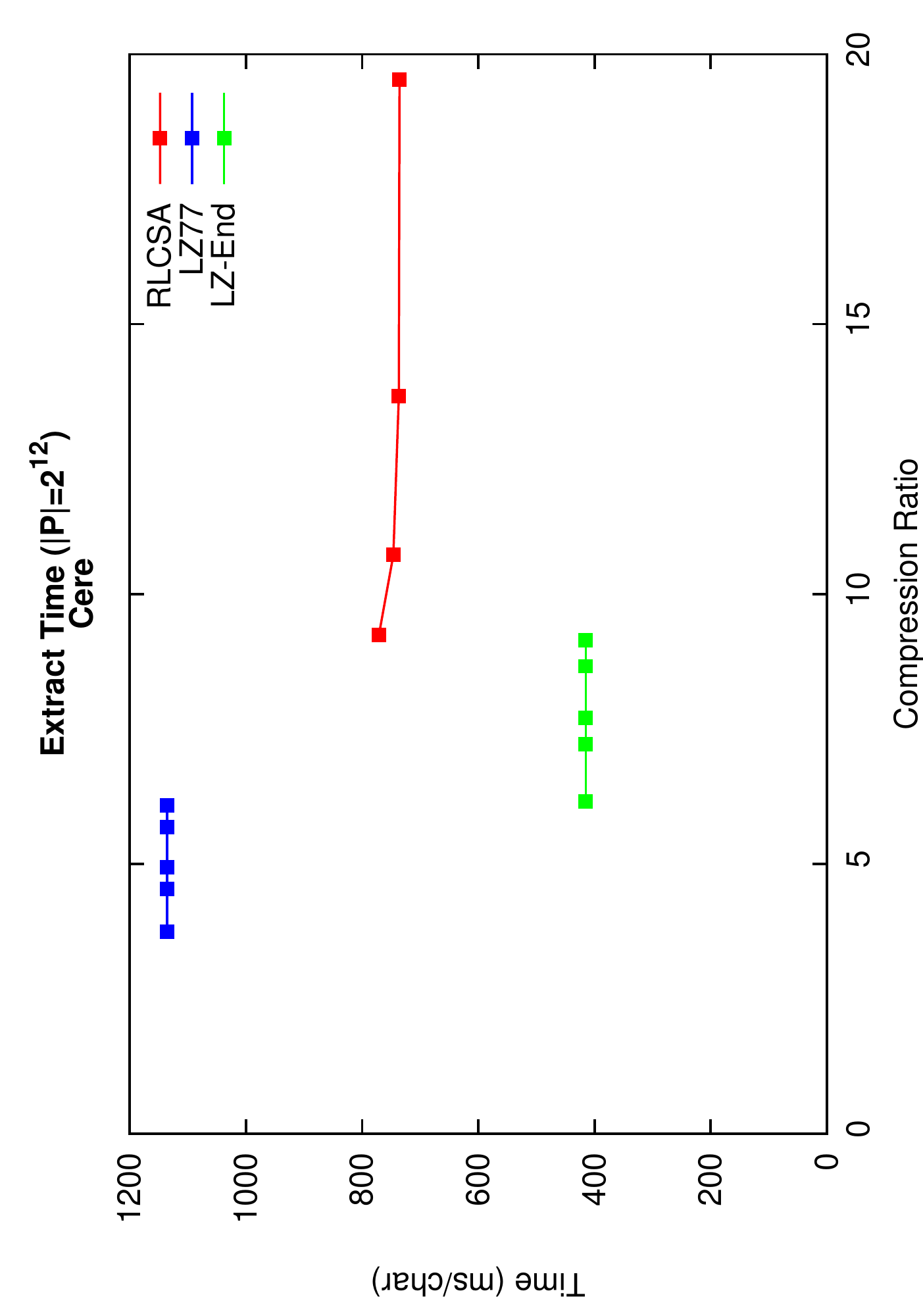}{Cere results}{res:cere}
\results{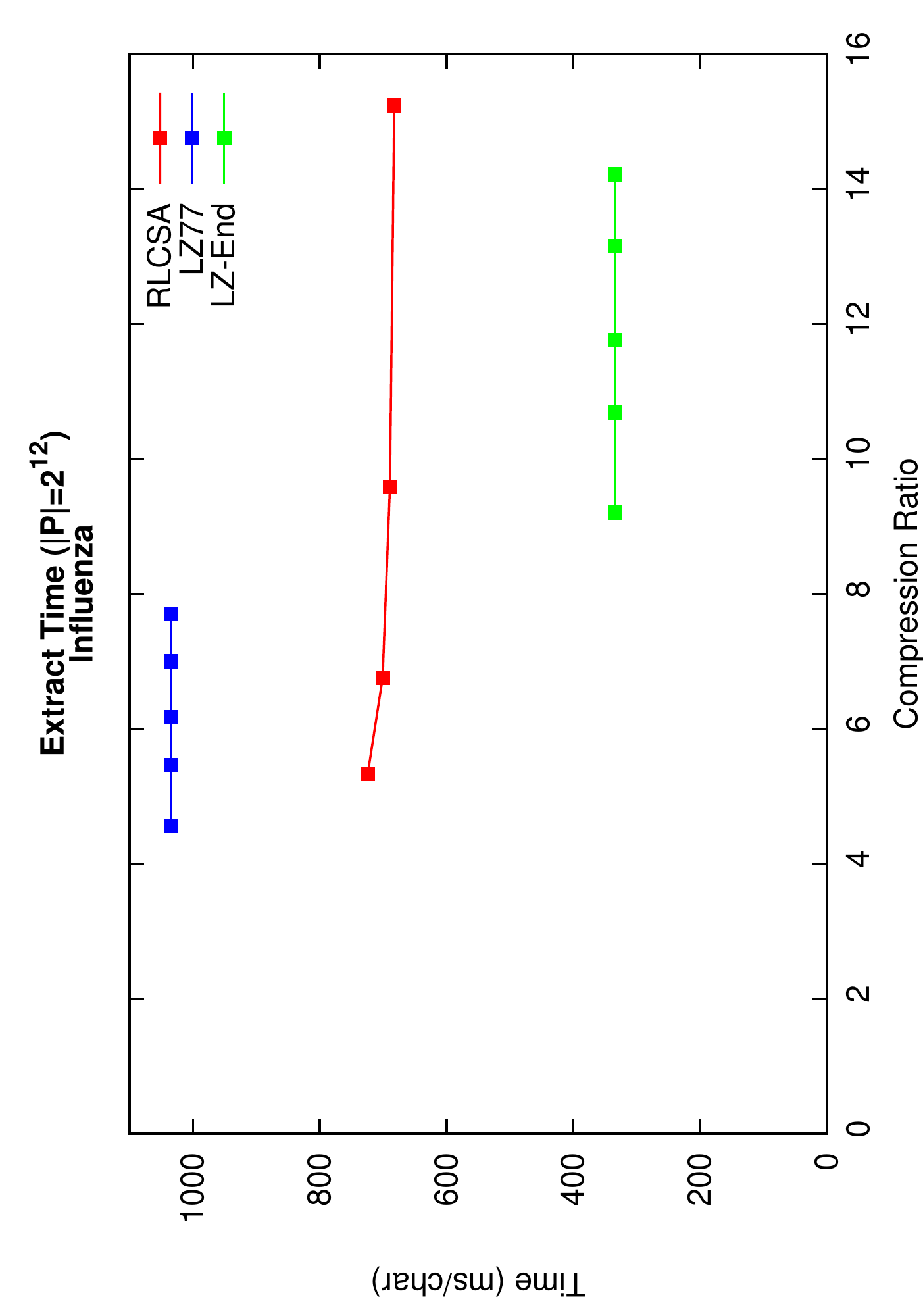}{Influenza results}{res:influ}
\results{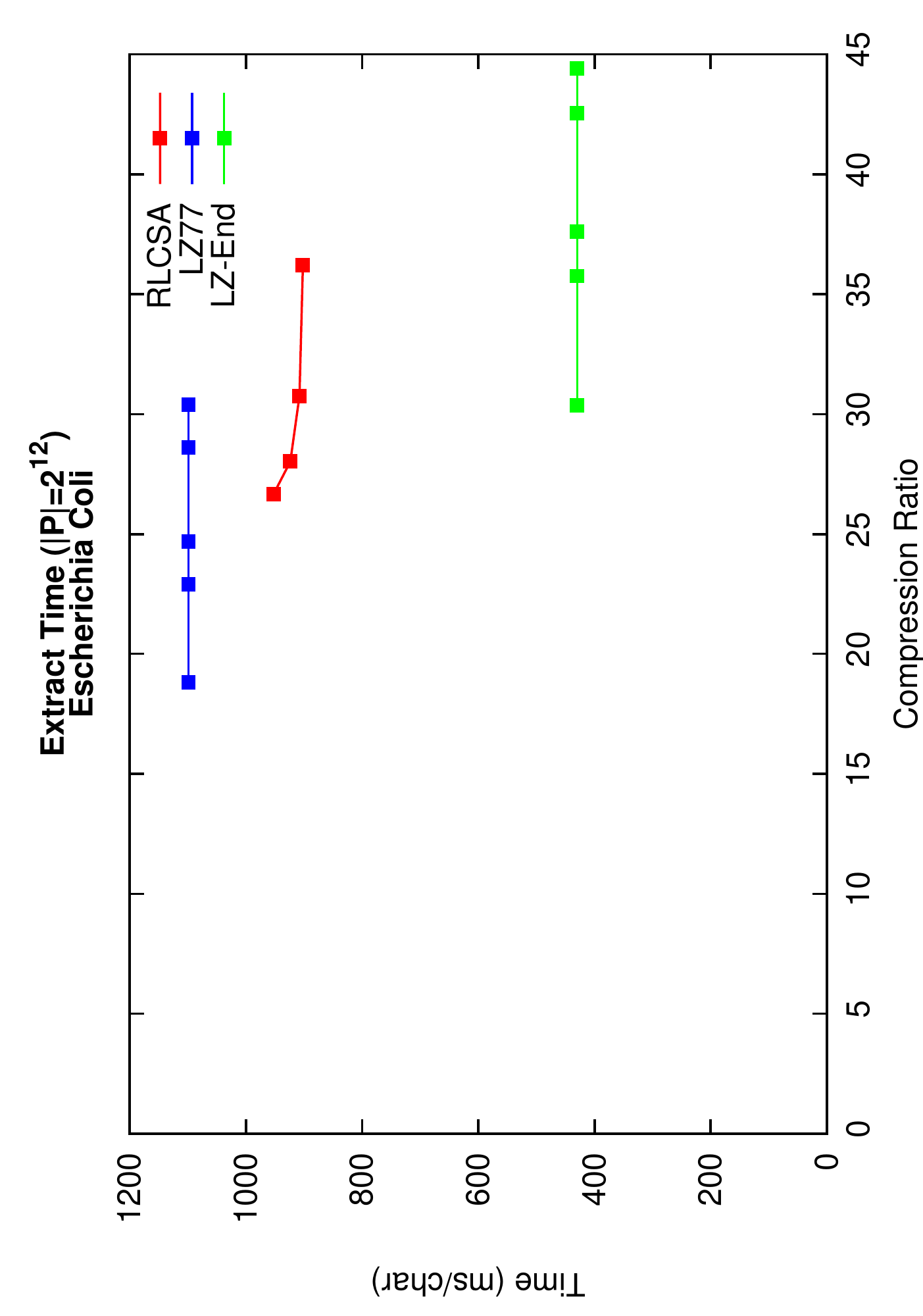}{Escherichia Coli results}{res:coli}
\results{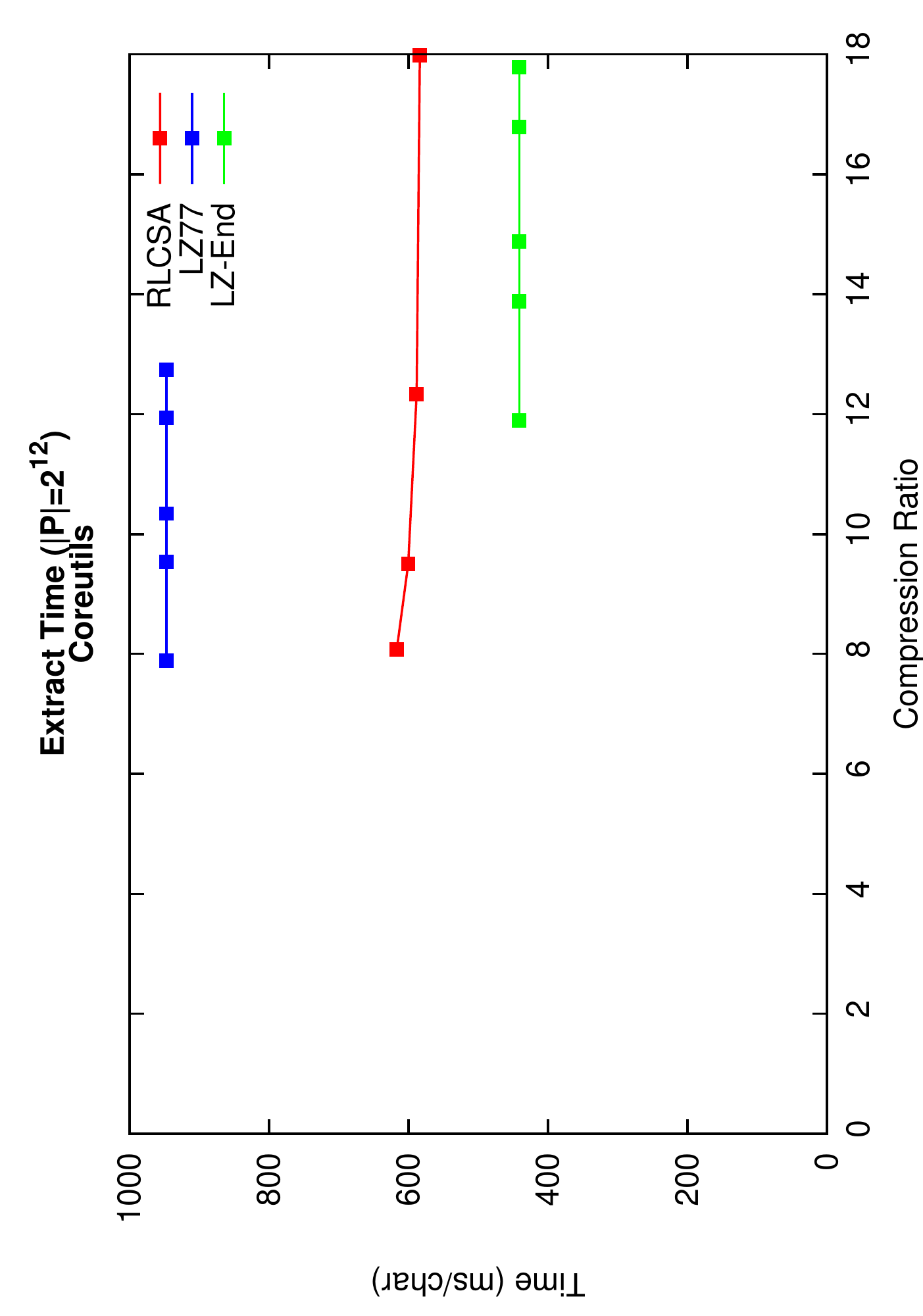}{Coreutils results}{res:coreutils}
\clearpage
\results{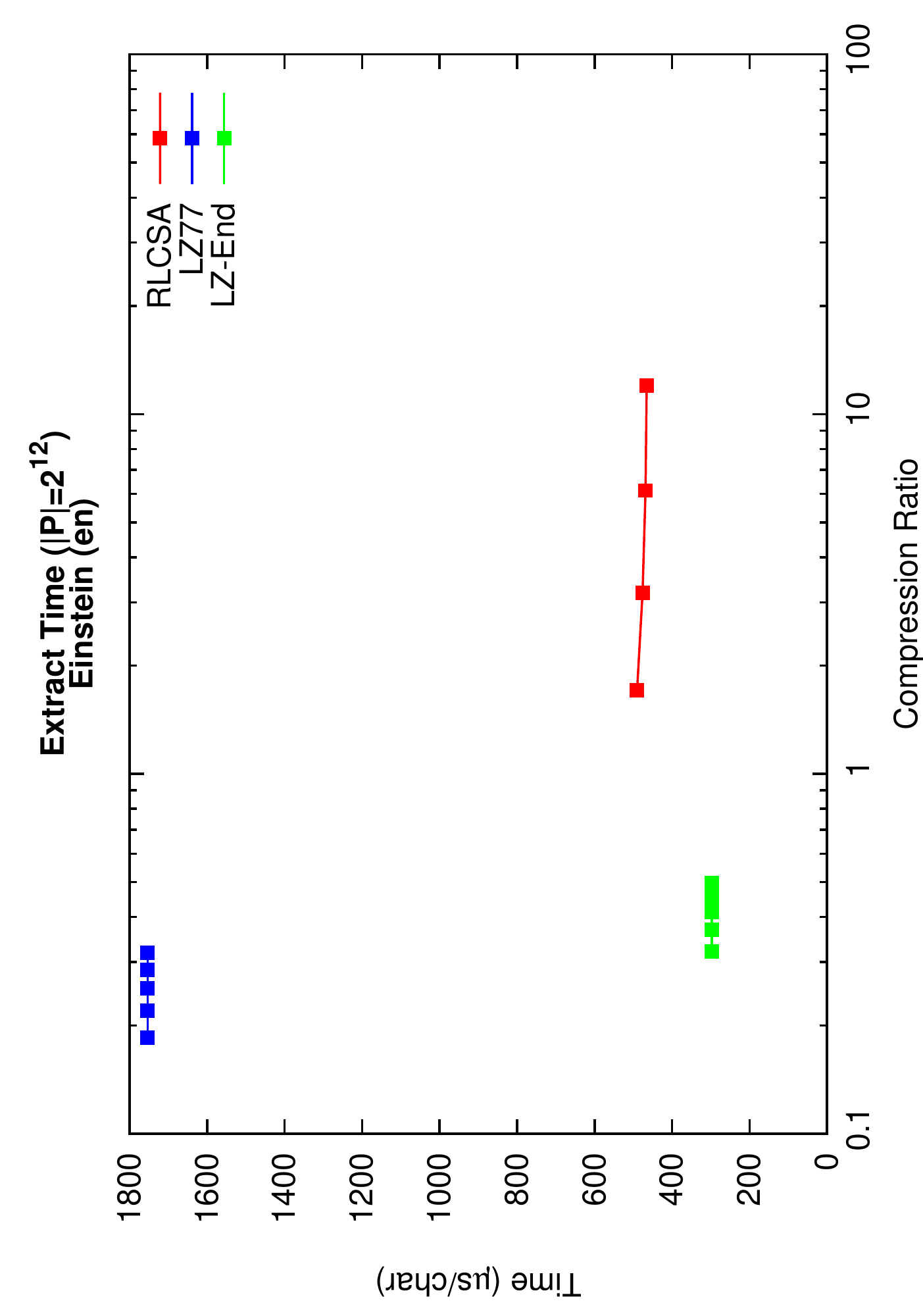}{Einstein (en) results}{res:wikien}
\results{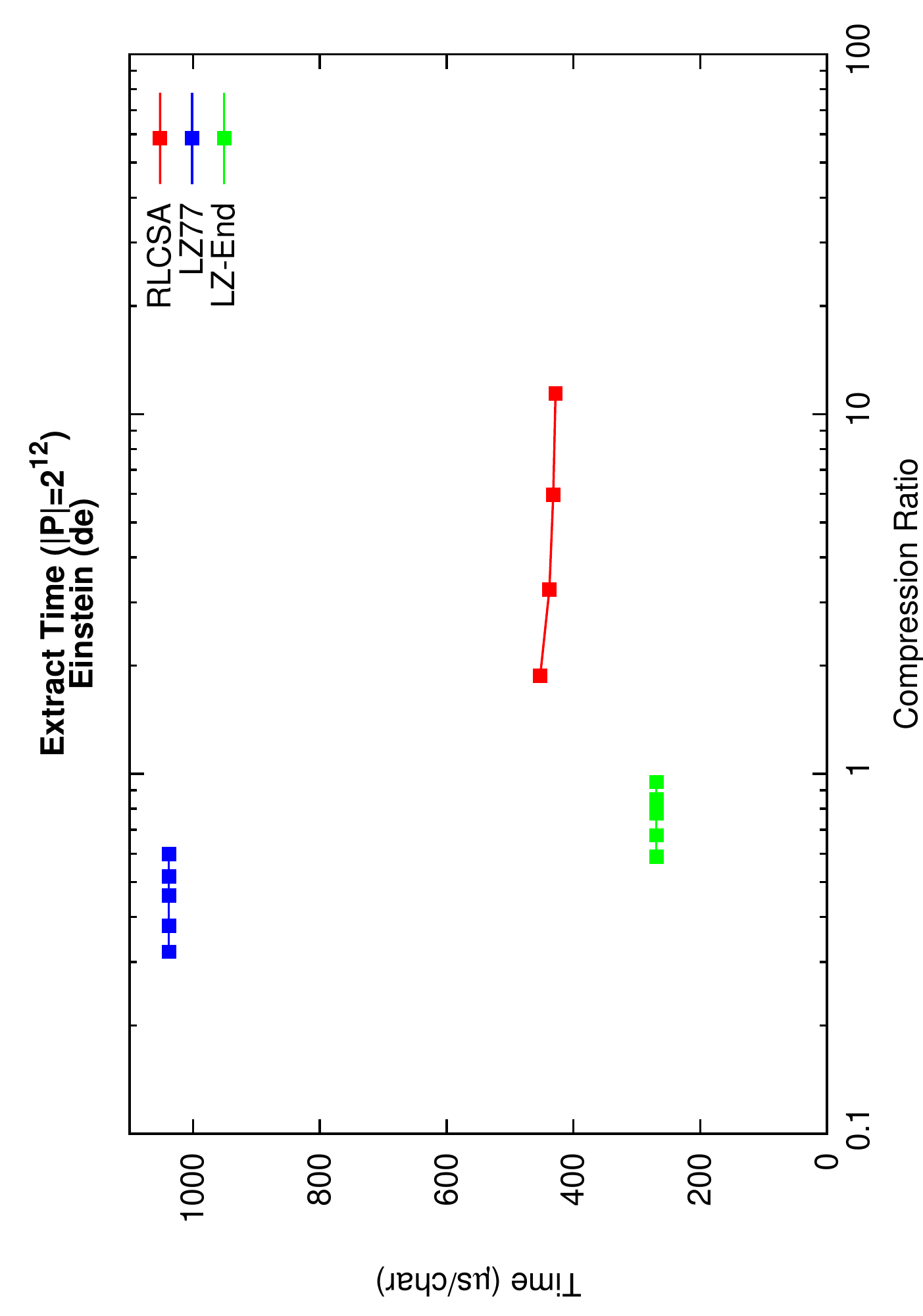}{Einstein (de) results}{res:wikide}
\results{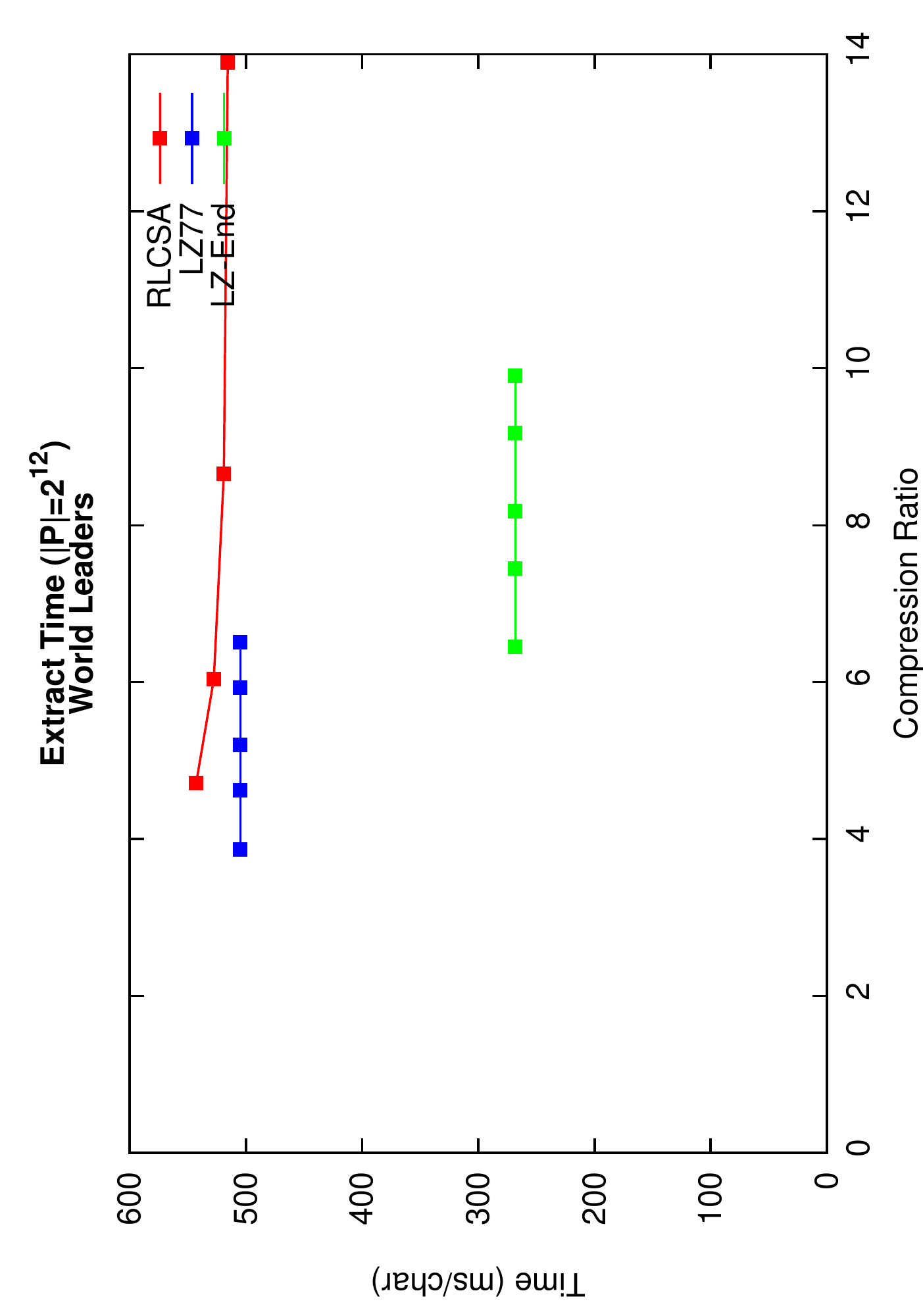}{World Leaders results}{res:leaders}

%% file: msc.bbl
\newcommand{\etalchar}[1]{$^{#1}$}
\begin{thebibliography}{ZdMNBY00}

\bibitem[ACNS10]{ACNSalenex10}
Diego Arroyuelo, Rodrigo C{\'a}novas, Gonzalo Navarro, and Kunihiko Sadakane.
\newblock Succinct trees in practice.
\newblock In {\em Proc. 11th Workshop on Algorithm Engineering and Experiments
  (ALENEX)}, pages 84--97. SIAM Press, 2010.

\bibitem[AN]{AN_LZ}
Diego Arroyuelo and Gonzalo Navarro.
\newblock Space-efficient construction of {L}empel-{Z}iv compressed text
  indexes.
\newblock Manuscript.

\bibitem[AN07]{AN07}
Diego Arroyuelo and Gonzalo Navarro.
\newblock Smaller and faster {L}empel-{Z}iv indices.
\newblock In {\em Proc. 18th International Workshop on Combinatorial Algorithms
  (IWOCA)}, pages 11--20. College Publications, UK, 2007.

\bibitem[AN10]{ANjea10}
Diego Arroyuelo and Gonzalo Navarro.
\newblock Practical approaches to reduce the space requirement of
  {L}empel-{Z}iv-based compressed text indices.
\newblock {\em ACM Journal of Experimental Algorithmics (ACM JEA)}, 2010.
\newblock To appear.

\bibitem[ANS06]{AN06}
Diego Arroyuelo, Gonzalo Navarro, and Kunihiko Sadakane.
\newblock Reducing the space requirement of {LZ}-index.
\newblock In {\em Proc. 17th Annual Symposium on Combinatorial Pattern Matching
  (CPM)}, LNCS 4009, pages 319--330, 2006.

\bibitem[ANS10]{ANSalgor10}
Diego Arroyuelo, Gonzalo Navarro, and Kunihiko Sadakane.
\newblock Stronger {L}empel-{Z}iv based compressed text indexing.
\newblock {\em Algorithmica}, 2010.
\newblock To appear.

\bibitem[AS99]{thuemorse}
Jean-Paul Allouche and Jeffrey Shallit.
\newblock The ubiquitous {P}rouhet-{T}hue-{M}orse sequence.
\newblock In {\em Proc. 1st International Conference on Sequences and their
  Applications (SETA)}, pages 1--16. Springer-Verlag, 1999.

\bibitem[B{\etalchar{+}}08]{dna_sim}
David~R. Bentley et~al.
\newblock Accurate whole human genome sequencing using reversible terminator
  chemistry.
\newblock {\em Nature}, 456(7218):53--59, 2008.

\bibitem[Ban09]{lzb}
Mohammad Banikazemi.
\newblock {LZB}: Data compression with bounded references.
\newblock In {\em Proc. 19th Data Compression Conference (DCC)}, page 436. IEEE
  Computer Society, 2009.
\newblock Poster.

\bibitem[BDM{\etalchar{+}}05]{labeled_child}
David Benoit, Erik~D. Demaine, J.~Ian Munro, Rajeev Raman, Venkatesh Raman, and
  S.~Srinivasa Rao.
\newblock Representing trees of higher degree.
\newblock {\em Algorithmica}, 43(4):275--292, 2005.

\bibitem[BLN09]{DAC}
Nieves Brisaboa, Susana Ladra, and Gonzalo Navarro.
\newblock Directly addressable variable-length codes.
\newblock In {\em Proc. 16th International Symposium on String Processing and
  Information Retrieval (SPIRE)}, LNCS 5721, pages 122--130. Springer, 2009.

\bibitem[BM77]{BM77}
Robert~S. Boyer and J.~Strother Moore.
\newblock A fast string searching algorithm.
\newblock {\em Communications of the ACM}, 20(10):762--772, 1977.

\bibitem[BW94]{BW94}
Michael Burrows and David Wheeler.
\newblock A block sorting lossless data compression algorithm.
\newblock Technical Report 124, Digital Equipment Corporation, 1994.

\bibitem[CFMPN10]{CFMPNbibe10}
Francisco Claude, Antonio Fari{\~n}a, Miguel Mart{\'{\i}}nez-Prieto, and
  Gonzalo Navarro.
\newblock Compressed $q$-gram indexing for highly repetitive biological
  sequences.
\newblock In {\em Proc. 10th IEEE Conference on Bioinformatics and
  Bioengineering (BIBE)}, pages 86--91. IEEE Press, 2010.

\bibitem[Cha88]{Cha88}
Bernard Chazelle.
\newblock Functional approach to data structures and its use in
  multidimensional searching.
\newblock {\em SIAM Journal on Computing}, 17(3):427--462, 1988.

\bibitem[CIT08]{CIT08}
Maxime Crochemore, Lucian Ilie, and Liviu Tinta.
\newblock Towards a solution to the "runs" conjecture.
\newblock In {\em Proc. 19th Annual Symposium on Combinatorial Pattern Matching
  (CPM)}, pages 290--302. Springer-Verlag, 2008.

\bibitem[Cla96]{Cla96}
David Clark.
\newblock {\em Compact Pat Trees}.
\newblock PhD thesis, University of Waterloo, 1996.

\bibitem[CLL{\etalchar{+}}05]{CLLP+05}
Moses Charikar, Eric Lehman, Ding Liu, Rina Panigrahy, Manoj Prabhakaran, Amit
  Sahai, and Abhi Shelat.
\newblock The smallest grammar problem.
\newblock {\em IEEE Transactions on Information Theory}, 51(7):2554--2576,
  2005.

\bibitem[CN09]{CN09}
Francisco Claude and Gonzalo Navarro.
\newblock Self-indexed text compression using straight-line programs.
\newblock In {\em Proc. 34th International Symposium on Mathematical
  Foundations of Computer Science (MFCS)}, LNCS 5734, pages 235--246. Springer,
  2009.

\bibitem[CN10]{CN10}
Francisco Claude and Gonzalo Navarro.
\newblock Self-indexed grammar-based compression.
\newblock {\em Fundamenta Informaticae}, 2010.
\newblock to appear.

\bibitem[CPS08]{CPS08}
Gang Chen, Simon~J. Puglisi, and William~F. Smyth.
\newblock {L}empel-{Z}iv factorization using less time \& space.
\newblock {\em Mathematics in Computer Science}, 1(4):605--623, June 2008.

\bibitem[FG89]{FG89}
Edward~R. Fiala and Daniel~H. Greene.
\newblock Data compression with finite windows.
\newblock {\em Communications of the ACM}, 32(4):490--505, 1989.

\bibitem[FH07]{FH07}
Johannes Fischer and Volker Heun.
\newblock {A New Succinct Representation of RMQ-Information and Improvements in
  the Enhanced Suffix Array}.
\newblock In {\em Proc. 1st International Symposium on Combinatorics,
  Algorithms, Probabilistic and Experimental Methodologies (ESCAPE)}, volume
  4614 of {\em LNCS 4614}, pages 459--470. Springer-Verlag, 2007.

\bibitem[FM05]{FM05}
Paolo Ferragina and Giovanni Manzini.
\newblock Indexing compressed text.
\newblock {\em Journal of the ACM}, 52(4):552--581, 2005.

\bibitem[FMMN07]{FMMN07}
Paolo Ferragina, Giovanni Manzini, Veli M{\"a}kinen, and Gonzalo Navarro.
\newblock Compressed representations of sequences and full-text indexes.
\newblock {\em ACM Transactions on Algorithms (TALG)}, 3(2):article 20, 2007.

\bibitem[FSS03]{rich}
Frantisek Franek, R.J. Simpson, and William~F. Smyth.
\newblock The maximum number of runs in a string.
\newblock In {\em Proc. Australian Workshop on Combinatorial Algorithms
  (AWOCA)}, pages 26--35, 2003.

\bibitem[GBYS92]{sa_gonet}
Gaston~H. Gonnet, Ricardo~A. Baeza-Yates, and Tim Snider.
\newblock New indices for text: Pat trees and pat arrays.
\newblock In {\em Information Retrieval: Data Structures {\&} Algorithms},
  pages 66--82. Prentice Hall, 1992.

\bibitem[GGMN05]{GGMNwea05}
Rodrigo Gonz\'alez, Szymon Grabowski, Veli M{\"akinen}, and Gonzalo Navarro.
\newblock Practical implementation of rank and select queries.
\newblock In {\em Poster Proc. Volume of 4th Workshop on Efficient and
  Experimental Algorithms (WEA)}, pages 27--38. CTI Press and Ellinika
  Grammata, 2005.

\bibitem[GGV03]{GGV03}
Roberto Grossi, Ankur Gupta, and Jeffrey~Scott Vitter.
\newblock High-order entropy-compressed text indexes.
\newblock In {\em Proc. 14th {A}nnual ACM-SIAM {S}ymposium on {D}iscrete
  {A}lgorithms (SODA)}, pages 841--850. SIAM Press, 2003.

\bibitem[GN08]{GN08}
Rodrigo Gonz{\'a}lez and Gonzalo Navarro.
\newblock Rank/select on dynamic compressed sequences and applications.
\newblock {\em Theoretical Computer Science}, 410:4414--4422, 2008.

\bibitem[GV05]{GV00}
Roberto Grossi and Jeffrey~Scott Vitter.
\newblock Compressed suffix arrays and suffix trees with applications to text
  indexing and string matching.
\newblock {\em SIAM Journal of Computing}, 35(2):378--407, 2005.

\bibitem[Ham86]{gibbs}
Richard~Wesley Hamming.
\newblock {\em Coding and Information Theory}.
\newblock Prentice-Hall, 1986.

\bibitem[IT06]{sst_constr}
Shunsuke Inenaga and Masayuki Takeda.
\newblock On-line linear-time construction of word suffix trees.
\newblock In {\em Proc. 17th Annual Symposium on Combinatorial Pattern Matching
  (CPM)}, pages 60--71. Springer-Verlag, 2006.

\bibitem[Jac89]{bp_jacobson}
Guy Jacobson.
\newblock Space-efficient static trees and graphs.
\newblock In {\em Annual IEEE Symposium on Foundations of Computer Science},
  pages 549--554. IEEE Computer Society, 1989.

\bibitem[K{\"a}r99]{Kar99}
Juha K{\"a}rkk{\"a}inen.
\newblock {\em Repetition-Based Text Indexes}.
\newblock PhD thesis, Department of Computer Science, Univeristy of Helsinki,
  Finland, November 1999.

\bibitem[KK99]{runslz}
Roman Kolpakov and Gregory Kucherov.
\newblock On maximal repetitions in words.
\newblock {\em Journal of Discrete Algorithms}, 1:159--186, 1999.

\bibitem[KM99]{KM99}
S.~Rao Kosaraju and Giovanni Manzini.
\newblock Compression of low entropy strings with {L}empel-{Z}iv algorithms.
\newblock {\em SIAM Journal on Computing}, 29(3):893--911, 1999.

\bibitem[KMP77]{KMP77}
Donald~E. Knuth, James~H. Morris, and Vaughan~R. Pratt.
\newblock Fast pattern matching in strings.
\newblock {\em SIAM Journal of Computing}, 6(2):323--350, 1977.

\bibitem[KN10]{KN10}
Sebastian Kreft and Gonzalo Navarro.
\newblock L{Z}77-like compression with fast random access.
\newblock In {\em Proc. 20th Data Compression Conference (DCC)}, pages
  239--248, 2010.

\bibitem[KPZ10]{KPZ10}
Shanika Kuruppu, Simon~J. Puglisi, and Justin Zobel.
\newblock Relative {L}empel-{Z}iv compression of genomes for large-scale
  storage and retrieval.
\newblock In {\em Proc. 17th International Symposium on String Processing and
  Information Retrieval (SPIRE)}, pages 201--206, 2010.

\bibitem[KS03]{KS03}
Juha K\"arkk\"ainen and Peter Sanders.
\newblock Simple linear work suffix array construction.
\newblock In {\em Proc. 30th International Colloquium on Automata, Languages
  and Programming (ICALP)}, LNCS 2719, pages 943--955, 2003.

\bibitem[KU96a]{KU96}
Juha K{\"a}rkk{\"a}inen and Esko Ukkonen.
\newblock Lempel-{Z}iv parsing and sublinear-size index structures for string
  matching.
\newblock In {\em Proc. 3rd South American Workshop on String Processing
  (WSP)}, pages 141--155. Carleton University Press, 1996.

\bibitem[KU96b]{KU96_sst}
Juha K\"{a}rkk\"{a}inen and Esko Ukkonen.
\newblock Sparse suffix trees.
\newblock In {\em Proc. 2nd Annual International Conference on Computing and
  Combinatorics (COCOON)}, pages 219--230. Springer-Verlag, 1996.

\bibitem[LM00]{LM00}
N.~Jesper Larsson and Alistair Moffat.
\newblock Off-line dictionary-based compression.
\newblock {\em Proc. IEEE}, 88(11):1722--1732, 2000.

\bibitem[Lot02]{ACW}
M.~Lothaire.
\newblock {\em Algebraic Combinatorics on Words}.
\newblock Cambridge University Press, 2002.

\bibitem[LZ76]{LZ76}
Abraham Lempel and Jacob Ziv.
\newblock On the complexity of finite sequences.
\newblock {\em IEEE Transactions on Information Theory}, 22(1):75--81, 1976.

\bibitem[Mai89]{Main89}
Michael~G. Main.
\newblock Detecting leftmost maximal periodicities.
\newblock {\em Discrete Applied Mathematics}, 25(1-2):145--153, 1989.

\bibitem[Man01]{Man2001}
Giovanni Manzini.
\newblock An analysis of the {B}urrows-{W}heeler transform.
\newblock {\em Journal of the ACM}, 48(3):407--430, 2001.

\bibitem[McC76]{McC76}
Edward~M. McCreight.
\newblock A space-economical suffix tree construction algorithm.
\newblock {\em Journal of the ACM}, 32(2):262--272, 1976.

\bibitem[MKI{\etalchar{+}}08]{MKIBS08}
Wataru Matsubara, Kazuhiko Kusano, Akira Ishino, Hideo Bannai, and Ayumi
  Shinohara.
\newblock New lower bounds for the maximum number of runs in a string.
\newblock In {\em Proc. Prague Stringology Conference (PSC)}, pages 140--145,
  2008.

\bibitem[MM93]{MM93}
Udi Manber and Gene Myers.
\newblock Suffix arrays: a new method for on-line string searches.
\newblock {\em SIAM Journal on Computing}, 22(5):935--948, 1993.

\bibitem[MN07]{MN07}
Veli M{\"a}kinen and Gonzalo Navarro.
\newblock Rank and select revisited and extended.
\newblock {\em Theoretical Computer Science}, 387(3):332--347, 2007.
\newblock Special issue on ``The Burrows-Wheeler Transform and its
  Applications''.

\bibitem[MNSV10]{RLCSA_journal}
Veli M{\"a}kinen, Gonzalo Navarro, Jouni Sir{\'e}n, and Niko V{\"a}lim{\"a}ki.
\newblock Storage and retrieval of highly repetitive sequence collections.
\newblock {\em Journal of Computational Biology}, 17(3):281--308, 2010.

\bibitem[Mor68]{patricia}
Donald~R. Morrison.
\newblock Patricia-practical algorithm to retrieve information coded in
  alphanumeric.
\newblock {\em Journal of the ACM}, 15(4):514--534, 1968.

\bibitem[MR01]{bp_munroraman}
J.~Ian Munro and Venkatesh Raman.
\newblock Succinct representation of balanced parentheses and static trees.
\newblock {\em SIAM Journal on Computing}, 31(3):762--776, 2001.

\bibitem[MRRR03]{MRRR03}
J.~Ian Munro, Rajeev Raman, Venkatesh Raman, and S.~Srinivasa Rao.
\newblock Succinct representations of permutations.
\newblock In {\em Proc. 30th International Colloquium on Automata, Languages
  and Computation (ICALP)}, LNCS 2719, pages 345--356. Springer, 2003.

\bibitem[Mun86]{Munro86}
J.~Ian Munro.
\newblock An implicit data structure supporting insertion, deletion, and search
  in {$O(\log n)$} time.
\newblock {\em Journal of Computer System Sciences}, 33(1):66--74, 1986.

\bibitem[Nav04]{Nav02}
Gonzalo Navarro.
\newblock Indexing text using the {Z}iv-{L}empel trie.
\newblock {\em Journal of Discrete Algorithms}, 2(1):87--114, 2004.

\bibitem[Nav08]{Nav08}
Gonzalo Navarro.
\newblock Indexing {L}{Z}77: {T}he next step in self-indexing.
\newblock {\em Keynote talk at Third {W}orkshop on {C}ompression, {Text}, and
  {A}lgorithms}, 2008.

\bibitem[Nav09]{Nav09}
Gonzalo Navarro.
\newblock Implementing the {LZ}-index: Theory versus practice.
\newblock {\em ACM Journal of Experimental Algorithmics (JEA)}, 13:article 2,
  2009.

\bibitem[NM07]{NM07}
Gonzalo Navarro and Veli M{\"a}kinen.
\newblock Compressed full-text indexes.
\newblock {\em ACM Computing Surveys}, 39(1):article 2, 2007.

\bibitem[OS07]{OS07}
Daisuke Okanohara and Kunihiko Sadakane.
\newblock Practical entropy-compressed rank/select dictionary.
\newblock In {\em Proc. 9th Workshop on Algorithm Engineering and Experiments
  (ALENEX)}. SIAM Press, 2007.

\bibitem[OS08]{OS08}
Daisuke Okanohara and Kunihiko Sadakane.
\newblock An online algorithm for finding the longest previous factors.
\newblock In {\em Proc. 16th {A}nnual {E}uropean {S}ymposium on {A}lgorithms
  (ESA)}, pages 696--707. Springer-Verlag, 2008.

\bibitem[PST07]{sa_construction}
Simon~J. Puglisi, William~F. Smyth, and Andrew~H. Turpin.
\newblock A taxonomy of suffix array construction algorithms.
\newblock {\em ACM Computing Surveys}, 39(2):4, 2007.

\bibitem[PWZ92]{PWZ92}
Eli Plotnik, Marcelo Weinberger, and Jacob Ziv.
\newblock Upper bounds on the probability of sequences emitted by finite-state
  sources and on the redundancy of the {L}empel-{Z}iv algorithm.
\newblock {\em IEEE Transactions on Information Theory}, 38(1):66--72, 1992.

\bibitem[RNO08]{RNO08}
Lu\'{\i}s M.~S. Russo, Gonzalo Navarro, and Arlindo~L. Oliveira.
\newblock Fully-compressed suffix trees.
\newblock In {\em {Proc. 8th Latin American Symposium on Theoretical
  Informatics (LATIN)}}, LNCS 4957, pages 362--373, 2008.

\bibitem[RO08]{ilzi}
Lu\'{\i}s M.~S. Russo and Arlindo~L. Oliveira.
\newblock A compressed self-index using a {Z}iv-{L}empel dictionary.
\newblock {\em Journal of Information Retrieval}, 5(3):501--513, 2008.
\newblock Special issue SPIRE 2006.

\bibitem[RRR02]{RRR02}
Rajeev Raman, Venkatesh Raman, and S.~Srinivasa Rao.
\newblock Succinct indexable dictionaries with applications to encoding $k$-ary
  trees and multisets.
\newblock In {\em Proc. 13th Annual ACM-SIAM Symposium on Discrete Algorithms
  (SODA)}, pages 233--242. SIAM Press, 2002.

\bibitem[Ryt03]{Rytter03}
Wojciech Rytter.
\newblock Application of lempel--ziv factorization to the approximation of
  grammar-based compression.
\newblock {\em Theoretical Computer Science}, 302(1-3):211--222, 2003.

\bibitem[Sad03]{Sad03}
Kunihiko Sadakane.
\newblock New text indexing functionalities of the compressed suffix arrays.
\newblock {\em Journal of Algorithms}, 48(2):294 -- 313, 2003.

\bibitem[SS82]{LZSS}
James~A. Storer and Thomas~G. Szymanski.
\newblock Data compression via textual substitution.
\newblock {\em Journal of the ACM}, 29(4):928--951, 1982.

\bibitem[SVMN08]{MNSV08}
Jouni Sir{\'e}n, Niko V{\"a}lim{\"a}ki, Veli M{\"a}kinen, and Gonzalo Navarro.
\newblock Run-length compressed indexes are superior for highly repetitive
  sequence collections.
\newblock In {\em Proc. 15th International Symposium on String Processing and
  Information Retrieval (SPIRE)}, LNCS 5280, pages 164--175. Springer, 2008.

\bibitem[Ukk95]{Ukk95}
Esko Ukkonen.
\newblock Constructing suffix trees on-line in linear time.
\newblock {\em Algorithmica}, 14(3):249--260, 1995.

\bibitem[Wei73]{Wei73}
Peter Weiner.
\newblock Linear pattern matching algorithms.
\newblock In {\em Proc. 14th Annual Symposium on Switching and Automata
  Theory}, pages 1--11, 1973.

\bibitem[Wel84]{LZW}
Terry~A. Welch.
\newblock {A Technique for High-Performance Data Compression}.
\newblock {\em Computer}, 17(6):8--19, 1984.

\bibitem[Wil91]{LZRW}
Ross~N. Williams.
\newblock An extremely fast ziv-lempel data compression algorithm.
\newblock In {\em Data Compression Conference}, pages 362--371, 1991.

\bibitem[WZ99]{vbyte}
Hugh~E. Williams and Justin Zobel.
\newblock Compressing integers for fast file access.
\newblock {\em Computer Journal}, 42(3):193--201, 1999.

\bibitem[ZdMNBY00]{ZMNBY00}
Nivio Ziviani, Edleno~Silva de~Moura, Gonzalo Navarro, and Ricardo Baeza-Yates.
\newblock Compression: A key for next-generation text retrieval systems.
\newblock {\em IEEE Computer}, 33(11):37--44, 2000.

\bibitem[ZL77]{ZL77}
Jacob Ziv and Abraham Lempel.
\newblock A universal algorithm for sequential data compression.
\newblock {\em IEEE Transactions on Information Theory}, 23(3):337--343, 1977.

\bibitem[ZL78]{ZL78}
Jacob Ziv and Abraham Lempel.
\newblock Compression of individual sequences via variable-rate coding.
\newblock {\em IEEE Transactions on Information Theory}, 24(5):530--536, 1978.

\end{thebibliography}
